\documentclass[12pt]{article}
\usepackage{fullpage}
\usepackage{amssymb}
\usepackage{times}
\usepackage{amsmath,amsthm}
\usepackage{graphicx}
\usepackage{xcolor}
\usepackage{mathtools}
\usepackage{ifthen}    
\usepackage{hyperref}
\usepackage{titlesec}

\newboolean{comment}
\setboolean{comment}{false}

\titlespacing\subsection{0pt}{6pt plus 4pt minus 2pt}{6pt plus 2pt minus 2pt}
\titlespacing\section{0pt}{6pt plus 4pt minus 2pt}{6pt plus 2pt minus 2pt}
\titlespacing\subsubsection{0pt}{6pt plus 4pt minus 2pt}{6pt plus 2pt minus 2pt}

\titleformat*{\section}{\large\bfseries} 
 \titleformat*{\subsection}{\bfseries}

\newcommand{\dnote}[1]{\ifthenelse{\boolean{comment}}
{\textcolor{purple}{ {\textbf{(Dorit: }#1\textbf{) }}}}
{}}
\newcommand{\snote}[1]{ \ifthenelse{\boolean{comment}}
{\textcolor{red}{ {\textbf{(Sandy: }#1\textbf{) } } } }
{}}

\def\IF{{\bf if}\ }

\def\ELSE{{\bf else}}
\def\WHILE{{\bf while}\ }
\def\FOR{{\bf for}\ }

\def\RETURN{{\bf return}\ }
\def\blo{\noindent
\begin{tabular}{@{\quad}l@{\quad}}
\begin{minipage}{1in}
\begin{tabbing}
\qquad\=\qquad\=\qquad\=\qquad\=\qquad\=\qquad\=\qquad\=\kill}
\def\elo{\end{tabbing}\end{minipage}\\\end{tabular}}

\long\def\comment#1{}

\newcommand{\fpqmaexp}{\mbox{FP}^{\mbox{QMA-EXP}}}
\newcommand{\fpnp}{\mbox{FP}^{\mbox{NP}}}

\newcommand{\fexpqmaexp}{\mbox{FEXP}^{\mbox{QMA-EXP}}}
\newcommand{\fexpnexp}{\mbox{FEXP}^{\mbox{NEXP}}}
\newcommand{\fpnexp}{\mbox{FP}^{\mbox{NEXP}}}
\newcommand{\qmaexp}{{\mbox{QMA-EXP}}}

\newcommand{\nexp}{{\mbox{NEXP}}}
\newcommand{\qma}{{\mbox{QMA}}}
\newcommand{\fexp}{{\mbox{FEXP}}}
\newcommand{\fp}{{\mbox{FP}}}
\newcommand{\calh}{{\cal{H}}}
\newcommand{\call}{{\cal{L}}}
\newcommand{\calk}{{\cal{K}}}

\newcommand{\cals}{{\cal{S}}}
\newcommand{\calt}{{\cal{T}}}
\newcommand{\caln}{{\cal{N}}}
\newcommand{\ket}[1]{|#1\rangle}
\newcommand{\bra}[1]{\langle#1|}

\newcommand{\ketbra}[2]{|#1\rangle\langle#2|}

\newcommand{\ketbrabig}[2]{\left|~#1~ \right\rangle \left\langle~#2~ \right|}
\newcommand{\myw}{\omega}

\newcommand{\vinit}{v\mbox{-init}}
\newcommand{\norm}[1]{\left\lVert #1 \right\rVert}

\newtheorem{theorem}{Theorem}[section]

\newtheorem{claim}[theorem]{Claim}
\newtheorem{lemma}[theorem]{Lemma}

\newtheorem{coro}[theorem]{Corollary}
\newtheorem{fact}[theorem]{Fact} 
\newtheorem{definition}[theorem]{Definition}

\newcommand{\twocellsvert}[2]{\begin{array}{|@{}c@{}|} \hline  #1 \\ \hline  
#2 \\ \hline \end{array}}

\newcommand{\threecellsL}[2]{\begin{array}{|@{}c@{}|@{}c@{}|} \hline  \leftBr &
	\begin{array}{@{}c@{}} #1 \\ \hline  #2 \\ \end{array}\\ \hline  \end{array}}
\newcommand{\threecellsR}[2]{\begin{array}{|@{}c@{}|@{}c@{}|} \hline
	\begin{array}{@{}c@{}} #1 \\ \hline  #2  \end{array} & \rightBr \\ \hline  \end{array}}
\newcommand{\fourcells}[4]{\begin{array}{|@{}c@{}|@{}c@{}|} \hline  #1 & #3 \\ \hline #2 & #4 \\ \hline   \end{array}}

\newcommand{\threecellsvert}[3]{\begin{array}{|@{}c@{}|} \hline  #1 \\ \hline  
#2 \\ \hline 
#3 \\ \hline
\end{array}}

\newcommand{\sixcells}[6]{\begin{array}{|@{}c@{}|@{}c@{}|} \hline  #1 & #2 \\ \hline  
#3 & #4 \\ \hline 
#5 & #6 \\ \hline
\end{array}}

\newcommand{\fourcellsR}[3]{\begin{array}{|@{}c@{}|@{}c@{}|} \hline
	\begin{array}{@{}c@{}} #1 \\ \hline  #2  \\ \hline  #3
	\end{array} & \rightBr \\ \hline  \end{array}}
\newcommand{\fourcellsL}[3]{\begin{array}{|@{}c@{}|@{}c@{}|} \hline  \leftBr &
	\begin{array}{@{}c@{}} #1 \\ \hline  #2 \\ \hline  #3 \\
	\end{array}\\ \hline  \end{array}}

\newcommand{\fivecellsL}[4]{\begin{array}{|@{}c@{}|@{}c@{}|} \hline  \leftBr &
	\begin{array}{@{}c@{}} #1 \\ \hline  #2 \\ \hline  #3 \\ \hline  #4 \\
	\end{array}\\ \hline  \end{array}}

\newcommand{\eightcells}[8]{\begin{array}{|@{}c@{}| @{}c@{}| @{}c@{}|@{}c@{}|@{}c@{}|@{}c@{}|} \hline
\leftBr &
	\begin{array}{@{}c@{}} #1 \\ \hline  #5   \end{array} &
	\begin{array}{@{}c@{}} #2 \\ \hline  #6   \end{array} &
	\begin{array}{@{}c@{}} #3 \\ \hline  #7   \end{array} &
	\begin{array}{@{}c@{}} #4 \\ \hline  #8   \end{array} &
	 \rightBr \\ \hline  \end{array}}

\newcommand\generic{\vcenter{\hbox{\includegraphics[height = 2.2ex]{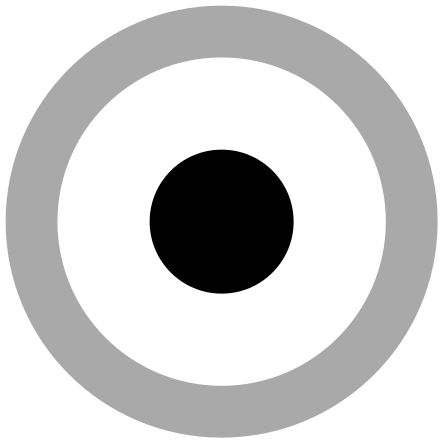}}}}
\newcommand\dead{\vcenter{\hbox{\includegraphics[height = 2.2ex]{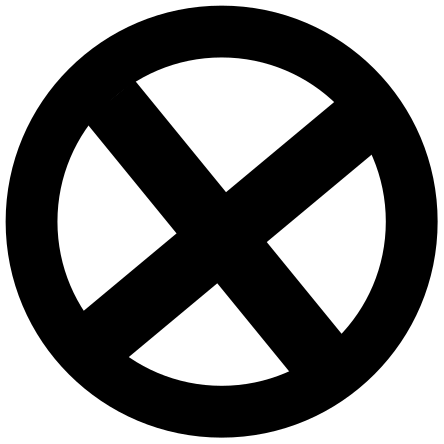}}}}
\newcommand\blank{\vcenter{\hbox{\includegraphics[height = 2.2ex]{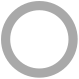}}}}
\newcommand\Dblank{\vcenter{\hbox{\includegraphics[height = 2.2ex]{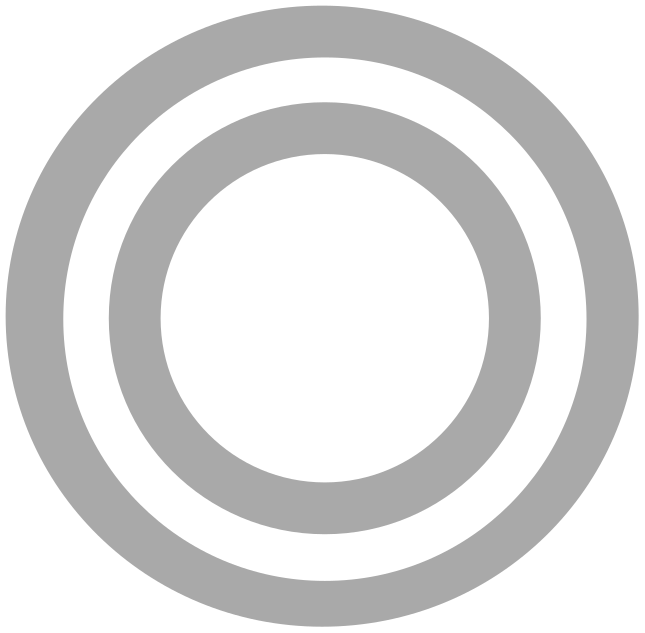}}}}
\newcommand\blankLtwo{\vcenter{\hbox{\includegraphics[height = 2.2ex]{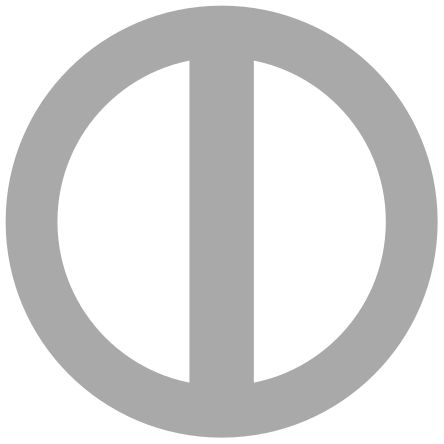}}}}
\newcommand\DblankLtwo{\vcenter{\hbox{\includegraphics[height = 2.2ex]{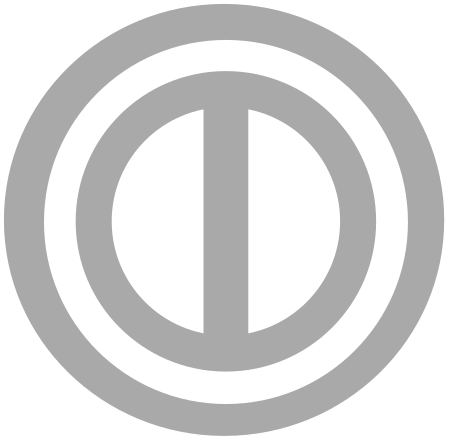}}}}
\newcommand\blankLthree{\vcenter{\hbox{\includegraphics[height = 2.2ex]{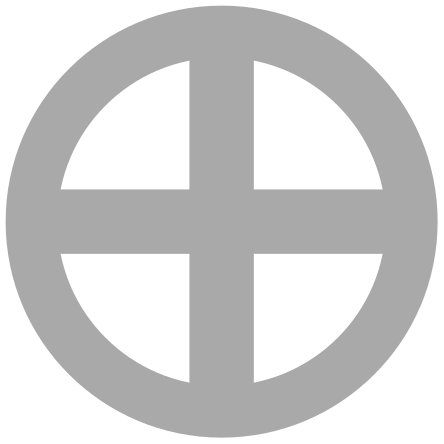}}}}
\newcommand\DblankLthree{\vcenter{\hbox{\includegraphics[height = 2.2ex]{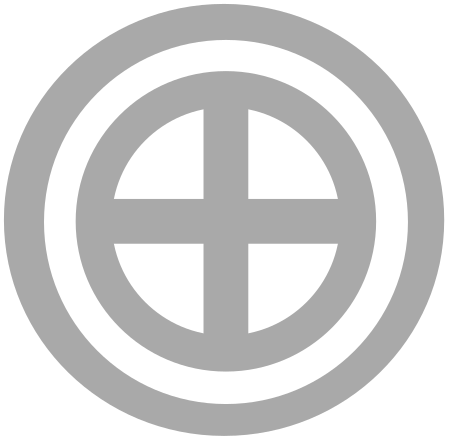}}}}

\newcommand\leftcorner{\vcenter{\hbox{\includegraphics[height = 2.2ex]{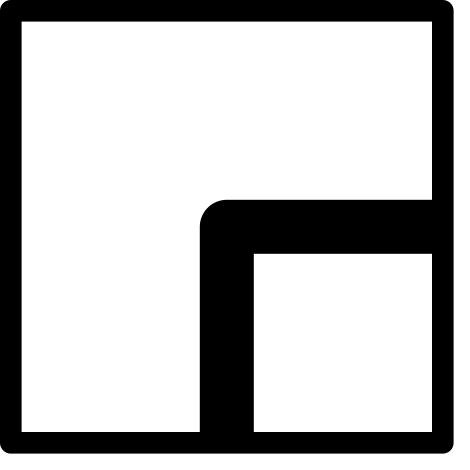}}}}
\newcommand\rightcorner{\vcenter{\hbox{\includegraphics[height = 2.2ex]{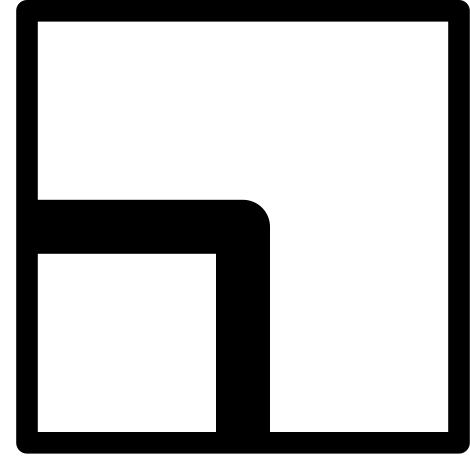}}}}
\newcommand\toprow{\vcenter{\hbox{\includegraphics[height = 2.2ex]{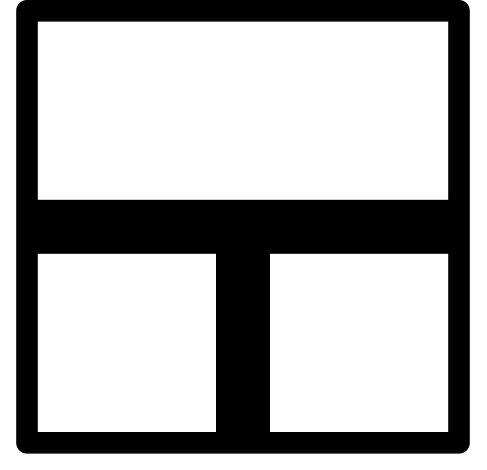}}}}

\newcommand\leftBr{\vcenter{\hbox{\includegraphics[height = 2.2ex]{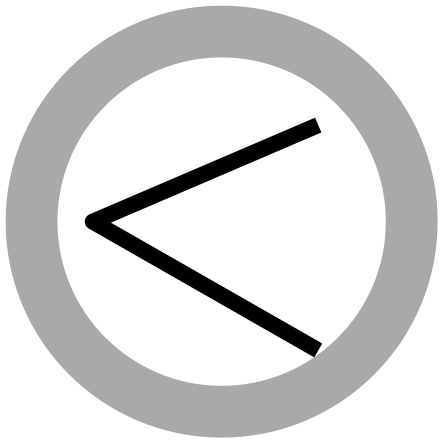}}}}
\newcommand\rightBr{\vcenter{\hbox{\includegraphics[height = 2.2ex]{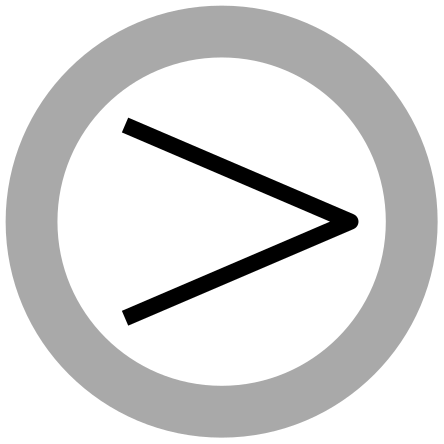}}}}

\newcommand\arrR{\vcenter{\hbox{\includegraphics[height = 2.2ex]{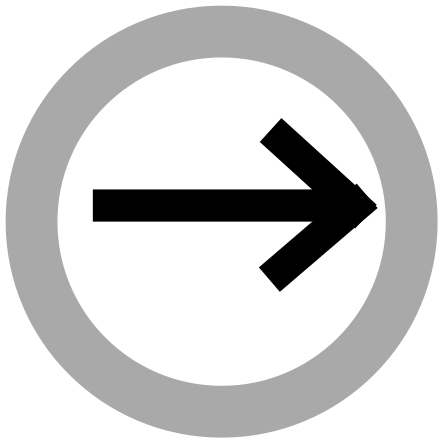}}}}

\newcommand\arrRone{\vcenter{\hbox{\includegraphics[height = 2.2ex]{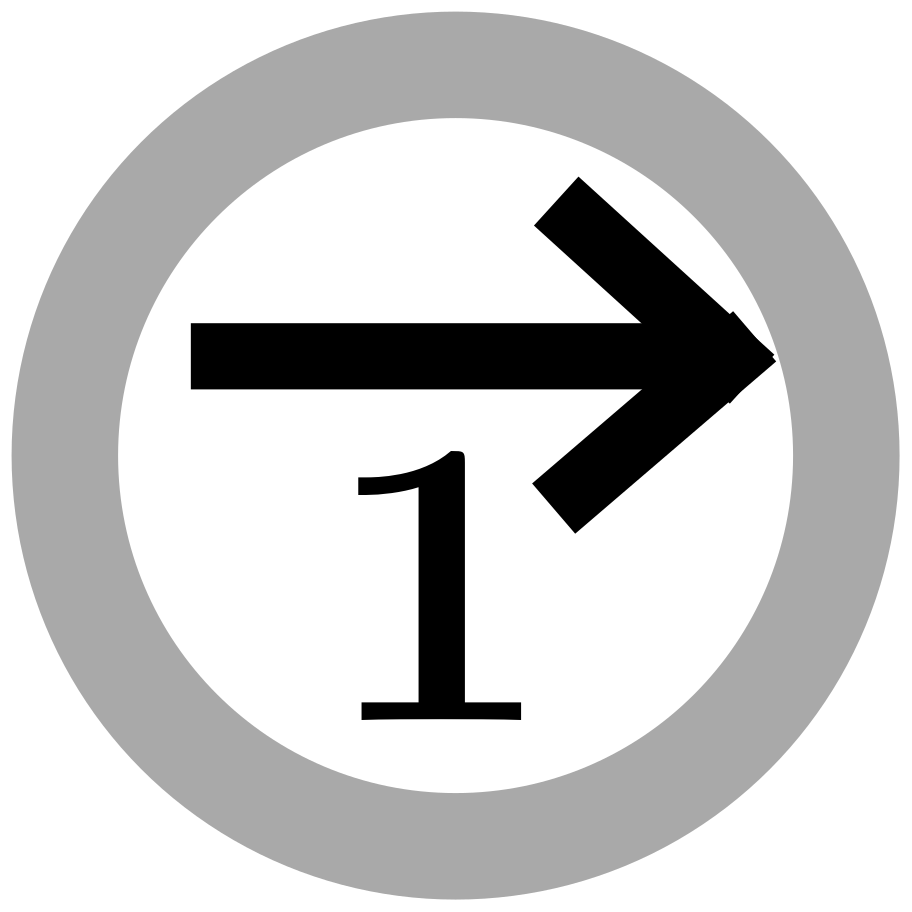}}}}

\newcommand\arrL{\vcenter{\hbox{\includegraphics[height = 2.2ex]{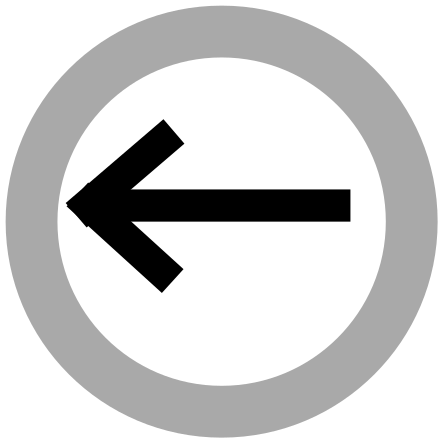}}}}
\newcommand\arrLtwo{\vcenter{\hbox{\includegraphics[height = 2.2ex]{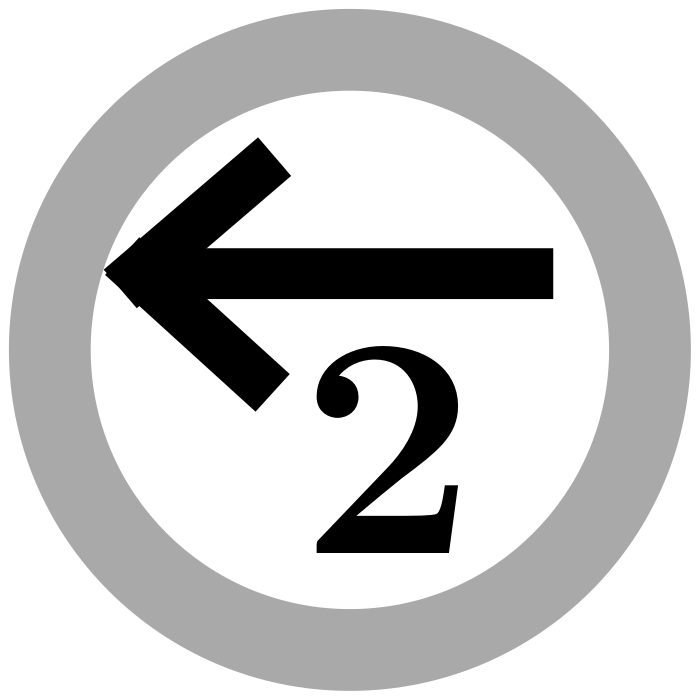}}}}
\newcommand\arrLfour{\vcenter{\hbox{\includegraphics[height = 2.2ex]{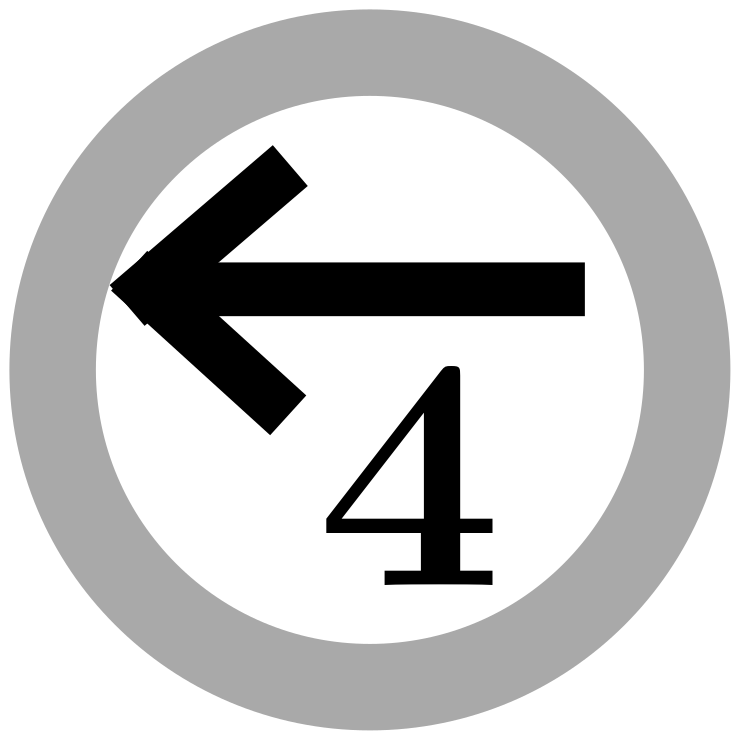}}}}
\newcommand\arrLfive{\vcenter{\hbox{\includegraphics[height = 2.2ex]{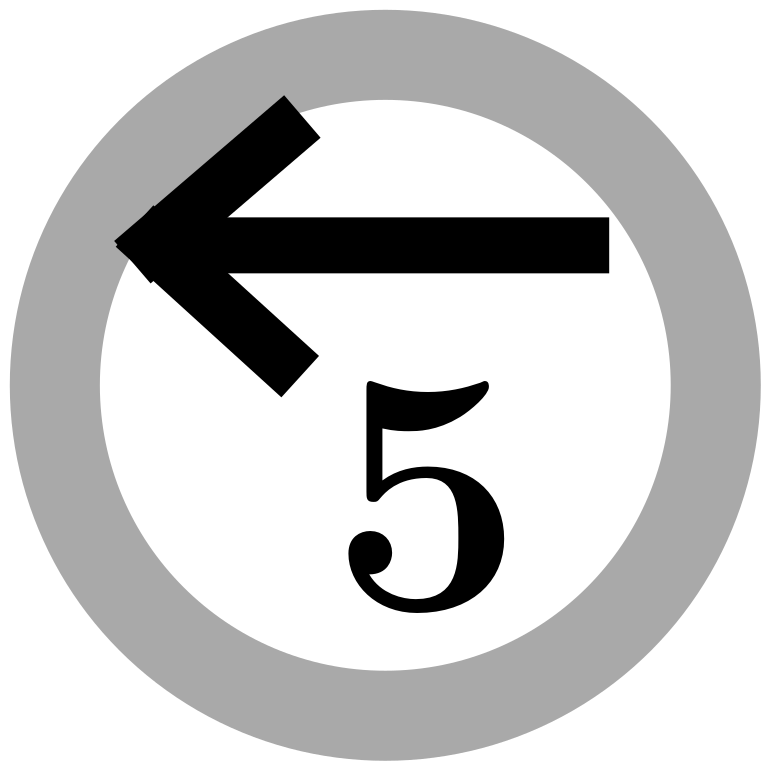}}}}
\newcommand\arrRi{\vcenter{\hbox{\includegraphics[height = 2.2ex]{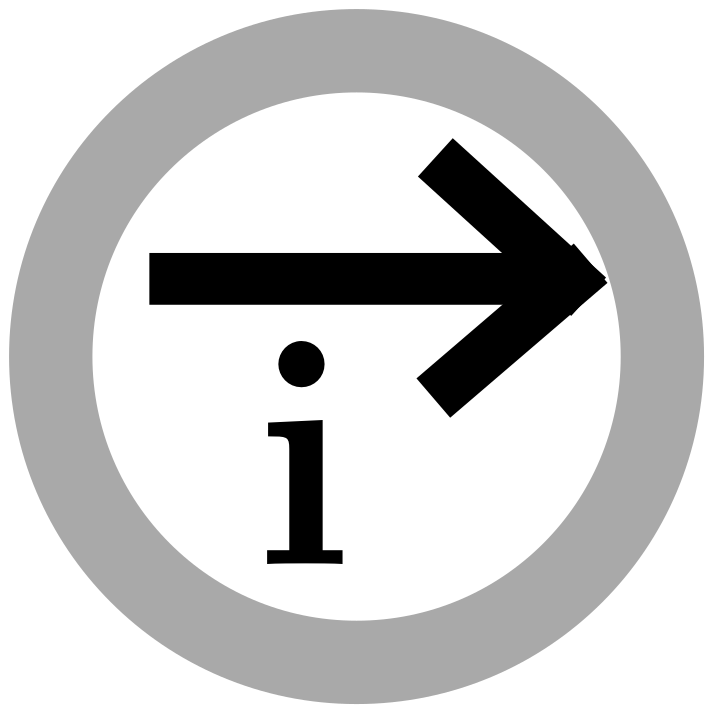}}}}
\newcommand\arrLi{\vcenter{\hbox{\includegraphics[height = 2.2ex]{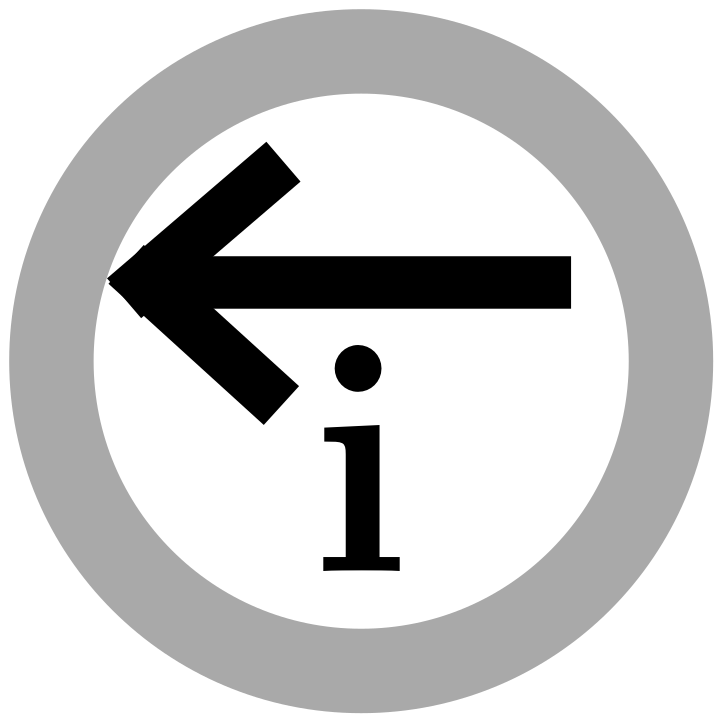}}}}

\newcommand\arrLone{\vcenter{\hbox{\includegraphics[height = 2.2ex]{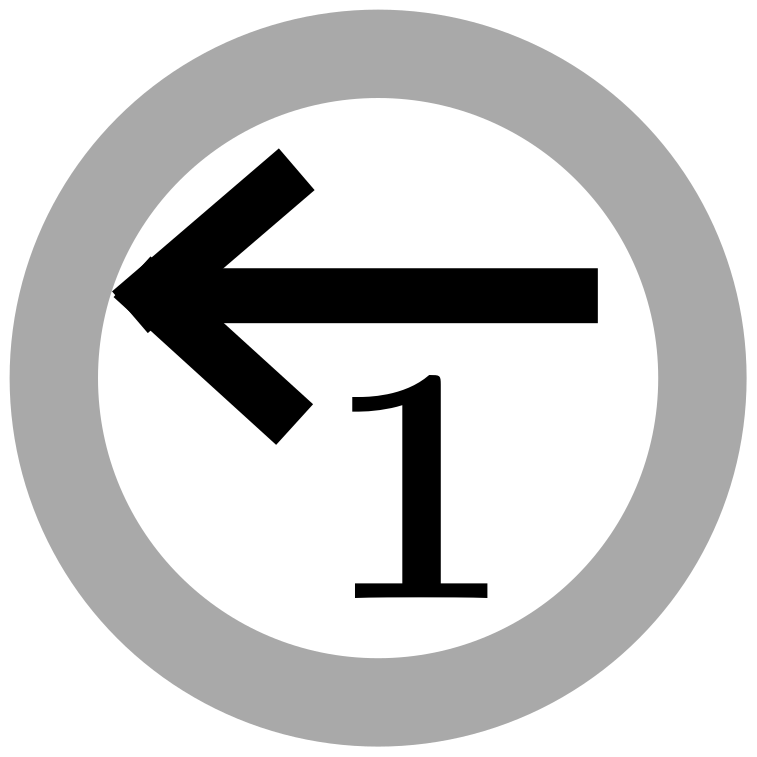}}}}

\newcommand\arrRtwo{\vcenter{\hbox{\includegraphics[height = 2.2ex]{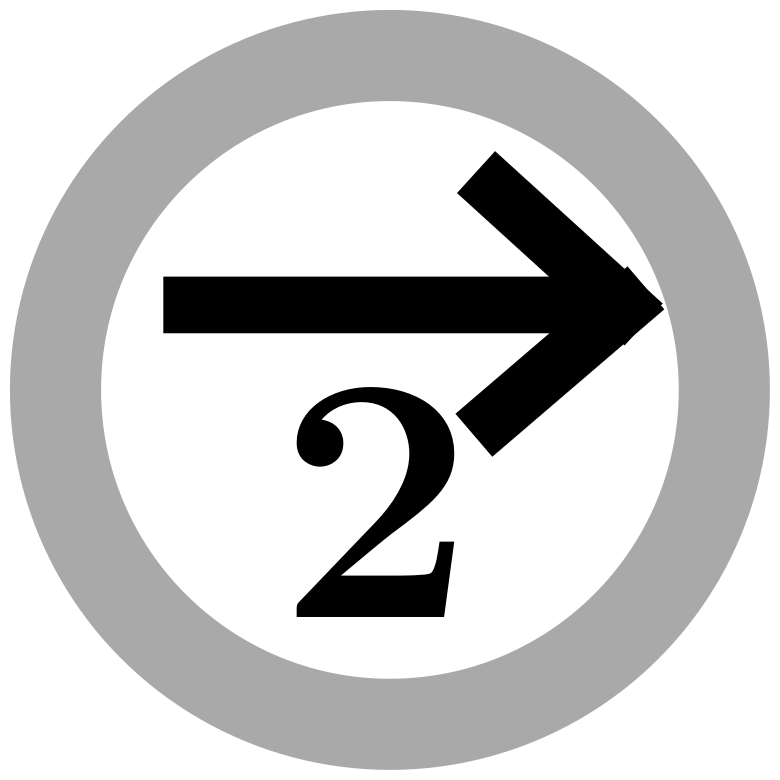}}}}
\newcommand\arrLthree{\vcenter{\hbox{\includegraphics[height = 2.2ex]{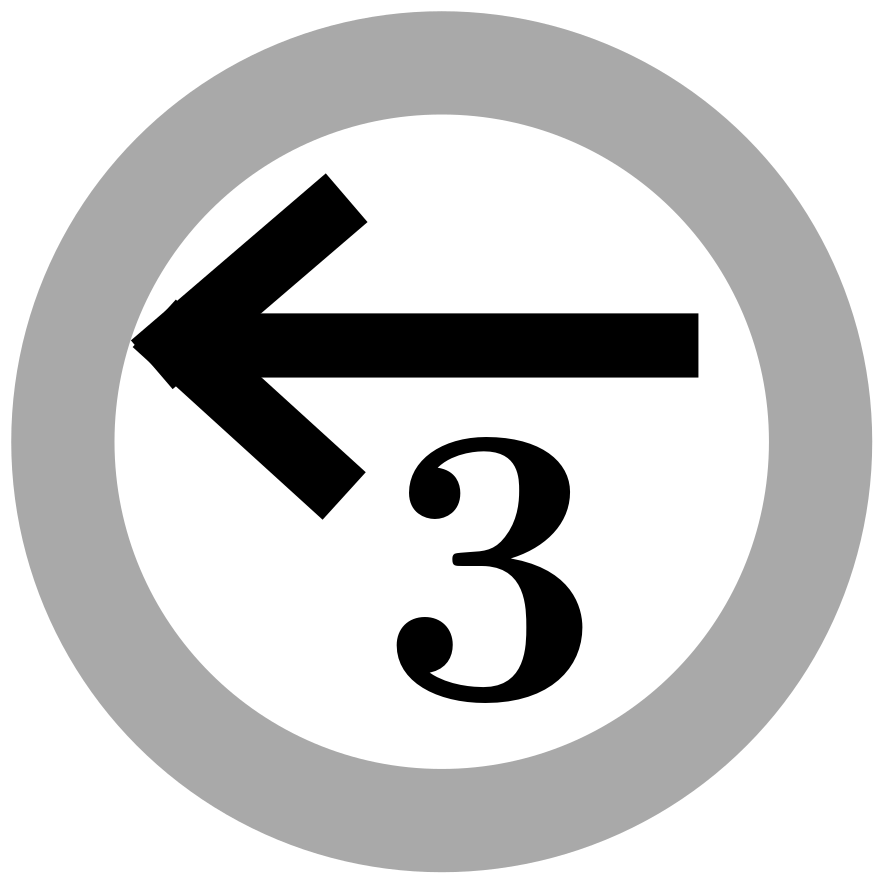}}}}
\newcommand\arrRthree{\vcenter{\hbox{\includegraphics[height = 2.2ex]{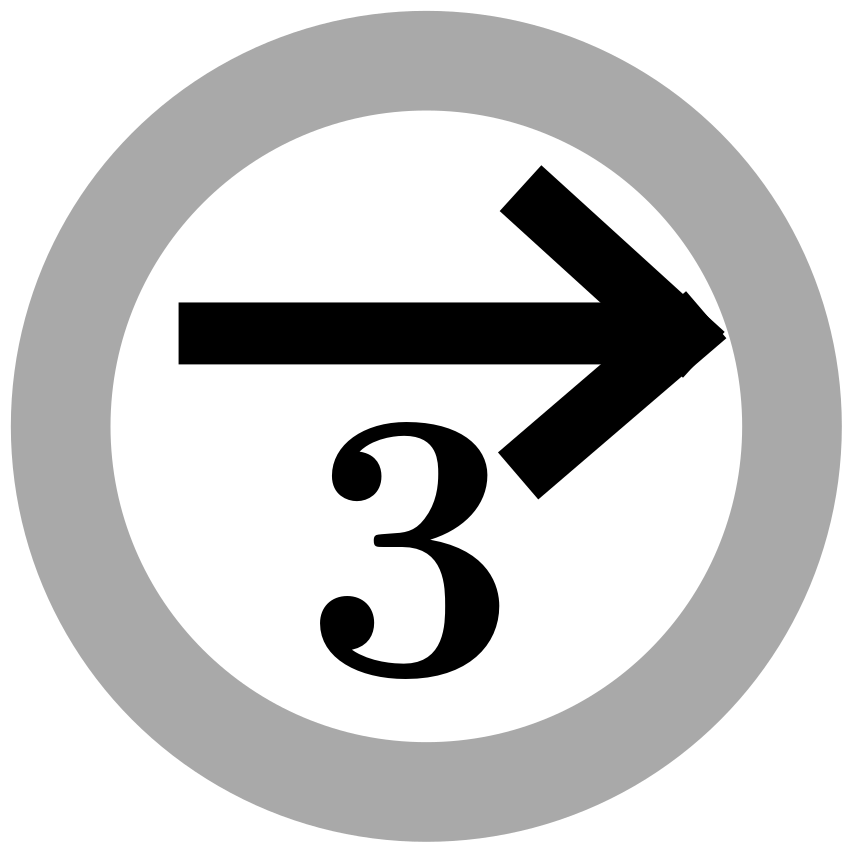}}}}

\newcommand\arrRfour{\vcenter{\hbox{\includegraphics[height = 2.2ex]{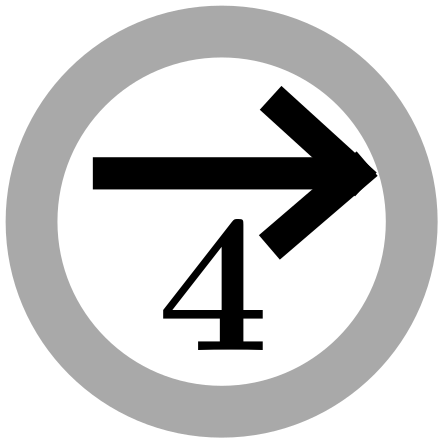}}}}
\newcommand\arrRfive{\vcenter{\hbox{\includegraphics[height = 2.2ex]{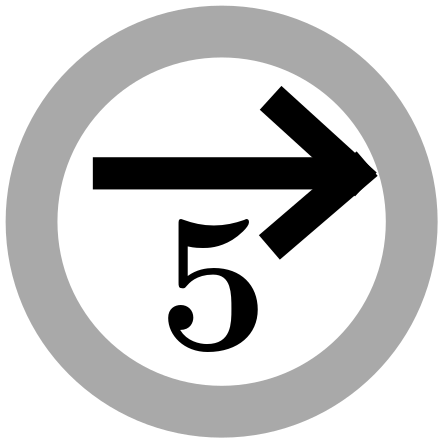}}}}
\newcommand\arrRsix{\vcenter{\hbox{\includegraphics[height = 2.2ex]{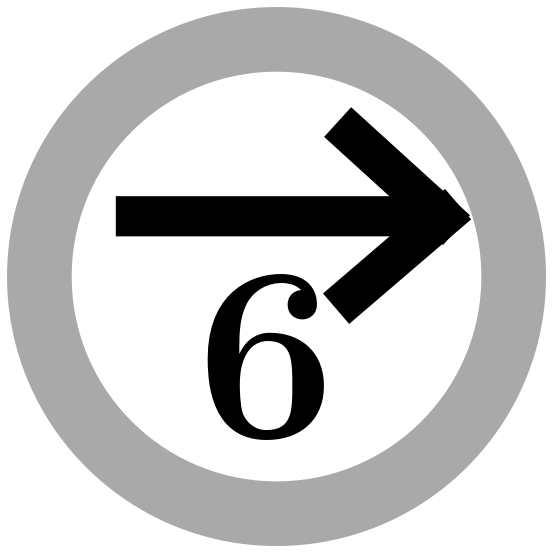}}}}
\newcommand\arrRseven{\vcenter{\hbox{\includegraphics[height = 2.2ex]{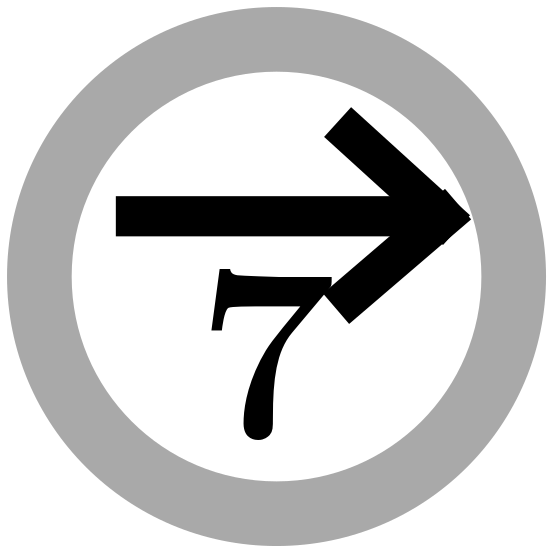}}}}
\newcommand\arrReight{\vcenter{\hbox{\includegraphics[height = 2.2ex]{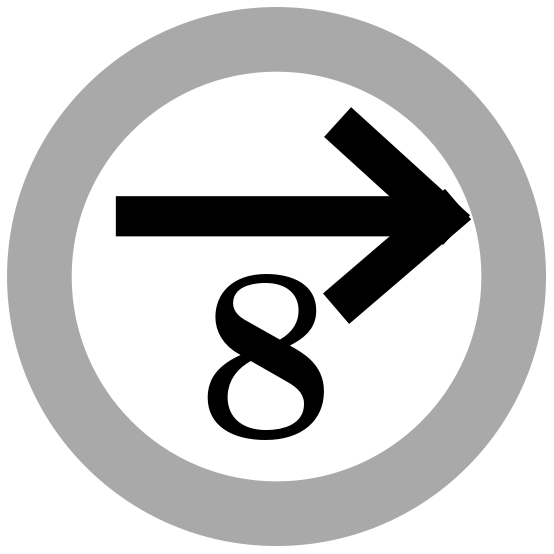}}}}
\newcommand\arrLsix{\vcenter{\hbox{\includegraphics[height = 2.2ex]{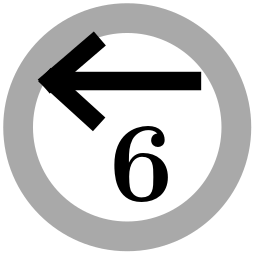}}}}
\newcommand\arrLseven{\vcenter{\hbox{\includegraphics[height = 2.2ex]{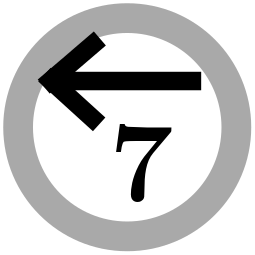}}}}
\newcommand\arrLeight{\vcenter{\hbox{\includegraphics[height = 2.2ex]{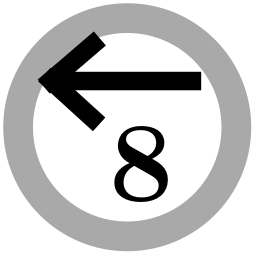}}}}

\newif\ifquantum
\quantumtrue

\author{
  Dorit Aharonov\thanks{{\tt dorit.aharonov@gmail.com}}
  \and
  Sandy Irani\thanks{{\tt irani@ics.uci.edu}}
}
\title{Hamiltonian Complexity in the Thermodynamic Limit}

\begin{document}

\maketitle
\thispagestyle{empty}

\abstract{
Despite immense progress in quantum Hamiltonian complexity in the past decade, little is known about the computational complexity of quantum physics at the thermodynamic limit.
In fact, even defining the problem properly is not straight forward. We study  the complexity of estimating the ground energy of a fixed, translationally-invariant (TI) Hamiltonian in the thermodynamic limit, to within a given precision; 
this precision (given by $n$ the number of bits of the approximation) is the sole input to the problem. 
Understanding the complexity of this problem captures how difficult it is for a physicist to measure or 
compute another digit in the approximation of a physical quantity 
in the thermodynamic limit.  
We show that this problem is contained in $\fexpqmaexp$ and is
hard for $\fexpnexp$. This means that the problem is {\it doubly} exponentially hard in the size of the input. 

As an ingredient in our construction, we study the problem of computing
the ground energy of translationally invariant finite 1D chains.
A single Hamiltonian term, which is a fixed parameter of the problem, is applied to every pair of particles in a finite chain.
In the finite case, the length of the chain is the sole input to the problem and the task is to compute an approximation of the ground energy. No thresholds are provided as in the standard formulation of the local Hamiltonian problem.
We show that this problem is contained in $\fpqmaexp$ and is
hard for $\fpnexp$. 
Our techniques employ a circular clock structure in which the ground energy is calibrated by the length of the cycle. This requires more precise expressions for the ground states of the resulting matrices than was required for previous QMA-completeness constructions and even exact analytical bounds for the infinite case which we derive using techniques from spectral graph theory. To our knowledge, this is the first use of the circuit-to-Hamiltonian construction which shows hardness for a function class.
}

\pagebreak
\thispagestyle{empty}

\tableofcontents
\thispagestyle{empty}

%
%

\setcounter{page}{0}
\pagebreak
\section{Introduction} 
Kitaev's fundamental  QMA-completeness result \cite{KSV02} initiated the field of quantum Hamiltonian complexity \cite{Gharibian_2015}; however its setting is still very far from the problems that naturally arise in condensed matter physics. 
Subsequent work brought QMA-completeness results 
closer to physical settings by extending from the general geometry of Kitaev's local Hamiltonians to 2D \cite{OT05}, and even 1D \cite{Aharonov_2009} lattices;  
Gottesman and Irani \cite{GI} showed that the 1D result holds even for translationally invariant (TI) systems, and even when 
the only input is $N$, the size of the system. 

Despite this important and fundamental progress,  such complexity results
are still  far from  capturing the primary challenges in computational many-body physics.
Physicists study finite systems by necessity,
but the problem which is typically of greatest interest is the estimation of physical quantities (energy density, two body correlations, etc.) in the {\it Thermodynamic limit} (TL)\footnote{This question was highlighted by I. Cirac in 
an online discussion in a SIMONS institute quantum workshop, 2017}. Their focus is not the estimation of physical quantities as a {\it function} of the size 
of the system, as in the QMA completeness results, but instead, they study a {\it particular} quantity of a Hamiltonian whose local terms are {\it fixed}, 
as  $N\longmapsto\infty$.

The breakthrough undecidability result for computing the spectral gap 
by Cubitt,  Perez-Garcia, Wolf
\cite{C15} and its  follow up \cite{Bausch_2020} provided 
the first study of the {\it computability} of  Hamiltonian problems in the TL. 
To the best of our knowledge, the
only existing result about  {\it computational complexity} in the TL, is that of  \cite{GI} who study this as a side result to their main finite case result
\footnote{
More specifically, the TL  problem studied in \cite{GI}  is parameterized by three polynomials, $r$, $p$ and $q$.
The input is an integer $N$  in binary representation and a Hamiltonian term $h$ acting on a pair of $d$-dimensional
particles whose entries are integer multiples of $1/r(N)$.
They show that it is a $\qmaexp$-complete problem
to determine whether the ground energy
density of the Hamiltonian resulting from applying $h$ to every pair of particles
in an infinite 1D chain 
is below $1/p(N)$ or above $1/p(N) + 1/q(N)$.}. 
However, importantly, \cite{GI} (as well as \cite{C15, Bausch_2020}) study the TL when the Hamiltonian term applied to each pair of particles is {\it input dependent}.
While the input-dependent Hamiltonian settings makes sense in the context of studying gappedness 
as a function of the Hamiltonian  parameters (as in \cite{C15, Bausch_2020}; see also the studies of {\it phase diagrams} for gappedness \cite{Bausch_2021,BravGoss15}), it seems much less justifiable in the common physical scenario of
approximating quantities of the ground state in the TL. In this context, physicists usually treat say, the AKLT model, as a different problem than, say, the Ising model. 

Here we initiate the study of computational complexity in the TL for fixed, input-independent physical systems.  
An immediate problem arises: it is not clear how to even define the  problem, namely, how to associate computational complexity to 
a problem whose 
quantity of interest is merely
a single, fixed number. Our natural approach is to have the input specify only the desired {\it precision} to which the number -- here the ground energy density in the TL -- is computed, in terms of number of bits. The required precision is the only input to the problem. In order to show computational hardness, we have 
 to show how to encode a full language into that single number.
 Note that this challenge does not arise in the case of input-dependent terms addressed in \cite{GI}, where 
the reduction can encode the answer to a decision problem on input $x$ into the input-dependent ground energy of the Hamiltonian resulting from the TI term $h_x$.

\subsection{Problem Definition and statement of result}
We will formally define the problem for the 2D grid; there are 
natural extensions to 1D and higher dimensions. 
The problem in 2D is parameterized by the dimension of each individual particle $d$, and two $d^2 \times d^2$ Hermetian matrices $h^{row}$ and $h^{col}$ denoting the energy interaction between two neighboring particles in the horizontal and vertical directions on a 2D grid. 
Then $H_{2D}(N)$ is the Hamiltonian of an $N \times N$ 2D grid of $d$-dimensional particles,
where the same fixed $d^2 \times d^2$ local terms, $h^{row}$ and  $h^{col}$, are applied to every pair in the horizontal and vertical directions.
%
$\lambda_0(H_{2D}(N))$ is the ground energy of $H_{2D}(N)$. 
The ground energy density of the system in the 
TL is defined as the following limit (we prove this limit always exists in Section \ref{sec:existLimit}):
\begin{definition}
{\bf Energy density in the TL:} We define $\alpha_0$ as the limit
$\alpha_0  = \lim_{N \rightarrow \infty} \frac{\lambda_0(H_{2D}(N))}{N^2}.$
\end{definition}
The main problem we consider is to compute $\alpha_0$ to within a given precision specified by the input:
\begin{definition}
{\sc Function-GED-2D for $(h^{row}, h^{col})$}

{\bf Input:} An integer $n$ expressed in binary.

{\bf Output:} A number $\alpha$ such that $|\alpha - \alpha_0| \le 1/2^n$.
\end{definition}

At first sight, it might seem counter intuitive that this problem is hard as the ground energy density for a fixed Hamiltonian in the TL is just a single number.
A hardness result would  need to embed a hard computational problem for all  instances
into a {\it specific number} $\alpha_0$. We address this by exploiting the fact that the ground energy density in the TL is an infinite precision number, and we can use different portions of its binary representation to encode the solution to different instances of the problem from
which we are reducing.
Our main result is:
\begin{theorem}
\label{th:infinite}
The problem {\sc Function-GED-2D}  is hard for $\fexpnexp$ under Karp reductions, and is contained in 
$\fexpqmaexp$. 
\end{theorem}

A remark is due regarding our chosen definition of the problem {\sc Function-GED-2D}. While it would have been possible to consider a decision version of this problem, such as some variant of determining the $n^{th}$ bit of $\alpha_0$, we believe the function version more naturally describes the problem encountered in physics.
Moreover, computing the $n^{th}$ bit essentially requires computing bits $1$ through $n-1$, as one can measure the  $n^{th}$ bit of the energy for a particular state, but in order to verify that the state being measured is close enough to the ground state, it seems 
necessary to verify that the first $n$ bits of the energy for the given state
correspond to the true ground energy density. 
We believe the 
computational complexity of determining a particular bit of $\alpha_0$
would still be characterized by an oracle class, which constitutes the hardest challenge in our proof.
As in  \cite{C15,Bausch_2020}, 
the proof is based on Kitaev's circuit-to-Hamiltonian construction, and works by embedding finite 1D chains into the 2D infinite lattice using Robinson tiles. 
The problem for the infinite grid thus reduces to a problem for finite 1D chains.
However, the constructions in \cite{C15,Bausch_2020}  require 
a "two-threshold version" of the finite 1D problem; by this we mean the standard QMA-type problem, in which one needs to decide whether a quantity is larger than some threshold or smaller than another. 
Therefore \cite{C15,Bausch_2020} can directly apply  techniques from the finite case (the main result) of \cite{GI}, 
 which addresses this two-threshold setting in the 1D finite TI case. 
However, the hardness result given here (Theorem \ref{th:infinite}) requires a different type of finite 1D problem, where the task is to approximate the ground energy {\it to some given precision}, and no threshold is given. 

 To this end we define an approximation version of the 1D finite problem and characterize its complexity. Our results in the finite case pertain to 1D systems, but they can be naturally generalized for higher dimensions. In the 1D case, there is a single $d^2 \times d^2$ Hermitian  matrix $h$ which parameterizes the problem, and $H_{1D}(N)$ is the Hamiltonian resulting from applying $h$ to every pair of neighboring particles in a $1D$ chain of length $N$. 
The function version of the finite TI  Hamiltonian problem in 1D ({\sc Function-TIH-1D}) is defined as follows:
\begin{definition}
{\sc Function-TIH-1D for $h$ and constant $c$}

{\bf Input:} An integer $N$ expressed in binary.

{\bf Output:} A number $\lambda$ such that $|\lambda - \lambda_0(H_{1D}(N))| \le 1/N^c$.
\end{definition}
The following theorem encapsulates our results for Function-TIH:
\begin{theorem}
\label{th:finite}
Function-TIH-1D problem is contained in $\fpqmaexp$
and is hard for $\fpnexp$ under Karp reductions.
\end{theorem}

The proof of Theorem \ref{th:finite}, which comprises the main technical effort of this paper, requires strengthening the finite TI results of \cite{GI}\footnote{We note that \cite{Bausch_2017} improves on \cite{GI} by reducing the dimension of the particles significantly. However, critically, their Hamiltonian is not input-independent.}  to handle approximation problems rather than two-threshold type problems.
 Ambainis \cite{Amb14}, and later also \cite{GPY19, GY19}, studied a related class of problems of  approximating quantities of groundstates, called APX-SIM, which they argue are better physically motivated than the standard two-threshold type local Hamiltonian problems. 
In those approximation problems, as in the one studied here, the absence of a given threshold presents an inherent challenge: one needs to verify that the state being measured is close enough to the ground state. As a result, the upper bounds on the complexity of these problems  all require some form of binary search with queries to a $\qma$ or $\qmaexp$ oracle \cite{Amb14}.
Indeed 
 the natural complexity classes for such problems are {\it oracular} ones. 
This oracular structure poses a technical challenge for lower bounds (assuming the lower bounds are proven in the stronger setting of Karp reductions), due to the fact that "no" answers from the oracle cannot be verified. 
To overcome this challenge, we make use of 
a technique pioneered by Krentel \cite{Krentel,Papa}
who proved that the optimization version of certain NP-hard problems are complete for $\fpnp$.
Our main technical contribution 
is implementing a TI version of Krentel's technique with sufficiently precise ground energy estimations as needed for the TL case. We 
elaborate on how we do this in Subsection
\ref{sec:proofoverviewfinite}.  


We remark regarding the exponential difference in complexity between the finite and the infinite case (Theorems \ref{th:finite} and \ref{th:infinite}). Roughly, the ground energy density in the TL in 2D can be estimated to within $\pm 1/N$
by solving a a finite grid of size $O(N)$ by $O(N)$ as shown in Lemma \ref{lem:gedbound}.
Since the input $n$ to the FUNCTION-GED-2D problem requires precision $1/2^n$ and is specified using $\log n$ bits, the complexity is doubly exponential in the input size.
In the finite case, the size of the system itself ($N$) is given in binary, so the complexity is only singly exponential.
From an expressibility perspective, in the finite case, every value of $N$ can be used to encode the solution to an instance of the problem from which we are reducing.
By contrast, in the infinite case, the most efficient reduction we can hope for is where each bit of $\alpha_0$ encodes an answer to a computational problem. In this case, the system size has to double for each input encoded.


\subsection{Proof Overview and Techniques: the Finite Case} \label{sec:proofoverviewfinite}

We start by giving an overview of the proof of Theorem \ref{th:finite} for computing the ground energy of finite 1D TI Hamiltonians. More details and references to the lemmas in the paper are given in Subsubsections \ref{sec:krenteloverview}-
\ref{sec:penalty}.  The
finite construction 
is used in the infinite case (Theorem
\ref{th:infinite}) as 1D finite Hamiltonians are layered on top of
a Robinson tiling of the infinite grid. An overview of the infinite case is given in Subsection \ref{sec:proofoverviewinf}.

We now consider
an arbitrary $f \in \fpnexp$, and describe how the reduction to  Function-TIH-1D  works. 
First, $f$ is associated with a {\em fixed} Hamiltonian term $h$ that operates on two $d$-dimensional particles. 
Let $H_N$ denote the Hamiltonian on a chain of $N$ $d$-dimensional particles resulting from applying $h$ to each neighboring pair in the chain. The reduction maps an input string $x$ of length $n$ for the function $f$, 
to a positive integer $N = N(x)$ such that $N(x)$ can be computed in time polynomial in $n$.
We will show that 
there is a polynomial $q$ and a polynomial time classical algorithm, which for any $x$ 
can compute $f(x)$ given a $1/q(n)$-approximation of
$\lambda_0 (H_N)$ (namely a value $E$, where $|E - \lambda_0 (H_N)| \le 1/q(n)$).

In the circuit-to-Hamiltonian construction, the Hilbert space of the entire chain consists of a Hilbert space which encodes the state
of the computation, tensored with a Hilbert space that contains a clock which regulates the process of transitioning from one configuration to another. In the absence of any additional penalty terms, the ground energy is $0$ and is achieved by
the state that is a uniform superposition of all states in the computation, entangled with the clock state for that point in time:
$\sum_{t=0}^L \ket{t} \ket{\phi_t(\text{init})}$. The
$\ket{t}$ denotes the state of the clock and $\ket{\phi_t(\text{init})}$ denotes the state of the computation
after starting in state $\ket{\text{init}}$ and progressing for $t$ clock steps. 
In our construction. the ground state of the Hamiltonian 
encodes the history of a computation which simulates the  polynomial time Turing Machine that computes the function $f$,
with access to a $\nexp$ oracle.

As mentioned above, the oracle calls pose a challenge in the circuit-to-Hamiltonian construction: 
the {\em no} guesses of the oracle responses cannot be verified. 
To overcome this, we use Krentel's accounting scheme \cite{Krentel,Papa} that applies a cost to every string $y$ representing guesses for the sequence of responses to all the oracle queries made. The accounting scheme needs to ensure that the minimum cost $y$ is equal to the correct sequence of oracle responses, $\tilde{y}$. 
{\em yes} and {\em no} guesses are treated differently, due to the fact that   
the verifier can check {\em yes} instances (and thus incorrect {\em yes} guesses can incur a very high cost), but {\em no} guesses, cannot be directly verified. In Krentel's scheme, {\em no} 
guesses 
incur a more modest cost, whether correct or not, and their cost 
must 
decrease exponentially. This is because the oracle queries are adaptive;  
an incorrect oracle response could potentially change all the oracle queries made in the future
and so it is important that the penalty for an incorrect guess on the $i^{th}$ query is higher than the energy
that could  potentially be saved on all future queries.
Thus,  the costs 
range
from a constant to exponential in $m$, where $m$ is the number of oracle calls; our main challenge is that 
a fixed translationally-invariant Hamiltonian cannot directly 
encode these costs in penalty terms.  

We address this issue as follows. The computation consists of repeated loops, each of which lasts $p(N)$ steps. In each repetition the computation simulates the verifier for all the {\em yes} oracle guesses and imposes an energy cost if any of these computations rejects (i.e., the {\em yes} guess was wrong). 
We define a function $T(x,y)$ (discussed below) and enforce that the number of loop  repetitions is $2T(x,y)+1$,  making the total length of the computation 
$L=(2T(x,y)+1)\cdot p(N)$ steps. The clock is  circular, making eigenvalues analysis easier, and hence we refer to the entire computation as a {\em cycle}. The Hamiltonian is thus block diagonal in those cycles, where each cycle, or block, corresponds to a computation initiated with different input parameters:
 the guess string
$y$,  a guess $T$ (in unary) of $T(x,y)$, a string $w$ corresponding to the witnesses needed for verifying the {\em yes} guesses of the oracle, and finally the initial configuration for the computation, which we denote by $\vinit$. 
The ground energy of the Hamiltonian is the minimum over the ground values of these blocks. 

To achieve the large penalty for incorrect {\em yes} guesses, we introduce a penalty for each  verification 
computation that ends with {\sc reject}.  This 
results in a periodic cost occurring once per iteration. We use spectral graph theory
to obtain a close-to-tight lower bound on the lowest eigenvalues of the relevant blocks, which are Laplacians with periodic $+1$'s on the diagonal; We show that the ground energy for blocks with incorrect {\em yes} guesses would then behave inversely with the square of the {\em period}'s length, and will thus be respectively large.  

To penalize incorrect {\em no} guesses, we introduce two consecutive $+1/2$ penalty terms on the diagonal for {\it every} computation (even correct ones).
We show 
that the ground energy of the resulting block (where all {\em yes} guesses are correct) is exactly \begin{equation}\label{eq:cosinus}
    1 - \cos(\pi/(L+1)),
    \end{equation} which is a function of $L$. 
Our idea is that we can vary the {\it length of the computation} $L$ to control the ground energy, in order to implement the energy penalty for the {\em no} guesses. 
This is achieved by defining the following function $T(x,y)$, for a given input $x$ and oracle guess sequence 
$y$:
\begin{equation}
    \label{eq:Txy}
    T(x,y) = f(x,y) + 2^m \cdot 4^{m+1} + 2^m \cdot \sum_{j=1}^m y_j \cdot 4^{m-j+1}
    \end{equation}
Note that a {\em yes} guess ($y_j = 1$) increases the value of $T$ and in turn of $L$, and thus decreases the lowest eigenvalue. Therefore {\em no} guesses have an implicit cost. The function $T$ has the required exponential structure so that
the cost of a {\em no} guess decreases exponentially with each query.
This fact, along with the fact that incorrect {\em yes} guesses incur a very high periodic cost,
guarantees that the smallest eigenvalue will correspond to a block with $T(x, \tilde{y})$, where $\tilde{y}$ is the string of correct oracle responses. 

The function $f(x,y)$ in the definition of $T(x,y)$ is the outcome of the computation on input $x$, when string $y$ is used for the oracle responses made by the Turing Machine computing $f$. If $\tilde{y}$ is the set of correct oracle responses, then $f(x) = f(x, \tilde{y})$. 
We show in Theorem \ref{th:finitehardness} that if the ground energy $(1 - \cos(\pi/(L+1)))$ can be computed to a $1/\mbox{poly}$ precision for a sufficiently high degree polynomial, then the value of $L = p(N)\cdot (2T(x, \tilde{y})+1)$ can be recovered, from which $T(x,\tilde{y})$ can be recovered.
Here we are assuming that $f(x,y)$ is of length at most $m$; this can be guaranteed using a standard padding argument (see Lemma \ref{lem:pad}). Thus, as can be seen from the expression for 
$T(x,y)$ given in Equation (\ref{eq:Txy}),  $f(x,\tilde{y})$
is just the low order bits of $T(x,\tilde{y})$. 

The Hilbert space of each particle in our construction 
is a tensor product of $6$ different spaces which are used to form $6$ tracks. 
There are also two additional particle 
states $\leftBr$ and $\rightBr$. We add constraints that enforce the condition that
for a low energy state, the leftmost particle must be in state
$\leftBr$, the right particle must be in state $\rightBr$
and the particles in between are not in states $\leftBr$ or $\rightBr$.
Three of the tracks are {\em clock} tracks which contain the state of the clock.
The other tracks are {\em computation} tracks that encode the configuration of a Turing Machine computation. 
An example is given in Figure \ref{fig:tracks}. The symbols in the figure above represent standard basis states for 
each portion of the Hilbert space.
\begin{figure}
    \centering
  \includegraphics[width=4.7in]{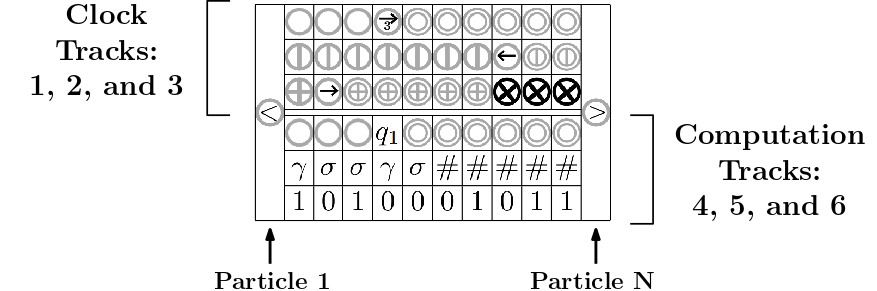}
    \caption{Sample configurations for the clock and computation tracks.
The computation tracks represent the configuration of a Turing Machine with the top track representing the state and location of the head, the second track representing the contents of the work tape, and the third track representing a read-only witness string. }
    \label{fig:tracks}
\end{figure}

In the remainder of this subsection, we provide more details on how the construction is implemented.

\subsubsection{The  Clock}\label{sec:krenteloverview}

In this subsubsection we focus on the Hamiltonian restricted to the subspace of the clock.
The clock used in this construction is a more elaborate version of the clock used in \cite{GI}.
Track $1$ has exactly one pointer that shuttles back and forth between the left and right ends of the chain according to the transition rules. The $\arrR$ pointer moves to the right and the $\arrL$ pointer moves to the left.  The Track $1$ pointers come in $8$ different varieties and are labeled with tags in the range from $1$ through $8$. There are two pointers 
with the  label $i$:  $\arrLi$ and $\arrRi$.
The different types of  Track $1$ pointers act as a means of program control as they trigger different operations on the three computation tracks (Tracks $4$, $5$, and $6$). 

In each iteration for Tracks $1$ and $2$, two different Turing Machines (described in the next subsection) are   run sequentially, each for $N-2$ steps. 
The Track $1$  pointer acts as a second hand for the clock.
In a round trip of a Track $1$ pointer, a single step of the Turing Machine  is executed
on the computation tracks.
The 
number of Turing Machine steps in the iteration  is regulated by the pointer in Track $2$ which is advanced by one location as the Track $1$ pointer sweeps by. Thus, the Track $2$ pointer acts as a minute hand for the clock. 
Some of the Track $1$ pointers act as the identity on the computation tracks and are used to check certain local conditions and impose an energy penalty if those conditions are not met.

We classify configurations for Tracks $1$ and $2$ as either {\em correct} or {\em incorrect} and show in Lemma \ref{lem:clockGraph2} (and see also Lemma \ref{lem:lbpath})
that the incorrect clock configurations will never be in the support of the ground state. 
There are exactly $p(N) = 4(N-2)(2N-3)$ correct clock configurations for Tracks $1$ and $2$.
We show in Lemma \ref{lem:clockGraph1} that starting from any correct clock configuration for Track $1$ and $2$, 
if the transition rules are applied $p(N)$ times, then
all the correct clock configurations for Tracks $1$ and $2$ are reached and we return to the initial configuration. We will refer to a sequence of $p(N)$ clock steps as a {\em iteration}.

The third clock track acts like an hour hand in that the pointer on Track $3$ moves by one location
each time an iteration is completed.
Track $3$ is used  to hold a {\it timer}, that will count the number of iterations. 
 The number of $\DblankLthree$ or $\blankLthree$ particles, denoted by $T$, is called the {\em timer length} for Track $3$, and the transition rules never alter the length of the timer. If the timer length is $T$, the iteration will be repeated $2T+1$ times, to form a {\em clock cycle}. This is done as follows.  

The Track $3$ pointer shuttles back and forth, as do the pointers on Tracks $1$ and $2$. However, the right-moving Track $3$ pointer will turn around when it reaches the  the left-most $\dead$ particle. This provides a way to control the number of iterations of Tracks $1$ and $2$ in each clock cycle. 
In the sample configuration shown in Figure \ref{fig:tracks}, the Track $3$ timer has length $6$. If the length of the timer is $T$, then Track $3$ transitions through
$2T+1$ configurations before repeating. Referring again to the example in Figure \ref{fig:tracks}, since the timer has length $6$, there will be 
$13$ iterations of Tracks $1$  and $2$ for each clock cycle.
Starting in a correct clock configuration in which Track $3$ has a timer length $T$,
after $(2T+1) \cdot p(N)$ time steps, the clock configuration will return to its original state completing a clock cycle. We note that 
in referring to a {\em clock configuration} without specifically restricting to a subset of the tracks, we are referring to a configuration for all three clock tracks. 

Lemmas \ref{lem:countcorrect} and \ref{lem:lbpath} together show that the space spanned by incorrect clock configurations have high energy. 
In the absence of the computation tracks, we show in Lemma \ref{lem:correctclock} that
the Hamiltonian restricted to correct clock configurations is block diagonal, with one block corresponding to each possible timer length $T$. The block corresponding to timer length $T$ is $1/2$ times the Laplacian of a cycle graph with $(2T+1) \cdot p(N)$ vertices. 
When the Hilbert space for the computation tracks is tensored with the Hilbert space for the clock, the state space expands but the Hamiltonian retains this block diagonal structure. In the expanded space that includes both the computation and the clock space,  every block is parametrized by its timer length $T$ in addition to the
initial configuration of the Turing Machine, denoted  by $\vinit$ (describing the state of Tracks $4$ and $5$)
as well as the read-only witness string on Track $6$. 
It will be convenient to divide the witness string into a string $y$ used as the guesses for the responses to the oracle queries and $w$ for the witness string used in simulating the verifier for the calls made to the oracle in which the guess is {\em yes} (for each such {\em yes} guess there will be a different witness string $w_i$).  Thus, every block is parameterized by the $4$-tuple $(T, \vinit, y, w)$. Penalty terms described in Subsection \ref{sec:penalty}
will be added in to ensure
that blocks corresponding to incorrect parameters have a high ground energy.
First we describe in Subsubsection \ref{sec:compoverview} the computation performed in each iteration.

\subsubsection{Overview of the Computation Embedded in the Hamiltonian}
\label{sec:compoverview}

Here we review the computation used in the construction. This computation is repeated $2T+1$ times.

\vspace{.08in}

\noindent
{\bf Stage $1$ - Binary Counter Turing Machine:}
A Turing Machine that increments a binary counter called $M_{BC}$ is run for $(N-2)$ TM steps, where $N$ is the length of the chain.  If the starting configuration has the string $"1"$ on the work tape, then after $(N-2)$ TM steps, there will be some string $x$ on the work tape representing the number of increment operations performed. $x$ is then used as the input to the next stage of the computation.

\vspace{.08in}

The idea of using a binary counter TM to translate the length of the 1D chain into an input string $x$ was used in \cite{GI} as well.
Define $N(x)$ to be  the number such that after $(N(x)-2)$ TM steps, the binary counter TM completes an increment operation with $x$ on the work tape. The function mapping $x$ to
$N(x)$ is the reduction. The complete specification of the binary counter Turing Machine as well as an explicit formula for $N(x)$ is given in Section \ref{sec:MBC}.

\vspace{.08in}

\noindent 
{\bf Stage $2$ - Timer and Verification Turing Machine:}
In the second stage of the computation, a Turing Machine called $M_{TV}$
simulates the verifier for oracle responses for which the guess response is  {\em yes}, 
and then computes the function $T(x,y)$ shown in Equation (\ref{eq:Txy}).
The first $m$ bits of the witness track
are used as the guess $y$ for the oracle responses and the remaining bits $w$ are used as witnesses for the verifier as needed. If $f \in \fpnexp$, then $f$ is computed by a polynomial time Turing Machine $M$ with access to an oracle for language $L \in \nexp$. The exponential time verifier for $L$ is called $V$.
First $M$ is simulated on input $x$ using the oracle responses given by $y$. This completely determines the oracle inputs $x_1, \ldots, x_m$. If $y$ guesses that an oracle response is {\em yes} ($y_i = 1$) then 
$V$ is simulated on input $x_i$ using bits from the witness $w$ from the witness track. If any of the 
oracle computations that correspond to {\em yes} guesses are rejecting, a character is placed on the work tape that will trigger a penalty term at the second checking phase, which occurs at the end of the computation (as explained in Subsection \ref{sec:penalty}). 
Next the function $T(x,y)$ is computed and written in unary on Track $3$ so  that it can be checked (again, in the second checking phase)  against the timer used to regulate the number of iterations of the clock. 
We argue in Lemma \ref{lem:existsMTV}  that this second stage of the computation  can also be computed in at most $(N-2)$ TM steps. 

\vspace{.08in}

After both stages, the entire process is reversed, so that the configuration of the computation tracks returns to its initial state with the string $1$ on the work tape. The pseudo-code for the computational process executed in the second $N-2$ TM steps is given in Figure \ref{fig:MTVpseudo}. The schedule for a single iteration is:
\begin{center}
\begin{tabular}{ll}
  1)~ Checking Phase $1$   &  4) ~Checking Phase $2$\\
 2) ~$N-2$ forward steps of $M_{BC}$  ~~~~~~  & 5)~ $N-2$ reverse steps of $M_{TV}$\\
3)~ $N-2$ forward steps of $M_{TV}$    & 6) ~$N-2$ reverse steps of $M_{BC}$
\end{tabular}
\end{center}

The checking phases and how they are used to implement the penalities are described in the next subsection (Subsubsection \ref{sec:penalty}). Each checking phase is implemented by
a single round trip of a pointer on Track $1$ which acts as the identity on the computation tracks
and triggers a penalty if certain local constraints are not met. Recall that the Track $1$ pointers come with $8$ different labels which act as a means of program control and are used to create the schedule given above.

\subsubsection{Penalty Terms and Lowest Eigenvalues}
\label{sec:penalty}

We now explain the checking phases, the conditions they impose, and how those conditions determine the lowest energy eigenvalue. 
If any of the conditions enforced in the checking phases are violated, that will trigger a penalty that occurs every $p(N)$ clock steps for a total of $2T+1$ times during one full iteration of the cycle of length $(2T+1) \cdot p(N)$. This will result in a block matrix corresponding to a cycle of length $L = (2T+1) \cdot p(N)$ with an additional $+1$ every $p(N)$ locations along the cycle. We use techniques from spectral graph theory to 
exploit the exact periodic structure of the penalty terms to 
show in Lemma \ref{lem:periodicEig} that the lowest energy eigenvalue of this matrix is at least
\begin{equation}
\frac 1 8 \left( 1 - \cos \left( \frac{\pi}{2p(N)+1} \right) \right) = \Theta \left( \frac 1 {(p(N))^2} \right).
\end{equation}
This lower bound is large enough to eliminate any blocks with these periodic costs.
We now explain how  we use this mechanism to create the large costs that penalize an incorrect {\em yes} guess for an oracle response as well as to check that the initial configuration and the timer lengths are correct. 

The first checking phase  takes place before the forward execution of the Turing Machines 
and verifies that the Turing Machine is starting in the correct initial configuration for Tracks $4$ and $5$
which is depicted below:
\begin{center}
      \begin{small}
\begin{tabular}{|@{}c@{}|@{}c@{}|@{}c@{}|}
\hline
$\leftBr$ &
\begin{tabular}{@{}c@{}|@{}c@{}|@{}c@{}|@{}c@{}|@{}c@{}|@{}c@{}|@{}c@{}|@{}c@{}|@{}c@{}|@{}c@{}}
$q_0$ &  $\Dblank$ & $\Dblank$ & $\Dblank$ & $\Dblank$ & $\Dblank$ & $\Dblank$ & $\Dblank$ & $\Dblank$ & $\Dblank$ \\
\hline
$1$ &  $\#$ & $\#$ &  $\#$ & $\#$ & $\#$ & $\#$ & $\#$ & $\#$ & $\#$ \\
\end{tabular}
& $\rightBr$ \\
\hline
\end{tabular}
\end{small}
\end{center}
If the computation begins in this special initial configuration then at the end of $N-2$ steps of the Turing Machine $M_{BC}$, the correct input $x$ for chain length $N$ will be written on the work tape. 

The second checking phase occurs after the forward computation of the two Turing Machine 
and checks that the timer on Track $3$ matches the calculation of $T(x,y)$ performed by $M_{TV}$. Thus, there will be a periodic cost unless the timer length $T$ is equal to $T(x,y)$ for the correct input $x$ and the $y$ written on the witness track. In addition, the second checking phase is used to check
that all the simulations  of the verifier performed by $M_{TV}$ are accepting. In other words, a penalty is applied if any of the verifier computations were rejecting.
Thus, if a block parameterized by $(T, \vinit, y, w)$ does not have a periodic cost, then the computation generates the correct input $x$, the timer $T$ is equal to $T(x,y)$, and all of the oracle guess where $y_i = 1$ ({\em yes} guesses) are correct and use a correct witness from $w$ in the verifying computation.

Now we discuss the mechanism used to apply the more modest penalty for {\em no} guesses for oracle responses.
 In our construction, every computation, regardless of the outcome of the computation, will incur two consecutive penalties of $+1/2$ at a particular point in the cycle. The resulting matrix is the propagation matrix for a cycle of length $L = (2T+1) \cdot p(N)$ with an additional $+1/2$ at two consecutive locations on the diagonal. We show in Lemma \ref{lem:cycleEig} that the lowest eigenvalue for this matrix is exactly
 $1 - \cos(\pi/(L+1))$.
Thus any block which does not have a periodic cost will have a smallest eigenvalue 
equal to $1 - \cos(\pi/(L+1))$, with
$L = (2T(x,y)+1) \cdot p(N)$. The block with the smallest eigenvalue will correspond to the $y$ that has the largest value of $T(x,y)$ subject to the condition that $y$ does not include any incorrect {\em yes} guesses. The exponential structure of the function $T(x,y)$ shown in Equation (\ref{eq:Txy})
guarantees that the block with the smallest eigenvalue uses the correct oracle guesses $\tilde{y}$, as shown in Lemma \ref{lem:groundenergy}.


Recovering the value of $f(x)$ from an approximation of the ground energy requires a much sharper analysis of the lowest eigenvalue 
than has been required in previous constructions. This analysis is given in Section \ref{sec:eigenvaluebounds}, and the recovery method is given, as mentioned above, in Theorem \ref{th:finitehardness}. 


\subsection{Proof Overview and Techniques: the Infinite Case} \label{sec:proofoverviewinf}

We use the  technique introduced by Cubitt, Perez-Garcia, and Wolf in \cite{C15} who incorporate an apreiodic tiling structure into the Hamiltonian term for the infinite plane. They use Robinson tiles \cite{R71}  which are a finite set of tiling rules that when applied to the infinite plane, force an aperiodic structure with squares of exponentially increasing size, as shown in Figure \ref{fig:RobinsonGrid}. Each square has size $4^k$, for positive integer $k$, and the density of squares of size $4^k$ in the limit of the infinite plane is  $1/4^{2k+1}$. 
Tiling rules are essentially a classical version of local Hamiltonians, so the tiling rules can be encoded into a 2D layer of the Hilbert space on the plane. As was done in \cite{C15}, we layer a TI 1D Hamiltonian  on top of one of the sides of all the squares. 
The tiling pattern on the lower layer determines where the 1D term is applied.
The effect of this structure is that the ground energy density $\alpha_0$ for the infinite plane is the sum of the ground energies for an infinite series of finite 1D systems divided by the density of each square size: 
\begin{equation}
\label{eq:alpha}
    \alpha_0 = \sum_{x=1}^{\infty} \frac{\lambda_0 (N_x)}{4 (N_x)^2},
\end{equation}
where $N_x = 4^{x^2}$ and $\lambda_0 (N_x)$ is the ground energy of a 1D chain of length $N_x$
with the Hamiltonian term from the 1D finite construction applied to each pair in the chain.
Note that it is essential here that the TI construction for finite 1D chains have a fixed Hamiltonian term that does not depend on the chain size since in the infinite construction the same Hamiltonian term is applied to every pair of neighboring
particles along each dimension.

As noted earlier, our reduction needs to encode an entire language in one, infinite-precision number $\alpha_0$.
In our construction, different portions of the binary representation of $\alpha_0$ encode the value of a function $f$
on different inputs. In particular, bits $4x^2$ through $4(x+1)^2$ encode the value of $f(x)$.
The top segment of each square of size $4^{x^2}$ in the Robinson tiling of the infinite plane is layered with
a TI 1D Hamiltonian whose ground energy encodes the value $f(x)$,
where the function problem $f$ is from the oracle complexity class $\fexpnexp$.
This 1D Hamiltonian is exactly the construction used in the finite case, except that the 
computation is in $\fexp$ instead of $\fp$ as in the finite case. The reason for the exponential
increase in complexity is that we are using a square of size $4^{x^2}$ to encode the computation
on input $x$. 

Each $\lambda_0(N_x)$ in Equation (\ref{eq:alpha}) is an irrational number,
so bits $4x^2$ through $4(x+1)^2$ which encode the value of $f(x)$
will also include the "overflow" from
the energy contributions of $f(x')$ for every $x' < x$. The values of $f$ on inputs $x' < x$,
then need to be calculated to the required precision in the classical calculation of the reduction and subtracted off from $\alpha_0$ in order to recover the bits required to reconstruct the value of $f(x)$. 
Using the analysis of our construction for the finite 1D case, we know that
$\lambda_0 (N_x) = (1 - \cos(\pi/(L_x+1))$, where $L_x$ is the integer equal to $p(N_x) \cdot (2T(x, \tilde{y})+1)$. Recall that $T(x, \tilde{y})$ encodes the desired value for $f(x)$.

We now sketch how the first $4(x+1)^2 + 2$ bits 
of $\alpha_0$ are sufficient to recover $\lambda_0 (N_x)$ by inductively 
subtracting off $\lambda_0 (N_{x'})/4(N_{x'})^2$ for every $x' < x$.
Note that dividing $\lambda_0 (N_{x'})$ by $4(N_{x'})^2$ effectively shifts the binary representation of $\lambda_0 (N_{x'})$ to the right by
$\log_2 [4(N_{x'})^2] = 4(x')^2 + 2$ bits.
Let $\overline{\alpha}_0$ be the current value of the sum. 
Initially, $\overline{\alpha}_0$ is equal to the first $4(x+1)^2 + 2$ bits 
of $\alpha_0$.
In an inductive step, one has subtracted off $\lambda_0(N_{x'})/4(N_{x'})^2$ for every
$x' < z$ for some $z \le x$. As shown in Figure, \ref{fig:alpha}, 
the first $4z^2+2$ bits of $\overline{\alpha}_0$ are $0$ and the next
$4(z+1)^2 - 4z^2$ bits are determined solely by $\lambda_0 (N_z)$.
We show in the proof of Theorem \ref{th:finitehardness} that this is enough information to recover the integer $L_z$ where
$\lambda_0 (N_z) = (1 - \cos(\pi/(L_z+1))$.
Once $L_z$ is recovered, $\lambda_0 (N_z)$
can be computed to any desired accuracy. In particular, it can be calculated up to $4(x+1)^2 + 2$ bits of precision
and subtracted off from $\overline{\alpha}_0$. The result is a new
$\overline{\alpha}_0$ which is equal to the sum $\sum_{x' = z+1}^x \lambda_0(N_{x'})/4(N_{x'})^2$ up to $4x^2 + 2$ bits of precision.
\begin{figure}[ht]
  \centering
  \includegraphics[width=5.0in]{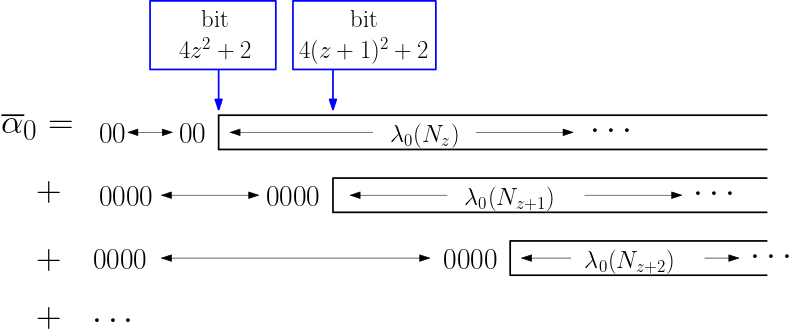}
\caption{The partial sum $\overline{\alpha}_0$. The first $4z^2+2$ bits of $\overline{\alpha}_0$ are $0$ and the next
$4(z+1)^2 - 4z^2$ bits are determined solely by $\lambda_0 (N_z)$.
This is enough information to recover the integer $L_z$ where 
$\lambda_0 (N_z) = (1 - \cos(\pi/(L_z+1))$. }
\label{fig:alpha}
\end{figure}
This illustrates why it is necessary to calculate the ground energy for each finite chain to an arbitrary level of precision: for every $z < x$, it is necessary to calculate $\lambda_0(N_z)$ to a precision of $4x^2+2$ bits in order to derive $f(x)$ from a $4x^2+2$-bit approximation of $\alpha_0$.

\subsection{Summary of main new technical contributions in the proof} \label{sec:new}

While the results in the paper make use of many previously known techniques  from Hamiltonian complexity,
there are several novel technical contributions introduced here which we hope will be useful elsewhere.

As mentioned above, to the best of our knowledge this is the first time that the length of the computation (which we control using the number of repetitions $T(x,y)$) is used to control the ground energy and make sure that the ground value indeed corresponds to the correct initial parameters.  

We believe this is also the first time in Hamiltonian complexity that such precise control over the eigenvalues was required. We need the high level of precision in several places in the proof. First, in the infinite case, as explained in Subsection \ref{sec:infinite}, 
the precise analysis of the lowest eigenvalue is essential in simultaneously encoding an infinite sequence of numbers (values $f(x)$ for an infinite sequence of $x$'s) into a single, infinite-precision number.
The high level of precision is needed so that given a precise enough estimate for $\alpha_0$ 
we can inductively subtract off the energy contributions of $f(x')$ for each $x' < x$, in the classical computation of the reduction; this then leaves enough information
to determine $f(x)$.

Secondly, 
even in our finite case construction, the fact that 
we are considering a function problem implies that we must distinguish between many different values for the cost function as opposed to extracting only a single bit as required for a decision problem. This requires a more precise 
handle on the lowest eigenvalue of the resulting matrix.

Last but not least, in the finite case, as explained in subsection \ref{sec:proofoverviewinf}, we use highly precise lower bounds to show that the energy penalty for incorrect {\em yes} oracle guesses, which we achieve by periodic penalties of period $p(N)$, is indeed larger than the ground energy in the case that all {\em yes} guesses are correct. 
The usual lower bound of  $\Omega(1/p(N)^3)$ used in previous circuit-to-Hamiltonian constructions does not suffice, as it needs to be larger than our upper bound on the ground value, which is $O(1/T^2P(N)^2)$
by Equation (\ref{eq:cosinus}). We cannot take $T$ to be arbitrarily large, since $T$ has to be computed in unary in less than $N$ steps and so it is bounded by roughly $\sqrt{N}$.



 The use of the circular clock also requires some additional care, as it is important that the computation
part of the state also returns to its initial state at the end of the cycle. By contrast, if the computation is a path, the ending state of the computation can be arbitrary.
We handle this by embedding a computation which executes a process in the forward direction for a certain number of steps. Since the action on the computation tracks is reversible, the computation can be undone by executing the same number of reverse steps. 

\subsection{Concluding Remarks, Related Work and Open Questions} \label{sec:related}

To our knowledge, this is the first time the complexity of calculating 
physical quantities in the
TL for a fixed Hamiltonian is characterized from a computational 
perspective. 
To this end we define a function problem which captures roughly the complexity of the task that a physicist encounters when attempting to calculate a physical quantity with one more bit of precision; we show, roughly, that it is doubly exponential.
 To 
 our knowledge this is also the first use of a circuit-to-Hamiltonian construction to show hardness for a {\em function} problem.




After completing this work, we  learned that  
Cubitt and Watson independently  considered the problem of   computing the ground energy density in the TL with a fixed Hamiltonian term \cite{CW21}. 
They study a classical two-threshold version of our GED problem (which they call GSED), and prove 
that the class of languages computable by an exponential time Turing Machine with access to GSED is contained
 in $\text{EXP}^{\text{NEXP}}$ and contains $P^{\text{NEEXP}}$.  Their hardness result thus holds for Turing reductions but not for Karp reductions, as in our reduction; using Turing reductions frees them from the need to employ Krentel's accounting scheme to embed the entire interaction with an oracle in a single instance, which is the primary technical contribution in our work. In their proof, the finite construction embedded in each square of the Robinson tiling  is  a NEEXP computation, as opposed to the $\fexpnexp$ computation implemented in our finite construction. 
 We note that we do not know how to prove Karp reduction results for the two-threshold problem variant. 
 The problem is that in a Karp reduction,
 a polynomial time reduction must take an instance $x$ of a language $L$ and
 create an instance of GSED that encodes whether $x \in L$. The instance of GSED must, by definition, include thresholds and computing the correct threshold for the $n^{th}$ bit of $\alpha_0$ essentially requires computing the first $n-1$ bits of $\alpha_0$.


Importantly, our results are not tight. 
We note that in our hardness results,  since we are reducing from functions $f \in \fpnexp$,
and $f\in\fexpqmaexp$, the verifier $V$ and hence all of the computations we are simulating are classical. 
We still need a quantum construction to execute the clock and create a ground state that is a superposition of the states of the computation at each clock step. 
We conjecture that the finite problem is in fact hard for $\fpqmaexp$, and the infinite one is hard for 
$\fexpqmaexp$. 
We believe that our constructions can be generalized in a straight-forward way to encode
the computation of quantum verifiers and quantum witnesses required for a $\qmaexp$ oracle
instead of the classical ones required for $\nexp$.
The main problem in completing  the quantum  hardness result is that $\qma$ and $\qmaexp$ are classes of
promise problems which means that the Turing Machine can make invalid oracle queries, and there is no guarantee on the responses for invalid queries. Recently, 
Gharibian and Yirka \cite{GY19} and later 
Watson, Bausch, and Gharibian \cite{WBG20} managed to get around the invalid queries issue in the  finite  APX-SIM problem, for quantities other than the ground-energy; however, in the context of measuring the ground energy itself, it is unknown how to do this since invalid queries have an uncontrolled effect on the ground energy.      
%
As written, the constructions in \cite{WBG20}, while TI, also have input-dependent 
Hamiltonian parameters, but this could probably be fixed. 
The issue of invalid queries, on the other hand, appears to be an obstacle for extending these results to the ground energy, even when the Hamiltonian terms are position-dependent as in \cite{GY19, GPY19}, as well as in the finite TI  constructions given here and in \cite{WBG20}.
It remains open to provide tight complexity bounds for computing the ground energy to within a given precision, even in the finite non-TI case. 
Another interesting open problem is to extend the works of \cite{WBG20} to the infinite case, and arrive at {\it tight} complexity results in the TL fixed Hamiltonian case, for measuring quantities other than energy.

We note that several recent works  \cite{C15, Bausch_2020, Bausch_2017,CubittMontanaro}
study a relaxed notion of TI (called {\it semi}-TI in \cite{AharonovZhou}), in which the Hamiltonian is a weighted sum of fixed terms; the weights are given as part of the input. This notion seems less natural for the TL setting where infinitely many weights need to be given in order to specify the problem.

\comment{
\subsection{Discussion and Open Problems} 
\label{sec:discussion}

The results are not tight; 
We believe  in fact, that the finite problem is hard for $\fpqmaexp$
and the infinite problem is hard for $\fexpqmaexp$. 
The complexity classes $\qmaexp$ and $\nexp$ differ in several ways.
The first is that in $\qmaexp$, the verifier is a quantum exponential time algorithm
and the witness provided to the verifier is a quantum state on an exponential number of qubits.
This contrasts with $\nexp$ where the verifier and witness are both classical.
The second important difference is that $\qmaexp$ is a class of promise problems
in which the instances can be classified as {\em valid} or {\em invalid}. Valid instances are further partitioned into {\em yes} and {\em no} instances. On a valid {\em no} instance, the verifier will reject with high probability regardless of the witness that is provided.
On a valid {\em yes} instance, there is a witness that will cause the verifier to accept with high probability. If an instance is invalid, there are no guarantees whatsoever on the behavior of the verifier. By contrast, $\nexp$ is a class of languages, in which every instances is a {\em yes} or {\em no} instance which has a corresponding guaranty on the behavior of the verifier. Our results can be likely extended to apply to verifiers which are Quantum Turing Machines that use quantum witnesses. However, the presence of invalid queries appears to be a more significant obstacle.
 Note that even though all the computations embedded in the Hamiltonian are classical in this construction,
quantumness is still required 
in order to create the computation state which is the superposition of each clock step of an encoded computation.


Interestingly,  the fact that $\qma$ is a class of promise problems instead of languages
was originally overlooked by Ambainis \cite{ambainis2014}
and later clarified by Gharibian and Yirka \cite{GY19}. 
If the oracle is for a promise problem, there are no guarantees on the responses
for invalid queries. This means that there could be more than one string of correct oracle queries, since we only require that responses be correct on valid instances.
Gharibian and Yirka show that for the problems they study, namely decision problems related to local measurements on the ground state
of a Hamiltonian, 
the ground state of their construction must correspond to one of the correct oracle response strings. 
This idea works fine if the quantity being measured is not the actual energy itself.
Note that the important point here is the property that distinguishes the state
to be measured (minimum energy) is different than the local observable
applied to the measured state. By contrast,  our problem is to compute the
energy of the lowest energy state.
In this case, a hardness result  requires encoding the answer to a computational problem into the ground energy itself.
It is unknown how to do this since invalid queries have an uncontrolled effect on the ground energy.  
This appears to be an obstacle, even when the Hamiltonian terms are position-dependent as in the constructions of \cite{GY19, GPY19}, as well as in the translationally-invariant constructions given here and in \cite{WBG20}.
It remains an open question to determine tight complexity bounds for computing the ground energy to within a given level of precision, even in the finite non-translationally-invariant case. 
}

{~}

\noindent{\bf Structure of remainder of paper} 
Section \ref{sec:finite} is devoted to the proof of Theorem \ref{th:finite}.
In the first subsection, we prove the far easier direction of containment, namely Function-TIH is in  $\fexpqmaexp$. The next subsection, Subsection \ref{finitehardnessoverview}, provides a detailed overview of the proof of the hardness
result,
which is the main technical effort in the paper. In this subsection we elaborate on the ideas which were sketched in the above overview of the proof, Subsection \ref{sec:proofoverviewfinite}.
Section \ref{sec:infinite} provides the proof of our main theorem, Theorem \ref{th:infinite},
starting again with a detailed overview of the proof 
in Subsection \ref{sec:hardnessoverviewinf}, which expands on the sketch provided above in Subsection \ref{sec:proofoverviewinf}.

\section{Computing Ground Energies for Finite TI Hamiltonians}\label{sec:finite} 

\subsection{Containment}

Using binary search with queries to a $\qmaexp$ oracle, it is possible in polynomial time
to determine $\lambda_0(H_N)$ to within $\pm 1/q(N)$, for
any polynomial $q(N)$. The algorithm which we present here for completeness is similar to the
binary search algorithm given by Ambainis in \cite{Amb14}.
There is a little subtlety regarding the fact that the queries
are made to a promise problem instead of a language.
The algorithm below shows how an oracle to a language in $\qmaexp$
can be used to compute the ground energy of $H_N$. 
The oracle computes a decision version of the translationally-invariant Hamiltonian problem:

\begin{quote}
{\sc Decision-TIH for $h$ and constant $c$}\\
{\bf Input:} A number $N$ specified with $n$ bits, and $\lambda$, specified to precision $1/2^{cn}$.\\
{\bf Output:} Accept if $\lambda_0(H_N) \le \lambda$ and Reject if $\lambda_0(H_N) \ge \lambda + 1/2^{cn}$.
\end{quote}

 Decision-TIH  is  in $\qmaexp$ for any $h$ and any $c$.
 The proof is a straight-forward generalization of Kitaev's proof that local Hamiltonian is in QMA \cite{KSV02}.
 The algorithm given here  computes the  Function-TIH in $\mbox{FEXP}^{\mbox{QMA-EXP}}$
 with access to an oracle for Decision-TIH. 
 We will assume that $I \cdot c_{offset}$ has been added to $h$
 to ensure that $h \succeq 0$. Since a Hamiltonian with $N$ particles
 has $N-1$ terms, then $(N-1) c_{offset}$ can be subtracted from the final answer to obtain the desired result.
 Note that {\sc Decision-TIH} and {\sc Function-TIH} are both
 parameterized by a constant $c$ governing the precision of the energy estimate. We will require an oracle for
 {\sc Decision-TIH} that has is parameterized by constant $c+1$
 if the desired constant 
 for {\sc Function-TIH} is $c$. \snote{New text in this paragraph, please check.}

\vspace{.2in}

\begin{figure}[ht]
\noindent
\begin{center}
\fbox{\begin{minipage}{\textwidth}
\begin{tabbing}
{\sc Algorithm for Function-TIH for $h$ and $c$}\\
{\sc Input:} Integer $N$.\\
(1)~~~~ \= Initialize $u_0 \leftarrow \norm{h} \cdot N$, $l_0 \leftarrow 0$, and $r \leftarrow (n+1)c + \lceil \log_2(u_0) \rceil$, where $n = |N|$.\\
(2)  \> \FOR $j = 1, \ldots, r$\\
(3)  \> ~~~~~ \=  Send two queries  to Decision-TIH for $h$ and $c+1$ with the following two input $\lambda$'s\\
(4)  \>\>$\lambda_1 = l_{i-1} + u_0/2^{j+1}$\\
(5)  \>\>$\lambda_2 = l_{i-1} + 2u_0/2^{j+1}$\\
(6)  \>\> \IF the result of query $1$ is Accept:\\
(7)  \>\>~~~~~ \=  $u_i \leftarrow l_{i-1} + u_0/2^j$ and $l_i \leftarrow l_{i-1}$\\
(8)  \>\> \IF the results of queries $1$ and $2$ are both Reject:\\
(9)  \>\>\> $l_i \leftarrow l_{i-1} + u_0/2^j$  and $u_i \leftarrow u_{i-1}$\\
(10)  \>\> \IF the result of query $1$ is Reject and the result of query $2$ is Accept:\\
(11)  \>\>\> $u_i \leftarrow  l_{i-1} + 3u_0/2^{j+1}$\\
(12)  \>\>\> $l_i  \leftarrow  l_{i-1} + u_0/2^{j+1}$\\
(13) \> \RETURN $(u_r + l_r)/2$
\end{tabbing}
\end{minipage}}
\end{center}
\caption{Pseudo-code for the Function-GED.}
\label{fig:BinSearch-pseudo}
\end{figure}

\begin{theorem}
Function-TIH $\in \fpqmaexp$ for any fixed $c$ and $h$.
\end{theorem}

\begin{proof}
We will show that
if $n = |N| \ge 2$, then after the final iteration of
the binary search algorithm in Figure \ref{fig:BinSearch-pseudo}, $\lambda_0 \in [l_r,u_r]$ and
$u_r-l_r \le 1/N^c$. 
The number of iterations is $O(\log(N))$. Each iteration makes two oracle queries
and does a constant number of arithmetic operations on number of length $O(\log(N))$.
\dnote{shouldn't we require r=cn to fit definition of decision-TIH?} \snote{Yes, I made the input $N$ a parameter of the algorithm
and initialized $r$ accordingly.}

Let $u_j$ and $l_j$ be the values of $u$ and $l$ after iteration $j$. 
We prove by induction on
$j$, that $\lambda_0 \in [l_j,u_j]$ and
$u_j -l_j = u_0/2^j$. Since the number of iterations $r$ is $(n+1)c + \lceil \log_2(u_0) \rceil$, this will guarantee that $u_r-l_r \le 1/2^{c(n+1)} \le 1/N^c$.
For the base case, notice that $\norm{H_N} \le \norm{h} \cdot N$.
By assumption $h \succeq 0$, so $\lambda_0 \in [0,u_0]$.
\dnote{missing base case j=1..} 

For the inductive step, we will consider three different cases, 
based on the outcomes of the queries.
The call to Decision-TIH for $h$ and $c+1$ has a promise gap of
$1/2^{(c+1)n}$. For $n \ge 2$, this is at most
$1/2^{cn+2}$.
Since $r$ is chosen to be $(n+1)c + \lceil \log_2(u_0) \rceil$,
$1/2^{cn+2} \le 1/2^{c(n+1)+2} \le u_0/2^{r+1} \le u_0/2^{j+1}$. Therefore, the queries to Decision-TIH
made in the $j^{th}$ iteration with threshold $\lambda$
will accept if $\lambda_0 \le \lambda$ and reject if $\lambda_0 \ge u_0/2^{j+1}$.

If the result of the first query is Accept, then  by the definition of Decision-TIH, $\lambda_0 \le \lambda_1 + u_0/2^{j+1} = l_{j-1} + u_0/2^j = u_j$. \dnote{in the first inequality we use $cn\ge j+1$ which comes from the connection of r and cn as in my dnote above} \snote{This is addressed above.}
By the inductive hypothesis, $\lambda_0 \ge l_{j-1} = l_j$.
$u_j - l_j = (l_{j-1} + u_0/2^j) - l_{j-1} = u_0/2^j$.

If the result of the second query is Reject,  then $\lambda_0 \ge \lambda_2 = l_{j-1} + 2u_0/2^{j+1} = l_j$. By the inductive hypothesis, $\lambda_0 \le u_{j-1} = u_{j}$.
Also, by the inductive hypothesis, $u_{j-1} - l_{j-1} = u_0/2^{j-1}$.
Therefore, $u_j - l_j = u_{j-1} - (l_{j-1} + 2u_0/2^{j+1}) = u_0/2^{j-1} - u_0/2^j = u_0/2^j$.

If the result of the first query is Reject and
the second query is Accept, then $\lambda_0 \ge \lambda_1 = l_{j-1} + u_0/2^{j+1} = l_j$. \dnote{I think the last equality is incorrect - the definition of $l_j$ depends on the result of the query on $\Lambda_1$ and so it is either determined by line 9 or 12 in the algorithm}
\snote{I think I fixed this by adding in the assumption that the second query is Accept.}
Since the result of the second query is Accept then, $\lambda_0 \le \lambda_2 + u_0/2^{j+1} = l_{j-1} + 3u_0/2^{j+1} = u_j$. 
Then $u_j - l_j = (l_{j-1} + 3u_0/2^{j+1}) - (l_{j-1} + u_0/2^{j+1}) = u_0/2^j$.
\end{proof}

\subsection{Hardness: Overview of the Construction}
\label{finitehardnessoverview}

We start with an arbitrary $f \in \fpnexp$ and construct a Hamiltonian term $h$ that operates on two $d$-dimensional particles. 
Let $H_N$ denote the Hamiltonian on a chain of $N$ $d$-dimensional particles resulting from applying $h$ to each neighboring pair in the chain. The reduction will map an input string $x$ of length $n$
to a positive integer $N = N(x)$ such that $N(x)$ can be computed in time polynomial in $n$.
We will show that there is a polynomial $q$ such that there is a polynomial time algorithm, which can compute $f(x)$ given a $1/q(n)$-approximation of
$\lambda_0 (H_N)$, namely a value $E$, where $|E - \lambda_0 (H_N)| \le 1/q(n)$.

We will also assume without loss of generality that if the last bit of $x$ is $0$, then
$f(x) = 0$. If $f$ does not have this property, then we can define a new function $f'$
such that $f'(x0) = 0$ and $f'(x1) = f(x)$. Any polynomial time algorithm to compute $f'$
can be used to compute $f$.

Since $f \in \fpnexp$, $f$ can be computed by a polynomial-time Turing Machine $M$
with access to an  oracle for language $L \in \nexp$. The verifier for $L$ is a  Turing Machine $V$ which
runs in time exponential in the size of the query input and uses a witness whose length is also exponential in the size of the query input.
We will assume that there are constants $c_1$ and $c_2$ such that on input $x$ of length $n$,
the number of queries made by $M$ is at most $c_1 n$ and if $M$ makes a query $\bar{x}$
to the oracle, the size of the witness $w$ required is at most $2^{c_2 n}$ and the running time
of $V$ on input $(\bar{x},w)$ is at most $2^{c_2 n}$. This assumption which follows from a standard padding argument, is given in Section \ref{sec:MTV}.


The remainder of this subsection gives an overview of the track structure of the Hilbert space
as well a review of the types of terms used in circuit-to-Hamiltonian constructions.
Subsection \ref{sec:clock} gives a detailed description of the the propagation terms for the clock tracks and analyzes the structure of the Hamiltonian restricted to the space of clock configurations.
Subsection \ref{sec:TMs} gives a detailed description of the two Turing Machines used in the construction and proves that they have the necessary properties required for the construction which embeds the computation in the ground state of a Hamiltonian.
Subsection \ref{sec:Ham} describes all of the components in the Hamiltonian term, including
how the propagation terms for the clock are enhanced to apply the right unitary operations on the computation tracks as well as additional penalty terms used to ensure that the ground state corresponds to a correct computation of the function $f$.
Finally Subsection \ref{sec:analysis} analyzes the ground energy of the resulting Hamiltonian and shows how the value $f(x)$ can be determined given an approximation of the ground energy of $H_N$ for $N = N(x)$.

\subsubsection{The Track Structure of the Hilbert Space}
\label{sec:track}
The Hilbert space of each particle is a tensor product of $6$ different
spaces which will be used to form $6$ tracks. There are also two additional
states $\leftBr$ and $\rightBr$. So the Hilbert space for particle
$i$ is
$$\calh_i = \{ \leftBr, \rightBr \} \oplus \left( \calh_{i,1} \otimes \calh_{i,2}\otimes \calh_{i,3}\otimes \calh_{i,4}\otimes \calh_{i,5}\otimes \calh_{i,6} \right).$$
Later, we will introduce constraints that enforce the condition that
for a low energy state, the leftmost particle must be in state
$\leftBr$, the right particle must be in state $\rightBr$
and the particles in between are not in states $\leftBr$ or $\rightBr$.
We will call all such standard bases states {\em bracketed}.
The space spanned by all bracketed states is $\calh_{br}$
and the final Hamiltonian is invariant on $\calh_{br}$. 
For any bracketed state, 
Track $k$ of the system of $N$ particles consists of the
tensor of $\calh_{i,k}$ as $i$ runs from $2$ to $N-1$. Thus,
the state of the system can be pictured as:

\begin{center}
\begin{small}
\begin{tabular}{|@{}c@{}|@{}c@{}|@{}c@{}|}
\hline
$\leftBr$ &
\begin{tabular}{@{}c@{}|@{}c@{}|clc|@{}c@{}|@{}c@{}}
$\arrRone$ & $\blank$ & $\cdots$ & Track 1: Clock second hand & $\cdots$ & $\blank$ &  $\blank$ \\
\hline
$\arrR$ & $\blankLtwo$ & $\cdots$ & Track 2: Clock minute hand & $\cdots$ & $\blankLtwo$ &  $\blankLtwo$ \\
\hline
$\arrR$ & $\blankLthree$ & $\cdots$  & Track 3: Clock hour hand & $\cdots$ & $\dead$ &  $\dead$ \\
\hline
$\Dblank$ & $q_0$   & $\cdots$ & Track 4: Tape head and state for TM $V$ & $\cdots$ & $\blank$   &  $\blank$   \\
\hline
1 & \# & $\cdots$ & Track 5: Turing Machine Work Tape  & $\cdots$ & \#&  \# \\
\hline
0/1 & 0/1 & $\cdots$ & Track 6: Quantum witness for $V$ & $\cdots$ & 0/1&  0/1\\
\end{tabular}
& $\rightBr$ \\
\hline
\end{tabular}
\end{small}
\end{center}

Tracks $1$ through $3$ are called {\em Clock Tracks}. Tracks $4$ through $6$ s are called {\em Computation Tracks}. The symbols in the figure above represent standard basis states for 
each portion of the Hilbert space of the particle $\calh_{i,j}$ whose meaning will be described later.
The states for the computation tracks will depend on the Turing Machine being simulated. We will use $\Gamma$ to denote the tape alphabet for the TM and $Q$ to denote the set of states. The standard basis for each $\calh_{i,j}$ are:

\begin{center}
\begin{tabular}{ll}
   $\calh_{i,1}$:   &  $\{ \blank, \Dblank\} \cup \{ \arrLi, \arrRi \mid i = 1, \ldots, 8\}$ \\
   $\calh_{i,2}$:  & $\{ \blankLtwo, \DblankLtwo, \arrR, \arrL \}$\\
   $\calh_{i,3}$:  & $\{ \blankLthree, \DblankLthree, \dead, \arrR, \arrL \}$\\
   $\calh_{i,4}$:   &  $\{ \blank, \Dblank\} \cup Q$ \\
   $\calh_{i,5}$:  & $\Gamma$~~~~(which includes $0$, $1$, and $\#$)\\
   $\calh_{i,6}$:  & $\{0, 1\}$
\end{tabular}
\end{center}

An overview of the structure of the clock is given in Subsection \ref{sec:clockOverview}. The details of the terms governing the clock are given in Section \ref{sec:clock}.
The computation tracks will represent the configuration of a Turing Machine. A typical configuration for the computation tracks will look like:
\vspace{.1in}
\begin{center}
\begin{small}
\begin{tabular}{|@{}c@{}|@{}c@{}|@{}c@{}|}
\hline
$\leftBr$ &
\begin{tabular}{@{}c@{}|@{}c@{}|@{}c@{}|@{}c@{}|@{}c@{}|@{}c@{}|@{}c@{}|@{}c@{}|@{}c@{}|@{}c@{}}
$\blank$ &  $\blank$ & $\blank$ & $q_1$ & $\Dblank$ & $\Dblank$ & $\Dblank$ & $\Dblank$ & $\Dblank$ & $\Dblank$ \\
\hline
$\gamma$ &  $\sigma$ & $\sigma$ &  $\gamma$ & $\sigma$ & $\#$ & $\#$ & $\#$ & $\#$ & $\#$ \\
\hline
$1$ &  $0$ & $1$ &  $0$ & $0$ & $0$ & $1$ & $0$ & $1$ & $1$ 
\end{tabular}
& $\rightBr$ \\
\hline
\end{tabular}
\end{small}
\end{center}
\vspace{.1in}
The symbols $\sigma, \gamma \in \Gamma$ are elements in  the tape alphabet of the TM, and $q_1 \in Q$ is a state of the TM.
Track $6$ contains a read-only binary string which stores the guesses for the oracle responses
and witnesses used in verifier computations. An overview of the computations performed on the computation tracks
is given in Subsection \ref{sec:twoTMs}.
The details regarding how the steps of the Turing Machine are encoded in the Hamiltonian are given in Section \ref{sec:implementingTM}.




\subsubsection{Circuit-to-Hamiltonian Preliminaries} 
\label{sec:circToHam}

The Hamiltonian is a sum of two body terms, and 
there are two types of  terms.
Type I terms will have the form
$\ket{ab}\bra{ab}$ where $a$ and $b$ are possible states.
This has the effect of adding an energy penalty to any
state which has a particle
in state $a$ to the immediate left of a particle in state $b$.
In this case, we will refer to $ab$ as an {\em illegal pair}.
We will sometimes consider the restriction of an illegal pair to a set of
tracks. In this case we implicitly mean that the Hamiltonian term acts as the identity
on the remaining (unspecified tracks).
For example the term $\ketbra{\arrL \arrL}{\arrL \arrL}$
applied to Track $2$ means that any pair of neighboring particles in which $\calh_{i,2} \otimes \calh_{i+1,2}$
is in state $\ket{ \arrL \arrL}$ will contribute $+1$ to the energy of the state.
The restriction of an illegal pair to a set of tracks will sometimes be referred to as an {\em illegal pattern}.
We will refer to standard basis states of the whole chain or a subset of the tracks as {\em configurations}. Any configuration which has an illegal pattern or illegal pair is said to be illegal.

We will sometimes use a regular expression to define a subset of standard basis states.
Throughout the construction, we will make use of the following lemma (Lemma 5.2 from
\cite{GI})

\begin{lemma}
\label{lem:regexp}
For any regular expression over the set of particle states in which each
state appears at most once, we can define a set of illegal pairs such that
for any  standard basis state $x$ for the system,  $x$ is  a substring of
a string in the regular set if and only if $x$ does not contains an illegal pair.
\end{lemma}

We will refine this a bit further for bracketed states:

\begin{lemma}
\label{lem:regexpwf}
{\bf [Regular Expressions for Standard Basis States]}
Consider any regular expression over the set of particle states in which each
state appears at most once. If the corresponding regular set only contains
bracketed strings, then we can use illegal pairs to ensure that
for any bracketed  standard basis state $x$ for the system, $x$ is in the regular set if and only if $x$ does not contain any illegal pairs.
\end{lemma}

\begin{proof}
If $x$ is a substring of $y$ and $x$ and $y$ are both bracketed, then $x = y$.
\end{proof}

Type II terms will have the form:
$\frac{1}{2}
(\ketbra{ab}{ab} + \ketbra{cd}{cd} -
\ketbra{ab}{cd} - \ketbra{cd}{ab})$.
These terms enforce that for any eigenstate with zero
energy, if there is a configuration $A$ with two neighboring
particles in states $a$ and $b$, there must be a configuration $B$ with
equal amplitude that is the same as $A$ except that $a$ and $b$ are
replaced by $c$ and $d$.
Even though a Type II term is symmetric, we 
break this symmetry and associate with each such term a direction by choosing and order on the pair of pairs. If we decide that $ab$ precedes $cd$, then 
we denote the term as: $ab \rightarrow cd$.
Type II terms are also referred to as
{\em transition rules}. We will say that configuration
$A$  transitions into configuration $B$ by rule $ab \rightarrow cd$
if $B$ can be
obtained from $A$ by replacing an occurrence of $ab$ with an occurrence of
$cd$. We say that the transition rule
applies to $A$ in the forward direction
and applies to $B$ in the backwards direction.
The sum of all the Type II terms is called $h_{prop}$ or the
{\em propagation Hamiltonian}.
Again, we may specify that a transition rule acts on a particular subset of the
tracks: $ab_{\calt} \rightarrow cd_{\calt}$, where $\calt$ is a subset of the tracks. This implicitly means that the Hamiltonian term acts as the identity on
the remaining tracks.
It will also sometimes be convenient to say that a transition rule applies to
every state except a particular standard particle basis state. In this case, we will
denote the rule by: $a(\neg b) \rightarrow cd$ to indicate that the rule applies to
every $ax$ pair, where $x \neq b$.

\subsection{The Clock Tracks} 
\label{sec:clock}

\subsubsection{Overview of the Clock Tracks}
\label{sec:clockOverview}

In referring to a {\em clock configuration} without specifically restricting to a subset of the tracks, we are referring to a configuration for all three clock tracks.
We start by defining the set of well-formed clock configurations. 
We  will show in Lemma \ref{lem:clockwellformed} that the Hamiltonian is closed on the span of well-formed configurations. 
Furthermore, penalty terms are added so that every
clock configuration that is not well-formed has an energy of at least $1$,
so we can restrict our attention to the subspace spanned by well-formed clock configurations.

\begin{definition}
\label{def:wellformed}
{\bf [Well-formed Clock Configurations]}
Any standard basis state of Tracks $1$ through $3$ is called a {\em clock configuration}.
A clock configuration is said to be {\em well formed} if the configuration has the following form for each track:
\begin{enumerate}
    \item Track $1$: $\leftBr \Dblank^* (\arrRi + \arrLi) ~\blank^* \rightBr$
\item Track $2$: $\leftBr \DblankLtwo^* (\arrR + \arrL) ~ \blankLtwo^* \rightBr$.
\item Track $3$: $\leftBr \DblankLthree^* (\arrR + \arrL) ~ \blankLthree^* \dead^* \rightBr$.
\end{enumerate}
\end{definition}

By Lemma \ref{lem:regexpwf}, we can use Type I constraints to give an energy
penalty to any clock state that is bracketed but not well-formed.

\begin{definition}
{\bf [Hamiltonian Terms to Enforce Well-formed Clock Configurations]}
\label{def:Hwf-cl}
Let $h_{wf-cl}$ denote the Hamiltonian terms with the constraints that give an energy
penalty for any clock state that is bracketed but not well-formed.
\end{definition}

A typical configuration of the clock tracks looks like:
\vspace{.1in}
\begin{center}
\begin{small}
\begin{tabular}{|@{}c@{}|@{}c@{}|@{}c@{}|}
\hline
$\leftBr$ &
\begin{tabular}{@{}c@{}|@{}c@{}|@{}c@{}|@{}c@{}|@{}c@{}|@{}c@{}|@{}c@{}|@{}c@{}|@{}c@{}|@{}c@{}}
$\blank$ &  $\blank$ & $\blank$ & $\arrRthree$ & $\Dblank$ & $\Dblank$ & $\Dblank$ & $\Dblank$ & $\Dblank$ & $\Dblank$ \\
\hline
$\blankLtwo$ &  $\blankLtwo$ & $\blankLtwo$ &  $\blankLtwo$ & $\blankLtwo$ & $\blankLtwo$ & $\blankLtwo$ & $\arrL$ & $\DblankLtwo$ & $\DblankLtwo$ \\
\hline
$\blankLthree$ &  $\arrR$ & $\DblankLthree$ &  $\DblankLthree$ & $\DblankLthree$ & $\DblankLthree$ & $\DblankLthree$ & $\dead$ & $\dead$ & $\dead$ \\
\end{tabular}
& $\rightBr$ \\
\hline
\end{tabular}
\end{small}
\end{center}
\vspace{.1in}
The constraints in $h_{wf-cl}$ ensure that
Track $1$ has exactly one pointer. The transition rules described in this section will cause the Track $1$ pointer to shuttle back and forth between the left and right ends of the chain according to the transition rules. The $\arrR$ pointer moves to the right, and the $\arrL$ pointer moves to the left. 
A sequence of clock configurations in which the Track $1$ pointer starts at one  end of the chain, sweeps to the other end and back will be referred to as a {\em round trip}. The Track $1$ pointers come in $8$ different varieties and are labeled with tags in the range from $1$ through $8$. A Track $1$ pointer labeled $i$ is called an $i$-pointer. There are two $i$-pointers for each $i$: $\arrLi$ and $\arrRi$.
The different types of  Track $1$ pointers act as a means of program control as they trigger different operations on the three computation tracks (Tracks $4$, $5$, and $6$).
Each right-to-left sweep of an $i$ pointer, for $i = 1,2 ,3$, advances a Turing Machine configuration on the computation tracks by exactly one TM step.
The particular TM move that is executed depends on the label of the $1$-pointer.
A left-to-right sweep of an $i$-pointer for $i = 5, 6, 7$, applies the reverse steps of the Turing Machines. The $4$-pointers and $8$-pointers act as the identity on the computation tracks and are used to trigger Type I penalty terms which enforce that the computation in a low energy state has the correct input and output. Figure \ref{fig:cycle} shows a  graphical representation of the schedule of segments within an iteration of Tracks $1$ and $2$.

A maximal sequence of consecutive clock steps in which the Track $1$ pointer (left or right)
is labeled with $i$ is called an {\em $i$-segment}. 
The even numbered 
$i$-pointers make a single round trip, beginning and ending at the  same end of the chain, so even numbered segments last for $2(N-2)$ clock steps. The odd numbered pointers make multiple round trips, where the number of iterations is regulated by a pointer on Track $2$. 
The Track $1$ pointer is like the second hand of a clock while the Track $2$ pointer is like the minute hand.
In each round trip of an even numbered $i$-pointer, the Track $2$ pointer is advanced by one location. 
When the Track $2$ pointer reaches the end of the chain, the $i$-pointer transitions to an $(i+1)$-pointer. Figure \ref{fig:1segment} illustrates a single $1$-segment.
Odd numbered segments last for $(N-2)(2N-5)$ clock steps.
When the  $8$-pointer reaches the end of the chain, it
transitions to a $1$-pointer, forming a cycle of clock configurations for Tracks $1$ and $2$.

In Definition \ref{def:corrconfigs}, we classify configurations for Tracks $1$ and $2$ as either {\em correct} or {\em incorrect} and show that the incorrect clock configurations will never be in the support of the ground state.
There are exactly $p(N) = 4(N-2)(2N-3)$ correct clock configurations for Tracks $1$ and $2$.
We show in Lemma \ref{lem:clockGraph1} that starting with any of these, if we apply the transition rules $p(N)$ times,
we will transition through all the correct clock configurations for Tracks $1$ and $2$
and return to the initial configuration.

The third clock track acts like an hour hand in that the pointer on Track $3$ moves by one location
each time  Tracks $1$ and $2$ complete an iteration (every $p(N)$ clock steps). The Track $3$ pointer shuttles back and forth between the left end of the chain and the left-most $\dead$ particle. 
The number of $\DblankLthree$ or $\blankLthree$ particles in between the left end of the chain and the left-most $\dead$ particle
is the {\em timer length} for Track $3$. 
The transition rules never alter the Track $3$ timer length, so the Hamiltonian is closed on the space of clock configurations in which Track $3$ has timer length $T$ for a particular value $T$.
The timer length regulates the number of clock steps in a full cycle of the clock. 
If the length of the timer is $T$, then Track $3$ transitions through $2T+1$ configurations before repeating. For each Track $3$ configuration, Tracks $1$ and $2$ will complete an iteration lasting $p(N)$ clock
steps. Thus the total number of clock steps in a cycle for Tracks $1$, $2$, and $3$ is $(2T+1) \cdot p(N)$.
In the sample configuration shown above, the Track $3$ timer has length $6$ which will result in
$13$ iterations  for Tracks $1$  and $2$ for every cycle of the entire clock.



We can define a graph structure on the set of well-formed clock configurations, where edges represent transitions:

\begin{definition} {\bf [Configuration Graph]} The vertices of the 
{\em configuration graph} correspond to the well-formed clock configuration for  Tracks $1$, $2$, and $3$. There is a directed edge from configuration $c_1$ to configuration $c_2$, if $c_2$ can be reached from $c_1$ by the application of one transition rule in the forward direction. 
\end{definition} 

We show in Lemma \ref{lem:clockwellformed}
that every well-formed clock configuration has at most one in-coming and at most one out-going edge in the configuration graph. Thus, the configuration graph is a set of disjoint paths and cycles. In Definition \ref{def:correct} we further divided the well-formed clock configurations as  {\em correct} or {\em incorrect} configurations. (In order for a clock configuration for Tracks $1$, $2$, and $3$ to be correct, the configuration for Tracks $1$ and $2$ must be correct along with some additional constraints on Track $3$). 
In Lemma \ref{lem:correctclock} we show that  incorrect clock configurations form short paths that end in a configuration that incurs an energy penalty. Furthermore the correct clock configurations with timer length $T$ form cycles of length
$(2T+1) p(N)$. The final Hamiltonian will be closed on the space of states whose clock configurations are all within a connected component of the configuration graph, so we can reason about each component separately. 

Since correct configurations will form cycles under the transition rules, it will be convenient to define a Time $0$ in the cycle:

\begin{definition}
\label{def:time0}
{\bf [Time $0$]}
\begin{enumerate}
\item A clock configuration is at Time $0$ for Track $2$ if it's Track $2$ state is
$\leftBr  \arrR \blankLtwo^* \rightBr$
\item A clock configuration is at Time $0$ for Tracks $1$ and $2$ if it's Track $1$ state is
$\leftBr  \arrRone \blank^* \rightBr$ and its Track $2$ state is
$\leftBr  \arrR \blankLtwo^* \rightBr$.
\item A clock configuration is at Time $0$ for Tracks $1$, $2$, and $3$ if it's Track $1$ state is
$\leftBr  \arrRone \blank^* \rightBr$, it's Track $2$ state is
$\leftBr  \arrR \blankLtwo^* \rightBr$,
and its Track $3$ state is
$\leftBr  \arrR \blankLthree^* \rightBr$.
\end{enumerate}
\end{definition}

{\bf Clock Track Summary:}
The Hamiltonian is closed on the set of well-formed clock configurations.
Every clock configuration which is not well-formed has an energy penalty that is at least $1$. The well-formed clock configurations form a graph consisting of disjoint paths and cycles. Incorrect clock configurations are in paths of length at most $p(N)$ which end in a configuration with a penalty term.   This will be sufficient to show that the Hilbert space spanned by configurations of the entire chain that have an incorrect clock configuration will have energy at least
$\Omega(1/p(N)^2)$. There are exactly $(2T+1)p(N)$
correct clock configurations whose timer length is $T$ and those configurations form a {\em cycle} of length 
$(2T+1)p(N)$ in the configuration graph. The cycle consists of $2T+1$ {\em iterations}.
In a single iteration, the configurations for Tracks $1$ and $2$ start at Time $0$, transition through all the correct configurations for Tracks $1$ and $2$, and then return to Time $0$. The schedule of $i$-segments and the movement of the pointers on Track $1$ in a single iteration is depicted in Figure \ref{fig:cycle}.

In the remainder of this section, we  describe the propagation terms for the clock
as well as the Type I constraints that achieve these properties. These sum of the  Type I terms given in the remainder of this section will be denoted $h_{cl}$.

\subsubsection{Rules for 1-pointers}

The $1$-pointers make many round trips between the left and the right
ends of the chain. This is accomplished by the following four generic transition terms:
$$\arrRone \blank \rightarrow \Dblank \arrRone
~~~~~~~~~
\arrRone \rightBr \rightarrow \arrLone \rightBr~~~~~~~~~
\Dblank \arrLone \rightarrow \arrLone \blank~~~~~~~~~
\leftBr \arrLone \rightarrow \leftBr \arrRone$$
Note that these rules will be enhanced to trigger certain actions
on the other tracks.
The following sequence shows one round trip:
\begin{center}
    \begin{tabular}{cc}
    $\leftBr \arrRone \blank \blank \blank \rightBr$
    & $\leftBr \Dblank \Dblank \arrLone \blank   \rightBr$\\
    $\leftBr \Dblank \arrRone  \blank \blank \rightBr$
    & $\leftBr \Dblank  \arrLone \blank \blank  \rightBr$\\
    $\leftBr \Dblank \Dblank \arrRone  \blank \rightBr$
    & $\leftBr   \arrLone \blank \blank \blank  \rightBr$\\
    $\leftBr \Dblank \Dblank \Dblank \arrRone  \rightBr$
    & $\leftBr   \arrRone \blank \blank \blank  \rightBr$\\
    $\leftBr \Dblank \Dblank \Dblank \arrLone  \rightBr$
    & 
    \end{tabular}
\end{center}
In general, one round trip of a Track $1$ pointer will take $2(N-2)$ clock steps.

As the $\arrLone$ sweeps from right to left, it advances the counter
on Track 2. Thus, we replace the rule $\Dblank \arrLone \rightarrow \arrLone \blank$
with the following two rules for the left-moving
pointer:
$$\fourcells{\Dblank}{\neg \arrR}{\arrLone}{\generic}
\rightarrow \fourcells{\arrLone}{\neg \arrR}{\blank}{\generic}~~~~~~~~~~~~~~~~~~
\fourcells{\Dblank}{\arrR}{\arrLone}{\blankLtwo}
\rightarrow \fourcells{\arrLone}{\DblankLtwo}{\blank}{\arrR},$$
where $\generic$ represent an arbitrary basis state for the Hilbert space for a particle and track.
We will enforce constraints that a left-moving $1$-pointer
should never be vertically aligned with a right pointer 
on Track $2$, so the first rule above applies only if $\generic$
is not $\arrR$.
We will add a constraint to enforce that
during a $1$-segment,
the arrow on Track $2$ points right, so a right-moving $1$-pointer
advances to the right as long as the state underneath it on Track $2$
is not $\arrL$:
$$\fourcells{\arrRone}{\neg \arrL}{\blank}{\generic}
\rightarrow \fourcells{\Dblank}{\neg \arrL}{\arrRone}{\generic}$$
This rule replaces the generic rule $\arrRone \blank \rightarrow \Dblank \arrRone$
shown above.

The $1$-pointers continue to shuttle back and forth until the
$\arrR$ on Track $2$ reaches the right end of the chain, so
we have the transitions:
$$\threecellsR{\arrRone}{\generic} \rightarrow \threecellsR{\arrLone}{\generic}~~~~~~~~~
\leftBr \arrLone \rightarrow \leftBr \arrRone$$
Where $\generic \neq \arrR$ or $\arrL$.
The $1$-pointer transitions to a left-moving $2$ pointer only
when the $\arrR$ has reached the right end of the chain:
$$\threecellsR{\arrRone}{\arrR} \rightarrow \threecellsR{\arrLtwo}{\arrL}$$
The change on Track $2$ from $\arrR$ to $\arrL$ indicates that the
Track $2$ counter will now be moving from right to left. 
For each round trip of the $1$-pointer, the state on Track $2$
changes as follows:

\begin{center}
    \begin{tabular}{c}
    $\leftBr \arrR \blankLtwo \blankLtwo \blankLtwo \rightBr$\\
    $\leftBr \DblankLtwo \arrR \blankLtwo \blankLtwo \rightBr$\\
    $\leftBr \DblankLtwo \DblankLtwo \arrR \blankLtwo \rightBr$\\
    $\leftBr \DblankLtwo \DblankLtwo \DblankLtwo \arrR \rightBr$
    \end{tabular}
\end{center}

During all of the round trips with the $1$-pointer, Track $2$
should be in a state $\leftBr \DblankLtwo^* \arrR \blankLtwo^* \rightBr$, so
we make illegal any configuration in which $\arrRone$ is
on top of $\arrL$ on Track $2$.
Also, when the right-moving pointer on Track $2$ is advanced by the left-moving $1$-pointer on
Track $1$, the two pointers pass each other. Therefore it will be illegal to have a left-moving $1$-pointer
on top of a right-moving pointer on Track $2$. The following local patterns are illegal and there are no out-going
transitions from these configurations:
\begin{equation}
    \label{eq:hcl1}
    \twocellsvert{\arrRone}{\arrL}~~~~~~~~~~~~~~~~~~~~~~~~
\twocellsvert{\arrLone}{\arrR}
\end{equation}
The two illegal patterns in (\ref{eq:hcl1}) are added as Type I constraints
to $h_{cl}$.

Figure \ref{fig:1pointer} shows a summary of the transition rules 
for the $1$-pointers. Note that some of these will be further refined
when we describe how the right-moving Track $1$ pointers trigger changes
on the computation tracks. The transition rules are labeled TR-$i$ so that
we can refer to them later.

\begin{figure}
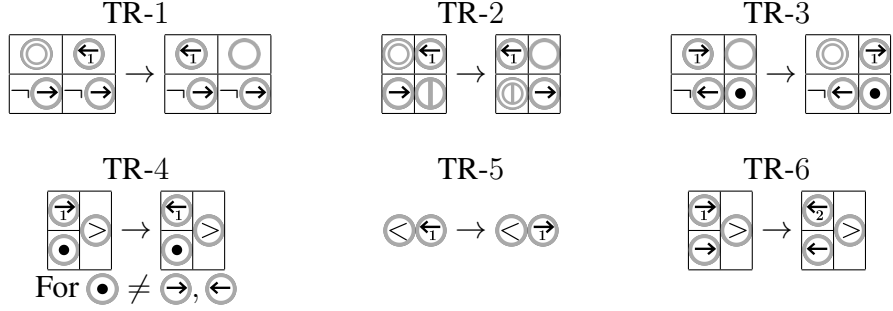

    \centering
    \begin{tabular}{ccc}
    TR-$1$ & TR-$2$ & TR-$3$\\
   ~~~~~ $\fourcells{\Dblank}{\neg \arrR}{\arrLone}{\neg \arrR}
\rightarrow \fourcells{\arrLone}{\neg \arrR}{\blank}{\neg \arrR}$~~~~~ & 
~~~~~$\fourcells{\Dblank}{\arrR}{\arrLone}{\blankLtwo}
\rightarrow \fourcells{\arrLone}{\DblankLtwo}{\blank}{\arrR}$~~~~~ & 
~~~~~$\fourcells{\arrRone}{\neg \arrL}{\blank}{\generic}
\rightarrow \fourcells{\Dblank}{\neg \arrL}{\arrRone}{\generic}$~~~~~\\
& & \\
TR-$4$ & TR-$5$ & TR-$6$\\
$\threecellsR{\arrRone}{\generic} \rightarrow \threecellsR{\arrLone}{\generic}$ & 
$\leftBr \arrLone \rightarrow \leftBr \arrRone$
& $\threecellsR{\arrRone}{\arrR} \rightarrow \threecellsR{\arrLtwo}{\arrL}$\\
For $\generic \neq \arrR, \arrL$ & & 
    \end{tabular}
    \caption{Transition rules involving $1$-pointers.}
    \label{fig:1pointer}
\end{figure}

A $1$-segment begins with the $\arrRone$ pointer at the left end of the chain
and ends with $\arrRone$ at the right end of the chain before it transitions
to $\arrLtwo$. Thus, the $1$-segment lasts for $N-3$ full round trips plus
on additional one-way trip from the left end of the chain to the right end of the chain. Figure \ref{fig:1segment} shows the first and last full
round trips of a $1$-segment as well as the last one-way trip.

\begin{figure}
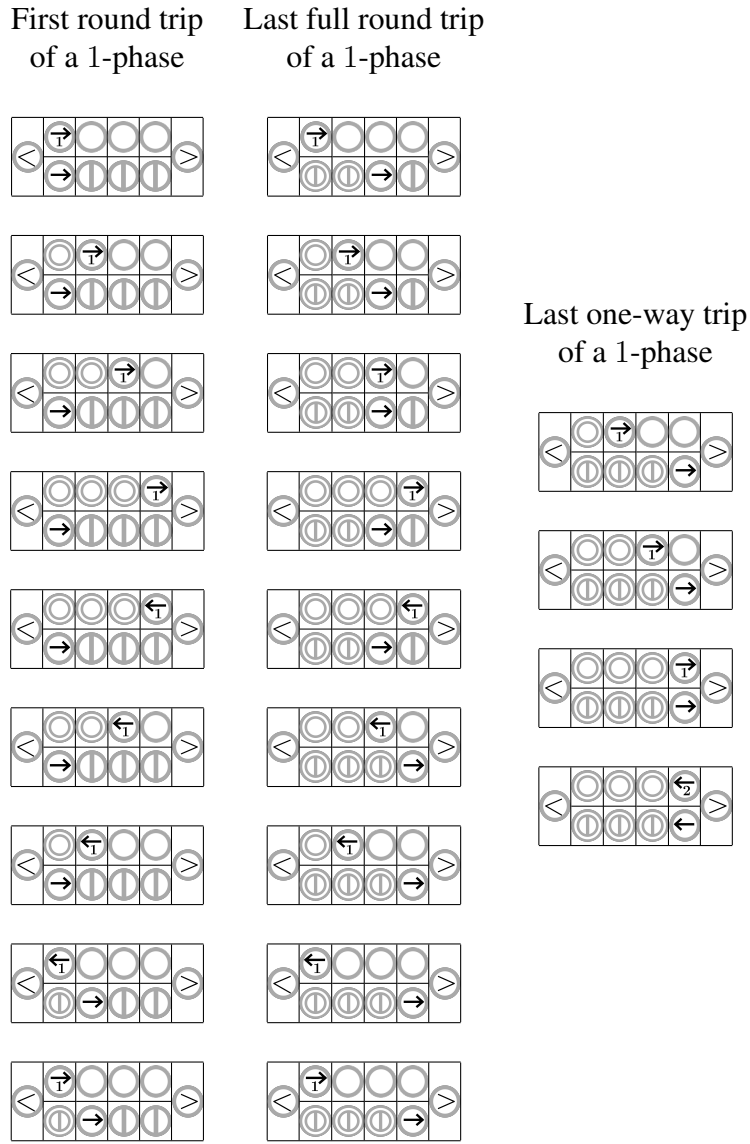

\begin{center}

\begin{tabular}{ccc}

\begin{tabular}{c}
First round trip \\
of a $1$-phase\\
\\
$\eightcells{\arrRone}{\blank}{\blank}{\blank}{\arrR}{\blankLtwo}{\blankLtwo}{\blankLtwo}$ \\  \\
$\eightcells{\Dblank}{\arrRone}{\blank}{\blank}{\arrR}{\blankLtwo}{\blankLtwo}{\blankLtwo}$ \\  \\
$\eightcells{\Dblank}{ \Dblank}{ \arrRone}{ \blank}{\arrR}{\blankLtwo}{\blankLtwo}{\blankLtwo}$ \\  \\
$\eightcells{\Dblank}{ \Dblank}{ \Dblank}{ \arrRone}{\arrR}{\blankLtwo}{\blankLtwo}{\blankLtwo}$ \\  \\
$\eightcells{\Dblank}{ \Dblank}{ \Dblank}{ \arrLone}{\arrR}{\blankLtwo}{\blankLtwo}{\blankLtwo}$ \\  \\
$\eightcells{\Dblank}{ \Dblank}{ \arrLone}{ \blank}{\arrR}{\blankLtwo}{\blankLtwo}{\blankLtwo}$ \\  \\
$\eightcells{\Dblank}{ \arrLone}{ \blank}{ \blank}{\arrR}{\blankLtwo}{\blankLtwo}{\blankLtwo}$ \\  \\
$\eightcells{\arrLone}{ \blank}{ \blank}{ \blank}{\DblankLtwo}{\arrR}{\blankLtwo}{\blankLtwo}$ \\  \\
$\eightcells{\arrRone}{ \blank}{ \blank}{ \blank}{\DblankLtwo}{\arrR}{\blankLtwo}{\blankLtwo}$ \\  \\
\end{tabular}

\begin{tabular}{c}
Last full round trip\\
of a $1$-phase\\
\\
$\eightcells{\arrRone}{\blank}{\blank}{\blank}{\DblankLtwo}{\DblankLtwo}{\arrR}{\blankLtwo}$ \\  \\
$\eightcells{\Dblank}{ \arrRone}{ \blank}{ \blank}{\DblankLtwo}{\DblankLtwo}{\arrR}{\blankLtwo}$ \\  \\
$\eightcells{\Dblank}{\Dblank}{ \arrRone}{ \blank}{\DblankLtwo}{\DblankLtwo}{\arrR}{\blankLtwo}$ \\  \\
$\eightcells{ \Dblank}{\Dblank}{\Dblank}{\arrRone}{\DblankLtwo}{\DblankLtwo}{\arrR}{\blankLtwo}$ \\  \\
$\eightcells{\Dblank}{\Dblank}{\Dblank}{\arrLone}{\DblankLtwo}{\DblankLtwo}{\arrR}{\blankLtwo}$ \\  \\
$\eightcells{\Dblank}{\Dblank}{\arrLone}{\blank}{\DblankLtwo}{\DblankLtwo}{\DblankLtwo}{\arrR}$ \\  \\
$\eightcells{\Dblank}{\arrLone}{\blank}{\blank}{\DblankLtwo}{\DblankLtwo}{\DblankLtwo}{\arrR}$ \\  \\
$\eightcells{\arrLone}{\blank}{\blank}{\blank}{\DblankLtwo}{\DblankLtwo}{\DblankLtwo}{\arrR}$ \\  \\
$\eightcells{\arrRone}{\blank}{\blank}{\blank}{\DblankLtwo}{\DblankLtwo}{\DblankLtwo}{\arrR}$ \\  \\
\end{tabular}

\begin{tabular}{c}
Last one-way trip\\
of a $1$-phase\\
\\
$\eightcells{\Dblank}{\arrRone}{\blank}{\blank}{\DblankLtwo}{\DblankLtwo}{\DblankLtwo}{\arrR}$ \\  \\
$\eightcells{\Dblank}{ \Dblank}{ \arrRone}{ \blank}{\DblankLtwo}{\DblankLtwo}{\DblankLtwo}{\arrR}$ \\  \\
$\eightcells{\Dblank}{ \Dblank}{ \Dblank}{ \arrRone}{\DblankLtwo}{\DblankLtwo}{\DblankLtwo}{\arrR}$ \\  \\
$\eightcells{\Dblank}{ \Dblank}{ \Dblank}{ \arrLtwo}{\DblankLtwo}{\DblankLtwo}{\DblankLtwo}{\arrL}$ \\  \\
\end{tabular}

\end{tabular}
\caption{The first round trip, the last full round trip,
and the last one-way trip of a $1$-phase. The very last transition
illustrates the transition to the next $2$-phase.}
\label{fig:1segment}
\end{center}
\end{figure}

There are a total of $N-3$ full round trips as the $\arrR$ on
Track $2$ moves from the left end of the chain to the right
end of the chain. Since each full round trip takes $2(N-2)$ clock steps,
this accounts for a total of $2(N-2)(N-3)$ clock steps.
Then in the last one-way  trip the $\arrRone$ on Track $1$
sweeps from the left end of the chain to the right end of the
chain ($N-3$ clock steps) and then the $\arrRone$ transitions
to $\arrLtwo$ (one final clock step). The total number of clock steps in
a $1$-segment is:
$$2(N-2)(N-3) + (N-3) + 1 = (N-2)(2N-5).$$

\begin{claim}
\label{cl:everycorrect1}
Starting from Time $0$ for Tracks $1$ and $2$, every combination of clock configurations
for Tracks $1$ and $2$ with a $1$-pointer is reached except those in which
Track $2$ has a left pointer $\arrL$ and those that contain
a $\arrLone$ on Track $1$ vertically aligned with the
$\arrR$ on Track $2$.
\end{claim}

\begin{proof}
Number the non-bracket particles  $1$ through $N-2$.
The $\arrR$ pointer on Track $2$ can be in any of the $N-2$ locations.
Suppose the $\arrR$ pointer on Track $2$ is on location $i$.
The Track $1$ pointer can be right-moving $\arrRone$ in any of the $N-2$ locations
or left-moving 
$\arrLone$ in any of the locations besides $i$. So there should be a total
of $(N-2)[(N-2) + (N-3)] = (N-2)(2N-5)$ configurations.
Starting at Time $0$ for Tracks $1$ and $2$, the Track $1$ pointer is $\arrRone$ in location $1$ and the Track $2$ pointer is $\arrR$  in location $1$.
The Track $1$ pointer sweeps to the right in every possible location
(by TR-$3$). When it hits the right end in location $N-2$ it changes
to $\arrLone$ (by TR-$4$), and then sweeps left to location $2$ (by TR-$1$). 
Every allowable configuration with the Track $2$ pointer in location $1$ is reached.
In the next clock step, the Track $1$ pointer and the Track $2$ pointer swap locations
(by TR-$2$).
In a general round trip, the Track $2$ pointer is in location $i$ and the Track $1$ pointer starts as $\arrLone$ in location $i-1$. The Track $1$ pointer sweeps left
to location $1$, transitions to $\arrRone$, sweeps right to location $N-2$,
transitions back to $\arrLone$ and then sweeps left to location $i+1$. Every allowable configuration with the Track $2$ pointer in location $i$ is reached.
In the next clock step, the Track $1$ pointer and the Track $2$ pointer swap locations
(by TR-$2$) and a new round trip begins. In each round trip, all $2N-5$ configurations
with the Track $2$ pointer location $i$ are reached.
In the last configuration, the Track $1$ pointer $\arrRone$ and Track $2$ pointer $\arrR$ are in location $N-2$. The Track $1$ pointer transitions to $\arrLtwo$
and the Track $2$ pointer transitions to $\arrL$ (by TR-$6$).
\end{proof}

\subsubsection{Rules for 2-pointers}

The $2$-pointers just make a single round trip starting and ending at the right end of the chain:
\begin{center}
    \begin{tabular}{ccc}
    TR-$7$ & TR-$8$ & TR-$9$\\
   ~~~~~ $\Dblank \arrLtwo \rightarrow \arrLtwo \blank$~~~~~ & 
~~~~~$\leftBr \arrLtwo \rightarrow \leftBr \arrRtwo$~~~~~ & 
~~~~~$\arrRtwo \blank \rightarrow \Dblank \arrRtwo$~~~~~
\end{tabular}
\end{center}

 When the right-moving $2$-pointer reaches the right end of the chain, it transitions
 to a left-moving three pointer. At this point the Track $2$ state should be
 $\leftBr \Dblank^* \arrL \rightBr$, so this transition only happens if the Track $2$
 symbol is $\arrR$:
 \begin{center}
    \begin{tabular}{c}
    TR-$10$\\
    $\threecellsR{\arrRtwo}{\arrL} \rightarrow \threecellsR{\arrLthree}{\arrL}$
    \end{tabular}
    \end{center}
    
 Moreover the following patterns is illegal and the Type I constraint that adds
 the energy penalty for this patterns is added to $h_{cl}$:
 \begin{equation}
    \label{eq:hcl2}
 \threecellsR{\arrRtwo}{\neg \arrL}
 \end{equation}

\subsubsection{Rules for 3-pointers}
The $3$-pointers also make many round trips between the left and the right
ends of the chain. We will add a constraint to enforce that
during a $3$-segment
the arrow on Track $2$ points left, so a left-moving $3$-pointer
advances to the left as long as the state underneath it on Track $2$
is not $\arrR$:
$$\fourcells{\Dblank}{\generic}{\arrLthree}{\neg \arrR}
\rightarrow \fourcells{\arrLthree}{\generic}{\blank}{\neg \arrR}$$
As the $\arrRthree$ sweeps from left to right, it advances the counter
on Track 2. Thus, we have the following two rules for the left-moving
pointer:
$$\fourcells{\arrRthree}{\generic}{\blank}{\neg \arrL}
\rightarrow \fourcells{\Dblank}{\generic}{\arrRthree}{\neg \arrL}~~~~~~~~~~~~~~~~~~
\fourcells{\arrRthree}{\DblankLtwo}{\blank}{\arrL}
\rightarrow \fourcells{\Dblank}{\arrL}{\arrRthree}{\blankLtwo}$$
We will enforce constraints that a right-moving $3$-pointer
should never be vertically aligned with a  left-moving pointer
on Track $2$, so the first rule above applies only if $\generic$
is not $\arrL$.

The $3$-pointers continue to shuttle back and forth until the
$\arrL$ on Track $2$ reaches the left end of the chain, so
we have the transitions:
$$\threecellsL{\arrLthree}{\generic} \rightarrow \threecellsL{\arrRthree}{\generic}~~~~~~~~~
 \arrRthree \rightBr \rightarrow \arrLthree \rightBr ,$$
where $\generic$ is any state other than $\arrL$ or $\arrR$.
For each round trip of the $3$-pointer, the state on Track $2$
changes as follows:

\begin{center}
    \begin{tabular}{c}
    $\leftBr  \DblankLtwo \DblankLtwo \DblankLtwo \arrL \rightBr$\\
    $\leftBr \DblankLtwo  \DblankLtwo \arrL \blankLtwo \rightBr$\\
    $\leftBr \DblankLtwo \arrL \blankLtwo \blankLtwo \rightBr$\\
    $\leftBr \arrL \blankLtwo \blankLtwo \blankLtwo  \rightBr$
    \end{tabular}
\end{center}

The $3$-pointer transitions to a right-moving $4$ pointer only
when the $\arrL$ has reached the left end of the chain:
$$\threecellsL{\arrLthree}{\arrL} \rightarrow \threecellsL{\arrRfour}{\arrR}$$
During all of the round trips with the $3$-pointer, Track $2$
should be in a state $\leftBr \DblankLtwo^* \arrL \blankLtwo^* \rightBr$.
Therefore, we make illegal any configuration in which $\arrLthree$ on Track $1$ is
on top of $\arrR$ on Track $2$.
Also, when the left-moving pointer on Track $2$ is advanced by the right-moving $3$-pointer on
Track $1$, the two pointers pass each other. Therefore it will be illegal to have a right-moving $3$-pointer
on top of a right-moving pointer on Track $2$. The following local patterns are illegal and there are no out-going
transitions from these configurations:
 \begin{equation}
    \label{eq:hcl3}
\twocellsvert{\arrLthree}{\arrR}~~~~~~~~~~~~~~~~~~~~~~~~
\twocellsvert{\arrRthree}{\arrL}
\end{equation}

Figure \ref{fig:3pointer} shows a summary of the transition rules 
for the $3$-pointers. Note that some of these will be further refined
when we describe how the right-moving Track $1$ pointers trigger changes
on the computation tracks. 

\begin{figure}
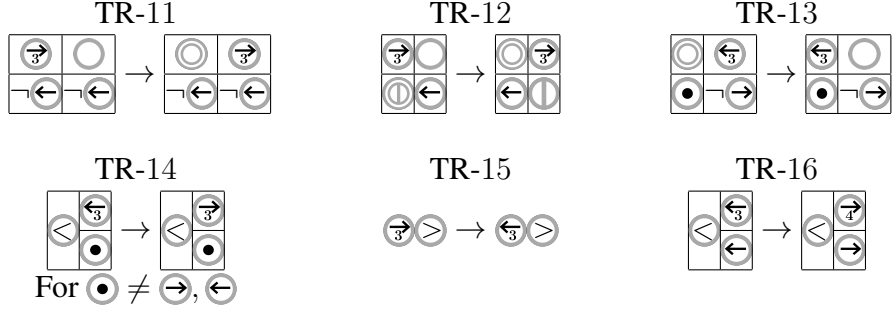

    \centering
    \begin{tabular}{ccc}
    TR-$11$ & TR-$12$ & TR-$13$\\
   ~~~~~ $\fourcells{\arrRthree}{\neg \arrL}{\blank}{\neg \arrL}
\rightarrow \fourcells{\Dblank}{\neg \arrL}{\arrRthree}{\neg \arrL}$~~~~~ & 
~~~~~$\fourcells{\arrRthree}{\DblankLtwo}{\blank}{\arrL}
\rightarrow \fourcells{\Dblank}{\arrL}{\arrRthree}{\blankLtwo}$~~~~~ & 
~~~~~$\fourcells{\Dblank}{\generic}{\arrLthree}{\neg \arrR}
\rightarrow \fourcells{\arrLthree}{\generic}{\blank}{\neg \arrR}$~~~~~\\
& & \\
TR-$14$ & TR-$15$ & TR-$16$\\
$\threecellsL{\arrLthree}{\generic} \rightarrow \threecellsL{\arrRthree}{\generic}$ & 
$\arrRthree \rightBr \rightarrow \arrLthree \rightBr$
& $\threecellsL{\arrLthree}{\arrL} \rightarrow \threecellsL{\arrRfour}{\arrR}$\\
For $\generic \neq \arrR, \arrL$ & & \\
    \end{tabular}
    \caption{Transition rules involving $1$-pointers.}
    \label{fig:3pointer}
\end{figure}

A $3$-segment begins with the $3$-pointer on Track $1$ at the
right end of the chain, does $N-3$ full round trips,
ending with the $3$-pointer at the right end of the chain.
Then it does a one-way trip in which the $3$-pointer
moves to the left end of the chain and then one final clock step in which the $\arrLthree$ on Track
$1$ transitions to $\arrRfour$.
$$\eightcells{\Dblank}{\Dblank}{\Dblank}{\arrLthree}{\DblankLtwo}{\DblankLtwo}{\DblankLtwo}{\arrL} \xrightarrow[\text{steps}]{2(N-2)(N-3)}
\eightcells{\Dblank}{\Dblank}{\Dblank}{\arrLthree}{\arrL}{\blankLtwo}{\blankLtwo}{\blankLtwo}
\xrightarrow[\text{steps}]{(N-3)}
\eightcells{\arrLthree}{\blank}{\blank}{\blank}{\arrL}{\blankLtwo}{\blankLtwo}{\blankLtwo}
\xrightarrow[\text{step}]{\text{one}}
\eightcells{\arrRfour}{\blank}{\blank}{\blank}{\arrR}{\blankLtwo}{\blankLtwo}{\blankLtwo}$$
Therefore,
a $3$-segment lasts for $(N-2)(2N-5)$ clock steps.

\begin{claim}
\label{cl:everycorrect3}
Starting from a configuration in which $\arrLthree$ is on the left end of the chain on Track $1$ and $\arrL$ is on the right end of the chain on Track $2$,  every combination of clock configurations
for Tracks $1$ and $2$ with a $3$-pointer on Track $1$
is reached by applications of the transition rules, except those in which
Track $2$ has a right pointer $\arrR$ and those that contain
a $\arrRthree$ on Track $1$ vertically aligned with the
$\arrL$ on Track $2$.
\end{claim}

\begin{proof}
Number the non-bracket particles  $1$ through $N-2$.
The $\arrL$ pointer on Track $2$ can be in any of the $N-2$ locations.
Suppose the $\arrL$ pointer on Track $2$ is on location $i$.
The Track $1$ pointer can be a left-moving $\arrLthree$ in any of the $N-2$ locations
or a right-moving $\arrRthree$ in any of the locations besides $i$. So there should be a total
of $(N-2)[(N-2) + (N-3)] = (N-2)(2N-5)$ configurations.
Starting at time $N-2$ for Tracks $1$ and $2$, the Track $1$ pointer is $\arrLthree$ in location $N-2$ and the Track $2$ pointer $\arrL$ is in location $N-2$.
The Track $1$ pointer sweeps to the left in every possible location
(by TR-$13$). When it hits the left end in location $1$ it changes
to $\arrRthree$ (by TR-$14$), and then sweeps right to location $N-3$ (by TR-$11$). 
Every allowable configuration with the Track $2$ pointer in location $N-2$ is reached.
In the next clock step, the Track $1$ pointer and the Track $2$ pointer swap locations
(by TR-$12$).
In a general round trip, the Track $2$ pointer is in location $i$ and the Track $1$ pointer starts as $\arrRthree$ in location $i+1$. The Track $1$ pointer sweeps right
to location $N-2$, transitions to $\arrLthree$, sweeps left to location $1$,
transitions back to $\arrRthree$ and then sweeps right to location $i-1$.
Every allowable configuration with the Track $2$ pointer in location $i$ is reached.
In the next clock step, the Track $1$ pointer and the Track $2$ pointer swap locations
(by TR-$12$) and a new round trip begins. In each round trip, all $2N-5$ 
allowable configurations
with the Track $2$ pointer location $i$ are reached.
In the last configuration, the Track $1$ pointer $\arrLthree$ and Track $2$ pointer $\arrL$ are in location $1$. The Track $1$ pointer transitions to $\arrRfour$
and the Track $2$ pointer transitions to $\arrR$ (by TR-$16$).
\end{proof}

\subsubsection{Rules for 4-pointers}

The $4$-pointers just make a single round trip starting and ending at the left end of the chain:
\begin{center}
    \begin{tabular}{ccc}
    TR-$17$ & TR-$18$ & TR-$19$\\
   ~~~~~ $\arrRfour \blank \rightarrow \Dblank \arrRfour$~~~~~ & 
~~~~~$\arrRfour \rightBr \rightarrow \arrLfour \rightBr$~~~~~ & 
~~~~~$\Dblank \arrLfour \rightarrow \arrLfour \blank$~~~~~
\end{tabular}
\end{center}
 When the right-moving $4$-pointer reaches the left end of the chain, it transitions
 to a right-moving five pointer. At this point the Track $2$ state should be
 $\leftBr  \arrR \blank^* \rightBr$, so this transition only happens if the Track $2$
 symbol is $\arrR$:
 \begin{center}
    \begin{tabular}{ccc}
    TR-$20$\\
   ~~~~~ $\threecellsL{\arrLfour}{\arrR} \rightarrow \threecellsL{\arrRfive}{\arrR}$~~~~~ 
\end{tabular}
\end{center}
 Moreover the following pattern is illegal and the Type I constraint that adds
 the energy penalty for this pattern is added to $h_{cl}$:
 \begin{equation}
    \label{eq:hcl4}
 \threecellsL{\arrLfour}{\neg \arrR}
 \end{equation}

\subsubsection{Rules for i-pointers for i = 5, 6, 7, 8}

The $5$-pointers make $N-3$ full round trips plus an additional one-way trip, starting at the left end of the chain and ending at the right end of the chain. Their rules are identical to the rules for the $1$-pointers, except that they eventually transition to a $6$-pointer when the Track $2$ pointer reaches the right end of the chain. TR-$1$ through TR-$5$ are duplicated with
$\arrRone$ and $\arrLone$ replaced with $\arrRfive$ and $\arrLfive$
and named TR-$21$ through TR-$25$. 
\begin{center}
    \begin{tabular}{ccc}
    TR-$26$\\
   ~~~~~ $\threecellsR{\arrRfive}{\arrR} \rightarrow \threecellsR{\arrLsix}{\arrL}$~~~~~ 
\end{tabular}
\end{center}
The Type I terms from (\ref{eq:hcl1}) are duplicated with
$\arrRone$ and $\arrLone$ replaced with $\arrRfive$ and $\arrLfive$ and added to $h_{cl}$.

The $6$-pointers make a single round trip, starting and ending at the right end of the chain.
The rules for these pointers are identical to the $2$-pointers given above, except that they transition to $7$-pointers when they reach the right end of the chain.
TR-$7$ through TR-$9$ are duplicated with
$\arrRtwo$ and $\arrLtwo$ replaced with $\arrRsix$ and $\arrLsix$
and named TR-$27$ through TR-$29$. 
\begin{center}
    \begin{tabular}{ccc}
    TR-$30$\\
   ~~~~~ $\threecellsR{\arrRsix}{\arrL} \rightarrow \threecellsR{\arrLseven}{\arrL}$~~~~~ 
\end{tabular}
\end{center}
The Type I terms from (\ref{eq:hcl2}) are duplicated with
$\arrRtwo$ and $\arrLtwo$ replaced with $\arrRsix$ and $\arrLsix$ and added to $h_{cl}$.
 
 The $7$-pointers make $N-3$ full round trips plus an additional one-way trip, starting at the right end of the chain and ending at the left end of the chain. Their rules are identical to the rules for the $3$-pointers, except that they eventually transition to an $8$-pointer when the Track $2$ pointer reaches the right end of the chain.
 TR-$11$ through TR-$15$ are duplicated with
$\arrRthree$ and $\arrLthree$ replaced with $\arrRseven$ and $\arrLseven$
and named TR-$31$ through TR-$35$. 
\begin{center}
    \begin{tabular}{ccc}
    TR-$36$\\
   ~~~~~ $\threecellsL{\arrLseven}{\arrL} \rightarrow \threecellsL{\arrReight}{\arrR}$~~~~~ 
\end{tabular}
\end{center}
The Type I terms from (\ref{eq:hcl3}) are duplicated with
$\arrRthree$ and $\arrLthree$ replaced with $\arrRseven$ and $\arrLseven$ and added to $h_{cl}$.

The $8$-pointers make a single round trip, starting and ending at the left end of the chain. The $8$-pointers also cause the Track $3$ pointer to advance one location. Therefore these rules
will be replaced in the next section
with rules that involve Track $3$ as well.
We will label these rules with the tag ``temp'' for now.
The rules for these pointers are identical to the $4$-pointers given above, except that they transition to a $1$-pointer at the end of the round trip. 
 TR-$17$ through TR-$19$ are duplicated with
$\arrRfour$ and $\arrLfour$ replaced with $\arrReight$ and $\arrLeight$
and named TR-$37$-temp through TR-$39$-temp. 
\begin{center}
    \begin{tabular}{ccc}
    TR-$40$-temp\\
   ~~~~~ $\threecellsL{\arrLseven}{\arrR} \rightarrow \threecellsL{\arrReight}{\arrR}$~~~~~ 
\end{tabular}
\end{center}
We will wait to re-number these transition rules because 

Figure \ref{fig:segmentsummary} summarizes the segments.
In general, the the different $i$-pointers will trigger one application of a transition function of a Turing Machine
on the computational tracks
as it sweeps from one end of the chain to the other. 
$i$-pointers for $i = 1, 2, 3$ trigger a TM step as they move from right to left. 
$i$-pointers for $i = 5, 6, 7$ trigger a TM step as they move from left to right. 
The labels on the pointers provide a way of controlling which Turing Machine transition function
is executed at which time. 
The $4$-segments and $8$-segments
will be used to check  conditions on the lower
tracks, triggering an energy penalty if any of the conditions are violated.

\begin{figure}
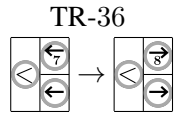

\begin{center}

\begin{tabular}{cccccc}
Segment & Start & End & Track $1$ State & Track $2$ State & Number of Steps\\
\hline
$1$ & Left & Right & $\leftBr \Dblank^* (\arrRone + \arrLone)
\blank^* \rightBr$ & $\leftBr \DblankLtwo^* \arrR \blankLtwo^* 
 \rightBr$ & $(N-2)(2N-5)$\\
$2$ & Right & Right & $\leftBr \Dblank^* (\arrRtwo + \arrLtwo)
\blank^* \rightBr$ & $\leftBr \DblankLtwo^* \arrL  
 \rightBr$ & $2(N-2)$\\
$3$ & Right & Left & $\leftBr \Dblank^* (\arrRthree + \arrLthree)
\blank^* \rightBr$ & $\leftBr \DblankLtwo^* \arrL \blankLtwo^*
 \rightBr$ & $(N-2)(2N-5)$\\
$4$ & Left & Left & $\leftBr \Dblank^* (\arrRfour + \arrLfour)
\blank^* \rightBr$ & $\leftBr \arrR \blankLtwo^* 
 \rightBr$ & $2(N-2)$\\
$5$ & Left & Right & $\leftBr \Dblank^* (\arrRfive + \arrLfive)
\blank^* \rightBr$ & $\leftBr \DblankLtwo^* \arrR \blankLtwo^* 
 \rightBr$ & $(N-2)(2N-5)$\\
$6$ & Right & Right & $\leftBr \Dblank^* (\arrRsix + \arrLsix)
\blank^* \rightBr$ & $\leftBr \DblankLtwo^* \arrL  
 \rightBr$ & $2(N-2)$\\
$7$ & Right & Left & $\leftBr \Dblank^* (\arrRseven + \arrLseven) \blank^*
 \rightBr$ & $\leftBr \DblankLtwo^* \arrL \blankLtwo^*
 \rightBr$ & $(N-2)(2N-5)$\\
$8$ & Left & Left & $\leftBr \Dblank^* (\arrReight + \arrLeight)
\blank^* \rightBr$ & $\leftBr  \arrR  \blankLtwo^*
 \rightBr$ & $2(N-2)$
\end{tabular}
\caption{A summary of the $i$-segments for $i = 0, \ldots, g$.}
\label{fig:segmentsummary}
\end{center}
\end{figure}

\subsubsection{Summary of Track 1 and 2 Clock Configurations}

Throughout the construction, we will need to maintain
that the transition rules are well defined for well-formed
clock configurations. For now, we prove this fact for
the transition rules that apply to Tracks $1$ and $2$.

\begin{lemma}
\label{lem:rulebound}
{\bf [Bound on Reachable Configurations in One Clock Step]}
For every  well formed clock configuration for Tracks $1$ and $2$,
$A$, there is at most one configuration that can be reached from $A$ by a transition rule applied in the forward direction and at most one configuration that can be reached from $A$ by a transition rule applied in the reverse direction
\end{lemma}

\begin{proof}
Every transition rule TR-$1$ through Tr-$40$ applies in the forward direction to a two-particle configuration in which the Track $1$ state is one of the following four cases:
$$\leftBr \arrLi~~~~~~~~~~~~~~~~~~
\Dblank \arrLi~~~~~~~~~~~~~~~~~~
\arrRi \blank~~~~~~~~~~~~~~~~~~
\arrRi \rightBr$$
If the clock configuration is well-formed then the Track $1$ state has
the form $\leftBr \Dblank^* (\arrRi + \arrLi) \blank^* \rightBr$.
In any clock configuration of that form, there is exactly one occurrence of any of the four patterns shown above. The additional constraints on the other tracks can only decrease the number of rules that apply in the forward direction to a given
Track $1$ configuration. Therefore, there is at most one location and at most one rule that can be applied in the forward direction to get a new clock configuration.

Every transition rule TR-$1$ through TR-$40$ applies in the backward direction to a two-particle configuration in which the Track $1$ state is one of the following four cases:
$$\leftBr \arrRi~~~~~~~~~~~~~~~~~~
\Dblank \arrRi~~~~~~~~~~~~~~~~~~
\arrLi \blank~~~~~~~~~~~~~~~~~~
\arrLi \rightBr$$
In any well-formed clock configuration, there is exactly one occurrence of these four patterns
on Track $1$. The additional constraints on the other tracks can only decrease the number of rules that apply in the reverse direction to a given Track $1$ configuration. Therefore, there is at most one location and at most one rule that can be applied in the reverse direction to get a new clock configuration.

\end{proof}

Eventually, we will define correct and incorrect clock configurations for all three tracks. However, it will be useful to first define properties required of the configurations for Tracks $1$ and $2$ and establish some properties about how they are connected by the transition rules.

\begin{definition}
\label{def:corrconfigs}
{\bf [Correct and Incorrect Clock Configurations for Tracks $1$ and $2$]}
A well formed clock configuration for Tracks $1$ ad $2$ is called {\em correct} if the following conditions hold
and is called {\em incorrect} otherwise:
\begin{enumerate} 
\item If the pointer on Track $1$ is an $i$-pointer for $i \in \{ 4, 8\}$, then Track $2$ is at Time $0$.
\item If the pointer on Track $1$ is an $i$-pointer for $i \in \{ 2, 6\}$, then Track $2$ is is in configuration $\leftBr \DblankLtwo^* \arrL \rightBr$.
\item If the pointer on Track $1$ is an $i$-pointer for $i \in \{ 1, 5\}$, then the pointer on Track $2$ is $\arrR$.
\item If the pointer on Track $1$ is an $i$-pointer for $i \in \{ 3, 7\}$, then the pointer on Track $2$ is $\arrL$.
\item If the pointer on Track $1$ is an $i$-pointer for $i = 1, 3, 5, 7$, then it is not vertically aligned with the pointer on Track $2$ and pointing on the opposite direction. 
\end{enumerate}
\end{definition}

\begin{lemma}
\label{lem:clockGraph1}
{\bf [Correct Configurations Form Cycles]}
Starting from the Time $0$ configuration for Tracks $1$ and $2$, the transition rules will
reach every correct configuration for Tracks $1$ and $2$ exactly once before returning to the Time $0$ configuration after
$p(N) = 4(N-2)(2N-3)$ clock steps.
\end{lemma}

\begin{proof}
Claim \ref{cl:everycorrect1} shows that starting from Time $0$
and applying the transition rules $(N-2)(2N-5)$ times, every correct
clock configuration with a $1$-pointer is reached, ending in the configuration:
$$\eightcells{\Dblank}{~~~~\cdots~~~~}{\Dblank}{\arrRone}{\DblankLtwo}{~~~~\cdots~~~~}{\DblankLtwo}{\arrR} \xrightarrow{\text{TR-}6}
\eightcells{\Dblank}{~~~~\cdots~~~~}{\Dblank}{\arrLtwo}{\DblankLtwo}{~~~~\cdots~~~~}{\DblankLtwo}{\arrL} $$
The rules for the $2$-pointers do not alter Track $2$. The $2$-pointer moves to the left and then to the right, reaching every location  with $\arrLtwo$ and then every location with $\arrRtwo$ until it reaches configuration:
$$\eightcells{\Dblank}{~~~~\cdots~~~~}{\Dblank}{\arrRtwo}{\DblankLtwo}{~~~~\cdots~~~~}{\DblankLtwo}{\arrL} \xrightarrow{\text{TR-}10}
\eightcells{\Dblank}{~~~~\cdots~~~~}{\Dblank}{\arrLthree}{\DblankLtwo}{~~~~\cdots~~~~}{\DblankLtwo}{\arrL} $$
Then by Claim \ref{cl:everycorrect3}, every correct configuration with a $3$-pointer is reached, ending in configuration:
$$\eightcells{\arrLthree}{\blank}{~~~~\cdots~~~~}{\blank}{\arrL}{\blankLtwo}{~~~~\cdots~~~~}{\blankLtwo} \xrightarrow{\text{TR-}16}
\eightcells{\arrRfour}{~~~~\cdots~~~~}{\blank}{\blank}{\arrR}{~~~~\cdots~~~~}{\blankLtwo}{\blankLtwo} $$
The rules for the $4$-pointers do not alter Track $2$. The $4$-pointer moves to the right and  then to the left, reaching every location  with $\arrRfour$ and then every location with $\arrLfour$ until it reaches configuration:
$$\eightcells{\arrLfour}{\blank}{~~~~\cdots~~~~}{\blank}{\arrR}{\blankLtwo}{~~~~\cdots~~~~}{\blankLtwo} \xrightarrow{\text{TR-}20}
\eightcells{\arrRfive}{\blank}{~~~~\cdots~~~~}{\blank}{\arrR}{\blankLtwo}{~~~~\cdots~~~~}{\blankLtwo} $$
Since the rules for the $5$, $6$, $7$, and $8$-pointers are the same as the
$1$, $2$, $3$, and $4$-pointers, the same argument holds for those segments. In the final clock step, the $8$-pointer transitions back to the $1$-pointer which puts the clock configuration back at time $0$:
$$\eightcells{\arrLeight}{\blank}{~~~~\cdots~~~~}{\blank}{\arrR}{\blankLtwo}{~~~~\cdots~~~~}{\blankLtwo} \rightarrow
\eightcells{\arrRone}{\blank}{~~~~\cdots~~~~}{\blank}{\arrR}{\blankLtwo}{~~~~\cdots~~~~}{\blankLtwo} $$
The $1$, $3$, $5$, and $7$-segments each take $(N-2)(2N-5)$ clock  steps. 
The $2$, $4$, $6$, and $8$-segments each take $2(N-2)$ clock steps. The total number of clock steps (and therefore the number of correct clock configurations for Tracks $1$ and $2$) is:
$$p(N) = 4[(N-2)(2N-5)] + 4[2(N-2)] = 4(N-2)(2N-3).$$
\end{proof}

\begin{lemma}
\label{lem:clockGraph2}
{\bf [Incorrect Configurations form Short Paths]}
Any  incorrect clock configuration for Tracks $1$ and $2$ will reach an illegal configuration 
from $h_{cl}$ in at most $2(N-2)$ clock steps from which there is no out-going transition.
\end{lemma}

\begin{proof}
If the Track $1$ pointer is a $1$-pointer or a $5$-pointer, and the clock configuration is incorrect, then either the Track $2$ pointer is $\arrR$ and is vertically aligned with $\arrLone$ or $\arrLfive$
or the Track $2$  pointer is $\arrL$ (in any location).
In the former case, the current configuration is illegal and
there is no outgoing transition, since TR-$1$ (TR-$21$)
applies only if the Track $2$
symbol under the $\arrLone$ ($\arrLfive$) is not $\arrR$.
In the latter case, the $1$-pointer (or $5$-pointer)
will keep moving until it is in a state $\arrRone$ (or $\arrRfive$) directly over the $\arrL$ on Track $2$.  This is an illegal configuration and there are no outgoing transitions since TR-$3$ and TR-$4$ (TR-$23$ and TR-$24$) do not apply if the Track $2$ pointer is in state $\arrL$. The number of clock steps until the illegal configuration is reached is at most $2(N-2)$, the time for a round trip for the Track $1$ pointer.

The symmetric argument holds if the Track $1$ pointer is a $3$ or $7$-pointer, with the role of left and right reversed.

The $4$-pointers and $8$-pointers make a single round trip starting and ending at the left end of the chain and do not cause any changes on Track $2$.
Eventually, the left end of the chain will be reached. If Track $2$ does not have $\arrR$ in the left end of the chain, the following illegal configuration will be reached within $2(N-2)$ clock steps, from which there is no outgoing transition:
$$\threecellsL{\arrLfour}{\neg \arrR}~~~~~~~~~~~\threecellsL{\arrLeight}{\neg \arrR}$$

The $2$-pointers and $6$-pointers make a single round trip starting and ending at the right end of the chain and do not cause any changes on Track $2$.
Eventually, the right end of the chain will be reached. If Track $2$ does not have $\arrL$ in the left end of the chain, the following illegal configuration will be reached within $2(N-2)$ clock steps, from which there is no outgoing transition:
$$\threecellsR{\arrRtwo}{\neg \arrL}~~~~~~~~~~~\threecellsR{\arrRsix}{\neg \arrL}$$
\end{proof}

\subsubsection{Track 3}

In the absence of any transition rules involving Track $3$,
a sequence of clock configurations in which Tracks $1$ and $2$
begin and end at Time $0$ form an iteration of length
$4(N-2)(2N-3)$ in the configuration graph. 
We now add transition rules which cause Track $3$ to change
every such iteration. This will result in the 
clock configurations for Tracks $1$ through $3$
forming a larger cycle in the configuration graph.
We will refer to a sequence of clock steps in which
the configurations on Tracks $1$ and $2$
begin and end at the same state as an {\em iteration}.

In a well-formed clock configuration, the configuration of Track $3$ has the form:
$$\leftBr \DblankLthree^* ( \arrR + \arrL) \blankLthree^* \dead^* \rightBr.$$
The transition rules will never change the number of $\dead$ particles
at the right end of the chain on Track $3$. The $\arrR$ and $\arrL$
pointers will move back and forth between the 
locations with $\DblankLthree$ or $\blankLthree$.
The {\em length} of the clock on Track 3 is the total number of
$\DblankLthree$ or $\blankLthree$ symbols and the transition
rules will leave the length fixed.
The only pointers from Track $1$
that interact with Track $3$ are the $8$-pointers.
Each sweep of the $8$-pointer will update the state of the
Track $3$ clock by moving the pointer over by one location or reversing the direction of the Track $3$ pointer.

If the Track $3$ pointer is facing right, it will advance as the $8$-pointer on
Track $1$ moves left. 
If the Track $3$ pointer is facing left, it will advance as the $8$-pointer on
Track $1$ moves right. 
\begin{center}
    \begin{tabular}{cc}
    TR-$37$ & TR-$38$ \\
   ~~~~~ $\sixcells{\Dblank}{\arrLeight}{\generic}{\generic}{\arrR}{\blankLthree}
\rightarrow \sixcells{\arrLeight}{\blank}{\generic}{\generic}{\DblankLthree}{\arrR}$~~~~~ & 
~~~~~$\sixcells{\arrReight}{\blank}{\generic}{\generic}{\DblankLthree}{\arrL} 
\rightarrow \sixcells{\Dblank}{\arrReight}{\generic}{\generic}{\arrL}{\blankLthree}$~~~~~ 
\end{tabular}
\end{center}
The arrow on Track $3$ changes direction when it reaches the end of the chain on
the left end or the $\dead$ particles on the right end. The $8$-pointer that points
in the same direction as the Track $3$ pointer will trigger the change in direction. Recall that the $8$-pointer also checks that the Track $2$ state is
at Time $0$, meaning that the pointer on Track $2$ is $\arrR$ and at the left end of the chain. Therefore, TR-$41$ applies only if the Track $2$ state
is $\arrR$.
\begin{center}
    \begin{tabular}{ccc}
    TR-$39$ & TR-$40$ & TR-$41$ \\
   ~~~~~ $\sixcells{\arrReight}{\blank}{\generic}{\generic}{\arrR}{\dead}
\rightarrow \sixcells{\Dblank}{\arrReight}{\generic}{\generic}{\arrL}{\dead}$~~~~~ & 
~~~~~$\fourcellsR{\arrReight}{\generic}{\arrR} 
\rightarrow \fourcellsR{\arrLeight}{\generic}{\arrL}$~~~~~ & 
~~~~~$\fourcellsL{\arrLeight}{\arrR}{\arrL} 
\rightarrow \fourcellsL{\arrRone}{\arrR}{\arrR}$~~~~~ 
\end{tabular}
\end{center}
Other than the configurations described in the four rules above, the $8$-pointers do not
trigger any other change on Track $3$. So in the absence of a pointer on Track $3$, the $8$-pointer just moves along:
\begin{center}
\begin{tabular}{c}
    \begin{tabular}{cc}
    TR-$42$ & TR-$43$  \\
   ~~~~~ $\sixcells{\arrReight}{\blank}{\generic}{\generic}{\generic}{\generic}
\rightarrow \sixcells{\Dblank}{\arrReight}{\generic}{\generic}{\generic}{\generic}$~~~~~ & 
~~~~~$\sixcells{\Dblank}{\arrLeight}{\generic}{\generic}{\neg \arrR}{\neg \arrR}
\rightarrow \sixcells{\arrLeight}{\blank}{\generic}{\generic}{\neg \arrR}{\neg \arrR}$~~~~~ 
\end{tabular}\\
\\
TR-$42$ applies when the bottom row $\generic \generic$ is in\\
$\{\blankLthree \dead, ~\blankLthree \blankLthree, ~\DblankLthree \DblankLthree, ~\arrR \blankLthree,
~\DblankLthree \arrR, ~\dead \dead \}$
\end{tabular}
    \begin{tabular}{ccc}
    TR-$44$ & TR-$45$  & TR-$46$ \\
~~~~~$\fourcellsR{\arrReight}{\generic}{\blank} 
\rightarrow \fourcellsR{\arrLeight}{\generic}{\blank}$~~~~~ & 
~~~~~$\fourcellsR{\arrReight}{\generic}{\dead} 
\rightarrow \fourcellsR{\arrLeight}{\generic}{\dead}$~~~~~ & 
~~~~~$\fourcellsL{\arrLeight}{\arrR}{\DblankLthree} 
\rightarrow \fourcellsL{\arrRone}{\arrR}{\DblankLthree}$~~~~~ 
\end{tabular}
\end{center}
It should never happen that the Track $1$ pointer is an $8$-pointer and is vertically aligned with the Track $3$
pointer, pointing in the opposite direction. It should also never happen that the $8$-pointer reaches the left
end of the chain and Track $2$ is not at Time $0$.
So the following three patterns are illegal:
\begin{equation}
\label{eq:hcl5}
    \threecellsvert{\arrLeight}{\generic}{\arrR}
~~~~~~~~~~~~~~~~~~~~~~~~\threecellsvert{\arrReight}{\generic}{\arrL}
~~~~~~~~~~~~~~~~~~~~~~~~\fourcellsL{\arrLeight}{\neg \arrR}{\generic}
\end{equation}
There should are no outgoing transitions from a configuration
that contains on of those patterns.
Also,  there is no out-going transition from the pattern below, so it is also added to $h_{cl}$ as an illegal pattern:
\begin{equation}
\label{eq:hcl6}
   \sixcells{\Dblank}{\arrLeight}{\generic}{\generic}{\arrR}{\dead}
\end{equation}

Now that all the transition rules that pertain to the clock tracks are complete, we are ready to analyze the structure of the configuration graph. The first lemma establishes that the configuration graph is well-defined since the set of transition rules is closed on the set of well-formed configuration. The second part of the lemma establishes a degree bound which implies that the configuration graph consists of a disjoint set of paths and cycles.

\begin{lemma}
\label{lem:clockwellformed}
{\bf [Degree Bound for the Configuration Graph]}
The transition rules are closed on the set of well-formed clock configurations. Furthermore, for every well-formed clock configuration, there is at most one transition rule that applies in the forward direction and at most one transition rule that applies in the reverse direction.
\end{lemma}

\begin{proof} 
The transition rules all preserve the number of pointers on each Track.
They also preserve the property that on Track $1$ $\Dblank$ particles are to the left of the pointer and $\blank$ particles are to the right of the pointer. Similarly, the transition rules maintain the property that on Track $2$, $\DblankLtwo$ particles are to the left of the pointer and $\blankLtwo$ particles are to the right of the pointer.
On Track $3$, the transition rules maintain the property that 
$\DblankLthree$ particles are to the left of the pointer and $\blankLthree$ particles are to the right of the pointer as well as the fact that $\blankLthree$ particles are to the left of $\dead$ particles.
Therefore when a transition rule is applied  to a well-formed clock configuration,
the result is another well-formed clock configuration.

Lemma \ref{lem:rulebound} shows that at most one configuration for Tracks $1$ and $2$ can be reached in the forward direction or the reverse direction from a well-formed clock configuration for Tracks $1$ and $2$. The additional constraints introduced by Track $3$ can only further limit the applicable transition rules. 
\end{proof}

We are ready now to describe the cycle of configurations reached by
starting in the Time $0$ configuration for Tracks $1$, $2$, and $3$ and repeatedly applying the transition rules.
If  $T$ is the length of the timer on Track $3$,
then the Track $3$ state at Time $0$ is $\leftBr \arrR \blank^T \dead^* \rightBr$.
There is one round trip of the $8$-pointers for every Track 1-2 iteration.
After $T$ round trips of the $8$-pointer, the arrow on Track $3$ has advanced
to the right end of the non-$\dead$ states (as shown below).
In the $T+1^{st}$ round trip of the $8$-pointer, the pointer on Track $3$ reverses directions but stays in the same place.
After $T-1$ more round trips, the pointer on Track $3$ is one place away from the left end of the chain.
Finally, in the next round trip, the Track $3$ pointer moves one position to the left and changes directions. The number above each arrow below 
shows the number of $8$-segments (which happen every $p(N)$ clock steps) to get from one Track $3$ configuration to the next.
$$\leftBr \arrR \blank^T \dead^* \rightBr  \xrightarrow{T} \leftBr  \Dblank^T \arrR \dead^* \rightBr
\xrightarrow{1} \leftBr  \Dblank^T \arrL \dead^* \rightBr
\xrightarrow{T-1} \leftBr \Dblank \arrL \blank^{T-1} \dead^* \rightBr \xrightarrow{1}
\leftBr  \arrR \blank^{T} \dead^* \rightBr
$$
Figure \ref{fig:middle8segment} shows the middle $8$-segment in which
the Track $3$ pointer stays in the same location and changes direction.
Figure \ref{fig:last8segment} shows the last $8$-segment in which
the Track $3$ pointer advances by one location and changes direction.
\begin{figure}
\begin{center}

\begin{tabular}{ccc}

\begin{tabular}{cc}
1) & $\eightcells{\arrReight}{\blank}{\blank}{\blank}{\DblankLthree}{\DblankLthree}{\arrR}{\dead}$~~~~ \\  \\
2) & $\eightcells{\Dblank}{\arrReight}{\blank}{\blank}{\DblankLthree}{\DblankLthree}{\arrR}{\dead}$ \\  \\
3) & $\eightcells{\Dblank}{\Dblank}{\arrReight}{\blank}{\DblankLthree}{\DblankLthree}{\arrR}{\dead}$ 
\end{tabular}

\begin{tabular}{cc}
4) & $\eightcells{\Dblank}{\Dblank}{\Dblank}{\arrReight}{\DblankLthree}{\DblankLthree}{\arrL}{\dead}$ ~~~~\\  \\
5) & $\eightcells{\Dblank}{\Dblank}{\Dblank}{\arrLeight}{\DblankLthree}{\DblankLthree}{\arrL}{\dead}$ \\  \\
6) & $\eightcells{\Dblank}{\Dblank}{\arrLeight}{\blank}{\DblankLthree}{\DblankLthree}{\arrL}{\dead}$ 
\end{tabular}

\begin{tabular}{cc}
7) & $\eightcells{\Dblank}{\arrLeight}{\blank}{\blank}{\DblankLthree}{\DblankLthree}{\arrL}{\dead}$ ~~~~\\  \\
8) & $\eightcells{\arrLeight}{\blank}{\blank}{\blank}{\DblankLthree}{\DblankLthree}{\arrL}{\dead}$ \\  \\
9) & $\eightcells{\arrRone}{\blank}{\blank}{\blank}{\DblankLthree}{\DblankLthree}{\arrL}{\dead}$ 
\end{tabular}

\end{tabular}
\caption{The middle $8$-segment of a complete cycle in which
the Track $3$ pointer changes from right-moving to left-moving.}
\label{fig:middle8segment}
\end{center}
\end{figure}
\begin{figure}
\begin{center}

\begin{tabular}{ccc}

\begin{tabular}{cc}
1) & $\eightcells{\arrReight}{\blank}{\blank}{\blank}{\DblankLthree}{\arrL}{\blankLthree}{\dead}$~~~~ \\  \\
2) & $\eightcells{\Dblank}{\arrReight}{\blank}{\blank}{\arrL}{\blankLthree}{\blankLthree}{\dead}$ \\  \\
3) & $\eightcells{\Dblank}{\Dblank}{\arrReight}{\blank}{\arrL}{\blankLthree}{\blankLthree}{\dead}$ 
\end{tabular}

\begin{tabular}{cc}
4) & $\eightcells{\Dblank}{\Dblank}{\Dblank}{\arrReight}{\arrL}{\blankLthree}{\blankLthree}{\dead}$ ~~~~\\  \\
5) & $\eightcells{\Dblank}{\Dblank}{\Dblank}{\arrLeight}{\arrL}{\blankLthree}{\blankLthree}{\dead}$ \\  \\
6) & $\eightcells{\Dblank}{\Dblank}{\arrLeight}{\blank}{\arrL}{\blankLthree}{\blankLthree}{\dead}$ 
\end{tabular}

\begin{tabular}{cc}
7) & $\eightcells{\Dblank}{\arrLeight}{\blank}{\blank}{\arrL}{\blankLthree}{\blankLthree}{\dead}$ ~~~~\\  \\
8) & $\eightcells{\arrLeight}{\blank}{\blank}{\blank}{\arrL}{\blankLthree}{\blankLthree}{\dead}$ \\  \\
9) & $\eightcells{\arrRone}{\blank}{\blank}{\blank}{\arrR}{\blankLthree}{\blankLthree}{\dead}$ 
\end{tabular}

\end{tabular}
\caption{The last $8$-segment of a complete cycle in which
the Track $3$ pointer advances by one location and changes direction. Tracks $1$ and $3$ are shown.}
\label{fig:last8segment}
\end{center}
\end{figure}
The number of $8$-segments until the clock configuration is again at Time $0$ for Tracks $1$, $2$, and $3$
is $2T+1$.
Each Track 1-2 iteration takes $p(N) = 4(N-2)(2N-3)$ clock steps, so starting at Time $0$,
the total number of clock steps for the clock configuration to return
to Time $0$ for Tracks $1$, $2$, and $3$ is $(2T+1)p(N)$.

The definition below partitions the set of well-formed clock configurations into
a {\em correct} set and an {\em incorrect} set. 
We will show that the set of $(2T+1)p(N)$ clock configurations with timer length $T$ that are reached in a cycle starting and ending at Time $0$ will be
exactly the set of correct configurations with timer length $T$.
The definition provides a characterization of the incorrect configurations which will allow us to ensure that each of the incorrect configurations are contained
in a short path ending in an illegal configuration
which incurs a penalty from $h_{cl}$. Since $H$ will be closed on each connected component of configurations, we can analyze each component separately
and lower bound the energy of the states
with incorrect clock configurations.

\begin{definition}
\label{def:correct}
{\bf [Correct Clock Configurations (for Tracks $1$, $2$, and $3$)]}
A well formed clock configuration for Tracks $1$, $2$ and $3$ is called {\em correct} if the following conditions hold
and is called {\em incorrect} otherwise:
\begin{enumerate} 
\item The configuration for Tracks $1$ and $2$ is correct.
\item If the pointer on Track $1$ is an $8$-pointer, then it is not vertically aligned with the pointer on Track $3$ pointing on the opposite direction. 
\item If the pointer on Track $3$ is $\arrL$ and in the leftmost location, then the Track $1$ pointer is an $8$-pointer.
\item If the pointer on Track $3$ is pointing right and next to a dead state ($\arrR \dead)$
and the Track $1$ pointer is an $8$-pointer, then the $8$-pointer is not vertically aligned with a $\dead$ state on Track $3$.
\end{enumerate}
\end{definition}

The first step is to count the number of correct configurations whose time length is $T$.

\begin{lemma}
\label{lem:countcorrect}
{\bf [Number of Correct Clock Configurations]}
The number of correct clock configurations with timer length $T$ is $(2T+1)p(N)$.
\end{lemma}

\begin{proof}
By Lemma \ref{lem:clockGraph1}, there are $p(N)$ correct clock configurations for Tracks $1$ and $2$.
There  are $2(T+1)$ possible
Track $3$ configurations in a well-formed clock state:
$$\leftBr \DblankLthree^j \arrR \blankLthree^{T-j} \dead^* \rightBr 
~~~~~~~~~~~~~~~~~~~~~~~~~~
\leftBr \DblankLthree^{T-j} \arrL \blankLthree^{j} \dead^* \rightBr,$$
where $j$ can range from $0$ to $T$.
We will start with $2(T+1)p(N)$ candidate correct configurations and subtract off the configurations that violate one of conditions $2$, $3$, or $4$ in Definition
\ref{def:correct}. Note that the the set of configurations that violate each one of those conditions is disjoint, so we can consider them separately.

Condition $2$ says that the $8$-pointer can not be vertically aligned with the $3$-pointer and pointing in the opposite direction. There are $T+1$ possible locations for the $3$-pointer and two directions. Each of these $2(T+1)$ configurations for the $3$-pointer eliminates one possible configuration for the $8$-pointer. (Subtract $2(T+1)$)

For condition $3$, there are $p(N) - 2(N-2)$ correct configurations for Tracks $1$ and $2$ in which the Track $1$ pointer is not an $8$-pointer. None of these are allowed when the Track $3$ pointer is all the way to the left of the chain and pointing to the left. (Subtract $p(N) - 2(N-2)$)

If the non-bracket particles are numbered from $1$ to $N-2$, then
the Track $3$ pointer is in position $T+1$ if it is immediately to the left of a $\dead$ particle. Condition $4$ says that the $8$-pointer can not be to the right of this location. So there are $(N-2)-(T+1)$ locations for the $8$-pointer that are not allowed and two directions for each location. (Subtract $2[(N-2)-(T+1)]$)

The total number of correct states is
$$2(T+1)p(N) - 2(T+1) - [p(N)-2(N-2)] - 2[(N-2)-(T+1)] = (2T+1)p(N).$$
\end{proof}

\begin{lemma}
\label{lem:correctclock}
{\bf [Structure of the Configuration Graph]}
In the configuration graph, the correct clock configurations with timer length $T$ form a cycle of length $(2T+1)p(N)$.
Every correct clock configuration for Tracks $1$ and $2$ occurs exactly $2T+1$ times in the cycle at intervals of
exactly $p(N)$.
All other well-formed clock configurations are in paths of length at most $p(N)$ that end in an illegal configuration.
\end{lemma}

Given the structure of the configuration graph defined in Lemma \ref{lem:correctclock}, we can define an enumeration of the correct clock states that will be useful later in the analysis of the final Hamiltonian:

\begin{definition}
\label{def:corrclockenum}
{\bf [Enumeration of Correct Clock Configurations]}
$\ket{c_{T,0,0}}$ is the clock configuration in which the Track $3$ timer has length $T$ and Tracks $1$, $2$, and $3$ are all at Time $0$. Note that the timer length is well-defined for $T \in \{0, \ldots, N-3\}$.
For $s \in \{0, \ldots, 2T\}$ and $t \in \{0,\ldots, p(N)-1\}$,
$\ket{c_{T,s,t}}$ is the clock configuration that is  reached after $s \cdot p(N) + t$ clock steps starting in $\ket{c_{T,0,0}}$.
$s$ uniquely determines the  configuration of Track $3$ and $t$ uniquely determines the configuration of Tracks $1$ and $2$.
\end{definition}

Here is the proof of Lemma \ref{lem:correctclock}.

\begin{proof} 
If the length of the timer on Track $3$ is $T$, then there are $2(T+1)$ possible
 Track $3$ configurations in a well-formed clock state:
$$\leftBr \DblankLthree^j \arrR \blankLthree^{T-j} \dead^* \rightBr 
~~~~~~~~~~~~~~~~~~~~~~~~~~
\leftBr \DblankLthree^{T-j} \arrL \blankLthree^{j} \dead^* \rightBr,$$
where $j$ can range from $0$ to $T$.
We can order the configurations from $0$ to $2T+1$ so that the configurations on the left come
first and each set is ordered by increasing $j$.
At Time $0$ for Tracks $1$, $2$, and $3$, Track $3$ is in configuration $0$.
Tracks $1$ and $2$ iterate through all correct configurations every $p(N)$ clock steps.
At each point that Time $0$ for Tracks $1$ and $2$ is reached, Track $3$ has advanced to the next
configuration, except in the last Track 1-2 iteration in which the configuration on Track $3$ goes 
from $2T$ to $2T+1$ and then to $0$ in a single $8$-segment(see Figure \ref{fig:last8segment}).
Because the configuration graph is a cycle, $(2T+1) p(N)$ unique configurations have been reached.

We will argue that any configuration reached by applying transition rules to the clock configuration at Time $0$ will be correct. By Lemma \ref{lem:countcorrect},
these must be exactly the set of correct clock configurations. We consider each condition in Definition \ref{def:correct} in turn.

\begin{enumerate}
\item {\bf Condition $1$:}
Lemma \ref{lem:clockGraph1} establishes that any configuration reached from Time $0$ for Tracks $1$ and $2$ is correct for Tracks $1$ and $2$.
\item {\bf Condition $2$:}
Note that Time $0$ for Tracks $1$, $2$, and $3$ is a correct configuration and none of the transition rules can result in an $8$-pointer on Track $1$ vertically aligned with the Track $3$ pointer pointing in the opposite direction. 
\item {\bf Condition $3$:}
If Track $3$ is in configuration $2T+1$ (with an $\arrL$ in the leftmost location), it must have entered that configuration in an $8$-segment since only the $8$-pointer can alter Track $3$.
In the last clock step of the $8$-segment, when the Track $1$ pointer transitions from an
$8$-pointer to a $1$-pointer,
TR-$41$ is applied which also flips the direction of the Track $3$ pointer.
Track $3$ remains unchanged until the next $8$-phase.
\item {\bf Condition $4$:} 
When Track $3$ first enters configuration $T$ ( which is $\leftBr \DblankLthree^T \arrR \dead^* \rightBr$) the $8$-pointer on Track $3$ is sweeping left, so afterwards
it will be to the left of the $\arrR$
on Track $3$. The next time, the Track $3$ configuration changes, the $8$-pointer is sweeping right.
TR-$39$ will cause the pointer on Track $3$ to change directions as it passes the Track $3$ pointer.
Therefore, when the $8$-pointer is over the $\dead$ states on Track $3$, the Track $3$ pointer will be
$\arrL$. 
\end{enumerate}

Now we need to argue that if a well-formed clock configuration is incorrect, it will reach an illegal configuration within $p(N)$ steps from which there will be no outgoing transition.
Again, we consider each condition from Definition \ref{def:correct} in turn.
\begin{enumerate}
\item {\bf Condition $1$:}
If the configuration is incorrect for Tracks $1$ and $2$, Lemma \ref{lem:clockGraph2} shows that an illegal configuration will be reached within $2(N-2) \le p(N)$ steps from which there is no outgoing transition.
\item {\bf Condition $2$:}
If the pointer on Track $1$ is an $8$-pointer that is vertically aligned with the pointer on Track $3$ pointing on the opposite direction, the current configuration is illegal due to term
 (\ref{eq:hcl5}) and there are no out-going transition rules that apply.
\item {\bf Condition $3$:}
Now suppose the pointer on Track $3$ is $\arrL$, in the leftmost location and  the Track $1$ pointer is not an $8$-pointer. The transition rules will advance Tracks $1$ and $2$ until the configuration below  is reached.
$$\begin{array}{|@{}c@{}| @{}c@{}| @{}c@{}|@{}c@{}|@{}c@{}|@{}c@{}|} \hline
\leftBr &
	\begin{array}{@{}c@{}} \arrLseven \\ \hline  \arrL  \\ \hline  \arrL  \end{array} &
	\begin{array}{@{}c@{}} \Dblank \\ \hline  \DblankLtwo  \\ \hline  \DblankLthree  \end{array} &
	\begin{array}{@{}c@{}} ~~~~\cdots~~~~ \\ \hline  ~~~~\cdots~~~~   \\ \hline  ~~~~\cdots~~~~ \end{array} &
	\begin{array}{@{}c@{}} \Dblank \\ \hline  \DblankLtwo  \\ \hline  \dead  \end{array} &
	 \rightBr \\ \hline  \end{array}
	  \rightarrow
\begin{array}{|@{}c@{}| @{}c@{}| @{}c@{}|@{}c@{}|@{}c@{}|@{}c@{}|} \hline
\leftBr &
	\begin{array}{@{}c@{}} \arrReight \\ \hline  \arrR  \\ \hline  \arrL  \end{array} &
	\begin{array}{@{}c@{}} \Dblank \\ \hline  \DblankLtwo  \\ \hline  \DblankLthree  \end{array} &
	\begin{array}{@{}c@{}} ~~~~\cdots~~~~ \\ \hline  ~~~~\cdots~~~~   \\ \hline  ~~~~\cdots~~~~ \end{array} &
	\begin{array}{@{}c@{}} \Dblank \\ \hline  \DblankLtwo  \\ \hline  \dead  \end{array} &
	 \rightBr \\ \hline  \end{array} $$
This will transition to an illegal configuration due to term
 (\ref{eq:hcl5}) from which there is no out-going transition rule. Since an iteration for Tracks $1$ and $2$ take $p(N)$ clock steps, an illegal configuration will be reached in at most $p(N)$ clock steps.
\item {\bf Condition $4$:}
Finally suppose that the Track $3$ configuration is $\leftBr \blankLthree^T \arrR \dead^* \rightBr$
and the Track $1$ pointer is an $8$-pointer to the left of the $\arrR$ on Track $3$. The $8$-pointer on Track $1$ will sweep right, turn around, and then sweep left. In at most $2(N-2)$ clock steps,  the following configuration is reached which is illegal due to the term
 (\ref{eq:hcl6}) and  there is no outgoing transition:
$$\sixcells{\Dblank}{\arrLeight}{\generic}{\generic}{\arrR}{\dead}$$
\end{enumerate}
\end{proof}
    
\subsection{Turing Machines Used in the Construction} 
\label{sec:TMs}

Tracks $4$ through $6$ are the computational tracks.
Track $4$ will store the state and head location for the Turing Machine,
Track $5$ will store the contents of the work tape, and Track $6$ will be a read-only track
storing the guess for the oracle queries and the witnesses for the verification
computations. The clock pointer on Track $1$ will trigger certain actions 
on the computation tracks, depending on the label of the pointer and the
state of the computation tracks. This subsection describes the two deterministic, reversible Turing Machines, $M_{BC}$ and $M_{TV}$, used in the construction.
We start with an overview in Subsection \ref{sec:twoTMs}
of the desired behavior of the two TMs
and then give more details in the remainder of the section.
Subsection \ref{sec:review} reviews basic concepts of  deterministic reversible TMs from \cite{BV97} and \cite{C15}.
Subsection \ref{sec:MBC} gives a detailed description of the binary counter TM, $M_{BC}$, including a closed formula for the reduction $N(x)$ that maps a string $x$ to the number of steps required by $M_{BC}$ to produce $x$.
 Section \ref{sec:MTV} outlines the properties required of $M_{TV}$ and argues
 that a reversible TM exists with those properties.
 Finally in Subsection \ref{sec:combineTM}, we describe several transformations that 
 we need to apply to $M_{BC}$ and $M_{TV}$ in order
 to encode the computation of the two Turing Machines into the Hamiltonian used in the construction.

\subsubsection{Overview of the two Turing Machines: $M_{BC}$ and $M_{TV}$.}
\label{sec:twoTMs}

The first Turing Machine is $M_{BC}$ which runs a binary counter.

\begin{quote}
{\bf Binary Counter Turing Machine:}
The initial configuration of $M_{BC}$ will have the head in state $q_0$
pointing to a $1$ followed by all blank symbols to the right.
The head of $M_{BC}$ will never move to the left of the original starting
location. $M_{BC}$ will repeatedly execute an increment operation.
Each increment operation starts with the head in state $q_0$ at the left
end of a binary string whose right-most bit is $1$ followed by blank symbols to the right.
After the increment operation, the head is back to the original location
in state $q_f$ and the binary string (interpreted as a the reverse of a binary number) has been incremented. $M_{BC}$ then transitions to
$q_0$ for the next increment operation.
The Turing Machine $M_{BC}$ is described in detail in Section \ref{sec:MBC}.
\end{quote}

The reduction will map $x$ to an integer $N$ such that
after $N-2$ steps of $M_{BC}$, the contents of the the work tape is $x$, and the head is in state $q_f$ in its original location. 
The function $N(x)$,  mapping $x$ to $N$, is the actual reduction and is
specified explicitly in Section \ref{sec:MBC}

The Turing Machine $M_{TV}$ performs two different operations in sequence.
The first operation starts with a string $x$  written on
the work tape
and a witness string  written on another auxiliary track
(which will be Track $6$ in our construction). 
The Turing Machine $M_{TV}$ starts in an initial state
$p_0$ with the head located at the leftmost input symbol. 
Then $M_{TV}$ will compute $m = c_1|x|$,
where $c_1$ is a constant that is determined by $M$ (hard-coded into the rules of $M_{TV}$).
Let $y$ denote the first $m$ bits of the witness string and 
$y_i$ the $i^{th}$ bit of $y$. $M_{TV}$ will simulate Turing Machine
$M$ on input $x$ using oracle responses denoted by $y$ to obtain the output
string $f(x,y)$. Note that by assumption $|f(x,y)| = m$.
Then $M_T$ will compute the following quantity and write that value
in unary on the worktape:
$$T(x,y) = f(x,y) + 2^m \left[ 4^{m+1} + \sum_{j=1}^m y_j \cdot 4^{m-j+1}\right]$$
The function $T(x,y)$ is used to employ Krentel's accounting scheme as described in 
Section \ref{sec:proofoverviewfinite}.

The second operation performed by $M_{TV}$ is to simulate $M$ with input $x$ and oracle queries $y$
from the witness tape. If the oracle responses are fixed, then the input to each
oracle query is also fixed. Let $x_i$ denote the input to the $i^{th}$ oracle query.
The next $m \cdot 2^{c_2 |x|}$ bits of the witness after $y$ are divided into $m$
segments of length $2^{c_2 |x|}$, where $c_2$ is a constant determined by $M$ and $V$.
Let $w_i$ denote the string formed by the bits in the
$i^{th}$ segment.
For each $i$ such that $y_i = 1$,  $M_{TV}$ runs the verifier $V$ on input
$(x_i, w_i)$.  Recall that $V$ is the verifier for language $L$, which is the oracle
language for the function $f \in \fpnexp$.

\begin{quote}
    {\bf Timer and Verification Turing Machine:}
    The  computation of $M_{TV}$ starts with $x$ written in binary 
on the worktape and finishes in the following configuration:
The first symbol is $\sigma_A$ if all of the verifier computations were accepting
and $\sigma_R$ if any of the verifier computations were rejecting.
Then $T(x,y)$ is written in unary on the worktape using the character $\sigma_X$.
The remainder of the work tape is arbitrary but can not contain the 
symbols $\sigma_A$, $\sigma_R$, or $\sigma_X$.
$M_{TV}$ will end in  state $p_f$ with the head back in the original start location.
The pseudo-code for $M_{TV}$ is given in Figure \ref{fig:MTVpseudo}.
\end{quote}

\subsubsection{Review of Reversible Turing Machines}
\label{sec:review}

Bernstein and Vazirani \cite{BV97} define several properties of Turing Machines and Quantum Turing Machines which we summarize here.
We will consider  {\em generalized} Turing Machines, also considered in \cite{C15},
in which the head can move left, right, or stay in the same place at a particular step. These options are denoted
by $\{L, R, N\}$. 
A Turing Machine is defined by a $5$-tuple:
$(Q, \Gamma, I, \delta, q_f, q_0)$.
$Q$ is the finite set of states, $\Gamma$ the tape alphabet, and
$I \subset \Gamma$ is the input alphabet. There is a special blank symbol $\# \in \Gamma - I$.
The tape consists of an infinite sequence of cells or locations, each of which stores
a symbol from $\Gamma$.
The {\em head} of the Turing Machine points to a particular location of the tape.
The {\em current symbol} is the symbol written in the location where the head is pointing.
A {\em configuration} of the Turing Machine consists of the contents of the tape, the location of the
head and the current state. 
For any $x \in I^*$, the initial configuration consists of the string $x$ written on
the tape with blank symbols everywhere else, and the head pointing to the leftmost symbol in $x$ in the state $q_0$.
The transition function $\delta$ maps pairs from $Q - \{q_f\} \times \Gamma$
to $Q \times \Gamma \times \{L, R, N\}$. In a single step of the Turing Machine, if the current symbol is $a$,  the state is $q$,
and $\delta(q,a) = (p, b, D)$, then the symbol $a$ will be overwritten by $b$, the new state will be
$p$ and the head will move one step in the direction $D$.
We say that a Turing Machine {\em halts} when it reaches the state $q_f$.
Since will be interested in Turing Machines that act as a permutation on a finite set of configurations, we will have transition functions that are also defined on pairs of the form 
$(q_f, a)$.

\begin{definition}
{\bf [Reversible Turing Machines]}
A (classical) Turing Machine is said to be {\em reversible} if each configuration has at most one predecessor. 
\end{definition}

\begin{definition}
{\bf [Unidirectional Turing Machines]}
A TM is said to be {\em unidirectional} if each state can be entered from only one direction. 
This means that if $\delta(q, a) = (p, b, D)$ and $\delta(q', a') = (p, b', D')$, then $D = D'$.
\end{definition}

Each state $p$ of a unidrectional TM has a unique direction
$D_p \in \{L, R, N\}$ that determines the direction from which
$p$ is entered.
Once these directions are defined for each state, the direction component in any transition rule
triple  $(p, b, D)$ is completely determined by the state $p$. 
The {\em reduced transition rule} $\delta_r: Q \times \Gamma \rightarrow Q \times \Gamma$
is well defined since the direction can be recovered from the output state  alone. 
It  will be helpful to refer collectively to the set of states that are reached from
the same direction:
\begin{align*}
    Q_L & = \{q \mid D_q = L\}\\
    Q_R & = \{q \mid D_q = R\}\\
    Q_N & = \{q \mid D_q = N\}
\end{align*}
Bernstein and Vazirani show local conditions that characterize whether a TM is reversible. 

\begin{theorem} {\bf [Theorem B.1 and Cor B.3 in \cite{BV97}]}
\label{th:reversible}
A TM $M$ is reversible if and only if $M$ satisfies the following two conditions:
\begin{enumerate}
    \item $M$ is unidirectional.
    \item The reduced transition function is one-to-one.
\end{enumerate}
\end{theorem}

We will make use of several other properties of Turing Machines given below.
The first two are standard from \cite{BV97}. The third is introduced here
for the purposes of our construction.

\begin{definition}
{\bf [Normal Form TM]}
A reversible TM is in {\em normal form} if
$$\forall a \in \Gamma~~~~~\delta(q_f, a) = (a, q_0, N).$$
\end{definition}

We will refer to the tape cell where the head is located in the initial configuration as the {\em start cell}.
Occasionally, it will be convenient to refer to a numbering of the tape cells. The start cell will be numbered $1$ and rest of the cells to the left and right are numbered accordingly.

\begin{definition}
{\bf [Proper TM Behavior on a Particular Input]}
A TM behaves {\em properly} on input $x$, if on input $x$,
the TM halts in the start cell and
the head never moves to the left of the start cell.
\end{definition}

\begin{definition}
{\bf [Compact TMs]}
For positive integer $c$, a TM is {\em $c$-compact}, if for every $t \ge 1$,
it reaches at most $t$ locations of the tape after $t+c$ steps.
\end{definition}

We will generally assume that a reversible Turing Machine has a transition function that is well-defined for every
pair in $Q \times \Gamma$. This is a reasonable assumption as the transition function can always be expanded if it is only
defined on a subset of $Q \times \Gamma$.

\begin{fact}
Let $M$ be a reversible Turing Machine whose transition function $\delta$ is defined on a subset  of $Q \times \Gamma$.
$\delta$ can be expanded so that it is well-defined on all of $Q \times \Gamma$ and the resulting Turing Machine is also reversible.
\end{fact}

\subsubsection{The Binary Counter Turing Machine}
\label{sec:MBC}

The Turing Machine tape in our construction will be composed of two tracks.
The {\em work track} will be stored in Track $5$  and
will
hold the contents of the work tape, each symbol of which comes from a set $\Sigma$.
The {\em witness track} will be stored in Track $6$ in the construction and
contains a read-only binary witness. Therefore the set of tape symbols $\Gamma$
is actually $\Sigma \times \{0, 1\}$. 
It will be convenient to separate the two tracks in describing a Turing Machine
configuration. Thus if $x$ and $\myw$ are two binary strings,
the input $(x,\myw)$ has $\myw$ on the witness track and input $x$ followed by blank
($\#$) symbols on the work track. The the left-most bits of $x$ and $\myw$ are vertically aligned. 
The initial configuration on input $(x,\myw)$ has the head pointing to the
leftmost bit of $x$ and $\myw$ in the start state $q_0$.
The Binary Counter Turing Machine $M_{BC}$ is unaffected by the contents of the witness track
(Track $6$), so we will only describe its behavior in terms of the contents of Track $5$.

We assume that $M_{BC}$ begins each increment operation in state $q_0$
at the left end of a binary string whose right-most bit is $1$. 
If $x$ is a binary string, then we interpret the string as 
representing the number $x^R$ in binary, where $x^R$ is the reverse of $x$.
The Turing Machine executes an increment loop,
and finishes in state $q_f$, back in the starting location,
with a string $x'$ written on the tape. The value of $(x')^R$ is
the value of $x^R$ plus $1$. Note that the condition that the binary string
$x'$ ends in $1$ will be maintained. In the construction, the starting configuration
will have the string $1$ on the tape, followed by $\#$ symbols, with the head in state $q_0$,
pointing to the $1$.
The transition rules for $M_{BC}$ are given in Figure \ref{fig:MBCrules}.
The pseudo-code for $M_{BC}$ is given in Figure \ref{fig:MBCpseudo}.

\begin{figure}[ht]
\centering
\begin{tabular}{|c|c|c|c|}
\hline
 & $0$ & $1$ & $\#$  \\
 \hline
 \hline
 $q_{0}$ & $(q_f, 1, N)$ & $(q_1, \#, R)$ & -  \\
 \hline
 $q_{1}$ & $(q_2, \#, R)$ & $(q_1, 0, R)$ & $(q_3, \#, R)$ \\
 \hline
 $q_{2}$ & $(q_4, 0, L)$ & $(q_4, 1, L)$ & - \\
 \hline
 $q_{3}$ & - & - & $(q_4, \#, L)$  \\
 \hline
 $q_{4}$ & - & - & $(q_5, 1, L)$ \\
 \hline
 $q_{5}$ & $(q_5, 0, L)$ & - & $(q_f, 0, N)$ \\
 \hline
\end{tabular}
\caption{A summary of the transition rules for $M_{BC}$. The transition rules are well-defined
as long the Turing Machine begins with the head pointing to the left end of a binary string whose rightmost bit is $1$, followed by  $\#$ symbols.
$M_{BC}$ reaches a final state in the position in which it began, with the binary
string incremented (in reverse notation) and $\#$ symbols to the right of the new string.
The rules can be extended so that they are well defined on every $Q \times \{0, 1, \# \}$ pair, so that the reduced transition function remains one-to-one and the Turing Machine
remains unidrectional. The rule $\delta(a, q_f) = (q_0, a, N)$ is included for every $a \in \{0, 1, \# \}$ so that the resulting Turing Machine is in normal-form.}
\label{fig:MBCrules}
\end{figure}

\begin{figure}[ht]
\noindent
\begin{center}
\fbox{\begin{minipage}{\textwidth}
\begin{tabbing}
(1)~~~~ \= Start in state $q_0$ at the left end of the string. \\
(2) \> \IF the current bit is $0$\\
(3)  \> ~~~~ \= Write a $1$ and transition to $q_f$, stay in place \\
(4) \> \ELSE \\
(5)  \> \> Write $\#$ (to mark the return location) move right into state $q_1$\\
(6) \> \> \WHILE the current bit is $1$, \\
(7) \> \> ~~~~ \= Write a $0$ and move right. Stay in state $q_1$\\
(8) \> \> \IF current symbol is  $0$\\
(9) \> \> \> Write a $\#$, transition to $q_2$ and move right.\\
(10) \> \> \ELSE ~~~(current symbol should be $\#$)\\
(11) \> \> \> Write a $\#$, transition to $q_3$ and move right.\\
(12) \> \> \IF current state is $q_2$, the current symbol should be a bit $b$.\\
(13) \> \> \> Write $b$, transition to $q_4$ and move left.\\
(14) \> \> \ELSE ~~~(current state is $q_3$, the current symbol should be   $\#$)\\
(15) \> \> \> Write a $\#$, transition to $q_4$ and move left.\\
(16) \> \> Current symbol should be $\#$. Write $1$, transition to $q_5$ and move left.\\ 
(17) \> \> \WHILE the current symbol is $0$, \\
(18) \> \> \> Write a $0$ and move left. Stay in state $q_5$\\
(19) \> \> \IF the current symbol is $\#$ (back at starting location), \\
(20) \> \> \> Write a $0$, transition to $q_f$, stay in place.
\end{tabbing}
\end{minipage}}
\end{center}
\caption{Pseudo-code for $M_{BC}$.}
\label{fig:MBCpseudo}
\end{figure}

\begin{lemma}
\label{lem:oneinc}
{\bf [Number of Steps for One Increment Operation of $M_{BC}$]}
$M_{BC}$ is a normal-form reversible TM.
Let $x$ be a binary string ending in $1$.
Let $l(x)$ denote the length of the maximal prefix of $x$ that is all $1$'s.
The number of steps $M_{BC}$ takes to reach $q_f$ on input $x$ is
$1$ if $l(x) = 0$ and is $2 l(x)+3$ if $l(x) > 1$.
\end{lemma}

\begin{proof}
If $l(x) = 0$, then the leftmost bit of $x$ is $0$. In this case,
$M_{BC}$ writes a $1$ and transitions to $q_f$ in the first step.

If $l(x) > 1$, then $M_{BC}$ writes a $\#$ to mark its return location and moves right into state $q_1$. The head continues to move right until it reaches a $0$ or $\#$ symbol. ($l(x)$ steps so far.) It writes a $\#$ as a placeholder for the $0$ or $\#$
and moves right. (One more step.) The new state will be $q_2$ or $q_3$ depending on whether it reached a $0$ or a $\#$.
At this point, $q_2$ should be pointing to a $0$ or $1$ (since the rightmost bit is $1$) and $q_3$ should be pointing to a $\#$. It re-writes the same symbol and
moves left to state $q_4$. (One more step.) The current symbol should be $\#$
because it holds the place holder for the $0$ or $\#$. It replaces the $\#$ with a $1$
and continues to move left past all the $0$'s until the next $\#$ is reached.
($l(x)$ more steps.) This $\#$ symbol marks the location where it started. It replaces the $\#$ with a $0$ and transitions to $q_f$ (One more step).
The total number of steps has been $2 l(x)+3$.
\end{proof}

For every $a \in \Gamma$, we add the rule $\delta(q_f, a) = (q_0, a, N)$
which makes $M_{BC}$ well formed. Then if $M_{BC}$ starts with an input of $1$, it
will repeatedly increment the string as long as it is allowed to run.
For example, suppose $M_{BC}$ starts   with string $x$:

\vspace{.1in}

\begin{tabular}{cccccccccccc}
 & $q_0$ & & & & & & &  &   & &  \\
 & $\downarrow$& & & & & & &  &   &  &  \\
$\vdash$ & 1 & 1 & 0 & 1 & \# & \#  & \#  & \# &  \# & $\dashv$  \\
\end{tabular}

\vspace{.1in}
After the increment operation, the state will be:

\begin{tabular}{cccccccccccc}
 & $q_f$ & & & & & & &  &   & &  \\
 & $\downarrow$& & & & & & &  &   &  &  \\
$\vdash$ & 0 & 0 & 1 & 1 & \# & \#  & \#  & \# &  \# & $\dashv$  \\
\end{tabular}

\vspace{.1in}

Then $M_{BC}$ transitions to $q_0$ in place and the increment operation begins again:

\vspace{.1in}

\begin{tabular}{cccccccccccc}
 & $q_0$ & & & & & & &  &   & &  \\
 & $\downarrow$& & & & & & &  &   &  &  \\
$\vdash$ & 0 & 0 & 1 & 1 & \# & \#  & \#  & \# &  \# & $\dashv$  \\
\end{tabular}

\vspace{.1in}

For a particular chain length $N$, the starting configuration for $M_{BC}$
will have the string $1$, followed by $N-3$ blank ($\#$) symbols
written on the work track and a binary witness string
$\myw$ of length $N-2$ written on the witness track. We will call this input
$(1,\myw)_N$. 

\begin{definition}
\label{def:Nofx}
{\bf [Formula for $N(x)$]}
$$N(x) = 5 n(x^R) - 2 w(x)  - 4[n(x^R)~\mbox{mod}~2],$$
where $x$ is a binary string, ending in $1$, 
$x^R$ is the reverse of binary string $x$,
$n(x)$ is the numerical value of the binary number represented by string $x$, and $w(x)$ is the number of $1$'s in $x$.
\end{definition}

\begin{lemma}
\label{lem:Nofxbound}
{\bf [$N(x)-2$ Steps of $M_{BC}$ Produces String $x$]}
Let $N(x) = N \ge 5$.
Starting with input $(1, \myw)_N$. 
If $M_{BC}$ is run for $N-2$ steps,
it will never leave the segment of tape cells from $1$ through $N-4$.
Moreover, at the end of the $N-2$ steps,
the configuration will be $(x,\myw)$ with the head pointing to the left-most bits of $x$ and $\myw$ in state $q_f$.
\end{lemma}

\begin{proof}
$M_{BC}$ takes two steps to add
 a $1$ to the leftmost position if a $0$ is present, beginning and ending in state $q_f$:
$$ (q_f/0)x1 \#^*  ~\rightarrow ~ (q_0/0)x1 \#^* ~\rightarrow ~ (q_f/1)x1 \#^*$$

Let $f(n)$ be the number of steps required to go from the state
$ (q_f/1) \#^*$ to $ (q_f/0) 0^{n-1} 1\#^*$,
which is effectively moving the $1$ over $n$ locations. Note that the cost is
the same if the starting configuration is $ (q_f/1) 0^n x \#^*$.
In order to end up in configuration $ (q_f/0) 0^{n-1} 1\#^*$,
the TM must first reach $(q_f/1)~1^{n-1}~\#^*$ and then perform an increment operation. 
In order to reach $(q_f/1)~1^{n-1}~\#^*$,
the first $1$ is moved over by $n-1$ (cost of $f(n-1)$), then a new $1$ is added
(cost of $2$), then that $1$ is moved over by $n-2$ (cost of $f(n-2)$), etc.
Thus the number of steps to reach $(q_f/1)~1^{n-1}~\#^*$ from
$(q_f/1) \#^*$ is $f(n-1) + f(n-2) + \cdots f(1) + 2(n-1)$. After one more
increment operation, the state will be $(q_f/0) 0^{n-1} 1\#^*$.
It takes $1$ step to get to $(q_0/1)~1^{n-1}~\#^*$, and then according to Lemma \ref{lem:oneinc}, an additional $2n+3$ steps to complete the increment.
Therefore the function
$f$ obeys the recurrence:
$$f(n) = \sum_{j=1}^{n-1} f(j) + 2(n-1) + 1 + (2n+3) = \sum_{j=1}^{n-1} f(j) + (4n+2).$$
The solution to this recurrence relation is $f(n) = 5 \cdot 2^n - 4$.
The initial condition is $f(1) = 6$, and it requires $6$ steps to go from
$ (q_f/1) \#^*$ to $ (q_f/0)  1\#^*$.

If the bits of $x$ are numbered from left to right $x_1 x_2 \cdots x_{n}$, then 
the number of steps to reach $ (q_f/x_1)x_2 \cdots x_n \#^*$ from $ (q_f/1) \#^*$ is
$$\sum_{j=2}^{n} x_j f(j-1) + 2(w(x)-1) = \sum_{j=2}^{n} x_j (5 \cdot 2^{j-1} - 4) + 2(w(x)-1)$$
If $x^R$ is even ($x_1 = 0$), then this value is
$5 n(x^R) - 2 w(x) - 2$. If $x^R$ is odd then this value is is $5 n(x^R) - 2 w(x) - 6$.
Note also that $M_{BC}$ starts in state $q_0$ instead of state $q_f$,
so we can subtract $1$ steps to get:
$$N(x) - 2 = 5 n(x^R) - 2 w(x) - 2 - 4(n(x^R)~\mbox{mod}~2).$$

Finally, we address whether the head tries to move out of the segment of tape symbols
from $1$ through $N(x)-3$ in the first $N(x)-2$ steps.
If $N(x) \ge 5$, then there are at least $3$ tape cells and the first increment
operation can be completed. 
The first time the head reaches tape cell $n \ge 4$ is in incrementing $1^{n-2}$. It takes $5(2^{n-2}-1) - 2(n-2) - 4$ steps to reach
$1^{n-2}$ and end in state $q_f$. It takes another $n$ steps to reach
tape cell $n$. We need that the total number of steps
$5 (2^{n-2}-1) -n \ge n+2$, which holds for every $n \ge 4$.

On the left end of the tape, $M_{BC}$ starts with a binary string ending in $1$.
In each increment operation in which it starts with a string ending in $1$, it behaves
properly. When it reaches $q_f$, it will again have a string ending in $1$. It then transitions to  $q_0$ and behaves properly on the next increment operation. 
\end{proof}

It will be  useful in the analysis to have a lower bound on $N(x)$ as a function of $|x|$.

\begin{lemma}
\label{lem:Nlowerbound}
For $|x| \ge 2$, $N(x) \ge 2^{|x|}$.
\end{lemma}

\begin{proof}
Consider a string $x$ of length $n$. $N(x)$ is minimized for
$x = 0^{n-1}1$ or $10^{n-2}1$. Any additional $1$'s in $x$ add
at least $10$ to $N(x)$ in the first term and only subtract $2$ from the second term.
For $10^{n-2}1$, $N(x) = 5(2^{n-1}+1) - 10$ which is at least $2^n$ for any $n \ge 2$.
For $0^{n-1}1$, $N(x) = 5(2^{n-1}) - 4$ which is also at least $2^n$ for any $n \ge 2$.
\end{proof}

\subsubsection{The Turing Machine $M_{TV}$ used in the Hamiltonian for Finite Chains}
\label{sec:MTV}

In this section we establish
that there is a reversible normal-form Turing Machine $M_{TV}$ that satisfies the conditions required in the construction. Recall that we are reducing from $f \in \fpnexp$.
The polynomial time Turing Machine that computes $f$ is called $M$ and the query language is $L$.
$L$ is contained in $\nexp$ and the exponential time verifier for $L$ is called $V$.
We will assume that there are constants $c_1$ and $c_2$ such that on input $x$,
the number of queries made by $M$ is at most $c_1 |x|$ and if $M$ makes a query $\bar{x}$
to the oracle, the size of the witness required is at most $2^{c_2 n}$ and the running time
of $V$ on input $(\bar{x},w)$ is at most $2^{c_2 |x|}$. This assumption justified by the following lemma:

\begin{lemma}
\label{lem:pad}
{\bf [Padding Lemma]}
$f \in \fpnexp$, then for any constants, $c_1$ and $c_2$,
$f$ is polynomial time reducible to a function $g \in \fpnexp$ such that $g$ can computed by
polynomial-time Turing Machine $M$
with access to a $\nexp$ oracle for language $L$. The verifier for $L$ is a  Turing Machine $V$.
Moreover, on input $x$ of length $n$, $M$ runs in $O(n)$ time, makes
at most $c_1 n$ queries to the oracle. Also, the length of the queries made to the oracle is at most $c_1 n$ and the running time of $V$ as well as the size of the witness required for $V$ on any query made by $M$ is $O(2^{c_2n})$. In addition the length of the output of $g$ is at most $c_1 n$.
\end{lemma}

\begin{proof}
Suppose that $M$ runs in $r_1(n)$ time and the running time of $V$ is $O(2^{r_2(n)})$.
Therefore, on input $x$, the number of queries made by $M$ as well as the length of the input
to the oracle, as well as the length of the output $f$ are at most $r_1(|x|)$. Also, the running time as well as the size of witness required by
$V$ is at most $O(2^{r_2(r_1(|x|))})$.

The reduction will pad an input $x$ to $f$ with $1 0^{t(n)} 1$, for a polynomial $t$ chosen below.
We will refer to this substring as the {\em suffix}.
The algorithm for $g$ will verify that the length of the suffix  is in fact $t(n)$, where $n$ is the number of bits not included in the suffix.
If not, then the value of $g$ is $0$. If the input to $g$ does have a suffix with the correct length
then it erases the suffix and simulates $M$ on the rest of the string. $t(n)$ will be chosen to be large enough so that this computation can be done in $O(t(|x|))$ time.
Note that the input length to $g$ is now $\bar{n} = t(|x|) + |x| + 2$. 

$t(n)$ is chosen so that $r_1(|x|) \le c_1 t(|x|) < c_1 \bar{n}$, so that the running time of 
$M$ and hence the number of oracles queries made, as well as the length of the inputs to those oracle queries
are all bounded  by $c_1 \bar{n}$.  The running time of $V$ on any of the query inputs is
at most 
$c 2^{r_2(r_1(|x|))}$ for some $c$. $t$ will also be chosen to be large enough so that
$$c 2^{r_2(r_1(|x|))} \le 2^{c_2 t(|x|)} < 2^{c_2 \bar{n}}.$$
\end{proof}

The pseudo-code in Figure \ref{fig:MTVpseudo} gives a detailed description 
for the behavior of the Turing Machine $M_{TV}$. Then Definition \ref{def:correctoutput}
defines the output of $M_{TV}$ required in the construction. 
\begin{figure}[ht]
\noindent
\fbox{\begin{minipage}{\textwidth}
\begin{tabbing}
{\sc Input:} $(x, \myw)$  \\
(1)  ~~~~\=  Compute $N = N(x)$\\
(2) \> $m$ is the number of oracle queries made by $M$ on input $x$\\
(3) \> $r$ is the size of the witness used by $V$  on any oracle query generated by $M$ on input $x$\\
(4) \> $m$ and $r$ are hard-coded functions of $|x|$ determined by $M$ and $V$.\\
(5)  \> Simulate Turing Machine $M$ on input $x$\\
(6)  \> ~~~~~ \= Use $y$ for the responses to the oracle queries, where $y$ denotes the first $m$ bits of $\myw$\\
(7)  \> \>Simulation generates $x_1, \ldots, x_m$, inputs to oracle queries\\
(8)  \> {\sc Reject} $= $ FALSE\\
(9)   \> \FOR $i = 1 \ldots m$\\
(10) \> \> \IF $y_i = 1$\\
(11) \> \> ~~~~~ \= Simulate $V$ on input $(x_i, w_i)$\\
(12) \> \> \> ~~~~~ \= $w_i$ is the string formed by bits $m +(i-1)r+1$ through $m + ri$ of $\myw$\\
(13) \> \> \> If $V$ rejects on input $(x_i, w_i)$\\
(14) \>\>\>\> {\sc Reject} $= $ TRUE\\
(15) \>\IF (\sc Reject)\\
(16) \>\> Write $\sigma_R$ in the left-most position of the work tape\\
(17) \>\ELSE\\
(18) \>\> Write $\sigma_A$ in the left-most position of the work tape\\
(19) \> Compute $T(x,y)$\\
(20) \> Write the value of $T(x,y)$ in unary with $\sigma_X$ symbols starting in the second position of the work tape\\
{\sc Output:} $(\{\sigma_A,\sigma_R\} (\sigma_X)^{T(x,y)} ; x,\myw)$  \\

\end{tabbing}
\end{minipage}}
\caption{Pseudo-code for the Turing Machine $M_{TV}$.}
\label{fig:MTVpseudo}
\end{figure}

The following definition allows us to refer to the correct output of $M_{TV}$ whose behavior depends on the input pair $(x, \myw)$ as well as Turing Machines $M$ and $V$:

\begin{definition}
\label{def:correctoutput}
{\bf [Correct Output of $M_{TV}$]}
Define {\sc Out}$(x, \myw, M, V)$ as the correct output of $M_{TV}$ on input $(x, \myw)$ using Turing Machines $M$ and $V$, as described in
Figure \ref{fig:MTVpseudo}.
\end{definition}

The following lemma establishes 
that there is a reversible normal-form Turing Machine $M_{TV}$ that satisfies the conditions required in the construction.

\begin{lemma}
\label{lem:existsMTV}
{\bf [Existence of $M_{TV}$]}
There is a reversible normal-form Turing Machine $M_{TV}$ such that 
\begin{enumerate}
    \item $M_{TV}$ behaves properly on all $(a, b)$, where $a, b \in \{0,1\}^*$.
    \item $M_{TV}$ is $1$-compact.
    \item For all but a finite number of binary string $x$, if $N = N(x)$
and  $\myw \in \{0,1\}^{N-4}$, then $M_{TV}$ halts on input $(x,\myw)$ in at most
$N - 3$ steps. 
\item $M_{TV}$ has {\sc Out}$(x, \myw, M, V)$ on its work tape when it reaches its final state.
\end{enumerate}
\end{lemma}

\begin{proof}
Let $n = |x|$.
By Lemma \ref{lem:Nlowerbound}, $N(x) \ge 2^{|x|} = 2^n$ for any $|x| \ge 2$.

We first consider the time to compute the function outlined in Figure \ref{fig:MTVpseudo}
on a RAM and then argue about the overhead in converting the algorithm to a reversible Turing Machine.
Simulating $M$ on input $x$ takes time $r(n)$ for some polynomial $r$.
By Lemma \ref{lem:pad}, we can assume that 
the number of queries made to $V$ is at most $m \le c_1 n$ and the running time of each call to $V$ takes at most time
$2^{c_2 n}$. Thus, the total time for running the verifier is at most $c_1 n 2^{c_2n}$.
Finally the time to compute $T(x,y)$ and write the resulting value in unary is $O(4^m)$ which is $O(2^{2 c_1 n})$.
Thus, for any constant $d$, the constants $c_1$ and $c_2$ can be chosen so that the running time to compute
the function computer by $M_{TV}$ is $O(N^d)$.

The RAM algorithm can be converted to a TM computation with at most polynomial overhead.
Finally, by Theorem B.8 from \cite{BV97}, the TM can be converted to a reversible TM that behaves properly on all inputs
with quadratic overhead. Therefore the constant $d$ can be chosen to be small enough so that the running time
of the reversible TM to compute $M_{TV}$ is $o(N)$ which will be at most
$N - 3 $ for all but a finite number of strings $x$.

Finally, to satisfy the second assumption, we can add one additional state which cause the reversible TM to stay in place for one step before beginning its computation.
\end{proof}

\subsubsection{Adjustments to the Two Turing Machines}
\label{sec:combineTM}

In this subsection, we describe several transformations that 
 we need to apply to $M_{BC}$ and $M_{TV}$ in order
 to encode the computation of the two Turing Machines into the Hamiltonian used in the construction.
  One issue to be addressed is
 that when we translate the two Turing Machines into the transition rules of the Hamiltonian,
 they will both operate on the same Hilbert space for Tracks $4$, $5$, and $6$, so we need to expand
 the state set and tape alphabet for the two Turing Machines to make them equal. States and symbols need to be added in such a way that each Turing Machine behaves as expected on its original tape symbols and state set and so that they also remain reversible.
 Lemma \ref{lem:addsymbols} describes how to add tape symbols and then Fact \ref{lem:addstates} (at the end of the subsection) describes how to add states.
 
 The computation encoded in the Hamiltonian will execute $M_{BC}$ for $N-2$ steps
 and then $M_{TV}$ for $N-2$ steps.  In the finite construction,
we can choose a particular $N$ so that after $N-2$ steps, $M_{BC}$ is at the end of an increment operation and
the desired $x$ is on the work tape. In the infinite case, we will have less control over the number of steps. In particular it may happen that $M_{BC}$ is in the middle of an increment operation after $N-2$  steps.  In this case,
we will want the second Turing Machine to finish the increment operation before starting on its own computation. In order to dovetail one TM with the next,
 we formally define a combined Turing Machine (which we call $M_{BC}+M_{TV}$). 
 The Turing Machine $M_{BC}+M_{TV}$ will allow $M_{TV}$ to pick up the computation after $N-2$ steps of $M_{BC}$,
 finish the current increment operation of $M_{BC}$, and then transition to $M_{TV}$'s primary tasks of computing the timer length and verifying {\em yes} oracle responses.
 
 Finally, we need to address the fact that $M_{TV}$ may finish its computation before the end of
 the $N-2$ steps. Since $M_{TV}$ is reversible, it will continue to execute some unknown operations
 until the clock runs out, which could potentially alter the contents of the work tape. In order to guarantee that the tape contents contain the correct output after exactly $N-2$ steps of $M_{TV}$,
 we introduce a time-wasting process which ensures that the operation of the Turing Machine is  reversible and   maintains the contents of the tape for at least an additional $N-2$ steps after the final state is reached.
 Lemma \ref{lem:gentv} establishes that the transformed $M_{BC}+M_{TV}$ has the required properties
 for the construction.

The first step is to expand the tape alphabet of $M_{BC}$ and $M_{TV}$ so that they
have exactly the same tape alphabet $\Gamma$. The first lemma shows that this can be done
while maintaining the  properties that the Turing Machines are reversible and in normal-form.

\begin{lemma}
\label{lem:addsymbols}
{\bf [Adding Symbols]}
Let $M = (Q, \Gamma, \delta, q_0, q_f)$ be a normal-form  reversible TM.
Let $a$ be a symbol that is not in $\Gamma$.
Let  $M' = (Q, \Gamma \cup \{a\},\delta', q_0, q_f)$, where
for every  $b \in \Gamma$,
\begin{align*}
    \delta'(q, b) & = \delta(q, b) &~\mbox{for}~ q \in Q \\
    \delta'(q, a) & = (a, q, D_q) &~\mbox{for}~ q \in Q - \{ q_0, q_f \}\\
    \delta'(q_0, a) & = (a, q_f, D_{q_f})\\
    \delta'(q_f, a) & = (a, q_0, D_{q_0})\\
\end{align*}
Then $M$ and $M'$ have the same behavior starting from any configuration that does not have the symbol $a$
anywhere on the tape. If $M$ is reversible and in normal-form, then $M'$ is as well.
\end{lemma}

\begin{proof}
By induction on the number of steps. If the current tape contents do not contain $a$, then the current symbol will be some $b \in \Gamma$. $M'$ will not write $a$ in the next step because 
$\delta'(q,b) = \delta(q, b)$ and $M$ does not have the symbol $a$ in its tape alphabet.

The reduced transition rule is still one-to-one. Moreover, the TM remains unidirectional. 
If $M$ is in normal form, then the new rule $\delta'(q_f, a)  = (a, q_0, D_{q_0})$ ensures that $M'$ will also be in normal form.
\end{proof}

Once the tape alphabets from $M_{BC}$ and $M_{TV}$ are expanded so that they are the same,
the TM resulting from dovetailing the two together is well defined.

\begin{definition}
\label{def:dovetail}
{\bf [Dovetailing Two TMs]}
Let $M_1 = (Q_1, \Gamma, I, \delta_1, q_0, q_f)$ and $M_2 = (Q_2, \Gamma, I, \delta_2, p_0, p_f)$ be two normal-form  reversible TMs.
Suppose  also that $M_1$ and $M_2$ use the same set of tape symbols $\Gamma$ and  that $Q_1 \cap Q_2 = \emptyset$.  
Define $M_1 + M_2$ to be the Turing Machine with state set $Q_{1+2} = Q_1 \cup Q_2$, and for every $a \in \Gamma$,
\begin{align*}
    \delta_{1+2} (q, a) & = \delta_1 (q, a) & \mbox{for}~q \in Q_1 - \{q_f\}\\
    \delta_{1+2} (p, a) & = \delta_2 (p, a) & \mbox{for}~p \in Q_2 - \{p_f\}\\
    \delta_{1+2} (q_f, a) & = (p_0, a, N)\\
    \delta_{1+2} (p_f, a) & = (q_0, a, N)\\
\end{align*}
The initial state for $M_1 + M_2$ is $q_0$ and the final state is $p_f$.
\end{definition}

\begin{lemma}
\label{lem:dovetail}
If $M_1$ and $M_2$ are normal-form  reversible TMs, then $M_1+M_2$ is a 
normal-form  reversible TM.
\end{lemma}

\begin{proof}
We will show that $M_1 + M_2$ is unidirectional and has a one-to-one reduced transition function, which by Theorem \ref{th:reversible}, will imply that $M_1 + M_2$ is reversible.
Since $M_1$ and $M_2$ are themselves reversible they are unidirectional. Since $M_1+M_2$
preserves the direction of each state, $M_1+M_2$ is also unidrectional.
Note also that for $q \in Q_1 - \{q_f\}$, $\delta_1(q, a)$ does not map to state $q_0$, otherwise the reduced transition function $M_1$ would not be one-to-one, since $\delta_1(q_f, a) = (q_0, a, N)$ for every $a \in \Gamma$. Similarly for $p \in Q_2 - \{p_f\}$, $\delta_2(p, a)$ does not map to state $p_0$. This means that there is no conflict between 
the rules in the first two lines of the Definition \ref{def:dovetail}
and the new rules added in the second two lines. Since the reduced transition functions of $M_1$ and $M_2$ are one-to-one and $Q_1 \cap Q_2 = \emptyset$, there are no conflicts between any two rules specified in the first two lines.

The final state of $M_1+M_2$ is $p_f$ and the initial state is $q_0$.
Since $\delta_{1+2}$ maps $(p_f, a)$ to $(q_0, a, N)$ for every
$a \in \Gamma$, $M_1+M_2$ is in normal form.
\end{proof}

The computational process that will be embedded into the Hamiltonian will
execute $N-2$ steps of $M_{BC}$ followed by $N-2$ steps of $M_{BC}+M_{TV}$.
Lemma \ref{lem:existsMTV} guarantees that the second process will reach a final
state $p_f$ within $N-2$ steps for most $N$. At this point, the correct output is written on the 
work tape. However, it will be important for the construction that this correct output also be written
on the input tape after exactly $N-2$ steps. 

Let $q_0$ be the start state of
$M_{BC}+M_{TV}$ (since it was the original start state of $M_{BC}$)
and let $p_f$ be the final state of $M_{BC}+M_{TV}$ (since it was the original final state of $M_{TV}$). Recall that $M_{BC}+M_{TV}$ is in normal-form, so $\delta(p_f, a) = (q_0, a, N)$ for
every $a \in \Gamma$. Also 
 in a correct computation of $M_{BC}+M_{TV}$, 
 the TM halts with the head pointing to the start cell. Furthermore, there are two symbols $\sigma_A$ and $\sigma_R$ in $\Gamma$ such that when the TM halts,
 the start cell contains $\sigma_A$ or $\sigma_R$ and there are no other occurrences of $\sigma_A$ or $\sigma_R$ on the rest of the tape. 

\begin{definition}
{\bf [Transformed TM $\calt(M_{BC}+M_{TV})$]}
Let $M_{BC}+M_{TV} = (Q, \Gamma, I, \delta, p_f, q_0)$.
Define $\calt(M_{BC}+M_{TV}) = (Q \cup \{q_*\}, \Gamma, I, \delta', p_f, q_0)$
\begin{align*}
    \delta'(q,a) &= \delta(q,a) &~\mbox{except for}~ (q,a) \in \{ (q_f, \sigma_A), (q_f, \sigma_R)\}\\
    \delta'(p_f,a) &= (q_*,a,R) &~\mbox{for}~ a \in \{ \sigma_A, \sigma_R\}\\
    \delta'(q_*,b) &= (q_*,b,R) &~\mbox{for}~ b \in \Gamma - \{ \sigma_A, \sigma_R\}\\
    \delta'(q_*,a) &= (q_0,a,N) &~\mbox{for}~ a \in \{ \sigma_A, \sigma_R\}\\
\end{align*}
\end{definition}

\begin{lemma}
\label{lem:gentv}
{\bf [Properties of  $\calt(M_{BC}+M_{TV})$]}
Let $M_{TV}$ be a TM that satisfies the properties from Lemma \ref{lem:existsMTV}.
Then $\calt(M_{BC}+M_{TV})$ is unidirectional and its reduced transition function
is one-to-one.
For all but a finite number of binary string $x$, 
and every $\myw \in \{0,1\}^{N(x)-2}$, starting with input $(1,\myw)$,
if $N(x)-2$ steps of $M_{BC}$ followed by $N(x)-2$ steps of $\calt(M_{BC}+M_{TV})$
are executed, then
\begin{enumerate}
    \item the head of the Turing Machine never leaves the sequence of tape cells from $1$ through $N(x)-3$ in the course of executing the $2N(x)-4$ steps.
\item the contents of the work tape are {\sc Out}$(x, \myw, M, V)$ at the end of the the $2N(x)-4$ steps.
\end{enumerate} 
\end{lemma}

\begin{proof}
Since $M_{BC}+M_{TV}$ is a normal-form TM,  the direction associated with $q_0$ is $N$. All of the new rules added are therefore consistent with $D_{q_0} = N$ and $D_{q_*} = R$ and therefore $\calt(M_{BC}+M_{TV})$ is uni-directional. 

The rules $\delta(p_f, \sigma_A) = (q_0, \sigma_A, N)$ and $\delta(p_f, \sigma_R) = (q_0, \sigma_R, N)$
are removed and the rules $\delta(q_*, \sigma_A) = (q_0, \sigma_A, N)$ and $\delta(q_*, \sigma_R) = (q_0, \sigma_R, N)$ are added, so there is exactly one rule that maps on to each pair
$(q_0, \sigma_A)$ and $(q_0, \sigma_R)$. Finally the new rules
$\delta(p_f, \sigma_A) = (q_*, \sigma_A, R)$, $\delta(p_f, \sigma_R) = (q_*, \sigma_R, R)$,
and $\delta(q_*, b) = (q_*, b, R)$ for $b \not\in \{ \sigma_A, \sigma_R\}$,
are one-to-one. Since $q_*$ is a new state, these rules do not conflict with any of the previous rules from $M_{BC}+M_{TV}$. Therefore the new reduced transition function is one-to-one.

By Lemma \ref{lem:Nofxbound} the head never leaves the sequence of tape cells from $1$ through $N(x)-3$ in $N(x)-2$ steps of $M_{BC}$ as long as $N(x) \ge 5$.
At this point, the string $x$ is written on the work tape, the head is back in the starting location and the state is $q_f$. According to the definition of dovetailing, $M_{BC}+M_{TV}$ transitions from $q_f$ to $p_0$, which initiates
$M_{TV}$. In the next $N-3$ steps, $M_{TV}$ will not move past tape cell $N(x)-3$ due to the fact that it is $1$-compact.
Since $M_{TV}$ behaves properly on input $(x,\myw)$, it will not move to the left of tape cell $1$ on input $(x,\myw)$. According to Lemma \ref{lem:existsMTV},
for all but a finite number of $x$'s, $M_{TV}$ will reach $p_f$ within $N-3$ steps. Since $M_{TV}$ behaves properly on $(x,\myw)$, when it halts, the head
will be back at tape cell $1$. The content of the work tape at this point is {\sc Out}$(x, \myw, M, V)$.
Note that this output will have an $\sigma_A$ or $\sigma_R$ in cell $1$ and no other occurrences of $\sigma_A$ or $\sigma_R$ on the rest of the tape. Therefore, the current state/tape symbol pair is $(p_f,\sigma_A)$ or $(p_f,\sigma_R)$. According to the rules of $\calt(M_{BC}+M_{TV})$,
the current symbol will be re-written and the head will move right into state $q_*$. As long as no $\sigma_A$ or $\sigma_R$ symbols are encountered, the head will continue to move to the right, leaving the tape contents intact. Therefore, the end of the
$N(x)-3$ steps will be reached before the head reaches tape cell $N(x)-2$. 
At the end of the $N(x)-2$ steps of $\calt(M_{BC}+M_{TV})$, the tape contents
will still be {\sc Out}$(x, \myw, M, V)$.
\end{proof}

At this point $M_{BC}$ and $\calt(M_{BC}+M_{TV})$ have the same tape alphabet, but the state set for $M_{BC}$ is a subset of the state set for $\calt(M_{BC}+M_{TV})$. When we translate the transition functions for each TM into a Hamiltonian, they will have to be well-defined on the same state space. Therefore, we would like to expand the state set for
$M_{BC}$ so that its state set is the same as $\calt(M_{BC}+M_{TV})$.
The following fact allows us to add states without effecting the behavior of $M_{BC}$ on the original state set.

\begin{fact}
\label{lem:addstates}
{\bf [Adding States to a TM]}
Let $M = (Q, \Gamma, \delta, q_0, q_f)$ be a reversible normal-form Turing Machine. Let $M'$ be the Turing Machine with $Q' = Q \cup \{q_{new}\}$. $\delta'$ will be the same as $\delta$ for all $q \in Q$,
and $\delta'(q_{new}, a) = (q_{new}, a, N)$ for all $a \in \Gamma$.
Then $M'$ is a reversible normal-form Turing Machine that behaves exactly as $M$ on any configuration whose state is $q \in Q$.
\end{fact}

\subsection{The Hamiltonian}
\label{sec:Ham}

We are now ready to describe the terms in the Hamiltonian. 
Section \ref{sec:compwf} describes the set of constraints which  ensure that the computation tracks are well-formed, corresponding to a valid configuration of the Turing Machine. 
As the Track $1$ pointer sweeps from right to left, it will trigger one TM step on the computation tracks. Section \ref{sec:implementingTM} gives a generic description of the Hamiltonian terms that execute a TM step.
Then Section \ref{sec:propsegments} describes more specifically how each $i$-pointer effects the computation tracks. The propagation terms are given explicitly in this section.
We have already added in Section \ref{sec:clock} a number of terms  that operate only on the clock tracks.  Many of those will remain unchanged since they operate as the identity on the computation tracks. 
Some of them, however, will be replaced with terms that advance the clock as well as change the computation tracks.
The transition terms introduced in Section \ref{sec:clock} are numbered TR-$1$ through TR-$46$. We will indicate when a new transition rule replaces one of these transition rules. 
Finally, Section \ref{sec:addconstraints} describes additional penalty terms used to ensure that the lowest energy state corresponds to a correct computation.

\subsubsection{Local Constraints to Make the Computation Tracks Well-Formed}
\label{sec:compwf}

The computation that will be embedded in the Hamiltonian will consist of $N-2$ steps of
$M_{BC}$ followed by $N-2$ steps of $\calt(M_{BC}+M_{TV})$. 
It will be convenient to  have a more succinct name for $\calt(M_{BC}+M_{TV})$:
\begin{align*}
    M_{check} & = \calt(M_{BC}+M_{TV})
\end{align*}
As per Lemmas \ref{lem:addstates} and \ref{lem:addsymbols}, states and tape symbols have been added to both
$M_{BC}$ and $M_{check}$ to ensure that they have the same set of symbols
and the same tape alphabet.
The previous section
assumed a single track Turing Machine when in fact both Tracks $5$ and $6$ 
are Turing Machine tapes. 
Thus, the tape symbols for the Turing Machine consists of the
set $\Gamma = \Sigma \times \{0, 1\}$, where $\Sigma$ is the set of symbols for Track $5$.
The Turing Machines will always have the property that the state of Track $6$ is never
changed, so if the current symbol is
some $a \in \Gamma$ which represents states $[x,1]$ on Tracks $5$ and $6$,
then the symbol written will always be of the form $[y,1]$.
In the next section, we will need to add a duplicate copy of every $q_L \in Q_L$
in order to execute the left-moving TM steps. Let $Q'_L$ denote the set of added states. 

Here are the standard basis states for the part of the Hilbert space of each particle corresponding to Tracks $4$, $5$, and $6$. 
\begin{itemize}
    \item The standard basis states for Track $4$ are $\{\blank, \Dblank\} \cup Q \cup Q'_L$.
    \item  The set of standard basis states for Track $5$ is $\Sigma$. 
    \item The standard basis states for Track $6$ are  $\{0,1\}$.
\end{itemize}

We will add Type I terms that will give an energy penalty to any standard basis state
for the computation tracks that do not correspond to a well-defined Turing Machine configuration.
\begin{definition}
\label{def:wellformedcomp}
{\bf [Well-Formed Computation Configurations]}
  Let $C_N$ denote the set of standard basis states for the computation tracks
  of a 1D chain of length $N$ in which the state of Track $4$ has the form:     $\leftBr \Dblank^* (Q + Q'_R) \blank^* \rightBr$. These are the well-formed computation configurations.
\end{definition}

\begin{definition}
\label{def:Hwf-co}
{\bf [The Hamiltonian Term Enforcing Well-Formed Computations Configurations]}
Let $h_{wf-co}$ denote the Hamiltonian terms with the constraints that give an energy
penalty for any  state of the computation tracks that is not well-formed.
\end{definition}

\subsubsection{Implementing TM Steps in Propagation Terms}
\label{sec:implementingTM}

We show here how a single sweep of a pointer on Track $1$ from right to left will cause the TM configuration
on Tracks $4$ through $6$ to advance by one TM step. The details of this part of the construction, as well as the proof Claim \ref{cl:partialIsoT} below follow
\cite{GI} closely.  Note that since we are reducing from $\fpnexp$ instead
of $\fpqmaexp$, all of the Turing Machines in our construction, including the verifier $V$,
are classical reversible Turing Machines instead of Quantum Turing Machines. 

The sweep of the Track $1$ pointer begins from 
state $\arrL \rightBr$. The Track $1$ pointer
$\arrL$ then sweeps all the way to the
left end of the chain until the state $\leftBr \arrL$ is reached.
Each transition $\ket{\Dblank \arrL} \rightarrow \ket{\arrL \blank}$ will cause a unitary 
operation to be applied to the computation tracks below.
We will write the pointer  $\arrL$ generically without a number and later discuss  how the $i$-pointers work for different $i$.

We will need to execute the TM steps in which the head moves right or stays in the same place
separately than the moves in which the head moves left. In order to do this, we introduce a new state $q'_L$
for every $q_L \in Q_L$. We will call this set $Q'_L$. In a correct computation, it will never be
the case that the TM is in state $q'_L$ without a $\arrL$ directly above it on Track $1$.
It will also never be the case that there is a $\arrL$ in the same location as the head if the state
is some $q_R \in Q_R$. 

Define $S$ to be the set of standard basis states
associated with the computation tracks of two neighboring particles. 
$S$ is the set of standard basis states of the following
form:

\begin{center}
\begin{minipage}[c]{0.25\textwidth}
~~~~~~~~~$\fourcells{~p~}{a}{~q~}{b}$~~~~~~~~where:
\end{minipage}
\begin{minipage}[c]{0.6\textwidth}
$(p, q) ~\in ~(Q \times \{\blank\})~~ \cup ~~(\{\Dblank \} \times Q)~ ~\cup ~~
\{ \Dblank \Dblank, \blank \blank \}$ \\
$(a, b) ~\in~ 
\Gamma \times \Gamma$
\end{minipage}
\end{center}

We will be interested in two particular subsets of $S$.
$S_A$ is the subset of states 
from $S$ where $q \not\in Q_R$ and $p \not\in Q'_L$.
We also define $S_B$ to be the subset of states from $S$, where $q \not\in Q'_L$ and $p \not\in Q_R$.

We will show the transformation $P$ that maps $S_A$ to $S_B$.
The transformation $P$ works in two parts. At the moment that the
Track $1$ pointer moves left from the position just to the right of the head, we execute the move for that location. 
For every $\delta(q, a) = (q_R, b, R)$,  we write a $b$ into the old
location and move the head right into state $q_R$.
For every $\delta(q, a) = (q_N, b, N)$,  we write a $b$ into the old
location and change the state to  $q_N$ without changing
the location of the head.
For every $\delta(q, a) = (q_L, b, L)$,  we write a $b$ into the old
location and change the state to  $q'_L$ without changing
the location of the head. We need to defer the action of moving the head left until the Track $1$ pointer has access to the new location for the head. In the next step of the clock when the head is aligned with the $q'_L$, we move the head left and convert it to $q_L$.
In the rules shown below, the top row shows the state of Track $1$, the middle row shows the
state of Track $4$ and the bottom row shows the combined state of Tracks $5$ and $6$.
The first set  of transitions are defined for every $q \not\in Q'_L$:

\vspace{.1in}

\begin{center}
\begin{tabular}{ccc}
~~~$\delta(q, a) = (q_R, b, R)$~~~~~ & ~~~~~$\delta(q, a) = (q_N, b, N)$~~~~~ & ~~~~~$\delta(q, a) = (q_L, b, L)$~~~
\end{tabular}
\end{center}
\begin{equation}
\label{eq:TMmove2}
\left|~ \sixcells{\Dblank}{\arrL}{q}{\blank}{a}{\generic} ~\right\rangle \rightarrow \left|~ \sixcells{\arrL}{\blank}{\Dblank}{q_R}{b}{\generic}  ~\right\rangle
~~~~~~~~~
\left|~ \sixcells{\Dblank}{\arrL}{q}{\blank}{a}{\generic}~\right\rangle \rightarrow \left|~ \sixcells{\arrL}{\blank}{q_N}{\blank}{b}{\generic}~\right\rangle
~~~~~~~~~
\left|~ \sixcells{\Dblank}{\arrL}{q}{\blank}{a}{\generic}~\right\rangle \rightarrow \left|~ \sixcells{\arrL}{\blank}{q'_L}{\blank}{b}{\generic}~\right\rangle
\end{equation}

\vspace{.1in}

After this step, we are in a configuration 
in which the step has been performed,
except that moving the head left has been deferred.
If the TM is in a primed state that is aligned with the $\arrL$,
that triggers the execution of the right-moving step.
\begin{equation}
\label{eq:TMmove3}
    \left|~ \sixcells{\Dblank}{\arrL}{\Dblank}{q'_L}{\generic}{\generic} ~\right\rangle
\rightarrow
\left|~ \sixcells{\arrL}{\blank}{q_L}{\blank}{\generic}{\generic} ~\right\rangle
\end{equation} 
Otherwise, the head just sweep right leaving the other tracks unchanged:
\begin{equation}
\label{eq:TMmove4}
    \left|~ \sixcells{\Dblank}{\arrL}{\Dblank}{q_L}{\generic}{\generic} ~\right\rangle
\rightarrow
\left|~ \sixcells{\arrL}{\blank}{\Dblank}{q_L}{\generic}{\generic} ~\right\rangle
~~~~~~~~\left|~ \sixcells{\Dblank}{\arrL}{\Dblank}{q_N}{\generic}{\generic} ~\right\rangle
\rightarrow
\left|~ \sixcells{\arrL}{\blank}{\Dblank}{q_N}{\generic}{\generic} ~\right\rangle
\end{equation} 
\begin{equation}
\label{eq:TMmove5}
    \left|~ \sixcells{\Dblank}{\arrL}{\blank}{\blank}{\generic}{\generic} ~\right\rangle
\rightarrow
\left|~ \sixcells{\arrL}{\blank}{\blank}{\blank}{\generic}{\generic} ~\right\rangle
~~~~~~~~\left|~ \sixcells{\Dblank}{\arrL}{\Dblank}{\Dblank}{\generic}{\generic} ~\right\rangle
\rightarrow
\left|~ \sixcells{\arrL}{\blank}{\Dblank}{\Dblank}{\generic}{\generic} ~\right\rangle
\end{equation}

The criteria for $S_A$ is that the Track $4$ state of the left particle
is not in $Q'_L$ and the Track $4$ state of the right particle is not in $Q_R$.
Therefore, $P$ is well-defined for every state in $S_A$.
The criteria for $S_B$ is that the Track $4$ state of the right particle
is not in $Q'_L$ and the Track $4$ state of the left particle is not in $Q_R$.
Therefore, the range of $P|_{S_A}$ is in $S_B$.
Since 
the computation state for a pair is in $S_B$ after $P$ is applied to that pair,
the next pair over is in state $S_A$
and  $P$ can then be applied to this new pair.
Note that in a correct computation, if the Track $1$ state is $\ket{\Dblank \arrL}$, then
the computational state for those two particles will be in $S_A$ and if the 
Track $1$ state is 
$\ket{\arrL \blank}$, then
the computational state for those two particles will be in $S_B$. $P$ is one-to-one from $S_A$ to $S_B$. 
This is formalized in Claim \ref{cl:partialIsoT} below.
Since $S_A$ and $S_B$ are both subsets of $S$, $P$ can be extended to be a permutation on $S$.
The Hamiltonian term will then be:
$$I_{S} \otimes \ketbra{\arrL \blank}{\arrL \blank}
+ I_{S} \otimes \ketbra{\Dblank \arrL}{\Dblank \arrL}
+ P \otimes \ketbra{\arrL \blank}{\Dblank \arrL}
+ P^{\dag} \otimes \ketbra{\Dblank \arrL}{\arrL \blank}
$$

\begin{claim}
\label{cl:partialIsoT}
If the reduced transition rule $\delta$ is one-to-one, then Rules (\ref{eq:TMmove2}) through
(\ref{eq:TMmove5})
define a one-to-one function from $S_A$ to $S_B$. 
\end{claim}

\begin{proof}
Each rule defines $P$ 
on a disjoint subset of $S_A$. Moreover the images of all the rules are mutually disjoint, so we only need to show that each rule individually is one-to-one.

The rules in (\ref{eq:TMmove4}) and (\ref{eq:TMmove5}) act as the identity 
on the computational space of the two particles.
The rule in (\ref{eq:TMmove3}) is one-to-one.
Finally, we need to consider the rules given in
(\ref{eq:TMmove2}).
 Since the reduced $\delta$ function is one-to-one, then
 this transformation is also one-to-one.
\end{proof}

The two transition functions $\delta_{BC}$ and $\delta_{check}$ 
give
rise to permutations $P_{BC}$ and $P_{check}$ defined on $S$.

\begin{lemma}
\label{lem:TMonestep}
{\bf[ Left sweep of the $1$-pointer executes one step of the TM]}
Consider a Turing Machine $M$ with transition function $\delta$ that gives
rise to permutation $P$ according to Rules (\ref{eq:TMmove2}) through
(\ref{eq:TMmove5}). Suppose that the state of the computation tracks corresponds to a valid
configuration of $M$ with the head in state $q$ and current symbol $a$.
Suppose further that the following two conditions are satisfied:
\begin{enumerate}
    \item The head is not at the right-most location of the chain.
    \item The head is not at the left-most location of the chain or $\delta(q, a)$  does not cause the head to move left.
\end{enumerate}
Then one sweep of $\arrL$ from the right end of the chain to the left end of the chain
results in the configuration corresponding to one step of $M$.
\end{lemma}

\begin{proof}
If the computation tracks are in a valid TM configuration, then the configuration of Track $4$ has the form:
$\Dblank \cdots \Dblank ~ q ~ \blank \cdots \blank$. Since the head is not in the right-most location, there must be at least
one $\blank$ to the right of the $q$.
As the $1$-pointer $\arrL$ sweeps from right to left, the operation $P$ is applied
to each pair of particles from right to left. $P$ will act as the identity as long as
the Track $4$ is $\blank \blank$. Eventually, a configuration is reached in which Track $1$
is $\Dblank \arrL$ and Track $4$ is $q ~\blank$.  In the next move, as $\arrL$ moves one location to the left,
$P$ applies a rule from (\ref{eq:TMmove2}) which executes one move of the TM. 
If the head moves right into state $q_R$, the TM move is complete.
If the head stays in place and transitions to a state $q_N$, the move is also complete, but
the $1$-pointer  is now co-located with $q_N$. Rule (\ref{eq:TMmove4}) is applied and the $1$-pointer moves left (or changes direction if it is at the left end of the chain) without changing the computation tracks. 
If the head stays in place and transitions into a state $q'_L$, then according to the conditions of the lemma, the head is not in the left-most location and the move can be completed with 
Rule (\ref{eq:TMmove3}). The $1$-pointer  is now co-located with $q_L$. Rule (\ref{eq:TMmove4}) is applied and the $1$-pointer moves left without changing the computation tracks. 
$P$ acts as the identity for
the remainder of the left-sweep since the remaining pairs have $\Dblank \Dblank$ on Track $4$.
\end{proof}

\subsubsection{Propagation Terms for Each of the 8-Segments}
\label{sec:propsegments}

We will use the different $i$-pointers on Track $1$ to trigger different $\delta$
transition functions on Tracks $4$ through $6$. Recall from the clock construction
that the  $i$-segments vary according to whether they begin or end at the
left or right end of the chain and according to the number of back and forth round trips
performed by the $i$-pointer before advancing to the next segment. 
The list below shows the segments for an entire iteration for Tracks $1$ and $2$
and how many steps of each $\delta$ function are triggered on the computation
tracks during the segment. 
Segments $1$ through $3$ trigger a step of one of the two TMs as the $1$-pointer moves from right to left. A sweep of the $\arrL$ 
triggers the application of the corresponding $P$ operation to each pair
of particles from right to left.
Segments $5$ through $7$ trigger the inverse operation of the TMs.
A sweep of the $\arrR$ 
triggers the application of the corresponding $P^{-1}$ operation to each pair
of particles from left to right. This will cause the state of the computation tracks to return to their original state after each
iteration for Tracks $1$ and $2$.
Segments $4$ and $8$ are used for checking purposes only.
The $4$ and $8$-pointers act as the identity on the computation tracks. We will add Type I terms that cause a penalty
in the presence of a $4$ or $8$-pointer to penalize certain conditions on the computation tracks.

\begin{figure}[ht]
\label{fig:schedule}
\begin{center}
\begin{tabular}{ll}
  $1)$~~  $N-3$ steps of $\delta_{BC}$ ~~~~~~~~~~&  $5)$~~  $N-2$ steps of $(\delta_{check})^-1$\\
   $2)$~~ One step of $\delta_{BC}$ &  $6)$~~  One step of $(\delta_{BC})^{-1}$\\
   $3)$~~ $N-2$ steps of $\delta_{check}$  &  $7)$~~  $N-3$ steps of $(\delta_{BC})^{-1}$\\
   $4)$~~  Identity  &  $8)$~~  Identity
\end{tabular}
\end{center}
\caption{The Turing Machine steps executed in each $i$-segment.}
\end{figure}

Figure \ref{fig:cycle} shows a graphical representation of an iteration for Tracks $1$ and $2$, indicating the path of the $i$-pointer for each $i$-segment. Note that since the Track $1$ pointer for Segment $1$ begins
at the left end of the chain and ends at the right end of the chain,
it only triggers $N-3$ steps of $M_{BC}$. The single round trip made by the $2$-pointer causes one more step of $M_{BC}$ to be executed so that each Turing Machine always runs for exactly $N-2$ steps.
The same holds for Segments $6$ and $7$ with the inverse of $M_{BC}$.

\begin{figure}[ht]
  \centering
  \includegraphics[width=6.0in]{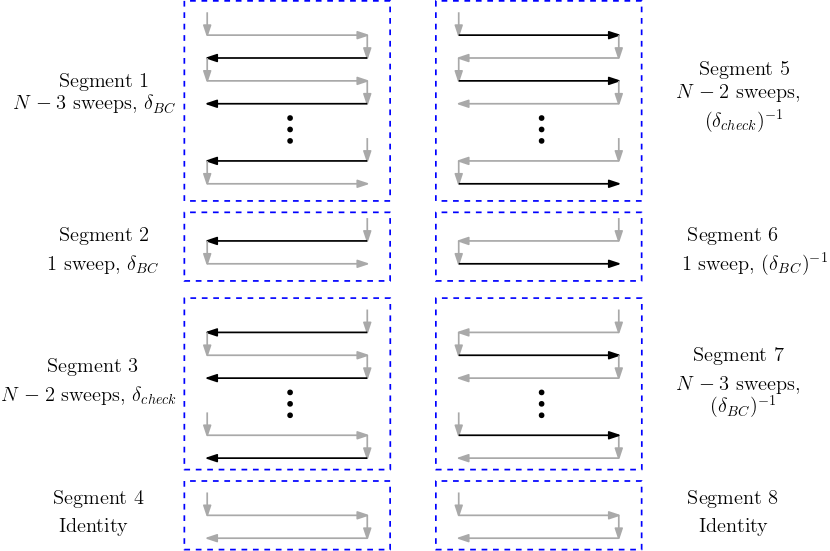}
\caption{The segments in an iteration for Tracks $1$ and $2$, along with the transition function that is triggered with each  sweep.
The dark arrows show the steps of the Track $1$ pointer that interact with the computation tracks. For Segments $1$ through $3$, the left-moving
$1$-pointer applies an operation to each pair of particles from right to left. For Segments $5$ through $7$, the right-moving
$1$-pointer applies an operation to each pair of particles from left to right.
The light arrows show the steps of the Track $1$ pointer which act as the identity on the computation tracks. The downward pointers correspond to a move at the end of the chain in which the Track $1$ pointer changes direction.
The horizontal pointers correspond to moves in the middle  of the chain in which the Track $1$ pointer moves right or left.}
\label{fig:cycle}
\end{figure}

We now define enhanced propagation terms for the clock states that act non-trivially on the computation space.
We will be replacing some of the transition rules defined in Section \ref{sec:clock} with enhanced terms that apply an 
operation on the computation tracks of the two particles to which the term applies. In each case, we will name the rule that
is removed and give the new term that will be added in its place.
The clock rules in Section \ref{sec:clock}
are expressed using arrows (such as $\Dblank \arrL \rightarrow \arrL \blank$) which implicitly represents the term:
$$\frac 1 2 [  \ketbra{\arrL \blank}{\arrL \blank}
+ \ketbra{\Dblank \arrL}{\Dblank \arrL}
+ \ketbra{\arrL \blank}{\Dblank \arrL}
+  \ketbra{\Dblank \arrL}{\arrL \blank}
]$$
These terms will be replaced by a term of the form:
$$\frac 1 2 [ I_{\cals} \otimes  \ketbra{\arrL \blank}{\arrL \blank}
+ I_{\cals} \otimes \ketbra{\Dblank \arrL}{\Dblank \arrL}
+ P \otimes \ketbra{\arrL \blank}{\Dblank \arrL}
+  P^{\dag} \otimes\ketbra{\Dblank \arrL}{\arrL \blank}
]$$
$\cals$ denotes the Hilbert spaces spanned by
$S$, the set of standard basis states for the computation tracks of two neighboring particles in a well-formed configuration.
In order to reduce notation, we will omit the factor of $1/2$ which should be applied to all the terms described in this subsection.

{\bf Segment $1$:} 
In the first segment, the $1$-pointers perform $N-3$ right-to-left sweeps in which we would like
to execute one step of $M_{BC}$. Transition Rules
TR-$1$ and TR-$2$ are placed by the following two terms:
\begin{align*}
     I_{\cals} \otimes \ketbrabig{\fourcells{\Dblank}{\neg \arrR}{\arrLone}{\neg \arrR}}{\fourcells{\Dblank}{\neg \arrR}{\arrLone}{\neg \arrR}}
& ~+~ I_{\cals} \otimes \ketbrabig{\fourcells{\arrLone}{\neg \arrR}{\blank}{\neg \arrR}}{\fourcells{\arrLone}{\neg \arrR}{\blank}{\neg \arrR}}\\
 +~ P_{BC} \otimes \ketbrabig{\fourcells{\arrLone}{\neg \arrR}{\blank}{\neg \arrR}}{\fourcells{\Dblank}{\neg \arrR}{\arrLone}{\neg \arrR}}
& ~+ ~(P_{BC})^{\dag} \otimes \ketbrabig{\fourcells{\Dblank}{\neg \arrR}{\arrLone}{\neg \arrR}}{\fourcells{\arrLone}{\neg \arrR}{\blank}{\neg \arrR}}
\end{align*}  

$$   I_{\cals} \otimes \ketbrabig{\fourcells{\Dblank}{\arrR}{\arrLone}{\blankLtwo}}{\fourcells{\Dblank}{\arrR}{\arrLone}{\blankLtwo}}
+ I_{\cals} \otimes \ketbrabig{\fourcells{\arrLone}{\DblankLtwo}{\blank}{\arrR}}{\fourcells{\arrLone}{\DblankLtwo}{\blank}{\arrR}}
 + P_{BC} \otimes \ketbrabig{\fourcells{\arrLone}{\DblankLtwo}{\blank}{\arrR}}{\fourcells{\Dblank}{\arrR}{\arrLone}{\blankLtwo}}
+ (P_{BC})^{\dag} \otimes \ketbrabig{\fourcells{\Dblank}{\arrR}{\arrLone}{\blankLtwo}}{\fourcells{\arrLone}{\DblankLtwo}{\blank}{\arrR}}$$


{\bf Segment $2$:} The next segment is the $2$-segment in which the $2$-pointer makes a single round trip, starting at the right end of the chain. As the $2$ pointer moves, right to left, we would like to trigger one step of $\delta_{BC}$. Transition Rule TR-$7$ is replaced by:
\begin{align}
I_{\cals} \otimes \ketbra{\arrLtwo \blank}{\arrLtwo \blank}
+ I_{\cals} \otimes \ketbra{\Dblank \arrLtwo}{\Dblank \arrLtwo}
& + P_{BC} \otimes \ketbra{\arrLtwo \blank}{\Dblank \arrLtwo}
+ (P_{BC})^{\dag} \otimes \ketbra{\Dblank \arrLtwo}{\arrLtwo \blank}
\end{align}

{\bf Segment $3$:} The next segment is the $3$-segment in which the $3$-pointer makes $N-3$ plus one half 
iterations between the two ends of the chain. Because the $3$-pointer starts on the right
end of the chain, there are  $N-2$ right-to-left sweeps. In each of these, we would like to execute one step of $M_{check}$.
The Rule TR-$13$ is replaced by:
\begin{align}
I_{\cals} \otimes \ketbrabig{\fourcells{\Dblank}{\generic}{\arrLthree}{\neg \arrR}}{\fourcells{\Dblank}{\generic}{\arrLthree}{\neg \arrR}}
& + I_{\cals} \otimes \ketbrabig{\fourcells{\arrLthree}{\generic}{\blank}{\neg \arrR}}{\fourcells{\arrLthree}{\generic}{\blank}{\neg \arrR}}\\
+ P_{check} \otimes \ketbrabig{\fourcells{\arrLthree}{\generic}{\blank}{\neg \arrR}}{\fourcells{\Dblank}{\generic}{\arrLthree}{\neg \arrR}}
& + (P_{check})^{\dag} \ketbrabig{\fourcells{\Dblank}{\generic}{\arrLthree}{\neg \arrR}}{\fourcells{\arrLthree}{\generic}{\blank}{\neg \arrR}}
\end{align}

{\bf Segment $5$:}  In the $5$-segment the $5$-pointer on Track $1$ makes $N-3$ plus one half iterations
starting at the left end of the chain, which results in $N-2$ left-to-right sweeps.
We use this segment to execute $N-2$ steps of $M_{check}$.
Transition Rule TR-$23$ is replaced with the rule shown below. Note that the rule of $P_{check}$ and $(P_{check})^{\dag}$ are 
switched so that $(P_{check})^{\dag}$ is applied in the forward direction and $P_{check}$ is applied in the reverse direction.
\begin{align}
I_{\cals} \otimes \ketbrabig{\fourcells{\arrRfive}{\neg \arrL}{\blank}{\generic}}{\fourcells{\arrRfive}{\neg \arrL}{\blank}{\generic}}
& + I_{\cals} \otimes \ketbrabig{\fourcells{\Dblank}{\neg \arrL}{\arrRfive}{\generic}}{\fourcells{\Dblank}{\neg \arrL}{\arrRfive}{\generic}}\\
+ (P_{check})^{\dag} \otimes \ketbrabig{\fourcells{\Dblank}{\neg \arrL}{\arrRfive}{\generic}}{\fourcells{\arrRfive}{\neg \arrL}{\blank}{\generic}}
& + P_{check} \ketbrabig{\fourcells{\arrRfive}{\neg \arrL}{\blank}{\generic}}{\fourcells{\Dblank}{\neg \arrL}{\arrRfive}{\generic}}
\end{align}
 

{\bf Segment $6$:} Segment $6$ is a single round trip with the Track $1$ pointer beginning
and ending at the right end of the track. The operation
$(P_{BC})^{\dag}$ is applied to each pair as the $1$-pointer sweeps from left to right:
\begin{align}
I_{\cals} \otimes \ketbra{\arrRsix \blank}{\arrRsix \blank}
+ I_{\cals} \otimes \ketbra{\Dblank \arrRsix}{\Dblank \arrRsix}
& + (P_{BC})^{\dag} \otimes \ketbra{\Dblank \arrRsix}{\arrRsix \blank}
+ P_{BC} \otimes \ketbra{\arrRsix \blank}{\Dblank \arrRsix}
\end{align}


{\bf Segment $7$:} The next segment is the $7$-segment in which the $7$-pointer makes $N-3$ plus one half 
round trips between the two ends of the chain. Because the $7$-pointer starts on the right
end of the chain, there are only $N-3$ left-to-right sweeps. 
Because the right-moving $1$-pointer also advances the pointer on Track $2$, there are two different transition rules that need to be augmented. The rule TR-$31$ is replaced by:
\begin{align}
I_{\cals} \otimes \ketbrabig{\fourcells{\arrRseven}{\neg \arrL}{\blank}{\neg \arrL}}{\fourcells{\arrRseven}{\neg \arrL}{\blank}{\neg \arrL}}
& + I_{\cals} \otimes \ketbrabig{\fourcells{\Dblank}{\neg \arrL}{\arrRseven}{\neg \arrL}}{\fourcells{\Dblank}{\neg \arrL}{\arrRseven}{\neg \arrL}}\\
+ (P_{check})^{\dag} \otimes \ketbrabig{\fourcells{\Dblank}{\neg \arrL}{\arrRseven}{\neg \arrL}}{\fourcells{\arrRseven}{\neg \arrL}{\blank}{\neg \arrL}}
& + P_{check} \ketbrabig{\fourcells{\arrRseven}{\neg \arrL}{\blank}{\neg \arrL}}{\fourcells{\Dblank}{\neg \arrL}{\arrRseven}{\neg \arrL}}
\end{align}
The rule TR-$32$ is replaced by:
\begin{align}
I_{\cals} \otimes \ketbra{\fourcells{\arrRseven}{\DblankLtwo}{\blank}{\arrL}}{\fourcells{\arrRseven}{\DblankLtwo}{\blank}{\arrL}}
& + I_{\cals} \otimes \ketbra{\fourcells{\Dblank}{\arrL}{\arrRseven}{\blankLtwo}}{\fourcells{\Dblank}{\arrL}{\arrRseven}{\blankLtwo}}
+ (P_{check})^{\dag} \otimes \ketbra{\fourcells{\Dblank}{\arrL}{\arrRseven}{\blankLtwo}}{\fourcells{\arrRseven}{\DblankLtwo}{\blank}{\arrL}}
& + P_{check} \ketbra{\fourcells{\arrRseven}{\DblankLtwo}{\blank}{\arrL}}{\fourcells{\Dblank}{\arrL}{\arrRseven}{\blankLtwo}}
\end{align}


\vspace{.1in}

\begin{definition}
\label{def:Hprop}
{\bf [$h_{prop}$ Sum of all the Propagation Terms]}
Let $h_{prop}$ denote the sum of all of the transition terms specified here
and in Section \ref{sec:clock} on the clock transition rules.
\end{definition}

\subsubsection{Additional Penalty Terms for the Finite Construction}
\label{sec:addconstraints}


We prove in Lemma \ref{lem:clockCycle} that
each clock step produces a unitary operation on the computation tracks, so that if the computation begins a iteration
for Tracks $1$ and $2$ in state $\ket{\phi}$, then after $p(N)$ clock steps of the transition function, Tracks $1$ and $2$ will be
back to their original configuration and the computation tracks will have returned to state $\ket{\phi}$. 
Thus, each cycle (for all three tracks) consists of the same computation repeated $2T+1$ times, where $T$ is the
length of the timer on Track $3$. 

The analysis in Section  \ref{sec:groundenergy}  defines  sets of standard basis states such that each set is parameterized by a $4$-tuple $(T, v, y, w)$, where $\ket{v}\ket{y}\ket{w}$ is a well-formed configuration for the computation tracks:
$\ket{v}$ represents a state for Tracks $4$ and $5$, $\ket{y}$ represents the first $m$ bits on the witness track, and
$\ket{w}$ represents the remaining $N-2-m$ bits on the witness track.
The states in the set parameterized by $(T, v, y, w)$
represent the evolution of the combined clock and computation state when the state of the 
computation at time $(0,0)$ is $\ket{v}\ket{w}\ket{w}$ and the Track $3$ timer has length $T$.
Let $H_{N,prop}$ be the Hamiltonian resulting from applying the propagation term $h_{prop}$ to every pair of neighboring particles in a 1D chain of length $N$.
We argue in the analysis that
 $H_{prop,N}$ is closed on each set of states.
 Moreover the
 matrix corresponding to $H_{prop,N}$ restricted to the set parameterized by $(T, v, y, w)$ is $1/2$ times the Laplacian of a cycle graph on $(2T+1) p(N)$ vertices.
 We will refer to the Hamiltonian restricted to the subspace spanned by one of the sets  as a {\em block}. 
 Since the remaining terms are diagonal in the set of standard basis states, we can analyze the lowest eigenvalue of each block separately. 

In the absence of any additional penalty terms, the ground state of each the  blocks has energy $0$
and is the uniform superposition of all the states in the set
corresponding to that block. In this subsection, we describe additional penalty terms that apply to states in these blocks. Some of the terms will result in
periodic costs that apply to at least $2T+1$ states in a set, occurring once every $p(N)$ clock steps.
These will be used to check for incorrect values for the parameters $T$, $\ket{v}$, $\ket{y}$, or $\ket{w}$.
The last term will apply to all the blocks and will be used to penalize shorter cycles over longer ones. In the analysis, this will be the mechanism for assigning an effective cost to {\em no} query answers.

The first set of  penalty terms  enforce the condition
that  the computation state at the start of a Track 1-2 iteration is the initial configuration. The condition is checked
as the $\arrLeight$ pointer moves from right to left which takes place just before $M_{BC}$ starts (in the forward direction).
In the correct initial configuration, the contents of Track $5$ are
$\leftBr  1 ~\#^*  \rightBr$, where $\#$ is the blank symbol.
The state should be $q_0$ and the head should be pointing to the cell with the $1$.

Throughout most of the discussion of the Turing Machine steps, we have combined the work tape (Track $5$) and the witness tape (Track $6$). However,  the first of the two rules given below
apply to Tracks $1$ and $5$ only and 
implicitly acts as the identity on all other tracks. The third rule
applies to Tracks $1$, $4$, and $5$:
$$h_{init} = \sum_{a \in \Gamma, ~x \in \Gamma - \{ \#\} } \ketbrabig{\fourcells{\Dblank}{a}{\arrLeight}{x}}{\fourcells{\Dblank}{a}{\arrLeight}{x}}
~~~~~+~~~~~\sum_{x \in \Gamma - \{ 1 \} } \ketbrabig{\threecellsL{\arrLeight}{x}}{\threecellsL{\arrLeight}{x}}
~~~~~+~~~~~
\ketbrabig{\fourcellsL{\arrLeight}{\neg q_0}{1}}{\fourcellsL{\arrLeight}{\neg q_0}{1}}
$$
The next set of terms check that the length of the time on Track $3$ is in fact equal to the
correct timer length as computed by $M_{TV}$. The condition is checked by the $4$-pointer $\arrLfour$
which completes its round trip just after the $N-2$ steps of $M_{TV}$ in the forward direction. 
At this point, 
the work tape will  have the following form, assuming that the computation started from a correct
initial configuration at time $0$: $$\leftBr (\sigma_A + \sigma_R) (\sigma_X)^T (\Gamma - \{\sigma_X, \sigma_A, \sigma_R\})^*  \rightBr$$ where $T = T(x,y)$.
Recall that a well-formed Track $3$ has the form
$\leftBr \DblankLthree^* \arrR \blankLthree^* \dead^* \rightBr$
or $\leftBr \DblankLthree^* \arrL \blankLthree^* \dead^* \rightBr$ and the length is the number of $\DblankLthree$ or
$\blankLthree$ symbols. 
We add a penalty if there is any
location in which Track $5$ has a symbol from
$\{\sigma_A, \sigma_R, \sigma_X\}$  and Track $3$ does not have a symbol
from
$\{\DblankLthree, \blankLthree, \arrR, \arrL \}$.
Similarly, we add a penalty if there is any
location in which Track $5$ has a symbol that is not from
$\{\sigma_A, \sigma_R, \sigma_X\}$  and Track $3$ has a symbol
from
$\{\DblankLthree, \blankLthree, \arrR, \arrL \}$.
This effectively gives a penalty if the length of the Track $3$
timer is anything other than $T(x,y)$.
The following rules apply to Tracks $1$, $3$, and $5$:
$$h_{length} = 
\sum_{\substack{s \in \{\DblankLthree, \blankLthree, \arrR, \arrL \}\\ a ~\not\in \{\sigma_A, \sigma_R, \sigma_X \}} } \ketbrabig{\threecellsvert{\arrLfour}{s}{a}}{\threecellsvert{\arrLfour}{s}{a}}
~~~~~+~~~~~
\sum_{b \in \{\sigma_A, \sigma_R, \sigma_X \} }
\ketbrabig{\threecellsvert{\arrLfour}{\dead}{ b}}{\threecellsvert{\arrLfour}{\dead}{b}}
$$
There is another penalty if the Verifier $V$ rejects on any of its computations which is also checked by the
$4$-pointer $\arrLfour$. Track $5$ of the
left-most particle  encodes
the result of the verifier computations. The symbol will be $\sigma_R$ if any of the verifier computations ended in a reject state, and we wish to penalize rejecting computations. 
The following rules apply to Tracks $1$ and $5$:
$$h_{V} = \ketbrabig{\threecellsL{\arrLfour}{\sigma_R}}{\threecellsL{\arrLfour}{\sigma_R}}.$$

Finally, there is a penalty even for correct computations and clock configurations for every cycle for Tracks $1$, $2$, and $3$. This cost is incurred whenever the left-moving $4$-pointer reaches the left end of the chain and the Tracks $2$ and $3$ pointers are also at the left end of the chain. This condition
occurs exactly once in each cycle.
The rule below applies to Tracks $1$, $2$, and $3$:
$$h_{final} = \frac 1 2 \ketbrabig{\fourcellsL{\arrLfour}{\arrR}{\arrR}}{\fourcellsL{\arrLfour}{\arrR}{\arrR}} + \frac 1 2 \ketbrabig{\fourcellsL{\arrRfive}{\arrR}{\arrR}}{\fourcellsL{\arrLfour}{\arrR}{s\arrR}}.$$

Note that $h_{final}$ adds a penalty  of $1/2$  to two consecutive clock states in every block.
The final $2$-particle Hamiltonain term is:
$$h = h_{prop} + h_{wf-cl} + h_{wf-co} + h_{cl} + h_{init} + h_{length} + h_V + h_{final}.$$
For each of these terms $h_*$, we will use $H_{N,*}$ to refer to the Hamiltonian on a chain of length $N$
obtained by applying the $2$-particle term $h_*$ to each pair of neighboring particles in the chain.

\subsection{Analysis of the ground energy of the Hamiltonian} 
\label{sec:analysis}

This section contains the analysis of the ground energy of the Hamiltonian and a proof that if the chain length is $N(x)$,
then the value $f(x)$ can be determined by a sufficiently accurate ($1/poly(N)$) estimate of the ground energy.
Section \ref{sec:cycliccomp}  establishes that the computation is cyclical in nature, meaning
that if the computation tracks start in state $\ket{\phi}$ at the beginning of a Track 1-2 iteration, then
the computation tracks will return to the same state $\ket{\phi}$ at the beginning of the next Track 1-2 iteration.
This computation is repeated $2T+1$ times over the course of the entire cycle (for Tracks $1$, $2$, and $3$).

Section \ref{sec:groundenergy} contains the main analysis of the ground energy of the Hamiltonian. 
The Hamiltonian has a block-diagonal structure which allows us to analyze the energy of each block independently.
We first eliminate blocks that correspond to incorrect clock states. Then we eliminate blocks in which
the initial configuration of the computation is incorrect.
Thus, we can assume that we have a correct computation which uses the correct string $x$ as the input.
A block is also parameterized by $y$ the first $m$ bits of the witness tape which are used as guesses for the responses of the oracle. We eliminate blocks in which the timer length on Track $3$ does not correspond to the 
correct $T(x, y)$ and  blocks in which $y$ contains an incorrect {\em yes} guess. These all have periodic costs that occur at least once every $p(N)$ clock steps in an iteration.
Finally, we give an exact expression for the lowest eigenvalue within  blocks which do not have a periodic cost, and show that blocks that correspond to the correct $y$ will have the smallest eigenvalue. Then we show how the value of $T(x, y)$ and
$f(x)$ can be extracted from a good approximation of this smallest eigenvalue.

Section \ref{sec:bracket} finally addresses the assumption we have been making all along that the state is in the span of all bracketed states. Additional terms are added to the Hamiltonian which ensure that the ground state of the Hamiltonian is in fact in
the span of all bracketed states. Theorem \ref{th:finitehardness} ties all the various pieces of the finite  construction together and gives
the final result on the hardness of {\sc Function}-TIH.
Section \ref{sec:eigenvaluebounds} gives the proofs for bounds on the lowest eigenvalue of various matrices used in the analysis.

\subsubsection{The Computation Repeats}
\label{sec:cycliccomp}

We prove here that the the computation tracks return to the same state
at the end of each iteration for Tracks $1$ and $2$.
The notation $\ket{c_{T, s, t}}$ refers to the enumeration of correct clock
configurations given in Definition \ref{def:corrclockenum}. The value of $s$ is in the range $0, \ldots, 2T$ and determines the 
configuration of the Track $3$ timer. The value of $t$ is in the range $0, \ldots, p(N)-1$
and determines the configurations for Tracks $1$ and $2$. 
Note that the last clock configuration in an iteration for Tracks $1$ and $2$
is $\ket{c_{T, s, p(N)-1}}$ in which the state of Track $1$ is $\leftBr \arrLeight \blank^* \rightBr$
and the state of Track $2$ is $\leftBr \arrL \blank^* \rightBr$.

\begin{lemma}
\label{lem:clockCycle}
{\bf [The Computation is Cyclic]}
Consider a sequence of clock steps  from $\ket{c_{T, s, p(N)-1}}$ to
$\ket{c_{T, s+1, p(N)-1}}$. If the computation tracks start out in a state $\ket{\phi}$,
then the state of the computation tracks is unchanged at the end of the sequence of clock steps.
That is, $p(N)$ clock steps applied to 
 $\ket{c_{T, s, p(N)-1}}\ket{\phi}$ will result in the state $\ket{c_{T, s+1, p(N)-1}}\ket{\phi}$
\end{lemma}

\begin{proof}
Let $P$ be a permutation applied to the computation tracks of two neighboring particles. The particles
will be numbered from left to right $0$ through $N-1$, so the non-bracketed particles are $1$ through $N-2$.
Applying $P$ to particles $i$ and $i+1$ will be denoted $P^{(i,i+1)}$.
In segments $1$ and $2$, each right-to-left sweep of the $1$-pointer applies $P_{BC}$ to every pair of particles:
$P^{(1,2)}_{BC} P^{(2,3)}_{BC} \cdots P^{(N-3, N-2)}_{BC}$. There are $N-3$ sweeps in Segment $1$ and one sweep in Segment $2$, so the state of the computation tracks after Segment $2$ will be:
$$(P^{(1,2)}_{BC} P^{(2,3)}_{BC} \cdots P^{(N-3, N-2)}_{BC})^{N-2} \ket{\phi}$$
Segment $3$ has $N-2$ right-to-left in which $P_{check}$ is applied to every pair of particles:
$$(P^{(1,2)}_{check} P^{(2,3)}_{check} \cdots P^{(N-3, N-2)}_{check})^{N-2}(P^{(1,2)}_{BC} P^{(2,3)}_{BC} \cdots P^{(N-3, N-2)}_{BC})^{N-2} \ket{\phi}$$
Then in Segment $5$, there are $N-2$ right-to-left sweeps in which $(P_{check})^{\dag}$
is applied to every pair of particles, starting with $(1,2)$ and ending with $(N-3, N-2)$,
which results in the overall operation:
$$((P^{(N-3, N-2)}_{check})^{\dag} \cdots (P^{(2,3)}_{check})^{\dag}   (P^{(1,2)}_{check})^{\dag})^{N-2}$$
Finally, Segments $6$ and $7$ apply a total of $N-2$ right-to-left sweeps in which
$(P_{BC})^{\dag}$ is applied to every pair of particles, resulting in:
$$((P^{(N-3, N-2)}_{BC})^{\dag} \cdots (P^{(2,3)}_{BC})^{\dag}   (P^{(1,2)}_{BC})^{\dag})^{N-2}$$
The net effect on the computation tracks is that they return to their orignal state $\ket{\phi}$.
 \end{proof}

\begin{coro}
\label{cor:clockCycle}
For any $0 \le s < 2T+1$ and any $0 \le t < p(N)$, and $\phi \in C_N$,
 $p(N)$ transition steps applied to 
 $\ket{c_{T, s, t}}\ket{\phi}$ will result in the state $\ket{c_{T, s+1, t}}\ket{\phi}$
\end{coro}

\subsubsection{Analysis of the Ground Energy}
\label{sec:groundenergy}

Let $\calh_{wf}$ denote the span of
the bracketed well-formed clock configurations tensored with the  span of $C_N$, the set of well-formed computation configurations on a chain of length $N$. It will be useful at this point to separate the state, head location, and work tape (Tracks $4$ and $5$)
from the witness track (Track $6$). We will also separate the first $m$ bits of the witness track $y$ from the remaining $N-2-m$ bits, which we refer to as $w$.
We will denote a well-formed standard basis state in $C_N$ as
$\ket{v} \ket{y}\ket{w}$, where $\ket{v}$ is a standard basis state for Tracks $4$ and $5$,
and $\ket{y}\ket{w}$ is a standard basis state for Track $6$.

According to Lemmas \ref{lem:clockGraph1} and \ref{lem:clockGraph2}, the configuration graph defined on the set of well-formed clock configurations 
consists of paths of relatively short length ($\le p(N)$) and cycles of the correct clock
states. We will first handle the the part of the Hilbert space in which the clock configuration
is in one of these short paths.
Let $\calh_{path}$ be the Hilbert space spanned by the states whose clock configuration
is in a path in the configuration graph.
$H_{N, prop}$ is the only term in the Hamiltonian that is not diagonal in the standard basis
and $H_{N, prop}$ is closed on $\calh_{path}$. Therefore, we can lower bound the the eigenvalues of
$H$ restricted to $\calh_{path}$ separately from the rest of the space.

\begin{lemma}
\label{lem:lbpath}
{\bf [Eliminating Incorrect Clock States]}
The lowest eigenvalue of $H_N$ restricted to $\calh_{path}$ is at least $(1 - \cos(\pi/(2p(N)+1)))$.
\end{lemma}

\begin{proof}
Let $\ket{c_0}, \ldots, \ket{c_l}$ be the clock states in  a path in the clock configuration graph.
According to Lemma \ref{lem:clockGraph2}, $l \le p(N)-1$.
Consider a state $\ket{v} \ket{y}\ket{w}$ for the computation tracks from $C_N$. 
Starting in state $\ket{c_0} \ket{v} \ket{y}\ket{w}$, 
after $t$ applications of the transition function, the state of the
system will be $\ket{c_t} \ket{v'}\ket{y}\ket{w}$, where $\ket{v'}$ is determined by $(c_0, v, y, w, t)$.
We will call this state $\ket{\phi(c_0, v, y, w, t)}$, where $\phi(c_0, v, y, w, t) \in C_N$.
Since the operations performed on the computation tracks are always permutations,
the following set of states is ortho-normal:
$$S(c_0) = \{ \ket{c_t}\ket{\phi(c_0, v, y, w, t)}: ~\mbox{for}~\ket{v}\ket{y}\ket{w} \in C_N, ~0 \le t \le l \}.$$
The propagation term $H_{N, prop}$ is block diagonal on the space spanned by this set. Each block is specified by a choice of $\ket{v} \ket{y}\ket{w}$. 
If we express $H_{N, prop}$ in the standard basis, then each block is 
$l \times l$ and is the Laplacian of a path of length $l$.

The set $S(c_0)$ is the only set of states in the basis of $\calh_{wf}$
that has clock states $\ket{c_0}, \ldots, \ket{c_l}$.
Since the terms $H_{N, init} + H_{N, length} + H_{N,V}$ are diagonal in the standard basis, they
are all closed on the span of $S(c_0)$. Since these terms are also positive semi-definite,
we need only lower bound the smallest eigenvalue of $H_{N, prop} + H_{N, cl}$ on this space.
According to Lemma \ref{lem:clockGraph2}, the final clock configuration in the path
is illegal, which means that
the upper-left element of the matrix will have
an additional $+1$ from $H_{N, cl}$.
According to Lemma \ref{lem:pathEig}, the smallest eigenvalue in this space
is at least $(1 - \cos(\pi/(2p(N)+1)))$.
\end{proof}

We can now focus on the space of correct clock states. 
According to Corollary \ref{cor:clockCycle}, if the system starts out in state
$\ket{c_{T, s, 0}} \ket{\phi}$, where $\phi \in C_N$,
then $p(N)$ steps later, the clock configuration will be $\ket{c_{T, s+1, 0}}$ and the computation
tracks will have returned to state $\ket{\phi}$.
Therefore, if we fix the computation state to be some $\ket{v}\ket{y}\ket{w} \in C_N$ at time $(0,0)$,
 the state of the computation tracks at some time $(s,t)$, depends only on $t$ and not on $s$ or $T$.
We will call this state $\ket{\phi(v, y, w, t)}$. 
Since the operations performed on the computation tracks are always permutations,
the following set of states is orth-onormal:
\begin{align*}
S_{cyc} =  \{ \ket{c_{T,s,t}}\ket{\phi(v, y, w, t)}: ~~\mbox{for} & ~~\ket{v}\ket{y}\ket{w} \in C_N, \\
&~~T \in \{0,\ldots,N-3\},~~ s \in \{0, \ldots, 2T\},~~ t \in \{0, \ldots, p(N)-1\}\}.
\end{align*}
Note that since the verifier $V$ is classical, $\ket{\phi(v, y, w, t)}$ is also a classical standard basis state.
The matrix for $H_{N, prop}$ expressed in the basis $S$ is block diagonal. Each block is
specified by a value for $T$ and starting configuration $\ket{v}\ket{y}\ket{w}$ for the 
computation tracks. The block has dimension $(2T+1)p(N)$ and is $1/2$ times the Laplacian of a cycle of length $(2T+1)p(N)$. All other terms are diagonal in this basis.

We will refer to the correct starting configuration for Tracks $4$ and $5$ as
$\ket{\vinit}$. In this correct starting configuration, the contents of Track $5$ are
$\leftBr ~1~\#^* \rightBr$, the state on Track $4$ is $q_0$, and the head
is at the same location as the $1$ on Track $5$.

There is a correct $x$ for the length $N$ of the chain. This is the string that is written on the computation track after $M_{BC}$ runs for $N-2$ steps starting in $\ket{\vinit}$.
The string $y$  is used as the guesses for the responses to the oracle calls. Thus, once $x$ and $y$ are fixed, the inputs to the oracles
$x_1, \ldots, x_m$ are determined. In addition, the function $T(x,y)$ is also determined.
Since the state of the witness track (including $\ket{y}$) does not change over time,
the values of $x_1, \ldots, x_m$ and $T(x,y)$ are well-defined for a block.
The verifier $V$ will be run on
 the oracle queries for the cases in which the string $y$ guesses that the query
string is in $L$.
Define the set
$$Y_{rej} = \{ y \mid ~\mbox{for some}~i\in [m], ~y_i = 1~\mbox{and}~x_i \not\in L\}.$$
In order words, $Y_{rej}$ is the set of oracle guesses  for which there is at least one oracle call
in which the guess is {\em yes} but the correct answer is {\em no}.

We will now partition $S_{cyc}$ into four different sets and consider the subspace spanned by each set separately. The states from $S_{cyc}$ will be put into one of the four sets in blocks. That is, each block is completely
contained in one of the four sets.
\begin{description}
    \item $S_1$ consists of the states of the form $\ket{c_{T,s,t}}\ket{\phi(v, y, w, t)}$ in which
     $\ket{v} \neq \ket{\vinit}$. These are the states from the blocks that do not start in the correct initial configuration.
     \item $S_2$ consists of the states of the form $\ket{c_{T,s,t}}\ket{\phi(p, y, w, t)}$, where  $\ket{v} =\ket{\vinit}$, but $T(x, y)$ is not equal to the length of Track $3$ timer.
     \item $S_3$ consists of the states of the form $\ket{c_{T,s,t}}\ket{\phi(v, y, w, t)}$ where  $\ket{v} =\ket{\vinit}$, $T(x, y)$ is  equal to the length of the Track $3$ timer,  and $y \in Y_{rej}$.
     \item $S_4$ consists of the states of the form $\ket{c_{T,s,t}}\ket{\phi(v, y, w, t)}$ where  $\ket{v} =\ket{\vinit}$, $T(x, y)$ is  equal to the length of the Track $3$ timer,  and $y \not\in Y_{rej}$.
\end{description}
$H_{N, prop}$ is closed on each $S_i$ because the states from each block are completely contained in $S_i$ or are disjoint from $S_i$. All of the other Hamiltonian terms besides $H_{N, prop}$  are all diagonal in the standard basis and are therefore diagonal on the basis $S_{cyc}$.

If $N = N(x)$ for some $x$, then
staring with the correct input configuration $\ket{\vinit}\ket{y}\ket{w}$, if we run the process
$M_{BC}$ for $N-2$ steps followed by $M_{check}$ for $N-2$ steps, the contents of the work track
should be {\sc Out}$(x, yw, M, V)$, the same tape contents as $M_{TV}$ on input
$(x, yw)$ when it halts (See Figure \ref{fig:MTVpseudo}).
Note that the second input parameter to $M_{TV}$ is $\myw$,  is the entire binary string on Track $6$, which is the concatenation of $y$ and $w$.
In order for the correct output to be produced, $M_{check}$ needs to hit a final state within $N-2$
steps and the computation needs to stay within the finite bounds of $C_N$.
Lemmas \ref{lem:existsMTV} and \ref{lem:gentv} establish that there is an $M_{TV}$ (used to form $M_{check}$)
so that the correct output is in fact produced for most $x$. Throughout this section,
the lemmas will assume that we the input $x$ is {\em good} in this sense.

\begin{lemma}
\label{lem:lb-s12}
{\bf [Eliminating States from $S_1$, $S_2$, and $S_3$]}
Suppose that 
the process $(M_{BC})^{N-2}$ followed
by $(M_{check})^{N-2}$ starting in configuration $\ket{\vinit}\ket{y}\ket{w}$
produces {\sc Out}$(x, yw, M, V)$ on the work track.
Then
the smallest eigenvalue of  $H$ in the span of $S_1 \cup S_2 \cup S_3$, is at least
$(1/8)(1 - \cos(\pi/(2p(N)+1)))$.
\end{lemma}

\begin{proof}
Since the terms are all closed on the spans of $S_1$, $S_2$, and $S_3$, we can lower bound each space separately. 
For a fixed $(T, x, y, w)$, consider the block defined by states of the form
$\ket{c_{T,s,t}}\ket{\phi(v, y, w, t)}$. If $\ket{v} \neq \ket{\vinit}$, then the configuration of Tracks $4$ and $5$ is
not equal to $\ket{\vinit}$ at time $(0,0)$. Since the $\arrLeight$ does not alter Tracks $4$ and $5$ as it sweeps left,
the state of Tracks $4$ and $5$ will not be equal to $\ket{\vinit}$
for the preceding $N-2$ time steps.  There will be at least
one point in time when the $\arrLeight$ is sweeping left that an illegal configuration from $H_{N, init}$ is hit. Call this time $t'$. Then for all $s$:
$$H_{N, init} \ket{c_{T,s,t'}}\ket{\phi(v, y, w, t')} = \ket{c_{T,s,t'}}\ket{\phi(v, y, w, t')}$$
Note that there may be more locations where an illegal configuration is reached but we need only consider one as the others can only increase the energy of a state.
The diagonal matrix corresponding to $H_{N,init}$ will have a 
$1$ at each location of the form
$t' + s \cdot p(N)$, where $s \in \{0, \ldots, 2T\}$. If we add this matrix to $H_{N, prop}$ which is
$1/2$ times the Laplacian of a cycle graph of length $(2T+1) \cdot p(N)$, then according
to Lemma \ref{lem:periodicEig},
the smallest eigenvalue of the sum of the matrices is at least $(1/8)(1 - \cos(\pi/(2p(N)+1)))$.
This reasoning can be applied to each block in the span of $S_1$ which gives a lower bound of
$(1/8)(1 - \cos(\pi/(2p(N)+1)))$ for any eigenvalue in this space.

A similar reasoning can be applied to the span of $S_2$ using the operator $H_{N, length}$.
If a block is inside $S_2$, then it is using the correct initial configuration $\ket{\vinit}$
which produces the correct input $x$. Therefore the string $y$ from the witness track and $x$
determine the correct value $T(x,y)$ for the length of the clock on Track $3$. 
This correct value is written in unary on the computation track at the end of the computation $M_{TV}$, which by assumption
of the lemma corresponds to the contents of the work track after $(M_{BC})^{N-2}$ followed
by $(M_{check})^{N-2}$  applied to $\ket{\vinit}\ket{y}\ket{w}$. The contents of the track are
$(\sigma_A+\sigma_R)(\sigma_X)^{T(x,y)}$. Since the first $\sigma_A$ or $\sigma_R$ symbol counts as a unary digit, the value encoded is $T(x,y)+1$. This is checked against the length of the prefix
$\DblankLthree^j \arrR \blankLthree^{T-j}$ or 
$ \DblankLthree^{T-j} \arrL \blankLthree^{j} $
on Track $3$, which is equal to the timer length $T$ plus $1$. (The plus $1$ comes from the pointer on Track $3$.)
The Track $1$ pointer $\arrLfour$
will reach an illegal configuration if and only if the Track $3$ clock does not have the correct length $T(x,y)$. 
As with the reasoning for $S_1$, there may be more than one such violation, but we need only consider one that
occurs regularly at times $(s, t')$ for some fixed $t'$ and $s \in \{0, \ldots, 2T\}$.
Since $S_2$ consists of those blocks in which the Track $3$ timer has the incorrect length, the smallest
eigenvalue of $H_{N, prop} + H_{N, length}$ in the span of $S_2$ will be at least $(1/8)(1 - \cos(\pi/(2p(N)+1)))$.

A similar reasoning can be applied to the span of $S_3$ using the operator $H_{N, V}$.
If a block is inside $S_3$, then it is using the correct initial configuration $\ket{\vinit}$
which produces the correct input $x$. Therefore the string $y$ from the witness track and $x$
determine the oracle queries $x_1, \ldots, x_m$. If for any $i \in \{1, \ldots, m\}$,
$y_i = 1$ and $x_i \not\in L$, the verifier will be run on input $x_i$ and will reject,
regardless of the string used as witness.
By assumption
of the lemma, the contents of the work track after the process $(M_{BC})^{N-2}$ followed
by $(M_{check})^{N-2}$ starting in configuration $\ket{\vinit}\ket{y}\ket{w}$ will be the correct output of $M_{TV}$
on input $(x, yw)$.
This means that  there will be an $\sigma_R$ symbol written in the first
location of the work tape  after Segment $4$.
In this case,  Track $1$ pointer $\arrLfour$
will reach an illegal configuration:
$$\fourcells{\Dblank}{\vdash}{\arrLfour}{\sigma_R}$$
Since $S_3$ consists of those blocks in which $y \in Y_{rej}$, the smallest
eigenvalue of $H_{N, prop} + H_{N, V}$ in the span of $S_3$ will be at least $(1/8)(1 - \cos(\pi/(2p(N)+1)))$.
\end{proof}

Finally, we arrive at the space spanned by $S_4$. 
Define
$$S_{y, w} = \{\ket{c_{T(x,y), s, t}} \ket{\phi(\vinit, y, w, t)}: 
0 \le s \le 2T(x,y), 0 \le t < p(N) \}$$
$S_{y,w}$ is the set of standard basis states for a single block,
where the initial state is $\ket{\vinit}\ket{y}\ket{w}$ and the Track $3$ timer
is the correct $T(x,y)$.
Define $S_y$ to be the union of $S_{y,w}$ over all possible $w$.
Note that $S_4$ is the union of all $S_{y}$, where $y \not\in Y_{rej}$.
All of the terms in $H$ are closed on the span of each $S_{y,w}$.
\begin{lemma}
\label{lem:Exy}
{\bf [Smallest Eigenvalues for Blocks in $S_4$]}
Suppose that 
the process $(M_{BC})^{N-2}$ followed
by $(M_{check})^{N-2}$ starting in configuration $\ket{\vinit}\ket{y}\ket{w}$
produces {\sc Out}$(x, yw, M, V)$ on the work track.
If $y \not\in Y_{rej}$, then the minimum eigenvalue of the space spanned by $S_y$ is exactly
$$    E(x,y) =  \left( 1 - \cos\left( \frac{ \pi}{L+1} \right)\right),$$
where $L = (2 T(x,y)+1) \cdot p(N)$.
\end{lemma}
\begin{proof}
All of the terms in $H$ except for
$H_{N, prop}$, $H_{N, V}$, and $H_{N, final}$ are $0$ on the entire space spanned by $S_4$
and are therefore $0$ on any $S_y$ such that $y \not\in Y_{rej}$.
Consider $H_{N, prop}$ restricted to the span of a particular
$S_{y,w}$ and expressed in the $S_{y,w}$ basis.
This matrix is 
$1/2$ times the Laplacian of a cycle of length $(2 T(x,y)+3) \cdot p(N)$.
The term $H_{N, final}$ adds a $+1/2$ to two consecutive locations along the diagonal. Since $H_{N, V}$
is semi-positive definite,  by Lemma \ref{lem:cycleEig}, the smallest eigenvalue for this block
is at least $E(x,y)$. 

$V$ is only run on those $i$ such that $y_i = 1$. If $y \not\in Y_{rej}$,
we know that each such $x_i$ is in $L$ which means that there is a witness which will cause
$V$ to accept. Therefore, there is a $w$ such that $H_{N, V}$ applied to every state in
$S_{y,w}$ is $0$. The lowest eigenvalue of $H$ restricted to the space spanned by
this $S_{y,w}$ is exactly the lowest eigenvalue of $H_{N, prop} + H_{N, final}$
restricted to this space, which by Lemma \ref{lem:cycleEig}, is $E(x,y)$.
\end{proof}

We are finally ready to establish that the smallest eigenvalue occurs within blocks
$S_y$, where $y$ is the set of correct oracle responses, which we refer to as $\tilde{y}$.

\begin{lemma}
\label{lem:groundenergy}
{\bf [The Ground Energy Corresponds to Blocks with Correct Oracle Guesses]}
Suppose that 
the process $(M_{BC})^{N-2}$ followed
by $(M_{check})^{N-2}$ starting in configuration $\ket{\vinit}\ket{y}\ket{w}$
produces {\sc Out}$(x, yw, M, V)$ on the work track.
The smallest eigenvalue for $H_N |_{\calh_{br}}$  is $E(x,\tilde{y})$, where
$\tilde{y}$ are the correct oracle responses for input $x$.
\end{lemma}

\begin{proof}
Recall that throughout the construction, we analyzed the smallest eigenvalue for $H_N$
restricted to the space spanned by all bracketed states.

Note that since $\tilde{y} \not\in Y_{rej}$,
according to Lemma \ref{lem:Exy}, 
there is an eigenvector whose eigenvalue is
$E(x,\tilde{y})$, so we only need to show that every other eigenvalue is at least $E(x,\tilde{y})$. 

The smallest eigenvalue of any state that is perpendicular to $\calh_{wf}$
is at least $1$
from $H_{N, wf-cl}$ and $H_{N, wf-co}$. 
From Lemma \ref{lem:pathEig}, the smallest eigenvalue of any state in $\calh_{path}$
is at least $1 - \cos(\pi/(2p(N)+1))$ which according to Claim
\ref{cl:boundcompare} will always be larger than $E(x,y)$ for any
$y$.
By Lemma \ref{lem:lb-s12}, the smallest eigenvalue of any state
in the span of $S_1 \cup S_2 \cup S_3$ is at least $(1/8)(1 - \cos(\pi/(2p(N)+1)))$ which according to Claim
\ref{cl:boundcompare} will always be larger than $E(x,y)$ for any
$y$.
Therefore,
the smallest eigenvalue for $H$ will correspond to an eigenvector in the span of $S_4$.

We will prove  that if $y \neq \tilde{y}$, then the smallest eigenvalue of $H$ restricted
to the span of $S_y$ is greater than $E(x, \tilde{y})$.
Suppose that the first $k$ bits of $y$ are the same as the first $k$ bits of $\tilde{y}$,
but $y_{k+1} \neq \tilde{y}_{k+1}$.
Note that the next oracle query $x_{k+1}$ will be the same for  $y$ and $\tilde{y}$.

If $x_{k+1} \not\in L$, then $\tilde{y}_{k+1} = 0$. If $y$ matches $\tilde{y}$ in the first $k$ bits and has $y_{k+1} = 1$, then $y \in Y_{rej}$ which means that $S_y$ is contained
in $S_3$ and the smallest eigenvalue of $H$ restricted to this space is larger than $E(x,y)$ for any $y$.

Now suppose that $x_{k+1} \in L$. Then $\tilde{y}_{k+1} = 1$. Consider a string $y$ that matches $\tilde{y}$ in the first $k$ bits and has $y_{k+1} = 0$. 
The value of $T(x,y)$ for any such $y$ must be at most:
$$T(x,y) = 2^{m} + 2^m \left[ 4^{m+1} + \sum_{j=k+2}^m  4^{m-j+1} + \sum_{j=1}^k \tilde{y}_j \cdot 4^{m-j+1} \right]$$
Meanwhile, the value of $T(x,\tilde{y})$ will be at least
$$T(x,\tilde{y}) =  2^m \left[ 4^{m+1} + 4^{m-k} + \sum_{j=1}^k \tilde{y}_j \cdot 4^{m-j+1} \right]$$
The lower bound for $T(x, \tilde{y})$ is larger than any $T(x,y)$, where $y$ matches $\tilde{y}$ in the first
$k$ bits and does not match $\tilde{y}$ on bit $k+1$. Therefore $E(x,y) > E(x, \tilde{y})$
and the smallest eigenvalue for $H$ restricted to $S_y$ is greater than $E(x, \tilde{y})$.
\end{proof}

\subsubsection{Bracketed States and the Final Reduction}
\label{sec:bracket}

Throughout the construction, we have analyzed the Hamiltonian restricted to the space
spanned by all bracketed states ($\calh_{br}$). 
These states have the leftmost particle in state $\leftBr$,
right rightmost particle in state $\rightBr$ and none of the particles in between in state
$\leftBr$ or $\rightBr$. We now add an additional term to $h$ to ensure that the ground state is in
$\calh_{br}$. 
We will denote a basis of the $d$-dimensional particles by $\ket{\leftBr}$, $\ket{\rightBr}$,
and $\ket{1}, \ldots, \ket{d-2}$. Let
\begin{align*}
    \bar{h} = h &  + \sum_{j-1}^{d-2} \frac 1 4 
\left( \ketbra{j}{j} \otimes I +  I \otimes \ketbra{j}{j} \right) + I \otimes \ketbra{\leftBr}{\leftBr} + \ketbra{\rightBr}{\rightBr} \otimes I
\end{align*}
Let $\bar{H}_N$ denote the Hamiltonian resulting from applying $\bar{h}$ to each pair of neighboring particles in a chain of length $N$.

\begin{lemma}
\label{lem:bracket}
$\lambda_0(\bar{H}_N) = \lambda_0( H_N |_{\calh_{br}} ) + \frac {N-2} 2$
\end{lemma}

\begin{proof}
The first two terms in $\bar{h}$ give a 1/4 penalty to each end particle if it is not in a bracket states and 1/2 penalty to each middle particle if it is not in a bracket state. The next two terms cause a penalty of $+1$ for any pair with a $\leftBr$ on the right and a $+1$ for any pair with a $\rightBr$ on the left.

Since $\bar{H}_N$ is closed on $\calh_{br}$ and $\calh_{br}^{\perp}$, we can consider each space separately. We will first consider the energy penalty from the additional terms in $\bar{h}$.
The bracketed states will have an additional energy of exactly $(N-2)/2$
from the additional terms, since they add  $1/2$ for every particle, except the two at the ends. Therefore, the lowest energy of any state in $\calh_{br}$ will be $(N-2)/2 + \lambda_0(H_N|_{\calh_{br}}) < (N-2)/2 + 1/4$. Note that the bound in Lemma \ref{lem:groundenergy},
on $\lambda_0(H_N|_{\calh_{br}})$ is less than $1/4$ for any $N \ge 4$.

Consider a state which is not bracketed. We will only consider the energy from the additional terms which will give a lower bound for the energy for the state. 
If the particle on the left is not in state $\leftBr$, then replacing it with
$\leftBr$ will cause the energy to go down by at least $1/4$. Similarly, If the particle on the right is not in state $\rightBr$, then replacing it with
$\rightBr$ will cause the energy to go down by at least $1/4$.
If there is a particle in a bracketed state in the middle, replacing it with a non-bracketed state will cause the  energy to go down by at least $1/2$. The process can be repeated until the state is bracketed at which point the energy from the additional terms will be exactly $(N-2)/2$. Therefore any non-bracketed configuration will have energy at least $(N-2)/2 + 1/4$.
\end{proof}

We are finally ready to put the pieces together to prove that {\sc Function}-TIH is hard for $\fpnexp$.

\begin{theorem}
\label{th:finitehardness}
{\bf [{\sc Function}-TIH is hard for $\fpnexp$]}
For any $f \in \fpnexp$, there is a Hamiltonian $h$  that operates on two $d$-dimensional particles. 
Let $H_N$ denote the Hamiltonian on a chain of $N$ $d$-dimensional particles resulting from applying $h$ to each neighboring pair in the chain. There is a polynomial time computable function $N(x)$ and polynomial $q = 32N^4 p(N)^4 = O(N^{12})$ such that,
the value of  $f(x)$ can be computed in  time polynomial in $|x|$,
given a value $E$ such $|E - \lambda_0 (H_N)| \le 1/q(N)$.
\end{theorem}

\begin{proof}
We have shown a construction for $h$ based on the Turing Machine $M$ that computes $f(x)$ and the TM $V$ that is the verifier for the oracle language $L$
in $\nexp$. 
The function $N$ is the function given in Definition \ref{def:Nofx}. 
Suppose that for some string $x$,
$N = N(x)$ and let $m$ be the number of oracle queries made by $M$ on input $x$.
Suppose further that for every $\myw \in \{0,1\}^{N-4}$,
the process $(M_{BC})^{N-2}$ followed
by $(M_{check})^{N-2}$ starting in configuration $\ket{\vinit}\ket{y}\ket{w}$
produces {\sc Out}$(x, \myw, M, V)$ on the work track,
where $yw = \myw$ and  $|y| = m$.
Then by Lemma 
\ref{lem:groundenergy},  
the ground energy of $H_N|_{\calh_{\calh_{br}}}$ is 
\begin{equation}
\label{eq:groundenergy}
    E(x, \tilde{y}) =  \left( 1 - \cos \left( \frac{ \pi}{L+1} \right) \right),
\end{equation}
where $L = (2T(x, \tilde{y}) + 1)\cdot p(N)$ and $\tilde{y}$ is the string denoting the correct oracle responses on input $x$. Lemma \ref{lem:bracket}
establishes that the state corresponding to the ground energy is in $\calh_{br}$. 

Lemmas \ref{lem:existsMTV} and \ref{lem:gentv} together establish that there is a Turing Machine $M_{TV}$
that results in Turing Machine $M_{check}$ such that
for all but a finite set of $x$'s, the conditions of 
Lemma \ref{lem:groundenergy} are met. Namely, it is in fact the case that
for $N = N(x)$ and for every $\myw \in \{0,1\}^{N-4}$,
the process $(M_{BC})^{N-2}$ followed
by $(M_{check})^{N-2}$ starting in configuration $\ket{\vinit}\ket{y}\ket{w}$
  produces {\sc Out}$(x, \myw, M, V)$ on the work track.
  
  Lemma \ref{lem:TMonestep} establishes that each right-to-left sweep of the pointer for Segments
  $1$, $2$, and $3$ executes one step of the TM corresponding to the operation $P$ applied to the computation tracks. The conditions of the lemma are satisfied since Lemma
  \ref{lem:gentv} allows us to assume that the head of the Turing Machine stays within the sequence of locations from $1$ through $N-3$.
  Segment $1$ and $2$ together apply $N-2$ steps of $M_{BC}$ and Segment $3$  applies $N-2$ steps of $M_{check}$. Therefore the propagation terms of the Hamiltonian implement the TM computations on the computation tracks in each iteration of the clock for Tracks $1$ and $2$ and the ground energy for $H_N$ is equal to the expression given in Equation (\ref{eq:groundenergy}).
  
For the finite set of $x$ for which the construction does not work, a finite look-up table can be used to determine $f(x)$.
For the other $x$'s, suppose you are given a value $E$ where $|E - E(x,\tilde{y})| < 1/2N^4 p(N)^4$.
We will show that the value of $f(x,\tilde{y}) = f(x)$ can be computed in polynomial time.
According to Claim \ref{cl:cosdiff},
if $L$ and $L'$ are distinct positive integers, then
$$\left|\cos\left( \frac{ \pi}{L+1} \right) - \cos\left( \frac{ \pi}{L'+1}\right) \right|
\ge \frac{1}{L^4}.$$
Therefore, the value of $L$ is uniquely determined, given a value $E$ such that
$|E - E(x,\tilde{y})| \le 1/2L^4$,
where $L = (2 T(x,\tilde{y})+1) \cdot p(N)$. Note that in order for the construction to work,
$T(x,y)$ must be written in unary on the work tape, which means that $T(x,y) \le N-2$.
Therefore, $L \le 2N p(N)$, and as long as the estimate is within $32 N^4 p(N)^4$ of the true
ground energy, $L$ can be uniquely recovered.

Given $E$, the value of integer $L$ can be found by binary search in time $O(\log L)$.
Since the value of $N$ is known, then $T(x,\tilde{y})$ can be recovered. The low order bits in
the binary representation of $T(x,\tilde{y})$ are exactly $f(x, \tilde{y})$, due to the assumption that the length of the output $f(x)$ is at most $m$, the number of oracle queries,
which is a valid assumption due to the padding argument given in Lemma \ref{lem:pad}.
The value of all the numbers in the calculations are polynomial in $N$
and therefore  take time polynomial in $\log N$, which by Lemma \ref{lem:Nlowerbound} is
polynomial in $|x|$.
\end{proof}

    \subsubsection{Eigenvalue Bounds}
    \label{sec:eigenvaluebounds}

This section contains proofs for the bounds on the smallest eigenvalues of various matrices used in the construction.
The end of the section contains a discussion of how these bounds relate to each other for the values of the variables that arise in the construction.

Let $C_L$ be the propagation matrix for a matrix for a cycle of length $L$. Note that $C_L$ is $1/2$ times the Laplacian
matrix for a cycle of length $L$. 
Let $P_L$ denote the propagation matrix for a path of length $L$.
We will number the rows and columns of an $L \times L$ matrix with indices $0$ through $L-1$.

\begin{definition}
{\bf [Periodic Diagonal Matrix]}
For positive integers $r$ and $s$, and $l \in \{0, \ldots, s-1\}$, let $D_{r,s,l}$ be a $rs \times rs$
diagonal matrix which is zero everywhere except that the diagonal entries
are $+1$
at locations $l +ks$ for $k \in \{0, \ldots, r-1\}$.
\end{definition}
Note that the matrix $D_{1,L,l}$ is zero everywhere, except for $+1$ in location $l$ on the diagonal.

\begin{lemma}
\label{lem:pathEig}
{\bf [Smallest Eigenvalue for the Path Graph Plus a $1/2$ Penalty]}
The smallest eigenvalue for $P_L + D_{1,L,0}$ is at least
$$\left( 1 - \cos \left(\frac{\pi}{2L+1} \right) \right)$$
\end{lemma}

\begin{proof}
\cite{DRS07} show that the smallest eigenvalue of $P_L + \frac 1 2 D_{1,L,0}$ 
is exactly 
$$\left( 1 - \cos \left(\frac{\pi}{2L+1} \right) \right).$$
Since the $D_{1,L,0}$ is positive semi-definite, the smallest eigenvalue
of $P_L + \frac 1 2 D_{1,L,0}$ is a lower bound for the smallest eigenvalue
of $P_L + D_{1,L,0}$. The lemma follows.
\end{proof}

\begin{lemma} 
\label{lem:cycleEig}
{\bf [Smallest Eigenvalue for the Cycle Graph Plus Two $1/2$ Penalties]}
The smallest eigenvalue of $C_L + \frac 1 2 D_{1,L,l} + \frac 1 2 D_{1,L,l+1~\mbox{mod}~L}$ is exactly 
$$\left( 1 - \cos\left( \frac{ \pi}{L+1} \right)\right).$$ 
\end{lemma}

\begin{proof}
It is sufficient to prove the lemma for $l = L-1$ due to the dihedral symmetry of the cycle.
A vector $\langle u_0, \ldots, u_{L-1} \rangle$ is an eigenvector of
$C_L + \frac 1 2 D_{1,L,l} + \frac 1 2 D_{1,L,l+1~\mbox{mod}~L}$ with eigenvalue $\lambda$ if the following constraints are satisfied:
\begin{align}
\label{eq:consttraint0}
  &  \frac 1 2 (u_{L-1} + u_1) = \left(\frac 3 2 - \lambda \right) u_0\\
\label{eq:consttraintk}
  &  \frac 1 2 (u_{k-1} + u_{k+1}) = \left(1 - \lambda \right) u_k & \mbox{for}~k = 1, \ldots, L-2\\
\label{eq:consttraintLm1}
  &  \frac 1 2 (u_{L-2} + u_0) = \left(\frac 3 2 - \lambda \right) u_{L-1}
\end{align}
Any vector of the form $u_k = e^{i (k+1) \theta} - e^{-i (k+1) \theta}$ satisfies the constraints in (\ref{eq:consttraintk})
for $\lambda = 1 - \cos \theta$.
Let $\omega = e^{\pi i/(L+1)}$. 
Define vector $\vec{u_j}$ to have entries
$u_{jk} = \omega^{(j+1)(k+1)}-\omega^{-(j+1)(k+1)}$.
Using the identity that $\omega^L = - \omega^{-1}$, it can be verified that $\vec{u_j}$
also satisfies (\ref{eq:consttraint0}) and (\ref{eq:consttraintLm1}) when $\lambda_j = 1 - \cos(\pi (j+1)/(L+1))$ for
$j \in \{0, \ldots, L-1\}$. This gives $L$ distinct eigenvalues for $C_L + \frac 1 2 D_{1,L,l} + \frac 1 2 D_{1,L,l+1~\mbox{mod}~L}$,
the smallest of which occurs when $j=0$.
\end{proof}

\begin{lemma} 
\label{lem:periodicEig}
{\bf [Smallest Eigenvalue Lower Bound for the Cycle Graph Plus a Periodic Penalty]}
If $L = rs$, then the smallest eigenvalue of $C_L + D_{r, s,l}$ is at least 
$$\frac 1 8 \left(1 - \cos \left( \frac{\pi}{2s+1} \right) \right)$$
\end{lemma}

We will need the following claim in the proof of
Lemma \ref{lem:periodicEig}. First we need to define the Fourier basis vectors:
$$\ket{f_k} = \frac{1}{\sqrt{L}} \sum_{j=0}^{L-1} \omega^{jk} \ket{j},$$
where  $\omega = e^{2 \pi i/L}$ is the $L^{th}$ root of unity.

\begin{claim}
\label{cl:mod}
Let $L = rs$. Then $\langle f_k | D_{r, s, 0} | f_j \rangle = 1/s$ if $k ~\mbox{mod}~ r = j ~\mbox{mod}~ r$, and $\langle f_k | D_{r, s, l} | f_j \rangle = 0$ otherwise.
\end{claim}

\begin{proof}
[Proof of Claim \ref{cl:mod}]
For the purposes of this proof, let $D = D_{r, s, 0}$. The $a^{th}$ entry along the diagonal is $1$
if and only if $a = xs$ for $x \in \{0, \ldots, r-1\}$.
Therefore 
$$D \ket{f_j} = \frac{1}{\sqrt{L}} \sum_{x=0}^{r-1} \omega^{j {xs}} \ket{xs} .$$
Suppose that $j = yr+m$, where $m = j \mod r$. Since $L = rs$ and $\omega$ is the $L^{th}$ root of unity,
$$\omega^{jxs} = \omega^{(yr+m)xs} = \omega^{yxrs} \cdot \omega^{xsm} = \omega^{xsm}.$$
Therefore $D \ket{f_j} =  1/\sqrt{L} \sum_{x=0}^{r-1} \omega^{mxs} \ket{xs}$.
Similarly $\bra{f_k} D = 1/\sqrt{L} \sum_{x=0}^{r-1} \omega^{-nxs} \bra{xs}$,
where $n = k \mod r$. 
$$\langle f_k |D | f_j \rangle = \frac{1} L \sum_{x=0}^{r-1} \omega^{(m-n)xs}.$$
Since $L = rs$, $\omega^s$ is the $r^{th}$ root of unity. Therefore, the  above is $0$ if $m \neq n$ and is equal to $r/L = 1/s$ if $m = n$.
\end{proof}

\begin{proof}
[Proof of Lemma \ref{lem:periodicEig}]
It is sufficient to prove the lemma for $l = 0$ due to the dihedral symmetry of the cycle.
Let $D_{r, s,0} = D$.
We will use $L$ to denote the size of the matrix: $L = rs$.

Suppose we express a state $\ket{\phi}$ in the Fourier basis: $\sum_{j=0}^{L-1} \alpha_j \ket{f_j}$.
By Claim \ref{cl:mod},
$$\langle \phi | D | \phi \rangle = \frac 1 s \sum_{x = 0}^{r-1} \sum_{j=0}^{s-1} \sum_{k=0}^{s-1} 
\alpha_{kr+x}^* \alpha_{jr+x}.$$
This means that $D$ is block-diagonal in the Fourier basis.
Each block is an $s \times s$ matrix, corresponding to a particular value for $x$,
and is $1/s$ times the all $1$'s matrix. 
Let $O_s$ be the $s \times s$ matrices of all $1$'s.
$D$ expressed in the Fourier basis and restricted to the
set $S_x = \{ \ket{f_{x}}, \ket{f_{r+x}}, \ldots \ket{f_{(s-1)r + x}}\}$
is exactly $(1/s) O_s$.
Meanwhile $C_L$ is diagonal in the Fourier basis.
When restricted to $S_x$, the $k^{th}$ entry along the diagonal is $1 - \cos(2 \pi (kr+x/L)))$. Call this matrix $F_x$.
We will first show that for all $x$, $F_x \succeq (1/4) F_0$,
where $\succeq$ means that every element along the diagonal of $F_x$ is greater than or equal to
the corresponding element in $(1/4) F_0$.
Then it will be sufficient the lower bound
$\lambda_0((1/s) O_s + (1/4)F_0)$.

To show $F_x \succeq (1/4) F_0$, we first consider the case that $x \le r/2$.
We will rotate each angle $2 \pi (kr+x)/L$ by $- 2 \pi x/L$. 
That is, $1 -\cos(2 \pi (kr+x)/L)$ is replaced by $1 -\cos(2 \pi (kr)/L)$ in the $k^{th}$
entry along the diagonal.
The resulting matrix is $F_0$. We will show for each entry, that the new value is at most  $4$ times the old value.
For $0 \le k < s/2$, $2 \pi (kr+x)/L < \pi$ and the new value is less than the old value.
For $s/2 \le k < s$, then $1 - \cos(2 \pi (kr+x/L))) = 1 - \cos(2 \pi ((s-k)r-x/L)))$
is replaced by $1 - \cos(2 \pi ((s-k)r/L)))$. Note that $s-k$ is positive and $x \le r/2$,
so the new angle $2 \pi (s-k)r/L$ is at most twice the original angle $2 \pi ((s-k)r-x)/L$.
Using the half angle formula for $\cos$, it can be verified that
$1/4(1 - \cos \theta) \le  (1 - \cos (\theta/2))$.

A symmetric argument can be applied if $r > x > r/2$, except that $1 - \cos(2 \pi (kr+x)/L)$
is replaced by $1 - \cos(2 \pi ((k+1)r)/L)$, which is a rotation by $2 \pi (r-x)/L$. The resulting matrix
is $\tilde{F}_0$, which is the same as $F_0$, except that the diagonal entries are rotated by one location so that the $(L-1)^{st}$
entry is $1 - \cos(2 \pi sr)/L) = 0$. Since $O_s$ is invariant under permutations of the rows and columns,
the entries of the new matrix can be rotated back to form $F_0$. Thus
the eigenvalues of $(1/s) O_s + (1/4)F_0$ are the same as the eigenvalues of $(1/s) O_s + (1/4)\tilde{F}_0$

Our task now is to lower bound $\lambda_0((1/s) O_s + (1/4)F_0 )$.
Let $U_s$ denote the unitary to change from the standard to the Fourier basis on a matrix of size $s \times s$. 
If we change $F_0$ back to the standard basis, we just get the cycle graph on $s$ vertices: 
$(U_s)^{-1} F_0 U_s = C_s$. 
Meanwhile if we apply the change of basis to $(1/s) O_s$,
we get $D(1,s,0)$, the all $0$'s matrix with one $+1$ entry in the upper left corner. Therefore, we can lower bound
$\lambda_0( D(1,s,0) + (1/4) C_s)$ instead.
Suppose that there is a $\ket{\phi}$ such that
$\langle \phi | D(1,s,0) +  C_s | \phi \rangle = c$.
The vector $\ket{\phi}$ when expressed in the standard basis is $\sum_{j=0}^{s-1} \alpha_j \ket{j}$. 
Let $\ket{\tilde{\phi}}$ be the length $2s$ vector obtained by doubling each entry from $\ket{\phi}$:
$$\ket{\tilde{\phi}} = \sum_{j=0}^{s-1} \frac{\alpha_j}{\sqrt{2}} (\ket{2j} + \ket{2j+1}).$$
We claim that
$$\langle \tilde{\phi} | (1/2) D(1,2s,0) + (1/2) D(1,2s,1) +  C_{2s} | \tilde{\phi} \rangle = \frac{c}{2}.$$
To see why this claim is true, consider partitioning 
the matrix $(1/2) D(1,2s,0) + (1/2) D(1,2s,1) +  C_{2s}$ into $s^2$ $2 \times 2$ submatrices.
Let $M_{jk}$ be the submatrix at the intersection of rows $2j$ and $2j+1$ and
columns $2k$ and $2k+1$. Let $E_{jk}$ be the entry in
row $j$ and column $k$ in 
$D_{1,s,0} +  C_s$. The key observation is that the sum of the entries in $M_{jk}$ is equal to $E_{jk}$.
Then
\begin{align*}
\frac 1 2 \langle \phi | D_{1,s,0} +  C_s | \phi \rangle & = \frac 1 2 \sum_{jk} \alpha_j^* \alpha_k E_{jk}\\
& =
\sum_{jk} \left( \begin{matrix} \alpha_j/\sqrt{2}\\ \alpha_j/\sqrt{2} \end{matrix} \right) M_{jk} ~(\alpha_k/\sqrt{2}, \alpha_k/\sqrt{2} )\\
& = \langle \tilde{\phi} | (1/2) D(1,2s,0) + (1/2) D(1,2s,1) +  C_{2s} | \tilde{\phi} \rangle
\end{align*}
Putting all the bounds together, we get that
\begin{align*}
    \lambda_0 (C_L + D(r, s,l)) & = \max_x \left\{
    \lambda_0( (1/s) O_s + F_x )
    \right\}\\
    & \ge \lambda_0( (1/s) O_s + (1/4) F_0 )\\
    & \ge \frac 1 4 \lambda_0( (1/s) O_s +  F_0 )\\
    & = \frac 1 4 \lambda_0( D(1,s,0) +  C_s )\\
    & \ge \frac 1 8 \lambda_0( (1/2) D(1,2s,0) + (1/2) D(1,2s,1) +  C_{2s} ) = \frac 1 8 \left(1 - \cos \left( \frac{\pi}{2s+1} \right) \right)
\end{align*}
The first equality follows from expressing 
$C_L + D(r, s,l)$ and observing the block diagonal structure.
The second inequality follows from the argument that for all $x$,
$F_x \ge (1/4) F_0$ by rotating the angles by $2 \pi x/rs$.
Line $3$ to $4$ above follows from rotating from the Fourier basis back to the standard basis.
The last equality follows from Lemma \ref{lem:cycleEig}.
\end{proof}

Lemma \ref{lem:periodicEig} is used to give a lower bound
for computations that have a periodic cost occurring every
$p(N)$ time steps. Lemma \ref{lem:cycleEig} gives an exact
expression for the smallest eigenvalue corresponding to correct computations. Thus, we need to establish that
the bound given in Lemma \ref{lem:cycleEig}, which is applied
for $L = (2 T(x,y) + 3)p(N)$ is always smaller than the lower bound from Lemma \ref{lem:periodicEig}, which is applied for
$s = p(N)$.

First, we will need the following fact:

\begin{fact}
\label{fact:cos}
For $0 \le \theta < 1$ and $c \ge 1$,
$$\left(1 - \cos \left( \frac{\theta}{c} \right) \right)
\le \frac 2 {c^2} \left(1 - \cos \theta \right).$$
\end{fact}

\begin{proof}
We will show that
$$\frac{c^2}{2} \le \frac{1 - \cos \theta }{1 - \cos (\theta/c) }.$$
Using the Taylor expansion for $\cos$, we have that
$1 - \cos \gamma = \sum_{j=1}^{\infty} (-1)^{j+1} \frac{\gamma^{2j}}{(2j)!}$.
We will show that for each pair of consecutive terms, the ratio is at least $c^2/2$:
$$\frac{\theta^{2x}}{(2x)!} - \frac{\theta^{2x+2}}{(2x+2)!}
\ge \frac{c^2}{2} \left[ \frac{(\theta/c)^{2x}}{(2x)!} - \frac{(\theta/c)^{2x+2}}{(2x+2)!} \right],$$
for $x \ge 1$. This is equivalent to showing that
$$1 - a \ge  \frac{1}{2c^{2x-2}} - \frac{a}{2c^{2x}},$$
where $a = \theta^2/(2x+2)(2x+1)$.
Using the fact that $0 \le a < 1/4$ and $c \ge 1$:
$$1 - a \ge \frac 3 4 \ge \frac 1 2 - \frac{a}{2c^{2x}}
\ge \frac{1}{2c^{2x-2}} - \frac{a}{2c^{2x}}.$$
\end{proof}

\begin{claim}
\label{cl:boundcompare}
For any $x$ and $y$ and $N \ge 4$,
$$ \left( 1 - \cos\left( \frac{ \pi}{(2 T(x,y) + 1)p(N)+1} \right)\right) < \frac 1 8 
\left( 1 - \cos\left( \frac{ \pi}{2p(N)+1} \right)\right).$$ 
\end{claim}

\begin{proof}
Since $N \ge 4$, $p(N) = 4(N-2)(2N-3) \ge 40$.
For any $x$ and $y$, $T(x,y) \ge 4$. Note that $m$ in the definition of $T(x,y)$
is the number of oracle calls made by Turing Machine $M$ on input $x$ and
$T(x,y)$ is always at least $2^{3m+2}$.
Therefore, even when $m = 0$, $T(x,y) = 4$.
$$ \left( 1 - \cos\left( \frac{ \pi}{(2 T(x,y) + 1)p(N)+1} \right)\right)
\le  \left( 1 - \cos\left( \frac{ \pi}{9 p(N)+1} \right)\right)$$
Since $p(N) \ge 40$, we know that $\pi/(2p(N)+1) < 1$ and 
$(9p(N)+1)/(2p(N)+1) > 4$. Applying Fact \ref{fact:cos},
$$  \left( 1 - \cos\left( \frac{ \pi}{9 p(N)+1} \right)\right) < \frac 2 {4^2} 
\left( 1 - \cos\left( \frac{ \pi}{2p(N)+1} \right)\right)
\le \frac 1 8 \left( 1 - \cos\left( \frac{ \pi}{2p(N)+1} \right)\right).$$
The bound in the Lemma follows.
\end{proof}

\section{Computing the Ground Energy Density in the Thermodynamic Limit}
\label{sec:infinite}

The parameters of the problem in the thermodynamic limit
(not the instance) include  two $d^2 \times d^2$ Hermetian
matrices $h^{row}_a$ and $h^{col}_a$ denoting the energy interaction between two neighboring particles in the horizontal and vertical directions on a $2$-dimensional grid.
The dimension of each individual particle is a constant $d$. 
Then $H_a (N)$ is defined the be the Hamiltonian of an $N \times N$ grid
of $d$-dimensional particles in which $h^{row}_a$ is applied to each neighboring pair in the horizontal direction
and $h^{col}_a$ is applied to each neighboring pair in the vertical direction.
$\lambda_0 (H_a (N))$ is the ground energy of $H_a (N)$.
 The ground energy density of $(h^{row}_a, h^{col}_a)$
 in the thermodynamic limit is defined to be:
$$\alpha_0 (h^{row}_a, h^{col}_a) = \lim_{N \rightarrow \infty} \frac{\lambda_0 (H_a (N))}{N^2}.$$
When it is clear from context, we will drop the subscript $a$ denoting
the particular Hamiltonian terms, in which case
$\alpha_0$ denotes the ground energy density of
$h^{row}$ and $h^{col}$ in the thermodynamic limit.
We can assume that the Hamiltonian terms are scaled so that $0 \le \alpha_0 \le 1$.
The problem itself will be characterized by $h^{row}$ and $h^{col}$. Here is the function
version of the ground energy density (GED) problem:
\begin{quote}
{\sc Function-GED for $(h^{row}, h^{col})$}\\
{\bf Input:} An integer $n$ expressed in binary.\\
{\bf Output:} A number $\alpha$ such that $|\alpha - \alpha_0| \le 1/2^n$.
\end{quote}

The following theorem encapsulates our results for Function-GED in the thermodynamic limit:
\begin{theorem}
Function-GED problem is contained in $\fexpqmaexp$
and is hard for $\fexpnexp$.
\end{theorem}

\subsection{Existence of the Limit}
\label{sec:existLimit}

Before proceeding to the complexity results, it is important to establish that the quantity $\alpha_0$ that we are trying to approximate always exists.

\begin{lemma}
\label{lem:gedbound}
The limit $\alpha_0 (h^{row}, h^{col}) = \lim_{N \rightarrow \infty} \frac{\lambda_0 (H_N)}{N^2}$ 
always exists. Moreover for every $N$, $$\left|\alpha_0 (h^{row}, h^{col}) - \frac{\lambda_0 (H_N)}{N^2}\right|
\le \frac{6(\norm{h^{row}} + \norm{h^{col}})+2}{N}.$$
\end{lemma} 
\begin{proof}
We can apply an additive offset to $h^{row}$ and $h^{col}$,
so that $h^{row} \ge 0$ and $h^{col} \ge 0$.
We will show that for any $\epsilon > 0$, there is a positive integer $N_0$ and an interval of size
$\epsilon$ that contains $\lambda_0 (H_N)/N^2$ for every $N \ge N_0$.

Fix a positive integer $\hat{N}$. Let $N_0 = (\hat{N})^2$ and let $\alpha$ denote $\lambda_0 (H_{\hat{N}})/(\hat{N})^2$.
 The value of $\hat{N}$ will be
chosen at the end to obtain the desired $\epsilon$. 
For any $N \ge N_0$, let $q = \lfloor N/\hat{N} \rfloor$ and $r = N ~\mbox{mod}~\hat{N}$.
The three facts that we will need about these variables are that $N - r = q \hat{N}$,
$q/N \le 1/\hat{N}$, and $r/N \le 1/\hat{N}$.

An $N \times N$ grid can be packed with $q^2$ disjoint $\hat{N} \times \hat{N}$ sub-grid. The energy
 of the adjacent pairs of particles
 within a single $\hat{N} \times \hat{N}$ sub-grid is at least
$\alpha \cdot (\hat{N})^2$. Therefore the energy density
across the entire $N \times N$ grid will be 
\begin{eqnarray*}
\frac{\lambda_0(H_{N})}{N^2} & \ge  &  \frac{q^2 \cdot \alpha \cdot (\hat{N})^2}{N^2}
= \frac{\alpha (N-r)^2}{N^2}\\
& \ge & \frac{\alpha N^2 - 2r N}{N^2} =
\alpha - \frac{2r}{N} \ge \alpha -\frac{2}{\hat{N}}.
\end{eqnarray*}

To get an upper bound on $\lambda_0 (H_N)/N^2$, consider
the state that is a $q \times q$ tensor product of the 
ground state for each
$\hat{N} \times \hat{N}$ sub-grid. The particles that are not contained
in an $\hat{N} \times \hat{N}$ sub-grid will be set to $\ket{0}$.
The cost  due to adjacent pairs of particles where both particles are contained in
a single $\hat{N} \times \hat{N}$ sub-grid
is at most $\lambda_0 (H_{\hat{N}})$. Therefore the total cost due to all pairs
of adjacent particles where both particles are contained in the same $\hat{N} \times \hat{N}$ sub-grid
is at most
$\alpha \cdot q^2 (\hat{N})^2 \le \alpha N^2$.

There are at most $2 \hat{N}$ pairs of adjacent particles that span the boundary
of a particular $\hat{N} \times \hat{N}$ sub-grid in the horizontal direction and
at most $2 \hat{N}$ in the vertical direction. The total cost due to pairs of adjacent
particles that span the boundary of any of the $\hat{N} \times \hat{N}$ sub-grids is at most
$(\norm{h^{row}} + \norm{h^{col}}) \cdot q^2 \cdot 2 \hat{N} \le  2 q (\norm{h^{row}} + \norm{h^{col}}) N$.

The number of grid locations in the $N^2$ grid that are not
contained in any of the $\hat{N} \times \hat{N}$ sub-grids is
$$N^2 - q^2 (\hat{N})^2 = N^2 - (N-r)^2 \le 2r N.$$
Each particle is adjacent to at most $4$ other particles, two in the horizontal direction and two in the vertical direction, so the total cost
due to adjacent pairs of particles where both particles are outside all of the $\hat{N} \times \hat{N}$ sub-grids
is at most $4 rN (\norm{h^{row}} + \norm{h^{col}}) $. Putting the three upper bounds together, we get that

\begin{eqnarray*}
\frac{\lambda_0(H_{N})}{N^2} & \le & \frac{\alpha N^2 + 2 q N(\norm{h^{row}} + \norm{h^{col}})  + 4 r N(\norm{h^{row}} + \norm{h^{col}}) }{N^2}\\
& = & \alpha + \frac{2q(\norm{h^{row}} + \norm{h^{col}})}{N} + \frac{4r(\norm{h^{row}} + \norm{h^{col}})}{N} 
\le \alpha + \frac{6 (\norm{h^{row}} + \norm{h^{col}})}{\hat{N}}
\end{eqnarray*}

We have shown that for every $N \ge (\hat{N})^2$,
$\lambda_0 (H (N))/N^2$ lies in the interval between
$\alpha - 2/\hat{N}$ and $\alpha + 6 (\norm{h^{row}} + \norm{h^{col}})/\hat{N}$. To prove the existence of the limit,
select $\hat{N}$  so that $(6 (\norm{h^{row}} + \norm{h^{col}}) + 2)/\hat{N} \le \epsilon$.

Since $\alpha_0 (h^{row}, h^{col})$ and $\alpha = \lambda_0 (H_{\hat{N}})/(\hat{N})^2$
both lie in the interval between
$\alpha - 2/\hat{N}$ and $\alpha + 6 (\norm{h^{row}} + \norm{h^{col}})/\hat{N}$, then $|\alpha_0 (h^{row}, h^{col}) - \alpha| 
\le (6 (\norm{h^{row}} + \norm{h^{col}}) + 2)/\hat{N}$.
\end{proof}

\subsection{Containment}
\label{sec:containment}





 
 \begin{theorem}
 Function-GED is in $\mbox{FEXP}^{\mbox{QMA-EXP}}$
 \end{theorem}
 
 \begin{proof}
 We can apply an additive offset to $h^{row}$ and $h^{col}$,
so that $h^{row} \ge 0$ and $h^{col} \ge 0$ which can then be subtracted from the final answer.
Lemma \ref{lem:gedbound} indicates that a finite system can be used to approximate the ground energy density in the thermodynamic limit, where the size of the finite system determines the error in the approximation. Therefore,
we will use the  algorithm given for Function-TIH in Figure \ref{fig:BinSearch-pseudo}.
The algorithm is described for a 1D chain of length $N$ but can be easily extended for a 2D $N \times N$ finite grid by setting the initial upper bound on ground energy to be $N^2 (\norm{h^{row}} + \norm{h^{col}})$. The fact that the oracle language Decision-TIH is  in $\qmaexp$ for 2D systems is also a straight-forward extension of Kitaev's proof that the local Hamiltonian problem is in $\qma$ \cite{KSV02}.
We will use a constant $c = 1$ for Function-TIH, which can be achieved since Theorem \ref{th:finite} says that Function-TIH is in $\fpqmaexp$ for any $c$.
Therefore, on input $N$, Function-TIH will return a value $\lambda(N)$ such that
$|\lambda(N) - \lambda_0(H_N)| \le 1/N$ which means that
$|\lambda(N)/N^2 - \lambda_0(H_N)/N^2| \le 1/N^3$. Meanwhile Lemma \ref{lem:gedbound}
indicates that the true ground energy density $\alpha_0$ is within $(2 + 6(\norm{h^{row}} + \norm{h^{col}}))/N$ of $\lambda_0(H_N)/N^2$.
In order to estimate $\alpha_0$ to within $1/2^n$, select $N$ large enough so that
$$\left|\frac{\lambda(N)}{N^2} - \alpha_0 \right|
\le \frac 1{N^3} + \frac{2 + 6(\norm{h^{row}} + \norm{h^{col}})}{N} \le \frac 1 {2^n}$$
$N$ will be $O(2^n)$. Therefore, the length of the input given to Function-TIH $|N|$ is $O(2^{|n|})$, where $|n|$ the length of the input to Function-GED. Therefore if Function-TIH is in $\fpqmaexp$, then Function GED will be in $\fexpqmaexp$.
 \end{proof}

\subsection{Preliminaries for the Hardness Reduction}
\label{sec:hardnessoverviewinf}

The construction from the finite case must be adapted in a few ways to work
for the infinite grid. Those changes are described in Section 
\ref{sec:modifications}. For now, we will call this modified Hamiltonian term $h_C$
and show how $h_C$ is used for the construction for the infinite grid. 
We will layer $h_C$ on top of the constraints of a Robinson tiling.
If the constraints of the Robinson tiling are satisfied, then
the infinite grid will be tiled with an infinite sequence of squares of size $4^k$ for every $k$. The density of squares of size $4^k$ is
$1/4\cdot 4^{2k}$.
Then $h_C$ is applied to every pair of particles in the horizontal direction on the top edge of each square.

We will start with a function $f$ computed by an exponential time Turing Machine $M$ with access to an oracle in $\nexp$. 
The language computed by the oracle will be called $L$, and $V$ will denote the exponential time verifier for $L$. The input to the function $f$ is a string $x$,
where $|x| = n$.

We assume without loss of generality that if the first bit of $x$ is $0$, then
$f(x) = 0$. If $f$ does not have this property, then we can define a new function $f'$
such that $f'(0x) = 0$ and $f'(1x) = f(x)$. Any polynomial time algorithm to compute $f'$
can be used to compute $f$. This allows us to refer to the value of input sting $x$ which is
just the number represented by $x$ in binary. Here is a formal statement of the hardness
result which we will prove:

\begin{theorem}
\label{th:mainInf}
For any $f \in \fexpnexp$, there is a polynomial $q$,
constant $d$,
and interaction terms $h^{row}$ and $h^{col}$ that operate on
pairs of $d$-dimensional particles, such that
if $\alpha_0$ is the energy density of $h^{row}$ and $h^{col}$
in the thermodynamic limit, then for all but a finite number of $x$'s, it is possible in time
polynomial in $2^{|x|}$, to compute $f(x)$ given a value $\alpha$,
where $|\alpha - \alpha_0| \le 1/2^{q(x)}$.
$q(x)$ is the value of the polynomial $q$ on the
binary number $x$.
\end{theorem}

\subsection{Robinson Tiling}
\label{sec:Robinson}

The Hamiltonian for the infinite case will have two different layers. The lower layer uses
constraints that impose a quasi-periodic tiling of the infinite plane.
A tiling can be seen as a classical Hamiltonian in which there are are a finite number of
tile types corresponding to classical basis states for a particle. A tiling system
also has rules
governing the cost of placing two tiles next to each other in the horizontal or vertical direction.
These correspond to Type 1 terms in the Hamiltonian that impose an energy penalty for two neighboring particles in a particular pair of standard basis states. 
The tiling rules devised by Robinson \cite{R71} impose
a  structure on the infinite plane as shown in
Figure \ref{fig:RobinsonGrid}.
\begin{figure}[ht]
  \centering
  \includegraphics[width=4.5in]{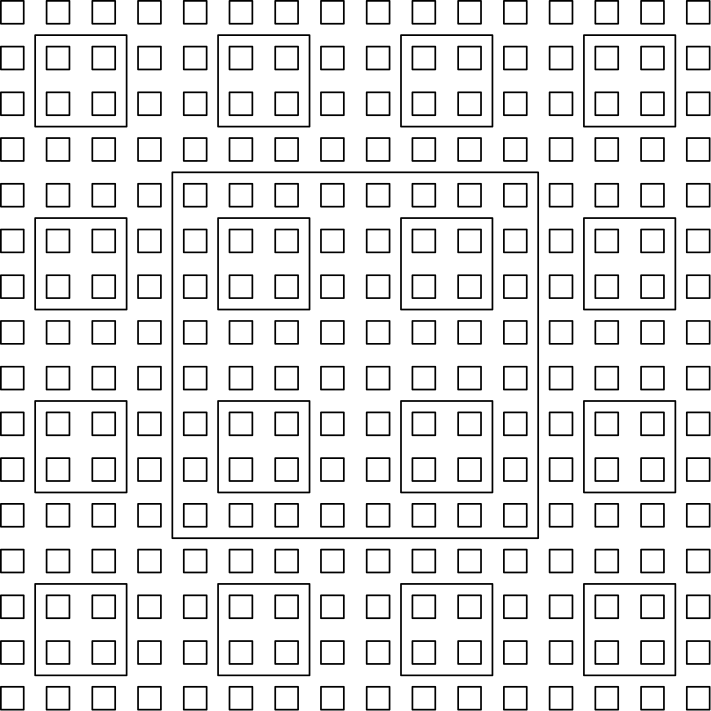}
\caption{A low energy Robinson tiling of the plane.}
\label{fig:RobinsonGrid}
\end{figure}
The essential feature of the tiling is that it has an infinite sequence of interlocking squares of increasing size. The {\em size} of a square refers to the length of one of the sides.  In the case of the Robinson tiles, the squares have size $4^k$ for $k \in \mathbb{Z}^+$. The density of squares of size $4^k$
is $1/4^{2k+1}$. 

Cubitt, Perez-Garcia, and Wolf were the first to make use of quasi-periodic tilings of the infinite plane in the context of Hamiltonian complexity \cite{C15}.
We use this same idea to layer the 1D Hamiltonian used in the construction for finite chains on top of the tile type for the top 
of the perimeter of a square. The $\leftcorner$ and $\rightcorner$ states in the
Robinson layer force the top layer to be in states $\leftBr$ and $\rightBr$, which in turn
forces the Hamiltonian to act on the top layer of the particles in between the
$\leftBr$ and $\rightBr$ as the Hamiltonian
used in our construction. The end result is  that if the energy  from
the lower layer is $0$, then the ground energy of the Hamiltonian will be the ground energy 
of the Hamiltonian from our construction on the top row of squares of size $4^k$ for every $k \ge 1$. The energy contribution is then weighted according to the the density of the squares of each size. \cite{C15} establishes
robustness criteria for the Robinson tiles so that it is more energetically favorable to have a correct
tiling on the lower Robinson layer and pay the price of the energy contribution from the finite chains
embedded in the squares. We are able to use the same robustness result which is summarized in Lemma \ref{lem:Robinson}.
We introduce some modifications of the construction for finite chains (described below) to get a $2$-particle term
$h_C$. If $H_{C}(N)$ denotes the Hamiltonian resulting from applying $h_C$ to each pair in a chain of length $N$,
then the net results is that  the ground energy density in the thermodynamic limit resulting from
layering $h_C$ on top of the Robinson
rules will be:
$$\alpha_0 = \sum_{k=1}^{\infty} \frac{\lambda_0 (H_{C}(4^k)|_{\calh_{br}})}{4^{2k+1}}.$$
This is established more formally below. The reduction then must show that if $\alpha_0$ can be approximated to a given level of precision, the value $\lambda_0 (H_{C}(4^k))$ can be determined and then the value for the desired function of $k$ can be extracted.

Let $\calh_R$ be the Hilbert space for each particle under the Robinson tiling. We will refer to this part of the space as the $R$-layer.
The standard basis for $\calh_R$ is just $\ket{t}$, for each tile type $t$. 
Let $\calh_C$ be the Hilbert space for each particle given in our construction. We will refer to this part of the space as the $C$-layer.
The Hilbert space for a particle will be $\calh_R \otimes ( \ket{0} \oplus \calh_C)$. We will refer to $\ket{0}$ as the {\em blank} state for the $C$-layer.
Let $H_R$ be the Hamiltonian resulting from the Robinson tiling rules.
Let $H_C$ be the Hamiltonian resulting from our construction.
The overall Hamiltonian acts as $H_R$ on the $R$-layer.
Then there are constraints that force the $C$-layer to be in state $\leftBr$
if and only if the $R$-layer is in state $\leftcorner$.
Similarly there are constraints that force the $C$-layer to be in state $\rightBr$
if and only if the $R$-layer is in state $\rightcorner$.
There are constraints between the layers that force low energy configurations
to be in a blank state on the $C$-layer, except  if the underlying tile type
on the lower layer is $\toprow$. The rules of the Robinson tiling will force
particles in between the $\leftcorner / \leftBr$ on the left and $\rightcorner / \rightBr$ on the right to be in state $\toprow$. The Hamiltonian on the top layer will act as $H_C$
if the tile type below is $\toprow$.

The details regarding the Robinson tiling are given in \cite{C15}.
They also detail the specific requirements required of
$H_C$ to be layered in this manner. 
We call a Hamiltonian that satisfies these constraints a
{\em Robinson-Embeddable} Hamiltonian. The term used in \cite{C15} is a
{\em Gottesman-Irani} Hamiltonian, named after the construction
from \cite{GI}.

\begin{definition}
\label{def:GI}
{\bf [Robinson-Embeddable Hamiltonian: Definition 50 from \cite{C15}]} Let $\calh$ be a finite-dimensional
Hilbert space with two distinguished orthogonal states $\ket{\leftBr}$ and
$\ket{\rightBr}$. A Robinson-Embeddable Hamiltonian is a 1D, translationally-invariant nearest-neighbor Hamiltonian $H(N)$ on a chain of length $N \ge 4$ which acts on
Hilbert space $(\calh)^{\otimes N}$ with local interaction $h$, which
satisfies the following properties:
\begin{enumerate}
    \item $h \ge 0$
    \item $[h, \ketbra{\leftBr \leftBr}{\leftBr \leftBr}] = 
    [h, \ketbra{\rightBr \rightBr}{\rightBr \rightBr}] = 
    [h, \ketbra{\leftBr \leftBr}{\rightBr \rightBr}] = 
    [h, \ketbra{\rightBr \rightBr}{\leftBr \leftBr}] = 0$
    \item $\lambda_0 (H (N) |_{\calh_{br}}) < 1$, where $\calh_{br}$ is the subspace of states with fixed boundary conditions $\ket{\leftBr}$ and $\ket{\rightBr}$ at the left and right ends of the chain respectively.
    \item For every positive integer $k$, $\lambda_0 (H (4^k) |_{\calh_{br}}) \ge 0$ and $\sum_{k=1}^{\infty} \lambda_0 (H (4^k) |_{\calh_{br}}) < 1/2$.
    \item $\lambda_0 (H (N) |_{\calh_{<}}) = \lambda_0 (H (N) |_{\calh_{>}}) = 0$, where $\calh_{<}$ and $\calh_{>}$ are the subspaces of states with, respectively, a $\ket{\leftBr}$ at the left end of the chain or a $\ket{\rightBr}$ at the right end of the chain.
\end{enumerate}
\end{definition}

We first need to show that the Hamiltonian given in our constructions 
satisfies the constraints of being a Robinson-Embeddable Hamiltonian:

\begin{lemma}
The Hamiltonian  $h_C$ in our construction is a Robinson-Embeddable Hamiltonian.
\end{lemma}

\begin{proof}
Although the Hamiltonian term $h_C$ used in the infinite case will be modified from the
construction used in the finite case, it will not be different in any of the features used in this proof. We will go through each property in turn:
\begin{enumerate}
    \item All terms are Type I or II terms as described in Section \ref{sec:circToHam} which have no negative eigenvalues.
    \item All of the transition rules that involve $\leftBr$ or $\rightBr$
    have the form: $\leftBr \generic \rightarrow \leftBr \generic$ or
    $\generic \rightBr \rightarrow \generic \rightBr$, where $\generic$ is not equal to $\leftBr$ or $\rightBr$. These propagation terms resulting from those transition rules all commute with the operators described in item $2$. All other terms are diagonal in the standard basis.
    \item According to Lemma \ref{lem:groundenergy}, The smallest eigenvalue of $H_N|_{\calh_{br}}$ is
    $$    E(x,\tilde{y}) =  \left( 1 - \cos\left( \frac{ \pi}{L+1} \right)\right),$$
where $L = (2 T(x,\tilde{y})+3) \cdot p(N)$. For $N \ge 4$, $L > 2$, and this value is always strictly less than $1$.
\item The lower bound $\lambda_0 (H_C (4^k) |_{\calh_{br}}) \ge 0$ holds because $h_C \ge 0$.
For the upper bound, we can plug in $p(N)$ to the expression above and use the fact that $T(x,y) \ge 0$ and the assumption that $N \ge 4$:
$$\lambda_0 (H_C (N) |_{\calh_{br}}) \le  \left( 1 - \cos\left( \frac{ \pi}{12(N-2)(2N-5)+1} \right)\right) \le \frac{1}{N}.$$
The sum $\sum_{k=1}^{\infty} 1/4^k$ is strictly less than $1/2$.
\item The standard basis states below are $0$ energy states for $\calh_{<}$ and $\calh_{>}$.
None of the transition rules apply since there is no pointer on Track $1$.
\begin{center}
\begin{small}
\begin{tabular}{|@{}c@{}|@{}c@{}|}
\hline
$\leftBr$ &
\begin{tabular}{@{}c@{}|@{}c@{}|@{}c@{}|}
$\Dblank$ & $\Dblank$ & $\Dblank \cdots$ \\
\hline
$\DblankLtwo$ & $\DblankLtwo$ & $\DblankLtwo \cdots$ \\
\hline
$\DblankLthree$ & $\DblankLthree$ & $\DblankLthree \cdots$   \\
\hline
$\Dblank$ & $\Dblank$   & $\Dblank \cdots$  \\
\hline
$\#$ & $\#$ & $\# \cdots$ \\
\hline
$0$ & $0$ & $0 \cdots$ \\
\end{tabular}
\\
\hline
\end{tabular}
~~~~~~~~~~~~~~~~~~~
\begin{tabular}{|@{}c@{}|@{}c@{}|}
\hline
\begin{tabular}{@{}c@{}|@{}c@{}|@{}c@{}|}
$\cdots \blank$ & $\blank$ & $\blank $ \\
\hline
$\cdots \blankLtwo$ & $\blankLtwo$ & $\blankLtwo $ \\
\hline
$\cdots \blankLthree$ & $\blankLthree$ & $\blankLthree $   \\
\hline
$\cdots \blank$ & $\blank$   & $\blank $  \\
\hline
$\cdots \#$ & $\#$ & $\#$ \\
\hline
$\cdots 0$ & $0$ & $0$ \\
\end{tabular}
&
$\rightBr$ 
\\
\hline
\end{tabular}
\end{small}
\end{center}
\end{enumerate}
\end{proof}

$\lambda_0(N)$ is defined to be $\lambda_0(H_{C}(N) |_{\calh_{br}})$, the ground energy of a bracketed
finite chain of length $N$ in which $h_{c}$ is applied to every pair of neighboring particles in the chain.
\cite{C15} prove a robustness bound for Robinson tiling which enables us
to assume  that the cost of $H_C$ on the finite segments dominates in the ground energy
density.

\begin{lemma}
\label{lem:Robinson}
{\bf (Lemma 52 from \cite{C15})}
Let $h^{row}_R$ and $h^{col}_R$ be the local interactions of the modified
Robinson tiling acting on $\calh_R \otimes \calh_R$
as described in \cite{C15}. 
For a given ground state configuration (tiling) of $H_R$, let $\call$
denote the set of all horizontal line segments of the lattice that lie between $\leftcorner$ and $\rightcorner$ tiles. Let $h_C$ denote
the local interaction of a Robinson-Embeddable Hamiltonian $H_C (N)$
as in Definition \ref{def:GI}.
Then there is a Hamiltonian $H$ on a 2D square lattice of width and height $L$ with nearest neighbor interactions $h^{row}$ and $h^{col}$
such that for any $L$, $\lambda_0(H(L))$, the ground energy of $H(L)$,  is in the interval.
\begin{equation}
    \label{eq:squareE}
    \left[  
    \sum_{k=1}^{\lfloor \log_4 (L/2) \rfloor}
    \left( \left\lfloor \frac{L}{2^{2n+1}} \right\rfloor
    \left( \left\lfloor \frac{L}{2^{2n+1}} \right\rfloor - 1 \right) \right) \lambda_0 (4^k),
    \sum_{k=1}^{\lfloor \log_4 (L/2) \rfloor}
    \left( \left\lfloor \frac{L}{2^{2n+1}} \right\rfloor
    \left( \left\lfloor \frac{L}{2^{2n+1}} \right\rfloor + 1 \right) \right) \lambda_0 (4^k)
    \right]
\end{equation}
\end{lemma}

The expression for the ground energy density in the thermodynamic limit which
we will need in our construction follows from Lemma 
\ref{lem:Robinson} and is encapsulated in the following Corollary. 

\begin{coro}
\label{cor:ged}
The ground energy density of $h^{row}$ and $h^{col}$ in the thermodynamic limit is
$$\alpha_0 = \sum_{k=1}^{\infty} \frac{\lambda_0 (4^k)}{4^{2k+1}}.$$
\end{coro}

\begin{proof}
Using the bounds in (\ref{eq:squareE}), the ground energy density of $H$
on an $L \times L$ square is contained in:
$$\frac{\lambda_0 (H(L))}{L^2} \in 
 \sum_{k=1}^{\lfloor \log_4 (L/2) \rfloor} \lambda_0 (4^k) \left(
\frac{1}{4^{2k+1}} \pm \frac{3}{2^{2k+1} \cdot L} \right)$$
Then
\begin{align*}
    \alpha_0 & = \lim_{L \rightarrow \infty}  \frac{\lambda_0 (H(L))}{L^2}\\
    & = \lim_{L \rightarrow \infty}   \sum_{k=1}^{\lfloor \log_4 (L/2) \rfloor} \lambda_0 (4^k) \left(
\frac{1}{4^{2k+1}} - \frac{3}{2^{2k+1} \cdot L} \right)
 = \lim_{L \rightarrow \infty}   \sum_{k=1}^{\lfloor \log_4 (L/2) \rfloor} \lambda_0 (4^k) \left(
\frac{1}{4^{2k+1}} + \frac{3}{2^{2k+1} \cdot L} \right)\\
& = \sum_{k=1}^{\infty} \frac{\lambda_0 (4^k)}{4^{2k+1}}
\end{align*}
\end{proof}

\subsection{The Final Reduction: Mapping  from 
the GED to f(x)}
\label{sec:infreduction}

A few modifications for the finite construction are required for use in the infinite construction. We will describe those modifications in general terms here and then state 
Lemma \ref{lem:energyksquare} that encapsulate the key facts required of the construction.
More specifics about the modifications as well as a proof of Lemma \ref{lem:energyksquare}
is given in Section \ref{sec:modifications}.
In this section we will complete the proof of the hardness result for the thermodynamic
limit using the claims from Lemma \ref{lem:energyksquare}.

The first change is that in the infinite construction,
we will not use squares of size $4^k$ for every $k$ to encode the result of a computation. We will only use  $k$'s that are a perfect square.
Thus, we will use a chain of size $4^{x^2}$ to encode $f(x)$.
This will allow us to use the bits from $x^2$ to $(x+1)^2$
in the energy density to encode $f(x)$.
We need to add a way for $\lambda_0(H_{inf}(4^k))$ to be $0$ if
$k$ is not a prefect square.

The second issue is the fact that the computation embedded in the Hamiltonian construction
for the finite case may not finish in exactly $N-2$ steps, where $N = 4^k$ is the length of the chain.
Recall, that the clock in the construction ensures that the process is executed for exactly $N-2$ steps.
We can make sure that $N$ steps is sufficient asymptotically, so that the construction will work on all but a finite
number of $k$, but we need to control what happens for small $N = 4^k$ when the computation may not finish. 
We alter the construction so that in the cases where the computation does not finish, the ground energy is $0$.
We define the set $\calk$ to denote the set of $k$ such that the computation does not finish in $N-2$ steps for $N = 4^k$.
We then prove that there is a Turing Machine that finishes the required computation in  $o(N)$ steps, so that $\calk$ is finite.
The analysis  follows the finite case very closely and we establish the following lemma similar to the
expression for the finite case.

\begin{lemma}
\label{lem:energyksquare}{\bf [Energy of the $k^{th}$ square as a function of $k$]}
For $N = 4^k$, where $k \in \calk$ or $k$ is not a perfect square, $\lambda_0 (4^k)  = 0$.
If $k$ is a perfect square, and $x = \sqrt{k}$, then 
$$    \lambda_0 (4^k) =  \left( 1 - \cos\left( \frac{ \pi}{L+1} \right)\right),$$
where $L = (2 T(x,\tilde{y})+3) \cdot p(4^k)$ and $\tilde{y}$ is the string denoting
the correct answers to the oracle queries made by $M$ on input $x$.
\end{lemma}

Lemma \ref{lem:reduction}
proven below  uses Lemmas \ref{lem:energyksquare}  and the results on Robinson tiling from Corollary \ref{cor:ged} to complete the proof of Theorem \ref{th:mainInf}.
In particular, it shows that if the energies of the finite segments are as described in Lemma \ref{lem:energyksquare}
and the ground energy density $\alpha$ is in fact the infinite sum given in Corollary \ref{cor:ged}, then
there is a polynomial $q(x)$ such that if we are
given $\alpha$ where
$|\alpha - \alpha_0| \le 1/ 2^{q(x)}$,
it is possible to determine $\lambda_0(4^{x^2})$ to within a sufficient precision so that the integer $L$ can be determined.
From $L$, $T(x, \tilde{y})$ can be computed which in turn encodes the value of $f(x)$.
Note that Lemma \ref{lem:reduction} makes the assumption that the Turing Machine $M$ that computes $f$ makes at most $2^{|x|}$ oracle calls on input $x$.
This assumption is justified by the padding argument given in Lemma \ref{lem:padinf}.
We will start by setting up some notation and 
a few claims about the functions used in the proof.

\begin{definition}
For positive integer $x$,
\begin{itemize}
    \item $N_x = 4^{x^2}$
    \item $T_x = T(x, \tilde{y})$, where $\tilde{y}$ is the set of correct responses to the oracle queries made by $M$ on input $x$.
    \item $$l(x) = \frac{1}{4 \cdot (N_x)^{2}}\left( 1 - \cos\left( \frac{ \pi}{(2 T_x+1) \cdot p(N_x)+1} \right)\right).$$
\end{itemize}
\end{definition}

We will use the following three  claims in the analysis:

\begin{claim}
\label{cl:Tbound}
Suppose that on input $x$, $M$ makes at most
$2^{|x|}$ oracle calls.
Then $T_x \le N_{x+1}/N_x$.
\end{claim}

\begin{proof}
If $m$ is the number of oracle calls made by $M$ on input $x$,
then the high order bit of $T(x,y)$ is in location $3m+2$ for every $y$.
Therefore $2^{3m+2} \le T(x,y) \le 2^{3m+3}$.
By the assumption of the claim, on input $x$, the number of oracle calls
is at most $2^{|x|} \le x$.
Therefore $T(x,y) \le 2^{3x+3}$.
This is upper bounded by $N_{x+1}/N_x \le 4^{(x+1)^2}/4^{x^2} = 4^{2x+1} = 2^{4x+2}$ for $x \ge 1$.
\end{proof}

\begin{claim}
\label{cl:fbound}
$$l(x+1) < \frac 1 2 \cdot l(x).$$
\end{claim}

\begin{proof}
If $m$ is the number of oracle calls made by $M$ on input $x$,
then the high order bit of $T(x,y)$ is in location $3m+2$ for every $y$.
Therefore $2^{3m+2} \le T(x,y) \le 2^{3m+3}$.
As long as $M$ makes as many queries on input $x+1$ as it does
on input $x$, $T_{x+1} \ge  T_x/2$. Since $x \ge 1$,
$N_{x+1} \ge 4^{2^2} = 64$, which implies that $p(N_{x+1})/p(N_x) \ge 4$.
Therefore $(2 T_x +1) p(N_x) \le (2 T_{x+1}+1) p(N_{x+1})$.
Therefore,
$$\left( 1 - \cos\left( \frac{ \pi}{(2 T_{x+1}+1) \cdot p(N_{x+1})+1} \right)\right) \le
\left( 1 - \cos\left( \frac{ \pi}{(2 T_x+1) \cdot p(N_x)+1} \right)\right)$$
Therefore $l(x+1)/l(x) \le (N_x/N_{x+1})^ = 4^{x^2}/4^{(x+1)^2} \le 1/2$.
\end{proof}

\begin{claim}
\label{cl:cosdiff}
If $x, x'$, and $y$ are positive integers greater than $5$ and $x \neq x'$, then
$$\left| \cos \left( \frac{ \pi}{yx+1}\right) - \cos \left( \frac{ \pi}{yx'+1}\right) \right| \ge \frac{1}{y^2 x^4}.$$
\end{claim}

\begin{proof}
The difference between the two $\cos$ terms is minimized for $x' = x+1$, in which case, using the first terms in the Taylor expansion for $\cos$, we have:
\begin{align*}
  &  \left| \cos \left( \frac{ \pi}{yx+1}\right) - \cos \left( \frac{ \pi}{yx'+1}\right) \right|\\
  & \ge
    \frac{\pi^2}{2(yx+1)^2} - \frac{\pi^2}{2(y(x+1)+1)^2}
    - \frac{\pi^4}{24 (yx+1)^4}\\
    & \ge \frac{\pi^2}{2} \left[  \frac{y^2}{(yx+1)^2(y(x+1)+1)^2} - \frac{\pi^2}{12 y^4 x^4} \right]\\
    & \ge \frac{\pi^2}{2} \left[  \frac{1}{2 y^2 x^4} - \frac{(\pi^2/12 y^2)}{y^2 x^4} \right]\\
  &  \ge \frac{\pi^2}{2} \left[ \frac{1}{2 y^2 x^4} - \frac{1}{4 y^2 x^4} \right] = \frac{\pi^2}{8} \left[ \frac{1}{ y^2 x^4} \right] > \frac{1}{ y^2 x^4}
\end{align*}
The inequality from the third to the fourth line uses the fact that for
$x \ge 5$ and $y \ge 5$, 
$(yx+1)^2(y(x+1)+1)^2 \le 2y^4 x^4$. 
The inequality from the fourth to the fifth line uses the fact that for
 $y \ge 5$, 
$\pi^2 /12 y^2 \le 1/4$. 
\end{proof}

\begin{lemma}{\bf [Computing $f(x)$ from the GED]}
\label{lem:reduction}
Suppose that
$$\alpha_0 = \sum_{k=1}^{\infty} \frac{\lambda_0 (4^k)}{4^{2k+1}}$$
and $\lambda_0(4^k)$ satisfies the conditions in Lemma \ref{lem:energyksquare}.
Suppose also that on input $x$, the Turing Machine $M$ that computes $f$ makes
at most $2^{|x|}$ oracle calls.
Then if $x^2 \not\in \calk$, there is a polynomial $q$ such that we can compute $f(x)$ in time that is polynomial in $x$, 
given $\alpha$ where
$|\alpha - \alpha_0| \le 1/ 2^{q(x)}$.
\end{lemma}

This lemma implies the proof of Theorem \ref{th:mainInf} 
in a straight-forward way. 

\begin{proof}
The polynomial $q$ is chosen so that 
$1/ 2^{q(x)} \le l(x+1)$. Let $\calk(k)$ be an indicator function which is equal to $1$ if $k \in \calk$ and is equal to $0$ otherwise. 
Since $\lambda_0 (4^k) = 0$ if $k$ is not a perfect square or if $k \in \calk$, we have that
$$\alpha_0 = \sum_{j=1}^{\infty} (1 - \calk(j^2)) l(j).$$

We will prove by induction on $s$ that for $s \le x$, we have a $\alpha_s$
that satisfies:
$$\left| \alpha_s - \sum_{j=s}^{\infty} (1 - \calk(j^2))l(j) \right|  \le \left( 1 + \frac s x \right) l(x+1).$$
We will show that if $s^2 \not\in \calk$,
then we can uniquely determine $T_s$ in polynomial time. Moreover, if $s < x$,
then we can find $\alpha_{s+1}$ such that
$$\left| \alpha_{s+1} - \sum_{j=s+1}^{\infty} (1 - \calk(j^2)) l(j) \right|  \le \left( 1 + \frac {s+1} x \right) l(x+1).$$
Note that we start with $\alpha = \alpha_{1}$ that satisfies the inductive hypothesis for $s = 1$.
When we reach $s = x$, as long as $x^2 \not\in \calk$, we will have an
 $\alpha_x$ from which we can determine $T_x$ which then determines $f(x)$.

Now to prove the inductive step.
If $s^2 \in \calk$, then
$$\sum_{j=s}^{\infty} (1 - \calk(j^2))l(j) = \sum_{k=j+1}^{\infty} (1 - \calk(j^2))l(j).$$
We can set $\alpha_{s+1} = \alpha_s$ and the inductive step is complete.

Now suppose that $s^2 \not\in \calk$. By Claim \ref{cl:fbound}, $\sum_{j=s+1}^{\infty} l(j) \le 2 l(s+1)$. Therefore
$$|\alpha_s - l(s)| \le 2 \cdot l(s+1) + \left( 1 + \frac s x \right) l(x+1)
\le 4 l(s+1)$$
The value $\lambda_0(N_s)$ is
$$ \left(1 - \cos \left( \frac{ \pi}{(2T_s + 1) p(N_s) +1}\right) \right).$$
Plugging in the definition for $l(s)$ and multiplying by $4(N_s)^2$, we get that
\begin{align}
    |4 (N_s)^2 \alpha_s - \lambda_0(N_s)| & \le \frac{(N_s)^2}{(N_{s+1})^2} \cdot 
    \left( 1 - \cos\left( \frac{ \pi}{(2 T_{s+1}+1) \cdot p(N_{s+1})+1}
    \right)\right)\\
    & \le \frac{(N_s)^2}{(N_{s+1})^2} \cdot 
    \left( \frac{ (\pi/2)^2}{4(p(N_{s+1}))^2 (T_{s+1})^2}
    \right)\\
    & \le \frac{(N_s)^2}{ (N_{s+1})^2 (p(N_{s+1}))^2 (T_{s+1})^2}
    \label{eq:Ebound}
\end{align}
By Claim \ref{cl:Tbound}, for any $T \neq T_s$, 
$$\left|   \left(1 - \cos \left( \frac{ \pi}{(2T_s + 1) p(N_s)+1 }\right) \right) -
\left(1 - \cos \left( \frac{ \pi}{(2T + 1) p(N_s) +1}\right) \right) \right|
\ge \frac{1}{ (p(N_s))^2 (T_s)^4}.$$
Therefore, if we can determine the value of $\lambda_0(N_s)$
to within $\pm 1/2 (p(N_s))^2 (T_s)^4$, the value of $T_s$ is uniquely determined
and can be computed by binary search in time that is poly-logarithmic in $N_s$, 
and therefore polynomial in $x$.
If we can determine $T_s$, then we can use the Taylor approximation for the
$\cos$ function to compute an approximation
$\tilde{l}(s)$
of $l(s)$ to within $\pm l(x+1)/x$. We can then set
$\alpha_{s+1} = \alpha_s - \tilde{l}(s)$.
\begin{align*}
\left| \alpha_{s+1} - \sum_{j=s+1}^{\infty} l(j) \right|
 = & ~~\left| \alpha_s - \tilde{l}(s) - \sum_{j=s+1}^{\infty} l(j) \right| \\
 \le  &~~ |l(s) - \tilde{l}(s)| + \left| \alpha_s - \sum_{j=s}^{\infty} l(j) \right| \\
 \le & ~~ \frac{l(x+1)}{x} + \left( 1 + \frac s x \right) l(x+1) = \left( 1 + \frac {s+1} x \right) l(x+1)
\end{align*}
It remains to show that the bound from (\ref{eq:Ebound}) is at most
$1/2 (p(N_s))^2 (T_s)^4$.
Thus, we need to show
\begin{equation}
\label{eq:toshow}
    \left[ \frac{ p(N_s) N_s }{ p(N_{s+1}) N_{s+1}} \right]^2
\le \left[ \frac{T_{s+1}}{2 (T_{s})^2 }\right]^2
\end{equation} 
Using the bounds from Claim \ref{cl:Tbound}, we get that
$T_s \le N_{s+1}/N_s$. Also,
The polynomial $p(N) = 4(N-2)(2N-3)$. 
Since $N_{s+1} \ge N_s \ge N_2 \ge 64$, 
 $p(N_s)/p(N_{s+1}) \le 1/4$.
Therefore 
$$
\left[ \frac{ p(N_s) N_s }{ p(N_{s+1}) N_{s+1}} \right]^2
\le 
\left[ \frac{ p(N_s)  }{ p(N_{s+1}) T_s} \right]^2 \le \frac{1}{(4 T_s)^2}$$
As long as $M$ makes as many oracle queries on input $(s+1)^2$ as it does on $s^2$, $2T_{s+1} \ge  T_s$.
The  the bound required in (\ref{eq:toshow}) follows.
\end{proof}

\subsection{Modifications to the Hamiltonian  from the Finite Case}
\label{sec:modifications}

The energy density in the thermodynamic limit will be 
$$\sum_{k=1}^{\infty} \frac{ \lambda_0(4^k)}{4^{2k+1}}.$$
With this expression in mind, we will describe required changes to
the construction for the finite case.
The first thing to notice is that the GED only encodes energies for chains of length $N = 4^k$, so the input to the function $f$ from which we are reducing will be $k$ which is logarithmic in $N$. The size of the input is $\log k$ which means that $N$, the size of the chain, is doubly exponential in the input length. This is why we are reducing from $\fexpnexp$, instead of $\fpnexp$ as we did in the finite case.

Furthermore, we will not use all the $k$'s to encode the result of a computation. We will only use the $k$'s that are a perfect square.
Thus, we will use a chain of size $4^{x^2}$ to encode $f(x)$.
This will allow us to use the bits from $x^2$ to $(x+1)^2$
in the energy density to encode $f(x)$.
We need to add a way for $\lambda_0(4^k)$ to be $0$ if
$k$ is not a prefect square. We handle this case by encoding
the fact that $k$ is not a prefect square in the output of 
of the verifier. Recall that $M_{TV}$ writes a symbol $\sigma_A$ or $\sigma_R$
as the output of its verification process. If $k$ is not a perfect square then we make sure that a different
character $\sigma_B$ is written in the first location
when $M_{TV}$ halts. 
The additional $+1$ in $H_{final}$ is only triggered if
the output of $M_{TV}$ starts with an $\sigma_A$ or an $\sigma_R$. 
If the initial symbol is neither $\sigma_A$ nor $\sigma_R$, then there will be no additional terms along the diagonal besides $H_{prop}$. The equal superposition of all the states in the computation will be a $0$ energy state.

Finally, we will need to alter the construction to ensure that no cost is incurred
if $M_{TV}$ does not finish its computation. 
Although $N-2$ steps will be enough for $M_{BC}+M_{TV-inf}$ to complete its computation asymptotically, for small $N = 4^k$, this may not be the case. 
We will manage this by replacing the  symbols $\{\sigma_A,  \sigma_R, \sigma_B, \sigma_X\}$ with symbols from $\{\gamma_A,  \gamma_R, \gamma_B, \gamma_X\}$ everywhere in the transition rules for
$M_{TV-inf}$. 
The $\gamma$ symbols do not trigger any penalty terms.
Then there is a post-processing step in which $\{\gamma_A,  \gamma_R, \gamma_B, \gamma_X\}$ are converted to $\{\sigma_A,  \sigma_R, \sigma_B, \sigma_X\}$ in a controlled way so that if the computation halts before it finishes, there will be a $0$ energy state.

After the binary counter Turing Machine runs for $N-2$ steps, we will run a process for $N-2$ steps that consists of the following three Turing Machines dovetailed together:
\begin{enumerate}
    \item $M_{BC}$ is the same as in the finite case and is described in Section \ref{sec:MBC}.
    \item $M_{TV-inf}$ will be an altered version of the $M_{TV}$ used in the construction for the finite systems. The pseudo-code for $M_{inf}$ is given in Figure \ref{fig:Minfpseudo}.
    \item $M_{post}$ is a post-processing step which ensures that the symbols that trigger a periodic cost appear in a controlled way so that if the computation stops before completion, the ground energy will be $0$. The pseudo-code for $M_{post}$ is given in Figure \ref{fig:Mpostpseudo}. The transition rules are specified in Figure \ref{fig:Mpostrules}.
\end{enumerate}

\begin{figure}[]
\noindent
\fbox{\begin{minipage}{\textwidth}
\begin{tabbing}
{\sc Input:} $(z, w)$\\
(1)~~~~ \= Compute $N = N(z)$   \\
(2)  \> Compute $k = \lfloor \log_4 N \rfloor$\\
(3)  \> Compute $x = \lfloor \sqrt{k} \rfloor$\\
(4) \> \IF $x^2 = k$\\
(5) \> ~~~~~ \= $m$ is the number of oracle queries made by $M$ on input $x$\\
(6) \> \> $r$ is the size of  the witness used by $V$ on any oracle query generated by $M$ on input $x$\\
(7) \> \> $m$ and $r$ are hard-coded functions of $|x|$ determined by $M$ and $V$.\\
(8) \> \>    Simulate Turing Machine $M$ on input $(x,y)$\\
(9)  \> \> ~~~~~ \= Use $y$ for the responses to the oracle queries, $y$ denotes the first $m$ bits of $w$\\
(10)  \> \> \>Simulation generates $x_1, \ldots, x_m$, inputs to oracle queries\\
(11)  \> \> {\sc Reject} $= $ FALSE\\
(12)   \> \> \FOR $i = 1 \ldots m$\\
(13) \> \> \> \IF $y_i = 1$\\
(14) \> \> \> ~~~~~ \= Simulate $V$ on input $(x_i, w_i)$\\
(15) \> \> \> \> ~~~~~ \= $w_i$ is the string formed by bits $m +(i-1)r+1$ through $m + ri$ of the witness track\\
(16) \> \> \> \> If $V$ rejects on input $(x_i, w_i)$\\
(17) \>\>\>\>  \>{\sc Reject} $= $ TRUE\\
(18) \> \>\IF (\sc Reject)\\
(19) \>\>  \>Write "$\gamma_R$" in the left-most position of the work tape\\
(20) \> \>\ELSE\\
(21) \>\>  \> Write "$\gamma_A$" in the left-most position of the work tape\\
(22) \> \>Compute $T(x,y)$\\
(23) \> \>Write the value of $T(x,y)$ in unary with "$\gamma_X$" symbols starting in the second position of the work tape\\
(24) \> \IF $l^2 \neq k$\\
(25) \> ~~~~~ \= Write "$\gamma_B$" in the left-most position of the work tape\\
(26) \> \> Reverse the computation and finish with $(\gamma_B, z,w)$ on the work tape\\
{\sc Output:}\\
\> \IF $k$ is not a perfect square $(\gamma_B, z, w)$\\
\> \IF $k$ is a perfect square $( \{\gamma_A, \gamma_R \} (\gamma_X)^{T(x,y)};z,w)$

\end{tabbing}
\end{minipage}}
\caption{Pseudo-code for the Turing Machine $M_{TV-inf}$.}
\label{fig:Minfpseudo}
\end{figure}

\begin{figure}[ht]
\noindent
\begin{center}
\fbox{\begin{minipage}{\textwidth}
\begin{tabbing}
(1)~~~~ \= \IF current symbol is not in $\{ \gamma_A, \gamma_R, \gamma_B\}$ \\
(2) \> ~~~~~ \= Transition to $s_f$, stay in place.\\
(3)  \> \ELSE \\
(4) \> \> move right into state $s_R$\\
(4)  \> \WHILE (current symbol is $\gamma_X$)\\
(5) \> \> Change symbol to $\sigma_X$ and move right in state $s_R$\\
(6) \> \IF current symbol is not $\gamma_X$\\
(7) \> \> Move left into $s_L$\\
(8)  \> Move left until a symbol from $\{ \gamma_A, \gamma_R, \gamma_B\}$ is reached\\
(9)  \> Change $\gamma_A/\gamma_B/\gamma_R$ to $\sigma_A/\sigma_B/\sigma_R$, transition to $q_f$, stay in place\\
\end{tabbing}
\end{minipage}}
\end{center}
\caption{Pseudo-code for the Turing Machine $M_{post}$.}
\label{fig:Mpostpseudo}
\end{figure}


\begin{figure}[ht]
\centering
\begin{tabular}{|c|c|c|c|c|c|c|}
\hline
 & $\gamma_A$ & $\gamma_R$ & $\gamma_B$ & $\gamma_X$ & $\sigma_X$ & $a \neq \gamma_A, \gamma_R, \gamma_B, \gamma_X$  \\
 \hline
 \hline
 $s_{0}$ & $(s_R, \gamma_A, R)$ & $(s_R, \gamma_R, R)$  & $(s_R, \gamma_B, R)$ & $(s_f, \gamma_X, N)$ & - & $(s_f, a, N)$ \\
 \hline
 $s_{R}$ & $(s_L, \gamma_A, L)$ & $(s_L, \gamma_R, L)$  & $(s_L, \gamma_B, L)$ & $(s_R, \sigma_X, R)$ &  - & $(s_L, a, L)$ \\
 \hline
 $s_{L}$ & $(s_f, \sigma_A, N)$ & $(s_f, \sigma_R, N)$  & $(s_f, \sigma_B, N)$  & * & $(s_L, \gamma_X, L)$ & * \\
 \hline
\end{tabular}
\caption{A summary of the transition rules for $M_{post}$. The transition rules are well-defined
as long as the tape contents do not contain $\{\sigma_A,  \sigma_R, \sigma_B, \sigma_X\}$ when $M_{post}$ begins.
The rules can be extended so that they are well defined on every $Q \times \Gamma$ pair, so that the reduced transition function remains one-to-one and the Turing Machine
remains unidrectional. The rule $\delta(a, q_f) = (q_0, a, N)$ is included for every $a \in \Gamma$ so that the resulting Turing Machine is in normal-form.}
\label{fig:Mpostrules}
\end{figure}

Lemma \ref{lem:addsymbols} is used so that $M_{BC}$, $M_{TV-inf}$, and $M_{post}$
all have the same tape alphabet $\Gamma$.
Then the same transformation that was applied to $M_{BC} + M_{TV}$ in the finite case
is applied here in the infinite case to $M_{BC} + M_{TV-inf}+M_{post}$,
with the exception that the final state for $M_{BC} + M_{TV-inf}+M_{post}$ is $r_f$
instead of $p_f$ for $M_{BC} + M_{TV}$.
Also the "time-wasting" subroutine is entered if the current character is in $\{\sigma_A, \sigma_R, \sigma_B\}$ when the final state $r_f$ 
of $M_{post}$ is reached.
The finite case did not have a $\sigma_B$, so the time-wasting subroutine was entered only if the
current character was $\sigma_A$ or $\sigma_R$ at the time that $M_{BC} + M_{TV}$ halts. More specifically, the rule $\delta(p_f, \sigma_B) \rightarrow (q_0, \sigma_B, N)$ is removed and the following two rules are added:
\begin{align*}
    \delta(p_f, \sigma_B) & \rightarrow (q_*, \sigma_B, R)\\
    \delta(q_*, \sigma_B) & \rightarrow (q_0, \sigma_B, N)\\
\end{align*}
These new rules preserve the properties of reversibility in exactly the same way as the original transformation did.
We will call this altered transformation $\calt'$ and define:
$$M_{check-i} = \calt' ( M_{BC} + M_{TV-inf} + M_{post})$$
The process embedded in the Hamiltonian for the infinite case will be $N-2$ steps of $M_{BC}$ followed
by $N-2$ steps of $M_{check-i}$.

To summarize what happens with $M_{check-i}$. Suppose that $N = 4^k$.
$M_{BC}$ is run for $N-2$ steps and ends in the middle of an increment operation.
$M_{check-i}$ will pick up where $M_{BC}$ left off and complete the increment operation, ending with some string $z$ on the work tape. $N(z)$ will be $4^k$ plus the few steps required to complete the increment operation. Even though $N(z) \neq 4^k$, there is enough information to recover $k$ from $z$
since $k$ will be $\lfloor \log_4 N(z) \rfloor$.
If $k$ is not a perfect square, then $M_{check-i}$ will end with the input
$(\gamma_B, z,\myw)$ on its work tape. In this case $M_{post}$ will convert $\gamma_B$ to $\sigma_B$ and then halt.
The time-wasting process is triggered which preserves the output of the tape until the clock runs out.
If $k = x^2$, for some integer $x$, then $M_{TV-inf}$ does exactly what $M_{TV}$
would do on input $(x, \myw)$, using the Turing Machines $M$ and $V$ for $f$.
Note that lines $(5)$ through $(23)$ of $M_{TV-inf}$ are exactly the same
as lines $(2)$ through $(20)$ of $M_{TV}$. The functions for $m$, the number of oracle calls and $r$ the size of the witnesses used for $V$ are different but both are hard-coded into $M$ and $V$. Therefore at the end of $M_{TV-inf}$'s computation, the output is exactly
{\sc Out}$(x, \myw, M, V)$ as given in Definition \ref{def:correctoutput} for the finite case,
with the exception that $M_{TV}$ has written its output with symbols $\{\gamma_A \gamma_R, \gamma_X\}$ instead of $\{\sigma_A \sigma_R, \sigma_X\}$.
Then $M_{post}$ changes the output $(\gamma_A + \gamma_R) (\gamma_X)^T$
to $(\sigma_A + \sigma_R) (\sigma_X)^T$.
This  is done in such a way that the $\sigma_A/\sigma_R$ symbol is the very last symbol written.

Lemma \ref{lem:correctoutput-inf} establishes that this process produces the correct output for the construction. 
It will be convenient to have a way to connect $N$, the size of the chain and the string
$z$ that will be one the work tape after $M_{check-i}$ finishes the last increment operation.

\begin{definition}
$\caln(x)$ is the set of positive integers $N$ such that $N \le N(x)$ and there is no other $x' \neq x$ such that $N \le N(x') < N(x)$.
\end{definition}

\begin{lemma}
\label{lem:correctoutput-inf}
{\bf [$M_{check-i}$ Produces the Correct Output]}
Suppose that $N = 4^k$ for $k \ge 2$.  Suppose also that $N \in \caln(z)$ and $M_{TV-inf}$ behaves properly on all binary input pairs $(z,\myw)$
and is $1$-compact.
Then the process that starts with input $(1,\myw)$ for $\myw \in \{0,1\}^{N-2}$,
runs $M_{BC}$ for $N-2$ steps and runs $M_{check-i}$ for $N-2$ steps will keep the head of the Turing Machine inside
tape cells $1$ through $N-3$. Also
\begin{enumerate}
    \item If $M_{BC} + M_{TV-inf}+M_{post}$ does not finish, then there will be no symbols from
$\{\sigma_A \sigma_R, \sigma_B, \sigma_X\}$ on the work tape.
\item If $M_{BC} + M_{TV-inf}+M_{post}$ does  finish and $k$ is not a perfect square, then the process finishes with $(\sigma_B, z,w)$ on its work tape.
\item If $M_{BC} + M_{TV-inf}+M_{post}$ does  finish and $k = x^2$, then the process finishes with tape contents {\sc Out}$(x, \myw, M, V)$ as given in Definition \ref{def:correctoutput} for the finite case.
\end{enumerate}
\end{lemma}

\begin{proof}
According the Lemma \ref{lem:Nofxbound}, after $N-2$ steps of $M_{BC}$, the head stays within tape cells $1$ through $N-4$. After $M_{check-i}$ takes over, the last unfinished increment operation is completed. The head can only reach one additional location past $N-4$ in a single increment operation which means the head has stayed within locations $1$ through $N-3$. At this point the head is back in location $1$ and $M_{TV-inf}$ is launched. Since $M_{TV-inf}$ behaves properly on all binary inputs, the head will not go to the left of tape cell $1$. There are at most $N-3$ steps remaining and since $M_{TV-inf}$ is $1$-compact, it will not go past location $N-3$ on the right.
When $M_{TV-inf}$ finishes, the head is back at tape cell $1$ at which point $M_{post}$ is launched. 
At this point, at least $2$ steps of $M_{check-i}$ have been run and there are at most $N-4$ steps
left to go. The head can not reach past location $N-3$ in only $N-4$ steps.
Also, the rules $M_{post}$ guarantee that the head will not move to the left of location $1$.
Furthermore, $M_{post}$ will end in the location in which is began.

The only time a character from $\{\sigma_A \sigma_R,  \sigma_B, \sigma_X\}$ is written, is in the last step of $M_{post}$, so if $M_{BC} + M_{TV-inf}+M_{post}$ does not finish, then there will be no occurrences of $\{\sigma_A \sigma_R,  \sigma_B, \sigma_X\}$ on the tape.

If $M_{BC} + M_{TV-inf}+M_{post}$ does finish and $k$ is not a perfect square, then
the contents of the tape after $M_{TV-inf}$ finishes is $(\gamma_B, z, w)$. $M_{post}$ will convert the
$\gamma_B$ to a $\sigma_B$ and then halt. The time-wasting process is launched and the tape contents will remain unchanged until the $N-2$ steps of $M_{check-i}$ are complete.

If $M_{BC} + M_{TV-inf}+M_{post}$ does finish and $k = x^2$, then
the contents of the tape after $M_{TV-inf}$ finishes is {\sc Out}$(x, \myw, M, V)$, except that every $\sigma_i$ is a $\gamma_i$. $M_{post}$ will convert the $\gamma$'s to $\sigma$'s and
the contents of the tape will be {\sc Out}$(x, \myw, M, V)$. When $M_{post}$ finishes, the head will be pointing to the $\sigma_A$ or $\sigma_R$ which will launch the time-wasting process and the tape contents will remain unchanged until the $N-2$ steps of $M_{check-i}$ are complete.
\end{proof}

\subsubsection{Existence of $M_{TV-inf}$ that Produces the Correct Output for Most N}

We now need to argue that there is a Turing Machine $M_{TV-inf}$ that has the required properties so that the resulting $M_{check-i}$ produces the correct output for most values of $N$. As with the finite case, we need some assumptions about the running time of $M$ and $V$ which $M_{TV-inf}$ simulates which can be proven using standard padding arguments. 

\begin{lemma}
\label{lem:padinf}
{\bf [Padding Lemma for $\fexpnexp$]}
If $f \in \fexpnexp$, then for any constants, $c_1$ and $c_2$,
$f$ is polynomial time reducible to a function $g \in \fexpnexp$ such that $g$ can computed by
an exponential-time Turing Machine $M$
with access to a $\nexp$ oracle for language $L$. The verifier for $L$ is a  Turing Machine $V$.
Moreover, on input $x$ of length $n$, $M$ runs in $O(2^{c_1 n})$ time, and makes
at most $2^{c_1 n}$ queries to the oracle. Also, the length of the queries made to the oracle is at most $2^{c_1 n}$ and the running time of $V$ as well as the size of the witness required for $V$ on any query made by $M$ is $O(2^{2^{c_2 n}})$
\end{lemma}

\begin{proof}
Suppose that $M$ runs in $O(2^{r_1(n)})$ time and the running time of $V$ is $O(2^{r_2(n)})$.
The reduction will pad an input $x$ to $f$ with $1 0^{t(n)} 1$, for a polynomial $t$ chosen below.
We will refer to this substring as the {\em suffix}.
The algorithm for $g$ will verify that the length of the suffix  is in fact $t(n)$, where $n$ is the number of bits not included in the suffix.
If not, then the value of $g$ is $0$. If the input to $g$ does have a suffix with the correct length
then it erases the suffix and simulates $M$ on the rest of the string.
Note that the input length to $g$ is now $\bar{n} = t(n) + n + 2$. 
The polynomial $t$ is chosen so that $2^{ c_1 \bar{n}} \ge 2^{c_1 t(n)} \ge 2^{r_1(n)}$, so that the running time of 
$M$ and hence the number of oracles made, as well as the length of the inputs to those oracles
are all bounded  by $2^{c_1 \bar{n}}$.  The running time of $V$ on any of the query inputs is
at most 
$2^{r_2(2^{r_1(n)})}$. $t$ will also be chosen to be large enough so that
$$2^{r_2(2^{r_1(n)})} < 2^{2^{c_2 t(n)}} < 2^{2^{c_2 \bar{n}}}.$$
\end{proof}

We need now to establish that $M_{TV-inf}$ can be chosen so that for all but a finite set of $k$, 
$M_{BC} + M_{TV-inf}+M_{post}$ finishes in $4^k-2$ steps for every witness $w$. 
\begin{definition}
\label{def:sufficient}
{\bf [Sufficient $N$ for Completing $M_{BC} + M_{TV-inf}+M_{post}$'s Computation]}
We say that $N$ is sufficient for $M_{BC} + M_{TV-inf}+M_{post}$ if for every 
$\myw \in \{0,1\}^{N-4}$,
starting with input $(1,\myw)$, if $M_{BC}$ is run for $N-2$ steps, followed by
$M_{BC} + M_{TV-inf}+M_{post}$ for $N-2$ steps, the final state $r_f$ for $M_{post}$ is reached.
\end{definition}

\begin{lemma}
\label{lem:existsMTVinf}
{\bf [Existence of $M_{TV-inf}$ With Required Running Time ]}
There is a reversible normal-form Turing Machine $M_{TV}$ such that 
\begin{enumerate}
    \item $M_{TV-inf}$ behaves properly on all $(a, b)$, where $a, b \in \{0,1\}^*$.
    \item $M_{TV-inf}$ is $1$-compact.
    \item For all but a finite number of  $k$, if $N = 4^k$ is sufficient for $M_{BC} + M_{TV-inf}+M_{post}$.
\item The behavior of $M_{TV-inf}$ is as described in pseudo-code given in Figure \ref{fig:Minfpseudo}.
\end{enumerate}
\end{lemma}

\begin{proof}
If $N = 4^k$, the time to complete the increment operation after $N-2$ steps of $M_{BC}$ is $O(k)$.

Next we consider the time to compute the function outlined in Figure \ref{fig:Minfpseudo}
on a RAM and then argue about the time to convert the algorithm to a reversible Turing Machine.
The initial arithmetic computations can be done in time that is polylogarithmic in $N$. 
Now suppose that $k = x^2$.
By Lemma \ref{lem:pad}, we can assume that the running time of $M$
is on the order of $2^{c_1 |x|} \le 2^{c_1 |k|} \le k^{c_1}$.
The number of queries made by $M$ and the size of the input to those questions is
also on the order of $k^{c_1}$. 
The running time of each call is on the order of $2^{2^{c_2 |x|}} \le 2^{k^{c_2}}$. Thus, the total time for running the verifier is on the order of $k^{c_1} 2^{k^{c_2}}$.
The time to compute $T(x,y)$ and write the resulting value in unary is $O(4^{k^{c_1}})$.

Finally, the time to run $M_{post}$ after the final state of $M_{TV-inf}$ is reached is $O(T(x,y))$ which
is also $O(4^{c_1 k})$.
Thus, for any constant $d$, the constants $c_1$ and $c_2$ can be chosen so that the running time to compute
the function computer by $M_{TV-inf}$ is $O(2^{dk}))$.

The RAM algorithm can be converted to a TM computation with at most polynomial overhead.
Finally, by Theorem B.8 from \cite{BV97}, the TM can be converted to a reversible TM that behaves properly on all inputs
with quadratic overhead. Therefore the constant $d$ can be chosen to be small enough so that the running time
of the reversible TM to compute $M_{BC} + M_{TV-inf}+M_{post}$ is $o(4^k)$. Since $N = 4^k$, the total running time will be at most
$N-2$ for all but a finite number of values for $k$.

Finally, to satisfy the second assumption, we can add one additional state which causes the reversible TM to stay in place for one step before beginning its computation.
\end{proof}

\subsubsection{Additional Penalty Terms for the Infinite Construction}

The clock and the propagation terms will be exactly the same as the finite case, except that $M_{check-i}$ is replaced by $M_{check-i}$ throughout.  We need to adjust the penalty terms so that the ground energy is $0$ for all $N$ in which the process
executed on the computation tracks does not have time to finish. Also, we only run the simulation on segments  of size $4^k$, where
$k$ is a perfect square. If $k$ is not a perfect square then we want the Hamiltonian on that segment to have a ground energy of $0$. We have set up the output of the Turing Machine  so that these situations can be detected.

In the correct initial configuration, the contents of Track $5$ are
$\leftBr  1 ~\#^*  \rightBr$, where $\#$ is the blank symbol.
The state should be $q_0$ and the head should be pointing to the cell with the $1$.
$h_{init}$ is exactly as it was for the finite case.
The first of the two rules given below
pertains to Tracks $1$ and $5$ only and 
implicitly acts as the identity on all other tracks. The third rule
pertains to Tracks $1$, $4$, and $5$:
$$h_{init} = \sum_{a \in \Gamma, ~x \in \Gamma - \{ \#\} } \ketbrabig{\fourcells{\Dblank}{a}{\arrLeight}{x}}{\fourcells{\Dblank}{a}{\arrLeight}{x}}
~~~~~+~~~~~\sum_{x \in \Gamma - \{ 1 \} } \ketbrabig{\threecellsL{\arrLeight}{x}}{\threecellsL{\arrLeight}{x}}
~~~~~+~~~~~
\ketbrabig{\fourcellsL{\arrLeight}{\neg q_0}{1}}{\fourcellsL{\arrLeight}{\neg q_0}{1}}
$$

In order to understand the remaining changes, it will be useful to analyze what the contents of Track $5$ (the work track) can be at various points in the computation. We will let $\bar{\Gamma}$ denote all the tape symbols $\Gamma$,
except for $\{\gamma_A, \gamma_R, \gamma_B, \gamma_X, \sigma_A, \sigma_R, \sigma_B, \sigma_X\}$. At any point in the computation of $M_{BC}$
the tape contents are always from $(\bar{\Gamma})^*$.
At any point in the computation of 
$M_{TV-inf}+M_{post}$, the tape contents will be in one of the following forms:
\begin{align}
\label{eq:firstcase}
    &  (\bar{\Gamma})^*  & M_{TV-inf}~\mbox{running}\\
    &  (\gamma_A + \gamma_R + \gamma_B) (\bar{\Gamma} + \gamma_X)^*   & M_{TV-inf}~\mbox{running}\\
    &  (\gamma_A + \gamma_R + \gamma_B) (\sigma_X)^* (\gamma_X)^* (\bar{\Gamma})^*   & M_{post}~\mbox{sweeping right}\\
\label{eq:fourthcase}
    &  (\gamma_A + \gamma_R + \gamma_B) (\sigma_X)^*  (\bar{\Gamma})^*   & M_{post}~\mbox{sweeping left}\\
\label{eq:fifthcase}
    &  (\sigma_A + \sigma_R + \sigma_B) (\sigma_X)^*  (\bar{\Gamma})^*   & \mbox{last step of}~M_{post}
\end{align}
 (\ref{eq:firstcase}) through (\ref{eq:fourthcase}),
show the possibilities for the tape contents up until the point where $M_{check-i}$ hits a final state.
The fifth case (\ref{eq:fifthcase}), corresponds to a finished $M_{check-i}$ which keeps it's tape contents for the remainder of the $N-2$ steps.
Therefore, we would like to have any symbol from the set $\Sigma = \{ \gamma_A, \gamma_R, \gamma_B, \sigma_A, \sigma_R, \sigma_B, \sigma_X \}$ count as a symbol for
checking the length of the Track $3$ timer. Recall that the Track $3$ timer always has the form
$\DblankLthree^* (\arrL + \arrR) \blankLthree^* \dead^*$. Thus, as the Track $1$ $4$-pointer sweeps left,
it will check that the symbol on Track $5$ is from the set $\Sigma$ if and only if
the Track $3$ state is from the set $\{ \blankLthree, \arrL, \arrR, \DblankLthree \}$.
$$h_{length} = 
\sum_{\substack{s \in \{\DblankLthree, \blankLthree, \arrR, \arrL \}\\ a ~\not\in \Sigma} } \ketbrabig{\threecellsvert{\arrLfour}{s}{a}}{\threecellsvert{\arrLfour}{s}{a}}
~~~~~+~~~~~
\sum_{b \in \Sigma }
\ketbrabig{\threecellsvert{\arrLfour}{\dead}{ b}}{\threecellsvert{\arrLfour}{\dead}{b}}
$$
If $M_{check-i}$ finishes and $N = 4^k$, where $k$ is not  a perfect square, then the number of $\sigma_X$'s be $0$ and the $h_{length}$ term will check that the Track $3$ timer is $0$.

$h_V$ which introduces a
 penalty if the verifier $V$ rejects on any of its computations does not change. The term below applies to Tracks $1$ and $5$:
$$h_{V} = \ketbrabig{\threecellsL{\arrLfour}{\sigma_R}}{\threecellsL{\arrLfour}{\sigma_R}}.$$

Finally, the penalty  for correct computations and clock configurations that occur every cycle for Tracks $1$, $2$, and $3$
should only happen if the $M_{check-i}$ finishes and $N = 4^k$, where $k$ is   a perfect square.
In this case there will be an $\sigma_A$ or an $\sigma_R$ in the leftmost position. The term below applies to Track $1$, $2$, $3$, and $5$:
$$h_{final} = \frac 1 2 \ketbrabig{\fivecellsL{\arrLfour}{\arrR}{\arrR}{\sigma_A}}{\fivecellsL{\arrLfour}{\arrR}{\arrR}{\sigma_A}}
+ \frac 1 2 \ketbrabig{\fivecellsL{\arrLfour}{\arrR}{\arrR}{\sigma_R}}{\fivecellsL{\arrLfour}{\arrR}{\arrR}{\sigma_R}} + \frac 1 2 \ketbrabig{\fivecellsL{\arrRfive}{\arrR}{\arrR}{\sigma_A}}{\fivecellsL{\arrLfour}{\arrR}{\arrR}{\sigma_A}}
+ \frac 1 2 \ketbrabig{\fivecellsL{\arrRfive}{\arrR}{\arrR}{\sigma_R}}{\fivecellsL{\arrLfour}{\arrR}{\arrR}{\sigma_R}}.$$

The final $2$-particle Hamiltonain term is:
$$h = h_{prop} + h_{wf-cl} + h_{wf-co} + h_{cl} + h_{init} + h_{length} + h_V + h_{final}.$$
For each of these terms $h_*$, we will use $H_{N,*}$ to refer to the Hamiltonian on a chain of length $N$
obtained by applying the $2$-particle term $h_*$ to each pair of neighboring particles in the chain.

\subsubsection{Analysis of the Ground Energy of the Modified Hamiltonian}

The analysis addresses three separate cases: $N = 4^k$ is not sufficient for $M_{BC} + M_{TV-inf}+M_{post}$,
$N = 4^k$ is  sufficient for $M_{BC} + M_{TV-inf}+M_{post}$ and $k$ is not a perfect square, and finally,
$N = 4^k$ is  sufficient for $M_{BC} + M_{TV-inf}+M_{post}$ and $k$ is  a perfect square.

\begin{lemma}
\label{lem:doesntfinish}
{\bf [Ground Energy for $N$ not Sufficient]}
Suppose that $N = 4^k \in \caln(z)$, and $N$ is not sufficient
for $M_{BC} + M_{TV-inf}+M_{post}$, then $\lambda_0 ( 4^k) = 0$.
\end{lemma}

\begin{proof}
First, we establish that there is a witness $\myw$ such that
starting with input $(1,\myw)$, at the end of the process ($N-2$ steps of $M_{BC}$ and $N-2$ steps of $M_{check-i}$), there are no symbols from
$\{ \sigma_A. \sigma_B, \sigma_R \}$ on the tape.
We consider two separate cases:
\begin{description}
\item {\bf Case 1: $N = 4$:} There are only $2$ tape symbols in the chain. $2$ steps of $M_{BC}$ are run followed by two steps of $M_{check-i}$. To handle this special case, we can add three time-wasting steps to $M_{BC}$ before it begins the increment operation:
\begin{align*}
    \delta(q_0, a) & = (q'_0, a, N)\\
    \delta(q'_0, a) & = (q''_0, a, N)\\
    \delta(q''_0, a) & = (q''''_0, a, N)\\
\end{align*}
Then $q'''_0$ replaces $q_0$ in the rules of $M_{BC}$. This adds $3n(x^R)$ to $N(x)$
which is only an increase of a constant factor. Lemmas \ref{lem:Nofxbound} and \ref{lem:Nlowerbound} still hold.
Therefore two steps of $M_{BC}$ followed by two steps of
$M_{check-i}$ will cause the head to stay in place for three steps and then move to the right one step. Only a symbols from $\{0, 1, \#\}$ are written.
\item {\bf Case 2: $N = 4^k$, where $k \ge 2$:}  According to Lemma \ref{lem:correctoutput-inf},
the tape head stays within locations $1$ through $N-3$ which means that the
Turing Machine steps are correctly executed in the Hamiltonian.
If $N$ is not 
sufficient for $M_{BC} + M_{TV-inf}+M_{post}$ then there is a $\myw$ for which $M_{BC} + M_{TV-inf}+M_{post}$ does not reach a final state.
On input $(1,\myw)$, there are no symbols from
$\{ \sigma_A. \sigma_B, \sigma_R \}$ on the tape at the end of the process.
\end{description}
Let $T$ be the number of symbols from the set $\{ \gamma_A, \gamma_R, \gamma_B, \sigma_X\}$
on the work tape  minus $1$ at the end of the process.
The ground state for this Hamiltonian will correspond to the cycle in which the clock has length $T$ and the computation
starts in state $(1,\myw)$ with the head in state $q_0$ at the left end of the chain. 
In this case $H_{N, init}$ will be $0$ because the initial configuration is correct.
$H_{N, length}$ will also be $0$ because the timer symbols on Track $3$ will correspond to the symbols from
the set  $\{ \gamma_A, \gamma_R, \gamma_B, \sigma_X\}$ on Track $5$. Furthermore, $H_{N,V}$ and $H_{N,final}$ will be $0$ because
there will never be an $\sigma_A$ or an $\sigma_R$ on the work tape since those symbols are written in the last step of $M_{check-i}$ before it hits a final state. The overall energy for this state will be $0$.
\end{proof}

\begin{lemma}
\label{lem:notsquare}
{\bf [Ground Energy for $N = 4^k$ where $k$ is not a Perfect Square]}
Suppose that $N = 4^k$  is  sufficient
for $M_{BC} + M_{TV-inf}+M_{post}$ and $k$ is not a perfect square.
Then $\lambda_0(4^k)$.
\end{lemma}

\begin{proof}
The ground state for this Hamiltonian will correspond to the cycle in which the clock has length $0$ and the computation
starts in state $(1,\myw)$ with the head in state $q_0$ at the left end of the chain. 
In this case $H_{N, init}$ will be $0$ because the initial configuration is correct.
$H_{N, length}$ will also be $0$ because the tape contents at the end of 
the process (in Segment $4$) will be $(\sigma_B, z, w)$, where $4^k \in \caln(z)$.
This will be consistent with a clock of length $0$.
Furthermore, $H_{N,V}$ and $H_{N,final}$ will be $0$ because
neither $\sigma_A$ nor $\sigma_R$ are on the work tape at the end of the process. The overall energy for this state will be $0$.
\end{proof}

\begin{definition}
{\bf [Set of $k$ where $4^k$ is not Sufficient]}
Let $\calk$ denote the set of $k$ which are not sufficient for $M_{check-i}$.
Lemma \ref{lem:existsMTVinf} shows that $M_{TV-inf}$ can be chosen so that $4^k$ is
sufficient for $M_{BC}+M_{TV-inf}+M_{post}$  for all but a finite set of $k$. We will fix a $M_{TV-inf}$ that satisfies those properties and let $\calk$ denote the finite set of $k$ such that $4^k$ is not sufficient for $M_{BC}+M_{TV-inf}+M_{post}$
\end{definition}

This final lemma characterizes $\lambda_0(4^k)$ for all $k$.

\begin{proof} {\bf [Proof of Lemma \ref{lem:energyksquare}]}
Lemmas \ref{lem:doesntfinish} and \ref{lem:notsquare} show that for $k \in \calk$ or when $k$ is not a perfect square, then
$\lambda_0(4^k) = 0$.
Finally, we have the case where $N = 4^k$, $k$ is a perfect square ($k = x^2$) and
$4^k$ is sufficient for $M_{BC} + M_{TV-inf}+M_{post}$. Then according to Lemma \ref{lem:correctoutput-inf} after $N-2$ steps of 
$M_{count}$ and $N-2$ steps of $M_{check-i}$ starting in configuration $(1,w)$, the contents of the tape will
be {\sc Out}$(x, w, M, V)$, corresponding to the correct output of $M_{TV-inf}+M_{post}$. Note that this is the same as the output for $M_{TV}$ (for the finite case) on input $x$ and $w$, simulating TMs $M$ and $V$.
The functions for $m$, the number of query calls, and $r$ the length of the witnesses for $V$ are a different function of $|x|$ but those functions are assumed to be hard-coded into $M$ and $V$.
The analysis of the ground energy now is the same as the finite case. Lemmas
\ref{lem:lb-s12}, \ref{lem:Exy}, and \ref{lem:groundenergy} all hold. In particular, Lemma \ref{lem:groundenergy} says that
$\lambda_0(4^k) = E(x,\tilde{y})$, where $\tilde{y}$ is the string denoting
the correct answers to the oracle queries made by $M$ on input $x$ and
$$E(x,y) =  \left( 1 - \cos\left( \frac{ \pi}{L+1} \right)\right),$$
where $L = (2 T(x,\tilde{y})+3) \cdot p(4^k)$.
\end{proof}

\vspace{.1in}



\section{Acknowledgements}

We are grateful to the Simons Institute for the Theory of Computing, at whose program on the ``The Quantum Wave in Computing'' this collaboration began.
We would also like to thank the following people for fruitful conversations regarding this work: Ignacio Cirac, Toby Cubitt, James Watson, Johannes Bausch, Sevag Gharibian, and David Perez-Garcia. We are especially grateful to Ignacio Cirac for highlighting the importance of studying quantum computational complexity in the thermodynamical limit.

\bibliographystyle{alpha}
\bibliography{references}

\end{document}